\documentclass[%
10pt,
twocolumn,
aps,
longbibliography,
showkeys,
 amsmath,
 amssymb,
 pra,
]{revtex4-2}
\usepackage[breaklinks=true]{hyperref}
\hypersetup{
    colorlinks,
    linkcolor={red!50!black},
    citecolor={blue!50!black},
    urlcolor={blue!80!black}
}

\usepackage{graphicx}
\usepackage{dcolumn}
\usepackage{bm}
\usepackage{tikz}
\usetikzlibrary{math,calc}
\usetikzlibrary{patterns}
\usepackage{quantikz}
\usepackage[all]{xy} %Arthur added this
\usepackage{enumerate}
\usepackage{amsmath,amssymb}
\usepackage{mathrsfs}
\usepackage{stackengine}
\usepackage[normalem]{ulem}
\usepackage{graphicx}
\usepackage{amsthm,epsfig}
\usepackage{float}
\usepackage{multirow}
\usepackage{relsize}
\usepackage{forest,adjustbox}
\usepackage{dsfont}%this is for mathbb{1} but you write mathds{1} 
\usepackage{soul}
\usepackage{calligra}%for scripty r
\usepackage{boldline}

\numberwithin{equation}{section}

\theoremstyle{plain}
\newtheorem{theorem}[equation]{Theorem}
\newtheorem{lemma}[equation]{Lemma}
\newtheorem{proposition}[equation]{Proposition}
\newtheorem{corollary}[equation]{Corollary}

\theoremstyle{definition}
\newtheorem{definition}[equation]{Definition}
\newtheorem{remark}[equation]{Remark}
\newtheorem{example}[equation]{Example}

\newcommand{\matr}{\mathbb{M}}%matrix algebra notation -- change as desired
\newcommand{\id}{\operatorname{id}}
\newcommand{\Tr}{\operatorname{Tr}}%trace
\newcommand{\dis}{\operatorname{dis}}
\DeclareMathOperator{\Ad}{\mathrm{Ad}} %Adjoint action operator 

\DeclareMathAlphabet{\mathcalligra}{T1}{calligra}{m}{n}
\DeclareFontShape{T1}{calligra}{m}{n}{<->s*[2.2]callig15}{}

\newcommand{\states}{\Alg{S}}
\newcommand{\unitary}{\mathrm{U}}

\DeclareFontFamily{OT1}{pzc}{}
\DeclareFontShape{OT1}{pzc}{m}{it}{ <-> s*[1.2] pzcmi7t }{}
\DeclareMathAlphabet{\mathpzc}{OT1}{pzc}{m}{it}
\newcommand{\Alg}[1]{\mathpzc{#1}}
\def\R{{{\mathbb R}}}
\def\C{{{\mathbb C}}}

\def\CP{{{\mathbb C}{\mathbb P}}}

\def\N{{{\mathbb N}}}
\def\Z{{{\mathbb Z}}}

\newcommand\define[1]{\emph{\textbf{#1}}}%italicize and bold-face %this seems like a good alternative

\DeclareMathOperator{\myargmin}{argmin}
\newcommand{\argmin}{\mathop{\myargmin}}

\begin{document}

\title{Quantum encodings that preserve persistent homology}

\author{Arthur J.~Parzygnat}
\email{arthurjp@mit.edu}
\affiliation{Experimental Study Group, Massachusetts Institute of Technology, Cambridge, Massachusetts 02139, USA}

\author{Andrew Vlasic}
\email{avlasic@deloitte.com}
\affiliation{
Deloitte Consulting LLP,
Tampa, Florida 33602, USA
}

\date{August 18, 2026}

\begin{abstract}
Given a data set with a notion of distance, such as a point cloud in Euclidean space, topological data analysis (TDA) uses techniques from algebraic topology and metric geometry to infer the topology of a hypothetical manifold from which the data are sampled. This inference is achieved by calculating topological invariants, some of which are difficult to compute classically. Meanwhile, quantum TDA utilizes quantum processes to extract the invariants used in making such inferences in an attempt to speed up the computations. Because applying transformations to the original classical dataset could alter the associated topological invariants, we investigate which quantum encodings would best preserve the invariants of the original dataset. This line of inquiry is distinct from standard approaches in quantum TDA, whose typical starting point is not from the classical dataset directly, but rather from the associated combinatorial objects, such as simplicial complexes, which typically demand a lot of resources to construct. We take the first step at a more direct approach by focusing on which quantum encodings acting directly on the data are admissible for applying quantum algorithms to extract topological features from classical datasets. 
\end{abstract}


\maketitle

\tableofcontents

\section{\label{sec:intro}Introduction}

Given a collection of data points, known as a \emph{point cloud}, such as in Figure~\ref{fig:1dmanifoldsTDA}, are the data points sampled from a probability distribution that is concentrated near a geometric object~\cite{ChazalMichelTDA21,NiSmWe08,ChCSMe11,BCOS16,GorbanTyukin2018}? 
One class of such geometric objects is the collection of smooth submanifolds of Euclidean space~\cite{SpivakCalc65,Munkres91,lee2010introduction,lee2013smooth,MilnorTopDiff1997}. This arises in problems often enough that the assumption that that the data lie along or near a submanifold (occasionally of a dimension significantly lower than the dimension of the ambient space) of the ambient space is referred to as the \emph{manifold hypothesis}~\cite{FeMiNa2016} (the manifold hypothesis can be interpreted more generally beyond ordinary smooth manifolds, such as unions of intersecting manifolds and changes in local dimension~\cite{JonesDG2024}).
Although not all described by manifolds, the geometric objects in Figure~\ref{fig:1dmanifoldsTDA} have features that are robust to small continuous deformations. Such features are known as topological invariants~\cite{hatcherbook2002,Munkres84,MayAlgTop99}, and they help to narrow down the geometric objects that could describe the data. Topological Data Analysis (TDA) is a subject that attempts to extract such topological invariants from data in order to infer the geometric structure of the data~\cite{wasserman2018topological,CarlssonTDA,oudot2017persistence,carlsson2021topological,ghrist2008barcodes}. An example is clustering, which corresponds roughly to the connected components of such a geometric object~\cite{Carlsson2010Clustering,memoli2011metric}. 

Besides clustering, TDA also extracts out latent features in data that describe higher levels of connectivity. It has been successfully applied to machine learning~\cite{hensel2021survey}, molecular science~\cite{wee2025review}, computational biology~\cite{abousamra2023topology,levenson2024advancing,torras2025topology}, and neuroscience~\cite{gardner2022toroidal,SchonsheckGiusti2025,Yoonetal2024,Guardamagna2026toroidal}, to name a few, occasionally surpassing other standard methods. Although graphs have been used as the standard model for analyzing the topology of networks, simplicial complexes and hypergraphs have been used more recently to capture more connectivity information in a network such as complex dynamics and higher-order interactions (such as what is believed to occur in the brain)~\cite{krishnagopal2023topology,millan2024triadic,millan2025topology}. As a simple toy example, if a network in the form of a graph has as its nodes authors, and edges representing that those two authors wrote a paper together, then a hypergraph could be used to represent teams of collaborators on a paper~\cite{Bianconi2021} (a simplicial complex, however, would not be a suitable model since if there are $k$ authors on a joint paper, that does not imply that every proper subset of those authors share a joint paper~\cite{su2025topological}). 

The extraction of topological invariants from data is typically done by implementing certain classical algorithms, along with several computational tools that can do this quite effectively~\cite{gudhimanual,scikittda2019,giotto-tda}. 
However, some of the classical TDA algorithms are notoriously computationally expensive, especially when computing higher-dimensional invariants, and so they tend to not scale well. For example, persistent homology essentially connects data points not just with edges as in a graph, but also with simplices of varying dimension, and consequently scales exponentially with the number of data points (see Chapter 5 in Ref.~\cite{oudot2017persistence} for more details). However, there are continuous efforts to increase the efficiency and scalability of individual algorithms \cite{simi2025a,ChoietalTDA2025}. 

\begin{figure}
\includegraphics{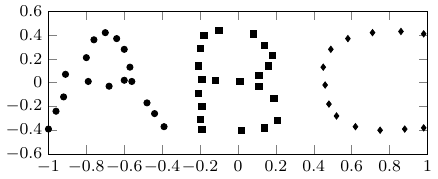}
\caption{A dataset obtained by sampling from probability distributions concentrated near geometric shapes representing the letters A, B, and C. The data point markers are chosen differently to illustrate what happens to this dataset under a quantum encoding later in Figure~\ref{fig:ABCqenc}.}
\label{fig:1dmanifoldsTDA}
\end{figure}

It is therefore natural to ask if there are quantum algorithms that can extract some of these topoligical invariants more efficiently. Indeed, such an idea was initially proposed in Ref.~\cite{LloydGarneroneZanardi16}, which focused on the invariants given by persistent homology~\cite{CarlssonTDA,Virk2022,oudot2017persistence}, one aspect of TDA. The extraction of these invariants can be calculated on a quantum computer in different stages. In Ref.~\cite{LloydGarneroneZanardi16}, the proposal is to first process the classical data by using classical algorithms to calculate the distances between the data points and to construct a simplicial complex, which is a combinatorial object associated with the data. Ref.~\cite{LloydGarneroneZanardi16} then encodes the combinatorial object as a quantum state and a Hermitian operator, 
and then proposes a quantum algorithm (essentially quantum phase estimation using Hamiltonian simulation) on that associated quantum state in order to extract the Betti numbers, which are topological invariants of homology~\cite{hatcherbook2002}. 

However, one can imagine, instead, encoding the classical data directly onto the quantum computer via a \emph{quantum encoding} (also called a \emph{quantum feature map}, \emph{quantum embedding}, \emph{data encoding}, \emph{data loading}, or \emph{state preparation}) and then processing the associated quantum states to calculate these invariants with minimal classical preprocessing. In doing so, it is extremely important to preserve, as much as possible, the structure of the classical data so as to not distort the geometric shapes from which the data are sampled. Indeed, it was observed in Ref.~\cite{vlasic2023qtda} that many commonly used quantum encodings distort the data, preventing their distances from being preserved. Figure~\ref{fig:ABCqenc} illustrates how the dataset in Figure~\ref{fig:1dmanifoldsTDA} looks like under the application of one such quantum encoding. The encoding effectively causes their associated persistent homologies to be noticeably different. In other words, if one is to directly encode the classical data onto a quantum computer by any of the standard encoding schemes, it appears that the persistent homology computed by the quantum computer could be quite different from the persistent homology of the given data. 

\begin{figure}
\includegraphics{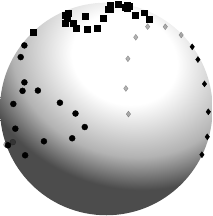}
\caption{The dataset from Figure~\ref{fig:1dmanifoldsTDA}, with data vectors in the form $x=(x_1,x_2)$, is mapped to the Bloch sphere by dense angle encoding $\ket{x}:=\cos\left(\frac{x_1}{2}\right)\ket{0}+e^{2ix_2}\sin\left(\frac{x_1}{2}\right)\ket{1}$ (the notation is discussed in more detail in Example~\ref{ex:denseangleencoding}, though we note that a slightly different scaling and centering convention is used there). 
The image of the dataset is shown under this mapping. Darker markers indicate that those points lie on the front side, while lighter shaded markers indicate that those points lie on the back side of the Bloch sphere. Upon inspection, it is not nearly as straightforward to identify that the dataset is describing the letters A, B, and C, thus highlighting that even a relatively simple smooth quantum encoding can distort a dataset.}
\label{fig:ABCqenc}
\end{figure}

In line with the perspective of structure-preserving encodings proposed in Ref.~\cite{PBVP24}, we should then search for quantum encodings that would preserve the geometric data as much as possible in order to obtain persistent homologies that are consistent with those that would be obtained from the original classical data via classical methods. More conservatively, in this paper, we propose conditions that should be satisfied by such an encoding in order to preserve the persistent homology as much as possible. Our proposal utilizes the stability theorems of persistent homology~\cite{Chazal2009Stable,CSEHstability2007,gardner2017stability,CCGGO09,Virk2022}, which relate the Gromov--Hausdorff distances between metric spaces to the bottleneck distances between the associated persistence diagrams or barcodes~\cite{ghrist2008barcodes} (these concepts will be reviewed in this paper). By analyzing various metrics on the space of quantum states, we propose some options for isolating such quantum encodings. 
One such proposal involves extending multidimensional scaling, a dimensional reduction procedure in classical distance geometry~\cite{DeLeeuwHeiser1982}, to geometries relevant for quantum systems. 

It is important to address why we aim to encode the raw classical data into a quantum state space directly, rather than first measuring the distances and encoding the combinatorial data from the associated simplicial complex. There are several reasons for this. 

\begin{enumerate}[(1)]
\item
First, it is interesting to consider the possibility of finding alternative quantum algorithms for TDA that bypass the resource-intensive task of constructing and storing the combinatorial data from the simplicial complex~\cite{oudot2017persistence,Berry2024QTDA}. Although speculative, it is conceivable that there could be ways in which one can apply quantum circuits that extract topological invariants from the quantum states more directly, rather than measuring distances on a quantum computer~\cite{Kuzmak2021,vlasic2023qtda}. Although this paper does not propose any such specific quantum algorithms, we provide evidence that any such algorithm should factor through a quantum encoding that preserves the topological features of the underlying geometric space from which the data are sampled. Therefore, in order to be more confident regarding the optimality of current proposed quantum algorithms for TDA~\cite{Berry2024QTDA}, it seems worthwhile to investigate such alternatives. 

\item Second, for applications in Quantum Machine Learning (QML), understanding how a given quantum encoding deforms the perceived topology of the data could yield insight into choosing an encoding scheme and subsequent network layers that avoid or minimize such deformations. This could therefore provide insight towards choosing quantum encodings that leverage structure in data to avoid the barren plateau problem or to avoid exponential concentration in quantum kernel methods~\cite{mcclean2018barren,larocca2024reviewbarrenplateausvariational,thanasilp2024exponential}.

\item 
Additionally, when we arrive at a time for analyzing \emph{quantum} data (rather than just classical data), it seems important to uncover an algorithm that works directly on quantum states without the need to measure the distance between them, since the latter is often a resource-expensive task, except under certain circumstances~\cite{Kuzmak2021,Rethinasamy2023}. 

\item
Finally, when presented with data in $\R^{n}$, the Euclidean distance between the data points might not always be the most appropriate notion of distance to use. Indeed, other notions of distance might more accurately describe the data~\cite{Perea2018Multiscale,ChCSMe11,GrandeSchaub2024,CoifmanLafon2006DM}. As such, constructing one simplicial complex from a single metric on the data and only using that information in a quantum algorithm for TDA may be too restrictive. 
\end{enumerate}

The structure of this paper is as follows.
Section~\ref{sec:quantumencodings} reviews some common examples of quantum encodings and definitions of distances between data points and the corresponding encoded quantum states. 
Section~\ref{sec:tda} reviews some aspects of persistent homology, which can be quickly skimmed to establish notation for those already familiar with TDA. It begins with a review of elements of category theory in Section~\ref{subsec:cats}. Section~\ref{subsec:HS-distance} reviews the Gromov--Hausdorff distance between two finite metric spaces, which mathematically represent point clouds. As a special example, we carefully explain how a quantum encoding generally alters the Gromov--Hausdorff distance. Section~\ref{subsec:persistenthomology} describes how a finite metric space gives rise to a persistent homology, which is a parametrized family of homology vector spaces that attempt to keep track of the metric space's topological features. Section~\ref{subsec:bottleneck-distance} describes the numerical invariants of persistent homology, which are called persistence diagrams, and it reviews the bottleneck distance, which measures the distance between persistence diagrams. Section~\ref{subsec:stability} then reviews some formulations of the Stability Theorems from topological data analysis, where such theorems are used to relate the Gromov--Hausdorff distances to the bottleneck distances. 

Having established the background, Section~\ref{sec:main} formulates the problem of quantum encodings for the preservation of invariants obtained from topological data analysis. Namely, in Section~\ref{sec:MAE}, we propose a quantum encoding that \emph{exactly} preserves the Euclidean distance of a finite point cloud in Euclidean space. However, because this quantum encoding involves preparing \emph{mixed} quantum states, we define an alternative approach that approximately preserves the Euclidean distance and only involves preparing pure quantum states in Section~\ref{sec:QMDS}. The latter proposed quantum encoding utilizes non-Euclidean metric multidimensional scaling to preserve distances in quantum state space. 
Section~\ref{sec:experiments} provides some experimental evidence related to the preservation of the topological invariants of the proposed quantum encodings and also illustrates the failure of other encodings to preserve the persistent homology. 
Section~\ref{sec:discussion} ends the manuscript with a discussion summarizing our findings and poses some open problems.

%%%%%%%%%%%%%%%%%%%%%%%%%%%%%%%%%%%%%%%%%%%%%%%%
\section{Quantum encodings and metric spaces}
\label{sec:quantumencodings}
%%%%%%%%%%%%%%%%%%%%%%%%%%%%%%%%%%%%%%%%%%%%%%%%

A good amount of information can be inferred from a dataset of vectors in Euclidean space, called a \emph{point cloud}. For example, one might want to identify local features associated with a specific data point that indicate a latent characteristic, as in binary classification. Additionally, one might like to infer a global feature, such as if the data points are sampled from some unknown low-dimensional geometric object, such as a manifold. Ideally, such conclusions should be robust to noise and the incorporation of additional data, and such methods are part of the subject of Topological Data Analysis (TDA)~\cite{CarlssonTDA,oudot2017persistence}. 

As for \emph{calculating} the invariants of TDA, instead of applying purely classical algorithms to analyze data, Quantum Machine Learning (QML) aims to provide methods and algorithms by leveraging properties of quantum mechanics to make certain conclusions about data~\cite{Biamonte2017QML,SchuldPetruccione21}. The first step in such a procedure is to encode the data as quantum states or quantum circuits~\cite{Aaronson15,schuld2021effect,lloyd20,KhAmSi24,PBVP24,LaRoseCoyle2020}. 

A natural question to ask is what happens to the information within a dataset as it is transferred to a quantum system. To make this question more precise, it is important to make sense of what we mean by information. Ref.~\cite{PBVP24} put forward the idea that such information is meant to be interpreted as mathematical structure itself, rather than some numerical quantity. In brief, structured spaces and encodings that preserve such structure are objects and morphisms, respectively, in a particular category~\cite{mac2013categories,riehl2016,Perrone2024}, while an encoding that loses the information of such a structure is more appropriately thought of as a morphism in a category that has forgotten that structure, with the two categories being related by a forgetful functor~\cite{PBVP24}. This is important in the context of TDA because if classical data are sampled from a particular manifold, then it should be the case that the same manifold should be inferred from the quantum encoded data. Hence, it is important to isolate the quantum encodings that preserve such structure. The present manuscript aims to address this point. 

To set the stage for making this question more precise, we first establish some notation~\cite{NiCh11,HallQTM13}.
Given a Hilbert space $\mathcal{H}$, let $\states(\mathcal{H})$ denote the set of states, i.e., trace-class positive operators with trace $1$, and let $\unitary(\mathcal{H})$ denote the set of unitary operators on $\mathcal{H}$. States are also called \emph{density matrices} in this work. The notations $\N$, $\Z$, $\R$, and $\C$ are used for natural numbers, integers, real numbers, and complex numbers, respectively. The adjoint of an operator $A$ on a Hilbert space $\mathcal{H}$ is denoted by $A^{\dag}$, and it is implicitly identified with the conjugate transposed matrix when $\mathcal{H}=\C^{n}$ for $n\in\N$.

\begin{definition}\label{def:encoding}
Let $X$ be a set, interpreted as a data domain. 
A \define{quantum state encoding} from $X$ into a Hilbert space $\mathcal{H}$ is a function $\rho:X\to\states(\mathcal{H})$. A \define{quantum unitary encoding} is a function $U:X\to\unitary(\mathcal{H})$. A \define{quantum encoding} refers to either of these. 
\end{definition}

A quantum encoding as defined above is also often called a \emph{quantum feature map}, \emph{quantum embedding}, \emph{data encoding}, \emph{data loading}, or \emph{state preparation}. These terms are sometimes used when emphasizing different aspects of the encoding. For example, \emph{state preparation} usually involves an algorithmic construction of a quantum circuit describing the encoding. Note also that a fixed state $\rho_{0}\in\states(\mathcal{H})$ together with a quantum unitary encoding $U:X\to\unitary(\mathcal{H})$ defines a quantum state encoding $\rho:X\to\states(\mathcal{H})$ via the formula $\rho(x):=U(x)\rho_{0}U(x)^{\dag}$ for all $x\in X$. The unitary $U(x)$ is then thought of as a quantum circuit that prepares $\rho(x)$. In most cases, $\rho(x)$ will correspond to a pure state, i.e., $\rho(x)=\ket{x}\bra{x}$ will be a rank-1 projector with $\ket{x}$ a unit vector.

%%%%%%%%%%%%%%%%%%%%%%%%%%%%%%%%%%%%%%%%%%%%%%%%
\subsection{Examples of quantum encodings}
\label{subsec:feature-maps}
%%%%%%%%%%%%%%%%%%%%%%%%%%%%%%%%%%%%%%%%%%%%%%%%

In this section, we review some common examples of quantum encodings. These examples include angle encoding, dense angle encoding, amplitude encoding, square-root encoding, diagonal encoding, and IQP encoding. There are also many closely related encodings, which include general qubit encoding and generalized amplitude encoding. Refs.~\cite{KhAmSi24,LaRoseCoyle2020,SchuldPetruccione21} provide many examples and contain additional references.
In the examples that follow, the notation $\ket{0}$ and $\ket{1}$ will be used to denote the vectors $(1,0)$ and $(0,1)$ in $\C^{2}$, respectively, unless otherwise stated. 
Additionally, the tensor product of $n$ copies of $\ket{0}$ will be denoted by any of $\ket{0\cdots0}:=\ket{0}\otimes\cdots\otimes\ket{0}=\ket{0}^{\otimes d}$, the first two being used when the number of qubits $d$ is clear from context. The unit sphere in $\R^{d}$ is denoted by $S^{d-1}$. Let $\Delta^{n}\subset\R^{n+1}$ denote the probability $n$-simplex in $\R^{n+1}$, i.e.,  
\begin{equation}
\label{eqn:nsimplex}
\Delta^{n}=\left\{x\in\R^{n+1}\;:\;x_{i}\ge0\;\forall i, \sum_{j=0}^{n}x_{j}=1\right\}.
\end{equation}

\begin{example}[Angle encoding]
\label{ex:angleencoding}
\emph{Angle encoding}, also called \emph{rotation encoding}, is a quantum encoding where the components of a data point $x=(x_1,\dots,x_d)\in\R^{d}$ are set to be the angles of a particular rotation operator $U(x)\in\unitary(\mathcal{H})$, where $\mathcal{H}=\C^{2^d}$~\cite{schuld21,schuld2019quantum}. 
This operator is then made to act on a fiducial state in $\mathcal{H}$, which is typically denoted $\ket{0\cdots0}$. Specifically, if $\sigma_j$ denotes the $j^{\text{th}}$ Pauli operator, with $j\in\{1,2,3\}$, for any unit vector $\vec{n}$ in $\R^3$, let 
\begin{equation}
\vec{n}\cdot\vec{\sigma}=\sum_{k=1}^{3}n_{k}\sigma_{k}
\end{equation}
denote the Pauli operator in the $\vec{n}$ direction. Then, for a collection $\{\vec{n}_{j}\}$ of $d$ such directions, set 
\begin{equation}\label{eq:angle}
     U(x) := \bigotimes_{j=1}^{d} \exp\left( -\frac{i x_j}{2} (\vec{n}_{j}\cdot\vec{\sigma}) \right),
\end{equation} 
so that the $j^{\text{th}}$ component of the data point $x\in\R^{d}$ controls the degree of rotation for the $j^{\text{th}}$ qubit around the $\vec{n}_{j}$ axis. Then, setting $\ket{x}:=U(x)\ket{0\cdots0}$, we obtain the associated quantum encoding $\rho:\R^{d}\to\states(\mathcal{H})$ given by 
\begin{equation}
\rho(x):=\ket{x}\bra{x}\equiv U(x)\ket{0\cdots0}\bra{0\cdots0}U(x)^{\dag}. 
\end{equation}
Note that already with a single qubit, the operator $\exp\left( -i \frac{x}{2} (\vec{n}\cdot\vec{\sigma}) \right)$ has a periodicity over the domain $\R$ with a fundamental domain $[-\pi,\pi)$ (an alternative option is $[0,2\pi)$, but we prefer $[-\pi,\pi)$ since it has $0$ as its center), since 
\begin{equation}
U(x+2\pi)\ket{0}\bra{0}U(x+2\pi)^{\dag}=U(x)\ket{0}\bra{0}U(x)^{\dag}
\end{equation}
(note: it is not true that $U(x+2\pi)=U(x)$, but rather the associated density matrix $\rho(x)$ satisfies this periodicity). 
If the data do not admit such a periodicity, then one may map the data first into $[-\pi,\pi)$ by some suitable function (such as centering and rescaling). As a special case, if $\vec{n}=(0,1,0)$, then this unitary becomes 
\begin{equation}
U(x)=\exp\left(-\frac{ix}{2}\sigma_{2}\right)=\begin{bmatrix}\cos\left(\frac{x}{2}\right)&-\sin\left(\frac{x}{2}\right)\\\sin\left(\frac{x}{2}\right)&\cos\left(\frac{x}{2}\right)
\end{bmatrix}
.
\end{equation}
Hence, the action of this unitary on the state $\ket{0}$ yields
\begin{equation}
\label{eqn:angleencodingstate}
U(x)\ket{0}=\cos\left(\frac{x}{2}\right)\ket{0}+\sin\left(\frac{x}{2}\right)\ket{1}
\end{equation}
with the associated density matrix given by 
\begin{equation}
\label{eqn:angleencodingdensitymatrix}
\rho(x)=\begin{bmatrix}\cos^2\left(\frac{x}{2}\right) & \cos\left(\frac{x}{2}\right)\sin\left(\frac{x}{2}\right)\\ \cos\left(\frac{x}{2}\right)\sin\left(\frac{x}{2}\right) & \sin^2\left(\frac{x}{2}\right)\end{bmatrix}.
\end{equation}
Note that one can recover $x$ from such a state by applying certain measurements. For example, measuring $\ket{1}$ would yield the probability $p:=\sin^2\left(\frac{x}{2}\right)$. From this, one has $\sqrt{p}=\pm \sin\left(\frac{x}{2}\right)$ so that $x=\pm2\arcsin\left(\sqrt{p}\right)\in[-\pi,\pi)$. In the case that $p\ne0$, to identify whether $x$ is positive or negative, we need to perform another measurement. One possibility is to measure $\ket{+}=\frac{1}{\sqrt{2}}\big(\ket{0}+\ket{1}\big)$ which has probability of occurrence given by $q:=\frac{1+\sin(x)}{2}$. Thus, combined with the previous possibilities for $x$, the combination $2q-1$ is positive if and only if $x$ is positive. Equivalently, $x$ is positive if and only if $q>\frac{1}{2}$.
A similar discussion, though possibly using different measurements, applies to the other qubits and for Pauli operators in arbitrary directions. 
\end{example}

\begin{example}[Dense angle encoding]
\label{ex:denseangleencoding}
A modification of angle encoding is \emph{dense angle encoding}, which uses fewer qubits~\cite{LaRoseCoyle2020}. 
This encoding begins with data in $\R^{d}$ and maps it to the state space of $\lceil\frac{d}{2}\rceil$ qubits, where $\lceil z\rceil$ is the least integer greater than $z$. Specifically, the dense angle encoding maps $x=(x_1,\dots,x_d)\in\R^{d}$ to 
\begin{equation}
\ket{x}=\bigotimes_{j=1}^{\lceil\frac{d}{2}\rceil}\Big(c(x_{2j-1})\ket{0}+e^{ix_{2j}}s(x_{2j-1})\ket{1}\Big),
\end{equation}
where $c(\theta):=\cos\left(\frac{\theta+\pi}{4}\right)$ and $s(\theta):=\sin\left(\frac{\theta+\pi}{4}\right)$. Note that if $d$ is odd, one sets $x_{d+1}=0$. The reason for our different convention as compared with Ref.~\cite{LaRoseCoyle2020} is because we will often work with centered data and also because the encoding is non-injective when $\theta=-\pi$ (and also when $\theta=\pi$ after removing the global phase). 
The associated quantum state encoding is then $\rho(x)=\ket{x}\bra{x}$, which is explicitly given by
\begin{equation}
\rho(x)=\bigotimes_{j=1}^{\lceil\frac{d}{2}\rceil}\begin{bmatrix}c^2(x_{2j-1})&e^{-ix_{2j}}cs(x_{2j-1})\\e^{ix_{2j}}cs(x_{2j-1})&s^2(x_{2j-1})\end{bmatrix},
\end{equation}
where $c^2(\theta):=c(\theta)^2$, $s^2(\theta):=s(\theta)^2$, and $cs(\theta):=c(\theta)s(\theta)$.
The periodicity causes classical data in different regions to be placed into the same range, so one must restrict to a suitable domain to guarantee injectivity (such injectivity does not hold in the example illustrated in Figures~\ref{fig:1dmanifoldsTDA} and~\ref{fig:ABCqenc} precisely when $x_2=0$). 
The fundamental domain over which the encoding is injective is $[-\pi,\pi)^{\lceil\frac{d}{2}\rceil}$. 
In addition, one can rescale $x_{2j-1}$ and $x_{2j}$, which would slightly alter what this domain would look like (Ref.~\cite{LaRoseCoyle2020} includes a factor of $4 \pi$, for example). 
Combining these two observations, one can use a large enough domain that covers the data by rescaling the argument and avoid encodings that are not injective on the data domain. 
\end{example}

Regarding both types of angle encoding, one must be cautious of induced edge effects due to periodicity, since data that are originally far apart could be brought closer together in quantum state space. We will address this point later in Sections~\ref{sec:main} and~\ref{sec:experiments} when defining distance-preserving encodings and when we analyze how these encodings alter TDA.

\begin{example}[Amplitude encoding]
\label{ex:ampencoding}
\emph{Amplitude encoding} begins with data on $S^{2^{d}-1}\subset\R^{2^{d}}$ and maps it by the inclusion into $\C^{2^{d}}$~\cite{schuld21}. Mathematically, $x\in S^{2^{d}-1}$ gets sent to the state $\ket{x}=\sum_{i=0}^{2^{d}-1}x_{i}\ket{i}$, where $\ket{i}$ is expressed in binary. In more detail, if $x=(x_0,x_1,\dots,x_{2^d-1})\in S^{2^{d}-1}$, rewrite the indices in binary so that they correspond to a sequence $n$ of $d$ $0$'s and $1$'s. For example, if $d=2$, then $n$ has four possibilities: $n\in\{00,01,10,11\}$. Therefore, the state encoding associated with amplitude encoding is
\begin{equation}
S^{2^{d}-1}\ni x\mapsto \rho(x):=\sum_{i,j=0}^{2^{d}-1}x_{i}\overline{x_{j}}\ket{i}\bra{j}.
\end{equation}
Circuit implementations for this form of state preparation are discussed in Refs.~\cite{ICKHC16,PleschBrukner2011}, where the number of steps grows as $\mathcal{O}(2^{d})$.
One can extend the previous definition to an encoding on all of $\R^{2^{d}}\setminus\{0\}$ by first normalizing the entries so that they lie on the unit sphere. In this case, amplitude encoding is given by 
\begin{equation}
\label{eqn:amplitudeencoding}
\R^{2^{d}}\setminus\{0\}\ni x\mapsto \frac{1}{\lVert x\rVert^2}\sum_{i,j=0}^{2^{d}}x_{i}\overline{x_{j}}\ket{i}\bra{j}.
\end{equation}
Of course, this mapping is not injective and can drastically distort data, as discussed in Ref.~\cite{PBVP24}. Alternative variations on amplitude encoding that do not distort the data as much are described in Refs.~\cite{PBVP24,SchuldPetruccione21}. 
\end{example}

A close variation on amplitude encoding involves starting with data whose entries can be interpreted as probabilities. 

\begin{example}[Square-root encoding]
\label{ex:squarerootencoding}
The \emph{square-root encoding}~\cite{GroverRudolph02,KayeMosca01,araujo2021divide,Herbert2021NoQuantumSpeedup,SoSc06,bradley2020,Iaconis2024} is a quantum encoding whose domain is $\Delta^{n}\subset\R^{n+1}$. 
If a basis of $\mathcal{H}=\C^{n+1}$ is written as $\ket{0},\ket{1},\dots,\ket{n}$, the encoding sends a point $x\in\Delta^{n}$ to the vector 
\begin{equation}
\label{eqn:sqrtencoding}
\ket{x}=\sum_{j=0}^{n}\sqrt{x_{j}}\ket{j}. 
\end{equation}
Therefore, the state encoding associated with the square-root encoding is the function $\rho_{\sqrt{}}:\Delta^{n}\to\states(\mathcal{H})$ sending $x$ to 
\begin{equation}
\rho_{\sqrt{}}(x)=\sum_{i,j=0}^{n}\sqrt{x_{i}x_{j}}\ket{i}\bra{j}. 
\end{equation}
It is worth comparing the square-root encoding with amplitude encoding. If the data begins on a simplex, both encoding schemes map the simplex $\Delta^{n}$ into a sphere $S^{n}$ but to different points. Figure~\ref{fig:ampversussqrt} illustrates how the two encodings differ in the case $n=1$. 
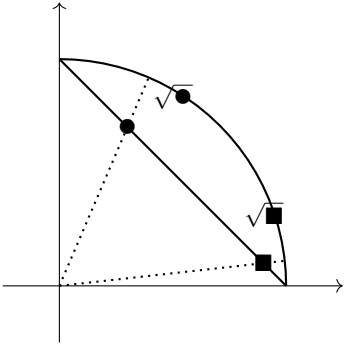
\begin{figure}
\begin{tikzpicture}[scale=3]
\def\px{0.3};
\def\py{0.9};
\draw[->] (-0.25,0) -- (1.25,0);
\draw[->] (0,-0.25) -- (0,1.25);
\draw[thick] (0,1) -- (1,0);
\draw[thick,domain=0:{pi/2}] plot ({cos(deg(\x))},{sin(deg(\x))});
\draw[thick,dotted] (0,0) -- ({\px},{1-\px}) -- ({\px/sqrt((\px)^2+(1-\px)^2)},{(1-\px)/sqrt((\px)^2+(1-\px)^2)});
\node at (\px,{1-\px}) {\Large $\bullet$};
\node[xshift=-4pt] at ({sqrt(\px)},{sqrt(1-\px)}) {$\sqrt{\text{\Large$\bullet$}}$};
\draw[thick,dotted] (0,0) -- ({\py},{1-\py}) -- ({\py/sqrt((\py)^2+(1-\py)^2)},{(1-\py)/sqrt((\py)^2+(1-\py)^2)});
\node at (\py,{1-\py}) {$\blacksquare$};
\node[xshift=-4pt] at ({sqrt(\py)},{sqrt(1-\py)}) {$\sqrt{\blacksquare}$};
\end{tikzpicture}
\caption{A visualization of the difference between amplitude encoding and square-root encoding as applied to the 1-simplex $\Delta^{1}\subset\R^{2}$. 
Two points, {\Large $\bullet$} and $\blacksquare$, are shown on the simplex.
The dotted lines are straight lines passing through the origin and these two points, respectively. The intersection of these dotted lines with the unit circle correspond to the points $\frac{\text{\Large$\bullet$}}{\lVert\text{\Large$\bullet$}\rVert}$ and $\frac{\blacksquare}{\lVert\blacksquare\rVert}$, which are not drawn. Instead, what are drawn are the square-root points $\sqrt{\text{\Large$\bullet$}}$ and $\sqrt{\blacksquare}$. The only points for which the amplitude and square-root encodings agree are for the points $(1,0)$, $(0,1)$, and $\frac{1}{2}(1,1)$.}
\label{fig:ampversussqrt}
\end{figure}

A circuit implementation of the square-root encoding when $n=2^{d}$ is given by the \emph{Grover--Rudolph method}, which encodes a probability vector with $2^{d}$ values into a state of a quantum system with $d$ qubits~\cite{GroverRudolph02,KayeMosca01}. The number of steps used for such state preparation goes as $\mathcal{O}(d)$ due to its iterative (``divide and conquer'') procedure. As a simple example, if $d=2$, begin with a data point $x=(x_0,x_1,x_2,x_3)\in\Delta^{3}$. The circuit used to prepare the state $\sqrt{x_{0}}\ket{00}+\sqrt{x_{1}}\ket{10}+\sqrt{x_{2}}\ket{01}+\sqrt{x_{3}}\ket{11}$ is given as follows. First, find the binary tree decomposition of the point $x$ as in Figure~\ref{fig:amp_enc_tree_general}. 
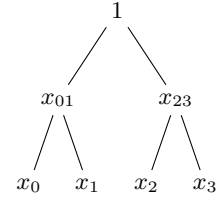
\begin{figure}
    \begin{forest}
      for tree={l+=.25cm} % increase level distance
      [$1$
        [$x_{01}$[$x_0$][ $x_1$ ]]
        [$x_{23}$[$x_2$][$x_3$]]
      ]
    \end{forest} 
\caption{\label{fig:amp_enc_tree_general} This is the binary tree decomposition of the data point $(x_0,x_1,x_2,x_3)\in\Delta^{3}$. Here, $x_{01}=x_0+x_1$ and $x_{23}=x_2+x_3$~\cite{araujo2021divide,de2024empirical}.}
\end{figure}
The quantum circuit then involves using this binary tree and some rotation gates, which are shown in Figure~\ref{fig:amp_enc_circ_general}. Note that the convention for the tensor product in our circuits is from top to bottom and the evolution is from left to right. 
\begin{figure} %[!htbp]
    \begin{quantikz}[thin lines] 
       \\ \lstick{$\ket{0}$} & \qw & \gate{ \Tilde{R}_y\Big(\sqrt{ \frac{x_0}{x_{01}}}\Big)}&  \gate{\Tilde{R}_y\Big( \sqrt{ \frac{x_2}{x_{23}} }\Big)} & \qw \\
       \lstick{$\ket{0}$} & \gate{\Tilde{R}_y\Big(\sqrt{x_{01}} \Big)}  & \octrl{-1} &\ctrl{-1}& \qw
    \end{quantikz} 
\caption{This circuit depicts the state preparation protocol for the square-root encoding in the case where the data lie on $\Delta^{3}\subset\R^{4}$~\cite{araujo2021divide,de2024empirical}. The rotation gates in this circuit are defined by $\Tilde{R}_y(\alpha) = R_y\big( 2 \arccos(\alpha) \big)$ (see Figure \ref{fig:amp_enc_tree_general} and the body for the remaining notation). The resulting state associated with this circuit is $\sqrt{x_{0}}\ket{00}+\sqrt{x_{1}}\ket{10}+\sqrt{x_{2}}\ket{01}+\sqrt{x_{3}}\ket{11}$.}
\label{fig:amp_enc_circ_general}
\end{figure}
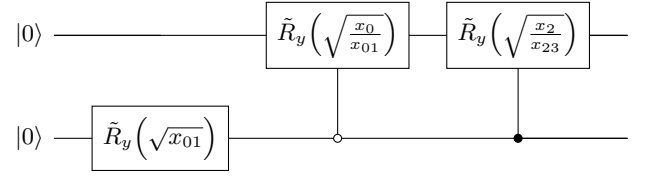
\end{example}

We note that both the amplitude encoding and square-root encoding can be viewed as examples of \emph{generalized amplitude encodings} in the sense of Ref.~\cite{LaRoseCoyle2020} (however, we do not use this terminology to avoid confusion with our previous definitions of amplitude and square-root encodings). 

\begin{example}[Diagonal encoding]
\label{ex:diagencoding}
The \emph{diagonal encoding} is a quantum encoding whose domain is $\Delta^{n}$. Just as for the square-root encoding, if a basis of $\mathcal{H}=\C^{n+1}$ is written as $\ket{0},\ket{1},\dots,\ket{n}$, the diagonal encoding sends a point $x=(x_1,x_2,\dots,x_n)\in\Delta^{n}$ to the density matrix (mixed state) 
\begin{equation}
\rho_{\mathrm{diag}}(x)=\sum_{j=0}^{n}x_{j}\ket{j}\bra{j}.
\end{equation}
Notice that this encoding lands in the space of mixed states, specifically those that are diagonal. Since the efficient preparation of mixed states is part of ongoing research~\cite{Sagastizabal2021,BLMT2024,Rall2023thermalstate,ChKaGi2023,Chen2023QTSP,Ding2025efficient,Rouze2025Optimal,ScandiAlhambra2026}, we will mainly use the diagonal encoding as a tool for comparisons to other encodings. 
\end{example}

\begin{example}[IQP encoding]
\label{ex:IQPencoding}
\emph{Instantaneous quantum polytime} (IQP) encoding is built off of Hadamard gates acting on $\ket{0}^{\otimes d}$, where $d$ denotes the number of qubits, followed by quantum gates diagonal in the Pauli $\sigma_{z}$ basis, followed by Hadamard gates~\cite{Shepherd_2009,bremner2016average,havlivcek2019supervised}. The terminology comes from the fact that Pauli $\sigma_{z}$ gates all commute among themselves, so that the unitary evolutions can be interpreted as happening instantaneously (with no requirement of a temporal order). Despite the commuting property of such diagonal unitary matrices, it was argued in Ref.~\cite{Shepherd_2009} that such encodings could still be used to achieve quantum advantage in certain tasks. However, Ref.~\cite{Rajakumar2025IQP} argued that such an advantage is no longer available when sufficient noise occurs and once a certain depth in the quantum circuit is achieved (i.e., after a sufficient number of layers). Nevertheless, the encoding is still of interest, and so we review its basic components. 

There are many options for which unitary operators to use between the two sets of Hadamard gates, and one concrete proposal was put forth in Refs.~\cite{havlivcek2019supervised,bremner2016average}, which goes as follows. Given a data domain of the form $\R^{d}$, and given a function $\phi_{S}:\R^{d}\to\R$ for every subset $S\subseteq\{1,\dots,d\}$ (examples will be given momentarily), set 
\begin{equation}
U_{\Phi}(x) := \exp\left( i \sum_{S \subseteq \{1,2,\ldots,d\}} \phi_{S}(x) \prod_{j\in S}Z_j \right), 
\end{equation}
where $Z_j$ denotes the Pauli $\sigma_{z}$ gate acting on the $j^{th}$ qubit, and $\Phi$ is used to denote the collection $\{\phi_{S}\}$ with $S$ a subset of $\{1,\dots,d\}$. Note that $U_{\Phi}(x)$ is a unitary operator acting on $\mathcal{H}:=(\C^{2})^{\otimes d}$. The \emph{IQP encoding} associated with such a family of functions $\Phi:=\{\phi_{S}\}$ is the unitary encoding $U:\R^{d}\to\unitary(\mathcal{H})$ of the form 
\begin{equation}\label{eq:iqp-full}
U(x) :=  H^{\otimes d} U_{\Phi}(x) H^{\otimes d}, 
\end{equation}
where $H$ is the Hadamard gate. This induces a quantum state encoding $\rho(x):=\ket{x}\bra{x}$, where 
\begin{equation}
\ket{x}:=H^{\otimes d} U_{\Phi}(x) H^{\otimes d}\ket{0}^{\otimes d}.
\end{equation}

As a particular example, the one that was used in Ref.~\cite{havlivcek2019supervised} and the one we also implement here, set $\phi_{\varnothing}=0$ and $\Phi_{S}=0$ for all subsets $S\subseteq\{1,\dots,d\}$ that have cardinality $\# S\ge 3$. Then, set 
\begin{equation}
\phi_{\{j\}}(x):=x_{j}
\end{equation}
for all $j\in\{1,\dots,d\}$ and 
\begin{equation}
\phi_{\{j,k\}}(x):=(\pi-x_{j})(\pi-x_{k})
\end{equation}
for all distinct $i,j\in\{1,\dots,d\}$.  
In other words, for such a family $\Phi$ of functions, given a data point $x=(x_1,\dots,x_d)\in\R^{d}$, we have
\begin{equation}
\label{eq:iqp-second}
U_{\Phi}(x) = \exp\left( i\left( \sum_{j=1}^d x_j Z_j + \sum_{\substack{k,l=1\\k<l}}^d \phi_{\{k,l\}}(x) Z_k Z_l \right)\right).
\end{equation} 
\end{example}

Many of the above-mentioned encoding techniques are designed for specific problems or problem classes. Rather than analyzing the efficacy of such encodings via commonly-used measures such as entanglement, expressibility, expressivity, and more~\cite{huang2021power,de2024empirical,sim2019expressibility,goto2020universal,schuld2021effect,schatzki2021entangled}, we will focus on preserving distances and the topological structure of data as defined by persistent homology. Therefore, it is important to discuss different notions of distance when data are viewed classically and after they have been encoded onto a quantum system. This is the subject we turn to next.

%%%%%%%%%%%%%%%%%%%%%%%%%%%%%%%%%%%%%%%%%%%%%%%%
\subsection{Distances between encoded data points}
\label{sec:metricspaces}
%%%%%%%%%%%%%%%%%%%%%%%%%%%%%%%%%%%%%%%%%%%%%%%%

A point cloud $X$ inside Euclidean space acquires some structure from Euclidean space. For example, the distance between two points in $X$ can be defined using the Euclidean metric. This endows $X$ with the structure of a finite metric space $(X,d_{X})$. Let us recall what this means~\cite{Bryant85}. 

\begin{definition}
A \define{metric space} $(X,d_{X})$ is a set $X$ together with a function $d_{X}:X\times X\to[0,\infty)$ satisfying the following conditions:
\begin{enumerate}[i.]
\item $d_{X}(x_1,x_2)=0$ if and only if $x_1=x_2$ 
\item $d_{X}(x_1,x_2)=d_{X}(x_2,x_1)$ for all $x_1,x_2\in X$
\item $d_{X}(x_{1},x_{2})\le d_{X}(x_{1},x_{3})+d_{X}(x_2,x_3)$ for all $x_1,x_2,x_3\in X$ (the \emph{triangle inequality}). 
\end{enumerate}
Such a function $d_{X}$ is called a \define{distance function}, \define{metric}, or \define{distance metric} on $X$. 
\end{definition}

Although data may be presented to us in $\R^{d}$, it is not always optimal to use the Euclidean distance to describe the similarity between the data points. As such, several examples of metric spaces will appear throughout this work, so we briefly review some relevant examples.

\begin{example}
Let $p\in[1,\infty)$ and $n\in\N$. On $\R^{n}$, the \define{$\ell^{p}$ metric} $d_{\ell^{p}}$ (sometimes called the \define{Minkowski distance} of order $p$) is defined by the \define{$\ell^{p}$-norm} of the difference of two vectors, i.e., 
\begin{equation}
d_{\ell^{p}}(x,y):=\lVert x-y\rVert_{p}:=\left(\sum_{i=1}^{n}(x_{i}-y_{i})^{p}\right)^{\frac{1}{p}}
\end{equation}
for all $x,y\in\R^{n}$. The \define{$\ell^{\infty}$ metric} $d_{\ell^{\infty}}$ is defined by 
\begin{equation}
d_{\ell^{\infty}}(x,y):=\lVert x-y\rVert_{\infty}:=\sup_{i}\big\{|x_{i}-y_{i}|\big\}.
\end{equation}
The \define{Euclidean distance} is $d_{\ell^{2}}$. 
These are examples of metrics associated with norms on vector spaces. 
\end{example}

\begin{example}
Let $\Delta^{n}\subset\R^{n+1}$ be the $n$-simplex from~\eqref{eqn:nsimplex}. There are many notions of distance on $\Delta^{n}$, since the set $\Delta^{n}$ is used to describe probability distributions on $n+1$ elements~\cite{amari2016information}. Since $\Delta^{n}$ is a subset of $\R^{n+1}$, one can restrict any of the $\ell^{p}$ metrics to $\Delta^{n}$. This provides one family of examples. As another example, the \define{Hellinger distance} between $x=(x_0,x_1,\dots,x_n)$ and $y=(y_0,y_1,\dots,y_n)$ in $\Delta^{n}$ is 
\begin{align}
d_{H}(x,y)&:=\left(\frac{1}{2}\sum_{j=0}^{n}\big(\sqrt{x_j}-\sqrt{y_j}\big)^2\right)^{\frac{1}{2}} \nonumber \\
&=\sqrt{1-B(x,y)},
\end{align}
where 
\begin{equation}
B(x,y):=\sum_{j=0}^{n}\sqrt{x_j y_j}
\end{equation}
is the \define{Bhatacharyya coefficient} (sometimes called the \define{affinity})~\cite{Bhattacharyya1946,FuvdG99,Borovyk2025,LuoZhang2004} (note that our convention for the Hellinger distance is such that the maximum distance is $1$ rather than $\sqrt{2}$). The square of the Bhattacharyya coefficient is called the \define{fidelity}. 
\end{example}

\begin{remark}
The Hellinger distance is not equal to any of the Minkowski distances when restricted to points in $\Delta^{n}$. However, if $x,y\in\Delta^{n}$, then 
$d_{H}(x,y)=\frac{1}{\sqrt{2}}d_{\ell^{2}}(\sqrt{x},\sqrt{y})$, 
where the notation $\sqrt{x}$ means $\sqrt{x}:=(\sqrt{x}_{0},\dots,\sqrt{x}_{n})$ for $x=(x_0,\dots,x_n)$ and similarly for $y$. 
\end{remark}

\begin{example}
\label{defn:quantumdistances}
Fix a finite-dimensional Hilbert space $\mathcal{H}$ and let $\rho$ and $\sigma$ be two density matrices on $\mathcal{H}$ (more generally, $\rho$ and $\sigma$ could be two operators). There are many notions of distance between $\rho$ and $\sigma$~\cite{PetzSudar1996,FuvdG99,Borovyk2025,DaLuHa2011,LuoZhang2004}. In what follows, we only mention a few. 
The \define{nuclear/trace/$S_{1}$ distance} between $\rho$ and $\sigma$ is 
\begin{equation}\label{eq:trace-distance}
d_{\Tr}(\rho,\sigma):=
\lVert\rho-\sigma\rVert_{S_{1}}\equiv
\Tr\left[\sqrt{(\rho-\sigma)^{\dag}(\rho-\sigma)}\right].
\end{equation}
The \define{Hilbert--Schmidt/Frobenius/$S_{2}$ distance} is
\begin{equation}\label{eq:HS-distance}
d_{\mathrm{HS}}(\rho,\sigma):=\lVert\rho-\sigma\rVert_{S_{2}}:=\sqrt{\Tr\left[(\rho-\sigma)^{\dag}(\rho-\sigma)\right]}.
\end{equation}
More generally, for $p\in[1,\infty)$, the \define{Schatten $p$ metric} 
is given by 
\begin{align}
d_{S_{p}}(\rho,\sigma)&:=\lVert\rho-\sigma\rVert_{S_{p}}\nonumber \\
&:=\left(\Tr\left[\left(\sqrt{(\rho-\sigma)^{\dag}(\rho-\sigma)}\right)^{p}\right]\right)^{\frac{1}{p}}.
\end{align}
The \define{operator/spectral norm distance} (the Schatten $\infty$ distance) is given by 
\begin{equation}
d_{S_{\infty}}(\rho,\sigma):=\lVert\rho-\sigma\rVert_{S_{\infty}}:=\sigma_{1}(\rho-\sigma), 
\end{equation}
where $\sigma_{1}(\rho-\sigma)$ denotes the largest singular value of $\rho-\sigma$. 
The \define{Bures fidelity distance} (also sometimes called the \define{Helstrom distance}) is 
\begin{equation}\label{eq:bures-distance}
d_{F}(\rho,\sigma):=\sqrt{1-\sqrt{F(\rho,\sigma)}}, 
\end{equation}
where
\begin{equation}
F(\rho,\sigma):=\left(\Tr\left[\sqrt{\sqrt{\rho}\sigma\sqrt{\rho}}\right]\right)^{2}
\end{equation}
denotes the \define{fidelity} between the states $\rho,\sigma\in\states(\mathcal{H})$ (note that our convention is such that the maximum distance is $1$ rather than $\sqrt{2}$). 
Meanwhile, the (quantum) \define{Hellinger distance} is 
\begin{equation}
d_{H}(\rho,\sigma)=\frac{1}{\sqrt{2}}\sqrt{\Tr[\sqrt{\rho}-\sqrt{\sigma}]}
=\sqrt{1-B(\rho,\sigma)},
\end{equation}
where 
\begin{equation}
B(\rho,\sigma):=\Tr[\sqrt{\rho}\sqrt{\sigma}]
\end{equation}
denotes the (quantum) \define{Bhattacharyya coefficient} (also called \define{affinity}~\cite{Borovyk2025,kholevo1972quasiequivalence}). 
The \define{Bures angle distance}~\cite{BengtssonZyczkowski2017}, also called the \emph{quantum Fisher information metric geodesic distance}~\cite{PiresQSL2016}, is 
\begin{equation}
d_{B\measuredangle}(\rho,\sigma)=\frac{2}{\pi}\mathrm{arccos}\left(\sqrt{F(\rho,\sigma)}\right).
\end{equation}
The \define{Wigner--Yanase distance}~\cite{GibiliscoIsola2003} is 
\begin{equation}
d_{WY}(\rho,\sigma)=\frac{2}{\pi}\mathrm{arccos}\big(B(\rho,\sigma)\big).
\end{equation}
The $\frac{2}{\pi}$ in the previous two distance metrics is to make it so that the maximum distance between quantum states is $1$.
All of the above distance functions are true metrics on the set $\states(\mathcal{H})$ of quantum states. They all have different geometric meaning. For example, the Bures fidelity distance measures the length of the shortest geodesic connecting $\rho$ and $\sigma$ within the cone of positive matrices, while the Bures angle distance measures the length of the shortest geodesic connecting $\rho$ and $\sigma$ within the space of density matrices~\cite{BengtssonZyczkowski2017}. 
\end{example}

\begin{remark}
Some of the distances introduced in Example~\ref{defn:quantumdistances} are more closely related to the fidelity when the inputs are restricted to the subsets of pure states. More specifically, if $\rho=\ket{\psi}\bra{\psi}$ and $\sigma=\ket{\phi}\bra{\phi}$ are rank-1 projectors, then 
\begin{equation}
F(\rho,\sigma)=\big|\langle\psi|\phi\rangle\big|^2
=B(\rho,\sigma).
\end{equation}
Therefore, 
\begin{equation}
\label{eqn:Trpure}
d_{\Tr}(\rho,\sigma)=2\sqrt{1-F(\rho,\sigma)}
\end{equation}
and
\begin{equation}
\label{eqn:HSpure}
d_{\mathrm{HS}}(\rho,\sigma)=\sqrt{2(1-F(\rho,\sigma))}
\end{equation}
for all pure states $\rho,\sigma$~%
%begin footnote 
\footnote{Identity~\eqref{eqn:HSpure} follows from hermiticity of the projection operators and  $\Tr[(|\psi\rangle\langle\psi|-|\phi\rangle\langle\phi|)^2]=2-2|\langle\psi|\phi\rangle|^2$. Identity~\eqref{eqn:Trpure} can be shown with matrix methods, for example, as follows. Since $\ket{\psi}$ and $\ket{\phi}$ span at most a two-dimensional space, it suffices to represent $\ket{\psi}$ by $\ket{0}$ and $\ket{\phi}$ by $\alpha\ket{0}+\beta\ket{1}$, where $|\alpha|^2+|\beta|^2=1$, since the fidelity is invariant under a unitary change in basis. Explicitly calculating $\sqrt{(\rho-\sigma)^{\dag}(\rho-\sigma)}$ yields the $2\times2$ matrix $\left(\begin{smallmatrix}\sqrt{|1-|\alpha|^2}&0\\0&|\beta|\end{smallmatrix}\right)$ in this representation. Taking the trace yields $2|\beta|$ since $1-|\alpha|^2=|\beta|^2$. Identity~\eqref{eqn:Trpure} then follows from comparing terms.}.
%end footnote
This shows that the Hilbert--Schmidt metric is proportional to the trace metric. The Bures fidelity metric is closely related, but not proportional to either of these because of the square root of fidelity appearing in its definition. 
The restriction of the Bures angle distance to the set of pure states is called the \define{Fubini--Study distance}, which is given by  
\begin{equation}
\label{eqn:FubiniStudy}
d_{FS}\big(\ket{\psi},\ket{\phi}\big)=\frac{2}{\pi}\mathrm{arccos}\Big(\big|\langle\psi|\phi\rangle\big|\Big)
\end{equation}
between pure states $\ket{\psi}$ and $\ket{\phi}$. 
\end{remark}

\begin{remark}
\label{rmk:ellpSchattenp}
At first glance, it might seem as if the Schatten $p$ metric on $n\times n$ matrices might coincide with the $\ell^{p}$ metric when those matrices are viewed as $n^{2}$-component vectors. This is true when restricted to diagonal $n\times n$ density matrices, and it is also true for arbitrary density matrices for $p=2$. The case for $p=2$ follows from the fact that 
\begin{equation}
\lVert A\rVert_{S_{2}}
=\sqrt{\Tr[A^{\dag}A]}
=\sqrt{\sum_{i,j}|A_{ij}|^2}
=\lVert A\rVert_{2}.
\end{equation}
However, more generally, the Schattan $p$ norm of $A$ coincides with the $\ell^{p}$ norm of the \emph{singular values} of $A$, i.e., $\lVert A\rVert_{S_{p}}=\lVert \Sigma\rVert_{p}$, where $A=U\Sigma V^{\dag}$ is a singular value decomposition of $A$. We compare these norms here in a table for completeness, because we will occasionally use both $S_{p}$ and $\ell^{p}$ distances for matrices, and it is important to note the subtle differences between them. 
\begin{center}
\begin{tabular}{c|c}
$S_{p}$ norms for matrices & $\ell^{p}$ norms for matrices \\
\hline
$\lVert A\rVert_{S_{\infty}}=\sup_{i}\sigma_{i}(A)$ & $\sup_{i,j}|A_{ij}|=\lVert A\rVert_{\infty}$\\
$\lVert A\rVert_{S_{2}}=\sqrt{\sum_{i}\sigma_{i}(A)}$ & $\sqrt{\sum_{i,j}|A_{ij}|^2}=\lVert A\rVert_{2}$\\
$\lVert A\rVert_{S_{1}}=\sum_{i}\sigma_{i}(A)$ & $\sum_{i,j}|A_{ij}|=\lVert A\rVert_{1}$
\end{tabular}
\end{center}
By convention, $\sup_{i}\sigma_{i}(A)=\sigma_{1}(A)$, since singular values are ordered from largest to smallest, but the form given in the above table is meant to draw a clearer comparison to the $\ell^{\infty}$ norm. One of the important differences between the $\ell^{p}$ norms and the $S_{p}$ norms is that the $S_{p}$ norms are left and right unitarily invariant, meaning that if $W_{1}$ and $W_{2}$ are unitary matrices, then $\lVert W_{1} A W_{2}\rVert_{S_{p}}=\lVert A\rVert_{S_{p}}$. This property is not generally true for the $\ell_{p}$ norms of matrices (except for $p=2$). 
\end{remark}

When we map classical data onto a quantum system by some quantum encoding $\rho:X\to\states(\mathcal{H})$, if we endow the quantum state space with a metric $d$, then as long as no two data points get mapped to the same quantum state, this induces a metric on $X$. 

\begin{definition}
Let $(X,d_{X})$ and $(Y,d_{Y})$ be two metric spaces. An \define{embedding} from $(X,d_{X})$ to $(\mathcal{Y},d_\mathcal{Y})$ is a function $f:X\to Y$ that satisfies $d_{Y}\big(f(x_{1}),f(x_{2})\big)=d_{X}(x_{1},x_{2})$ for all $x_1,x_2\in X$. In such a case, $f$ is also said to be \define{distance-preserving}. If $f$ is also surjective, then $f$ is said to be \define{isometry}. 
\end{definition}

\begin{example}
\label{ex:diagonalencoding}
The diagonal encoding $\rho_{\mathrm{diag}}:\Delta^{n}\to\states(\C^{n+1})$, as in Example~\ref{ex:diagencoding}, defines an embedding $(\Delta^{n},d_{\ell^{p}})\to(\states(\C^{n+1}),d_{S_{p}})$ sending the $\ell^{p}$ metric to the Schatten $p$ metric. The diagonal encoding also defines embeddings $(\Delta^{n},d_{H})\to(\states(\C^{n+1}),d_{H})$ and $(\Delta^{n},d_{H})\to(\states(\C^{n+1}),d_{F})$ sending the (classical) Hellinger distance to the (quantum) Hellinger and Bures fidelity distances. The former can be verified by a short computation, while the latter follow from the relations 
\begin{equation}
B(x,y)=B\big(\rho_{\mathrm{diag}}(x),\rho_{\mathrm{diag}}(y)\big)
\end{equation}
and
\begin{equation}
B(x,y)^2=F\big(\rho_{\mathrm{diag}}(x),\rho_{\mathrm{diag}}(y)\big)
\end{equation}
for all $x,y\in\Delta^{n}$.
\end{example}

\begin{example}
Given a unitary matrix $U\in\unitary(\mathcal{H})$ acting on a Hilbert space $\mathcal{H}$, let $\Ad_{U}:\states(\mathcal{H})\to\states(\mathcal{H})$ be the induced operation on density matrices given by $\Ad_{U}(\rho)=U\rho U^{\dag}$. Then $\Ad_{U}$ defines an isometry from $\states(\mathcal{H})$ to itself with respect to all the Schatten $p$ metrics, the Bures fidelity metric, and the Hellinger metric defined in Example~\ref{defn:quantumdistances}. 
\end{example}

Note that an embedding is automatically injective (one-to-one). Conversely, injective maps from a \emph{set} into a \emph{metric space} can be used to \emph{pull back} the metric onto the set, thereby endowing the domain with a metric. 

\begin{lemma}
\label{lem:inducedmetric}
Let $X$ be a set, $(Y,d_{Y})$ a metric space,  $X\xrightarrow{f}Y$ a function, and $d_{f}:X\times X\to[0,\infty)$ the function 
\begin{equation}
\label{eq:dXpullbackmetric}
d_{f}(x_1,x_2):=d_{Y}\big(f(x_1),f(x_2)\big)
\end{equation}
defined for all $x_1,x_2\in X$. 
Then, 
$d_{f}$ is a metric on $X$ if and only if $f$ is one-to-one. In such a case, $(X,d_{f})\xrightarrow{f}(Y,d_{Y})$ is an embedding. 
\end{lemma}

The proof is left as an exercise (see Ref.~\cite{PBVP24} for details). 

\begin{definition}
\label{defn:pullbackmetric}
Given a set $X$, a metric space $(Y,d_{Y})$, and an injective function $X\xrightarrow{f}Y$, the metric $d_{f}$ on $X$ constructed in Lemma~\ref{lem:inducedmetric} is called the \define{pull-back metric} or the \define{embedding metric} induced by $f$.
\end{definition}

The importance of the embedding metric is two-fold. First, if a dataset $X$ does not have a metric on it, then the embedding metric obtained from some encoding $f:X\to Y$, where $Y$ has a metric, enables the notion of distance to be defined between the points of $X$. An example to keep in mind is a quantum encoding $X\to\states(\mathcal{H})$, where the quantum state space $\states(\mathcal{H})$ is equipped with one of the distance functions from Example~\ref{defn:quantumdistances}, but where we might not have a notion of distance on $X$ or perhaps we do not trust that distance (eg.\ vectors in Euclidean space representing images~\cite{Perea2018Multiscale}). Second, if $X$ does already have a metric on it, say $d_{X}$, then there could be a discrepancy between the distances $d_{f}(x_1,x_2)$ and $d_{X}(x_1,x_2)$. This discrepancy is measured by the \emph{distortion}, which is defined later in Definition~\ref{defn:distortionf}, and is one of the ingredients in making sense of the Gromov--Hausdorff distance between two metric spaces, which in this case are the metric spaces $(X,d_{X})$ and $(X,d_{f})$. Notice that these two metric spaces have the same underlying set $X$ but different metrics. In other words, we can measure how different a given metric $d_{X}$ is from the embedding metric $d_{f}$, and this will play an important role in determining how much information is lost under the mapping $f$.

%%%%%%%%%%%%%%%%%%%%%%%%%%%%%%%%%%%%%%%%%%%%%%%%
\section{Topological Data Analysis}
\label{sec:tda}
%%%%%%%%%%%%%%%%%%%%%%%%%%%%%%%%%%%%%%%%%%%%%%%%

Typical finite subsets of Euclidean space, i.e., \emph{point clouds}, that arise in structured problems often contain intrinsic structure. For example, the data could be sampled from a probability distribution that is localized along some geometric object, such as a submanifold~\cite{FeMiNa2016}. Inferring what that manifold could be is the subject of topological inference, and knowing such information often indicates structure in a population that could be leveraged for certain tasks, such as generating data or specifying the proximity of one datum to another. It is usually difficult to infer the exact manifold from which the data are sampled, and therefore one typically resorts to computing \emph{invariants}. This helps to rule out certain manifolds and isolate plausible manifolds in order to better visualize the topology of a point-cloud or simply to have a better understanding of the source of the data~\cite{carlsson2021topological}. For example, Betti numbers, the dimensions of homology groups, provide one example of an effective technique. Although the homology does not uniquely characterize a manifold, it does provide a great deal of information. Other invariants, such as characteristic classes or homotopy groups, will not be investigated here, though there are many works in this direction~\cite{Rubio2002Constructive,Tinarrage2022Computing,Perea2018Multiscale,MedinaZhou2025,MemoliZhou2025,Zhou2025SullivanModelTDA}. 

Since our aim is to bring these techniques to the quantum setting, we will provide a somewhat thorough review of persistent homology, the main form of topological data analysis investigated in this work. We begin with some background on category theory, metric spaces and notions of distance between them, and some methods of inferring the topology of a manifold from a sampled point cloud. The key results that we review at the end are some stability theorems of TDA, which describe how such invariants change under small perturbations of the data, which will become important when we describe quantum encodings in a later section.

%%%%%%%%%%%%%%%%%%%%%%%%%%%%%%%%%%%%%%%%%%%%%%%%
\subsection{Categories, functors, and natural transformations}
\label{subsec:cats}
%%%%%%%%%%%%%%%%%%%%%%%%%%%%%%%%%%%%%%%%%%%%%%%%

The language of categories, functors, and natural transformations dramatically facilitates the discussions of topological data analysis, specifically persistent homology, primarily due to its origins stemming from algebraic topology~\cite{CarlssonTDA,edelsbrunner2002topological,MacLaneHomology}. As such, we briefly review some basic definitions and save our main examples for the following sections. Fortunately, we only need three definitions from category theory (see Ref.~\cite{PBVP24} for an introduction emphasizing examples coming from quantum encodings). 

\begin{definition}\label{def:category}
    A \define{category} $\mathbf{C}$ consists of (1) a collection of \define{objects}, (2) a set $\mathbf{C}(X,Y)$ of \define{morphisms} from object $X$ to object $Y$, elements of which are depicted as arrows $f:X\to Y$, (3) a \define{composition rule} $\mathbf{C}(X,Y)\times\mathbf{C}(Y,Z)\to\mathbf{C}(X,Z)$ sending morphisms $f~:~X\to~Y$ and $g:Y\to Z$ to a specified morphism $g\circ f:X\to~Z$. 
    Moreover, these items must satisfy the following two axioms. First, composition is \define{associative}, i.e., $h\circ (g\circ f)=(h\circ g)\circ f$ for all composable triples of morphisms $f:X\to Y$, $g:Y\to Z$, and $h: Z\to W$,
    Second, there exist \define{identity morphisms}, i.e., for every object $X$ there exists a morphism $\id_X:X\to X$ such that $f\circ \id_X=\id_Y\circ f=f$ for all morphisms $f:X\to Y$.
\end{definition}

\begin{example}
Let $(X,d_{X})$ be a finite metric space, i.e., a finite set $X$ together with a metric $d_{X}:X\times X\to[0,\infty)$. 
A \define{distance nonincreasing map} (also called a \define{short map}) from $(X,d_{X})$ to $(Y,d_{Y})$ is a function $f:X\to Y$ such that $d_{Y}\big(f(x),f(x')\big)\le d_{X}(x,x')$ for all $x,x'\in X$. 
The collection of finite metric spaces together with the collection of distance nonincreasing maps forms a category, which will be denoted by $\mathbf{FinMetSpace}$. 
Many of the finite metric spaces we will be dealing with in this manuscript are subsets of Euclidean space with the induced metric. The subsets in this case are often interpreted as point clouds, sampled data points. In the case that the data are not naturally embedded in a given Euclidean space, one can use a Multidimensional Scaling (MDS) algorithm to embed the data points into some Euclidean space as an approximation~\cite{Mead1992MDS,BoGr2005MDS,Torgerson1952MDS}. We will return to MDS in the context of quantum encodings later in Section~\ref{sec:QMDS}. 
The restriction to finiteness of the metric spaces is not necessary, and we will find it important to allow for metric spaces whose underlying sets have infinitely many elements.  
\end{example}

\begin{definition} 
\label{def:functor}
A \define{functor} $F:\mathbf{C}\to\mathbf{D}$ between categories consists of two assignments, the first being the assignment of an object $F(X)$ in $\mathbf{D}$ to each object $X$ in $\mathbf{C}$, and the second being the assignment of a morphism $F(f): F(X)\to F(Y)$ in $\mathbf{D}$ to each morphism $f: X\to Y$ in $\mathbf{C}$.
Moreover, these assignments must satisfy the following two axioms. First, \emph{composition is preserved}, i.e., $F(g\circ f)={F(g)\circ F(f)}$ for all  morphisms $f:X\to Y$ and $g:Y\to Z$ in $\mathbf{C}$. Second, \emph{identities are preserved}, i.e., $F(\id_X)=\id_{F(X)}$ for every object $X$ in $\mathbf{C}$.
\end{definition}

\begin{definition}
\label{defn:nattrans}
Let $\mathbf{C}$ and $\mathbf{D}$ be categories, and let $F,G:\mathbf{C}\to\mathbf{D}$ be two functors. A \define{natural transformation} $\eta$ from $F$ to $G$, written $\eta:F\Rightarrow G$, associates to each object $\mathcal{X}$ in $\mathbf{C}$ a morphism $\eta_{\mathcal{X}}:F(\mathcal{X})\to G(\mathcal{X})$ in $\mathbf{D}$ such that the diagram
\begin{equation}
\label{eq:naturality}
\xy 0;/r.25pc/:
   (-12.5,7.5)*+{F(\mathcal{X})}="FX";
   (12.5,7.5)*+{F(\mathcal{Y})}="FY";
   (-12.5,-7.5)*+{G(\mathcal{X})}="GX";
   (12.5,-7.5)*+{G(\mathcal{Y})}="GY";
   {\ar"FX";"FY"^{F(f)}};
   {\ar"GX";"GY"_{G(f)}};
   {\ar"FX";"GX"_{\eta_{\mathcal{X}}}};
   {\ar"FY";"GY"^{\eta_{\mathcal{Y}}}};
\endxy
\end{equation}
in $\mathbf{D}$ commutes, i.e., $G(f)\circ\eta_{\mathcal{X}}=\eta_{\mathcal{Y}}\circ F(f)$, for every morphism $\mathcal{X}\xrightarrow{f}\mathcal{Y}$ in $\mathcal{C}$. The commutativity of diagram~\eqref{eq:naturality} is often referred to as \define{naturality}. 
\end{definition}

%%%%%%%%%%%%%%%%%%%%%%%%%%%%%%%%%%%%%%%%%%%%%%%%
\subsection{The Gromov--Hausdorff distance between metric spaces}
\label{subsec:HS-distance}
%%%%%%%%%%%%%%%%%%%%%%%%%%%%%%%%%%%%%%%%%%%%%%%%

For two metric spaces $(X,d_X)$ and $(Y,d_Y)$ the Gromov--Hausdorff distance provides a method to calculate the difference between these spaces (cf.\ \cite[Section 7.3]{BuBuIv01}). To understand the Gromov--Hausdorff distance, we first define the Hausdorff distance. We assume familiarity with the basics of topology and metric spaces~\cite{Mu00,Ru76}.  

\begin{definition}
For a subset $A \subseteq X$, with $(X,d_{X})$ a metric space, denote 
\begin{equation}
A^{(r)} = \bigcup_{x\in A} B_r(x),
\end{equation}
where $B_r(x)$ is a ball of radius $r>0$ centered at the point $x$. Then, for two closed sets $A,B \subseteq X$, the \define{Hausdorff distance} is defined as 
\begin{equation}\label{eq:hausdorff}
    d_{H}(A,B) := \inf \left\{ r>0 \,\big|\, B \subset A^{(r)} \mbox{ and } A \subset B^{(r)} \right\}.
\end{equation}
\end{definition}

The Hausdorff distance gives the farthest distance that any point of $A$ is from the set $B$, or any point of $B$ is from the set $A$ if this distance is larger, i.e., 
\begin{equation}
d_{H}(A,B)=\max\left\{\sup_{a\in A} d_{X}(a,B),\sup_{b\in B} d_{X}(b,A)\right\}, 
\end{equation}
where 
\begin{equation}
d_{X}(a,B):=\inf_{b\in B}d_{X}(a,b)
\end{equation}
and similarly for $d_{X}(b,A)$. 

To make sense of a distance between two metric spaces $(X,d_{X})$ and $(Y,d_{Y})$ that are not a-priori subsets of a fixed metric space, one approach is to embed $(X,d_{X})$ and $(Y,d_{Y})$ into an ambient metric space $(Z,d_{Z})$ and then utilize the Hausdorff distance for subsets of $(Z,d_{Z})$. It is important for us to use embeddings here, since otherwise one could send all of $X$ and $Y$ to the same point rendering their distance to be zero. Moreover, many embeddings of $(X,d_{X})$ and $(Y,d_{Y})$ into $(Z,d_{Z})$ might exist (or none might exist), and so one should additionally change $(Z,d_{Z})$ to allow for greater flexibility. Minimizing over these possibilities leads to the Gromov--Hausdorff distance (this definition also applies for semi-metric spaces~\cite{KaltonOstrovskii1999}, which will be described later). 

\begin{definition}
\label{defn:GHdistance}
The \define{Gromov--Hausdorff distance} between two metric spaces $(X,d_X)$ and $(Y,d_Y)$ is given by 
\begin{equation}\label{eq:gromov-hausdorff}
    d_{GH}(X,Y) := \inf_{(Z,d_{Z})} \inf_{ \substack{f:X\hookrightarrow Z \\ g:Y\hookrightarrow Z}} d_H \Big( f(X), g(Y) \Big),
\end{equation}
where the innermost infimum is over all embeddings $f:X\hookrightarrow Z$ and $g:Y\hookrightarrow Z$, while the outermost infimum is over all finite metric spaces $(Z,d_Z)$. 
\end{definition}

Note that although the collection of metric spaces $(Z,d_{Z})$ is a proper class, the infimum can be equivalently calculated over the set of semi-metrics $d_{Z}$ on $Z=X\sqcup Y$, which is a set (for details, see Definition~\ref{defn:semimetricspace} below and Remark 7.3.12 in Ref.~\cite{BuBuIv01}). 
Alternatively, the Gromov--Hausdorff distance can also be expressed as the infimum over distortions of correspondences (cf.\ \cite[Theorem 7.3.25]{BuBuIv01} and Refs.~\cite{KaltonOstrovskii1999,Adams2025Hausdorff}). This is important because definition~\eqref{eq:gromov-hausdorff} involves all possible external metric spaces $(Z,d_{Z})$, making the expression difficult to compute. As such, we recall the relevant definitions. 

\begin{definition}
Given sets $X$ and $Y$, a \define{correspondence} between $X$ and $Y$ is a set $\mathfrak{R}\subseteq X\times Y$ such that for every $x\in X$, there exists at least one $y\in Y$ such that $(x,y)\in\mathfrak{R}$, and for every $y\in Y$, there exists at least one $x\in X$ such that $(x,y)\in\mathfrak{R}$. 
\end{definition}

\begin{definition}
Given metric spaces $(X,d_X)$ and $(Y,d_Y)$ and a correspondence $\mathfrak{R}$ between $X$ and $Y$, the \define{distortion} of $\mathfrak{R}$ is the number
\begin{equation}
\label{eq:distortionofcorrespondence}
\mathrm{dis}(\mathfrak{R})=
\sup_{(x,y),(x',y')\in\mathfrak{R}} \Big|d_{X}(x,x')-d_{Y}(y,y')\Big|.
\end{equation}
\end{definition}

\begin{example}
\label{ex:intervalandcircle}
Equip $X=[0,2\pi)$ with the Euclidean distance $d_{X}(\theta,\phi)=|\theta-\phi|$ and equip 
\begin{equation}
Y=S^{1}=\big\{e^{i\theta}\in\C\,:\,\theta\in[0,2\pi)\big\}
\end{equation}
with the geodesic distance
\begin{equation}
d_{Y}(e^{i\theta},e^{i\phi}):=\min\big\{|\theta-\phi|,2\pi-|\theta-\phi|\big\}. 
\end{equation}
These distances are visualized as heat maps in Figure~\ref{fig:heatmapintervalcircle}. 

\begin{figure}
\includegraphics[width=4cm]{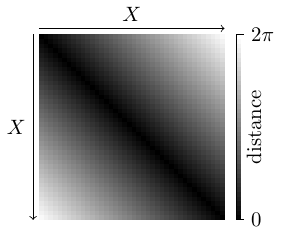}
\quad
\includegraphics[width=4cm]{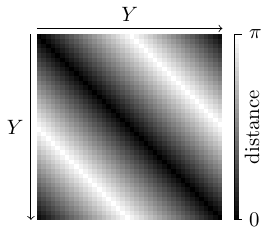}
\caption{Heat maps of the Euclidean distance on the interval $X=[0,2\pi)$ and geodesic distance on the circle $Y=S^1$, respectively. The sets have been discretized uniformly and the intensities have been rescaled so that the maximum distance between points is drawn in white and minimum distance between points is drawn in black. See Example~\ref{ex:intervalandcircle} for more details.}
\label{fig:heatmapintervalcircle}
\end{figure}

Let $\mathfrak{R}$ be the correspondence 
\begin{equation}
\mathfrak{R}=\big\{(\theta,e^{i\theta})\in X\times Y\,:\, \theta\in[0,2\pi)\big\}. 
\end{equation}
Therefore, in this case, the distortion of $\mathfrak{R}$ is $\mathrm{dis}(\mathfrak{R})=\pi$.
\end{example}

With these two definitions, the Gromov--Hausdorff distance between metric spaces $(X,d_{X})$ and $(Y,d_{Y})$ can also be expressed as 
\begin{equation}
d_{GH}(X,Y)=\frac{1}{2}\inf_{\mathfrak{R}}\big(\mathrm{dis}(\mathfrak{R})\big), 
\end{equation}
where the infimum is taken over all correspondences $\mathfrak{R}$ between $X$ and $Y$ (cf.\ Theorem 7.3.25 in Ref.~\cite{BuBuIv01}). 

Finally, we mention yet one more expression for the Gromov--Hausdorff distance defined in terms of $\epsilon$-isometries. This will involve defining distortion and co-distortion, concepts that we will use later. 

\begin{definition}
\label{defn:distortionf}
The \define{distortion of a function} $f:X\to Y$, where $(X,d_X)$ and $(Y,d_Y)$ are metric spaces ($f$ need not preserve distances) is the number given by 
\begin{equation}
\label{eqn:distortionoffunction}
\mathrm{dis}(f)=\sup_{x,x'\in X}\Big|d_{X}(x,x')-d_{Y}\big(f(x),f(x')\big)\Big|.
\end{equation}
\end{definition}

Note that the distortion of a \emph{surjective} function is a special case of~\eqref{eq:distortionofcorrespondence}, where the correspondence is given by 
\begin{equation}
\mathfrak{R}=\Big\{\big(x,f(x)\big)\;:\;x\in X\Big\}.
\end{equation}
Note, however, that $f$ need not be surjective for the definition in~\eqref{eqn:distortionoffunction}. In addition, the distortion of a correspondence need not arise from a surjective function.  

\begin{remark}
\label{rmk:distortionasdistancedifference}
Note that if $(X,d_{X})$ and $(Y,d_{Y})$ are finite metric spaces and $f$ is injective, then the distortion of $f$ coincides with the $\ell^{\infty}$ distance between the distance matrices associated with $(X,d_{X})$ and $(X,d_{f})$ (note, as discussed in Remark~\ref{rmk:ellpSchattenp}, that the $\ell^{\infty}$ metric is not the same as the Schatten $\infty$ metric). Namely, if we write $X=\{x_1,\dots,x_m\}$, set $D_{X}$ to be the $m\times m$ matrix with $ij$ entry $(D_{X})_{ij}=d_{X}(x_i,x_j)$, and set $D_{f}$ to be the $m\times m$ matrix with $ij$ entry $(D_{f})_{ij}=d_{f}(x_i,x_j)=d_{Y}(f(x_i),f(x_j))$, then 
\begin{equation}
\mathrm{dis}(f)=d_{\ell^{\infty}}(D_{X},D_{f}). 
\end{equation}
This formulation will be useful when we discuss multidimensional scaling for optimal quantum encodings in Section~\ref{sec:QMDS}. 
\end{remark}

\begin{definition}
Given metric spaces $(X,d_X)$ and $(Y,d_{Y})$ together with two functions $f:X\to Y$ and $g:Y\to X$, the \define{co-distortion of $f$ and $g$} is the number
\begin{equation}
\label{eqn:codistortion}
C(f,g)=\sup_{x\in X,y\in Y}\Big|d_{X}(x,g(y))-d_{Y}\big(f(x),y\big)\Big|.
\end{equation}
\end{definition}

For some intuition regarding the definition of co-distortion, note that if $\delta:=C(f,g)$ is small, then $d_{X}(x,g(f(x)))\le\delta$ for all $x\in X$ and $d_{Y}(f(g(y)),y)\le\delta$ for all $y\in Y$, which are obtained by setting $y=f(x)$ and $x=g(y)$ in~\eqref{eqn:codistortion}, respectively~\cite{memoli2012some}. In other words, if the co-distortion $C(f,g)$ is small, then $f$ and $g$ are close to being inverses of each other. 

\begin{definition}
Given metric spaces $(X,d_X)$ and $(Y,d_{Y})$, a positive number $\epsilon>0$, the pair of functions $f:X\to Y$ and $g:Y\to X$ form an \define{$\epsilon$-isometry} between $(X,d_X)$ and $(Y,d_{Y})$ iff 
\begin{equation}
\mathrm{dis}(f)\le\epsilon,
\quad
\mathrm{dis}(g)\le\epsilon,
\quad\text{ and }\quad
C(f,g)\le\epsilon.
\end{equation}
Such an $\epsilon$ isometry will be denoted 
\begin{equation}
X\;\xy 0;/r.25pc/:{\ar(0,1.0);(5,1.0)^{f}};{\ar(5,-1.0);(0,-1.0)^{g}};{\ar@{}(0,0);(5,0)|-{\epsilon}};\endxy \;Y\,
\quad\text{ or }\quad
X\;\xy 0;/r.25pc/:{\ar(0,0.5);(5,0.5)^{f}};{\ar(5,-0.5);(0,-0.5)^{g}};\endxy \;Y\,
\end{equation}
if we want to specify $\epsilon$ explicitly or if we want to leave $\epsilon$ unspecified, respectively.
\end{definition}

Alternatively, $f:X\to Y$ is part of an $\epsilon$-isometry if and only if $\dis(f)\le\epsilon$ and $f$ is \define{$\epsilon$-surjective}, i.e., $d_{Y}\big(y,f(X)\big)\le\epsilon$ for all $y\in Y$~\cite{BuBuIv01,BrBrKi2006GMDS}. Using the idea of an $\epsilon$-isometry, another equivalent expression for the Gromov--Hausdorff distance between compact metric spaces $(X,d_X)$ and $(Y,d_{Y})$ is given by (cf.\ \cite[Definition 5.5]{MemoliSinghal19} and Refs.~\cite{memoli2012some,KaltonOstrovskii1999})
\begin{align}
d_{GH}(X,Y)&=\frac{1}{2} \inf_{ \substack{f:X \to Y \\ g:Y \to X}} \sup\Big\{ \dis(f), \dis(g), C(f,g) \Big\} \nonumber\\
&=\frac{1}{2}\inf\Big\{\epsilon\ge0\,:\,\exists \; X\;\xy 0;/r.25pc/:{\ar(0,1.0);(5,1.0)^{f}};{\ar(5,-1.0);(0,-1.0)^{g}};{\ar@{}(0,0);(5,0)|-{\epsilon}};\endxy \;Y\,\Big\}.
\label{eq:gromov-hausdorff-final}
\end{align}

As mentioned earlier, the Gromov--Hausdorff distance formalizes a notion of distance between metric spaces. However, the Gromov--Hausdorff distance is not a metric in the usual sense for a few reasons. First, the collection of all metric spaces does not form a set. Secondly, to get a metric, we must restrict to certain isometry classes of compact metric spaces (compact here means in the sense of topology, with the topology induced by the metric~\cite{Mu00}---since we often work with finite sets representing data sets, compactness is automatic). Once we do this, we obtain a metric on the isometry classes of compact metric spaces (cf.\ Theorem 7.3.30.\ in Ref~\cite{BuBuIv01}). 

\begin{proposition}
The Gromov--Hausdorff distance defines a metric on the set of isometry classes of compact metric spaces. Namely, $d_{GH}$ is nonnegative-valued and satisfies the following properties. 
\begin{enumerate}[i.]
\item $d_{GH}(X,Y)=0$ if and only if $(X,d_{X})$ is isometric to $(Y,d_{Y})$.
\item
$d_{GH}(X,Y)=d_{GH}(Y,X)$ for all $(X,d_{X})$ and $(Y,d_{Y})$
\item
$d_{GH}(X,Y)\le d_{GH}(X,Z)+d_{GH}(Z,Y)$ for all $(X,d_{X})$, $(Y,d_{Y})$, $(Z,d_{Z})$. 
\end{enumerate}
\end{proposition}

In this paper, the metric spaces and functions we consider will often be of the following form. A point cloud in $\R^{n}$, for some $n\in\N$, will be denoted by $(X,d_{X})$, where $d_{X}$ is some suitable metric, typically Euclidean but not necessarily. Meanwhile, a quantum state space (or some subset thereof) will take the place of $(Y,d_{Y})$, so that $Y=\states(\mathcal{H})$ is the state space of some Hilbert space $\mathcal{H}$ and $d_{Y}$ is some suitable metric, such as the Bures fidelity metric. 
Before the encoding, which is a function $f:X\to Y$, is applied, distances between data points in the point cloud are measured using $d_{X}$. After the encoding, distances are measured between the encoded states, which use the metric $d_{Y}$. Even if $f$ is the restriction of a smooth mapping from $\R^{n}$, $f$ need not preserve distances, which means that the distortion $\mathrm{dis}(f)$ could be significant.

%%%%%%%%%%%%%%%%%%%%%%%%%%%%%%%%%%%%%%%%%%%%%%%%
\subsection{Persistence vector spaces and homology}
\label{subsec:persistenthomology}
%%%%%%%%%%%%%%%%%%%%%%%%%%%%%%%%%%%%%%%%%%%%%%%%

A point cloud in Euclidean space gives rise to a persistence vector space, which is a parameterized family of vector spaces connected to each other by linear maps. More specifically, the \emph{Vietoris--Rips complex} is obtained from a point cloud by isolating the set of points that are within a distance $\epsilon$ from each other. A set of $n$ points that are all within distance $\epsilon$ between any pair of points defines an $(n-1)$-simplex. For example, if two points are within a distance of $\epsilon$ from each other, then this defines a $1$-simplex, which is a line. If three points are within a distance of $\epsilon$ from each other, then this defines a $2$-simplex, which is a filled-in triangle. The set of all of these simplices then defines a vector space, where each simplex is treated as a linearly independent vector. The unions of certain simplices might specify topological features roughly describing holes, voids, etc. The \emph{birth time} of such a topological feature is the value of $\epsilon$ that the feature appears and the \emph{death time} of this feature is the value of $\epsilon$ such that the feature is filled in. With the increasing value of $\epsilon$, the birth-death interval $(b_{i},d_{i})$ of a particular structure labelled by $i$ gives the persistence diagram.
The robustness of the persistence diagrams is dictated by stability theorems.
In what follows, we review these ideas in more detail by following Refs.~\cite{CarlssonTDA,carlsson2021topological,edelsbrunner2002topological,ZoCa05,oudot2017persistence,tinarrage2021notes,MemoliSinghal19,Munkres84,Hausmann96}.  

\begin{definition}
An \define{abstract simplicial complex} is a set $\mathcal{S}$ of finite nonempty sets so that whenever $\sigma\in\mathcal{S}$ then every subset $\tau\subseteq\sigma$ is also an element of $\mathcal{S}$. An element of $\mathcal{S}$ is called an \define{abstract simplex}, or simply \define{simplex} for short. The \define{dimension} of a simplex $\sigma$ is equal to $\#\sigma-1$, where $\#\sigma$ denotes the number of elements in the set $\sigma$. If the dimension of $\sigma$ is $n$, then $\sigma$ is referred to as an \define{$n$-simplex} in $\mathcal{S}$. The set of all $n$-simplices in $\mathcal{S}$ is denoted by $\mathcal{S}_{n}$. The \define{dimension} of an abstract simplicial complex $\mathcal{S}$ is the largest integer $n$ such that $\mathcal{S}_{n}$ is nonempty (it could be infinite). 
\end{definition}

These definitions are motivated by abstracting some essential features of geometric simplices and simplicial complexes (cf. Sections 1 and 2 of Ref.~\cite{Munkres84}). Nevertheless, it is often useful to visualize abstract simplices as geometric simplices. 

There are many ways to construct abstract simplicial complexes. Here, we will focus on one obtained from a finite metric space~\cite{Hausmann96}. 

\begin{definition}
Given a finite metric space $(X,d_{X})$ and a real number $\epsilon\ge0$, the \define{Vietoris--Rips complex} is the abstract simplicial complex $\mathcal{R}^{\epsilon}(X,d_{X})$ defined as follows. First, the $0$-simplices of $\mathcal{R}^{\epsilon}(X,d_{X})$ are the points of $X$. Second, for $k\ge1$, a $k$-simplex in $\mathcal{R}^{\epsilon}_{k}(X,d_{X})$ is a subset $\sigma\subseteq X$ with $\#\sigma=k+1$ and $\mathrm{diam}(\sigma)<\epsilon$, where $\mathrm{diam}(\sigma)$ denotes the \define{diameter} of $\sigma$ and is defined as 
\begin{equation}
\mathrm{diam}(\sigma)=\sup_{x,x'\in\sigma}d_{X}(x,x').
\end{equation}
Occasionally, the notation $\mathrm{diam}(\sigma,d_{X})$ will be used to emphasize the metric. 
\end{definition}

Equivalently, $\sigma$ is a $k$-simplex in $\mathcal{R}^{\epsilon}_{k}(X,d_{X})$ iff $d_{X}(x,x')<\epsilon$ for all $x,x'\in\sigma$. 

The examples from Ref.~\cite{Hausmann96} (see also Ref.~\cite{WeiWeiPTL2025}) are particularly illuminating, as they offer a useful visual guide for the definitions (cf.\ Figure~\ref{fig:Hausmannunity}). 

\begin{figure}
$\vcenter{\hbox{%
\text{(a)}}}$
$\vcenter{\hbox{%
\begin{tikzpicture}
\def\n{8}
\def\e{0.5}%epsilon <2pi/n gives discrete set (divide diameter by 2 when drawing balls of radius)
\foreach \k in {1,2,...,\n} {
   \coordinate (\k) at ({cos(deg(2*pi*\k/\n))},{sin(deg(2*pi*\k/\n))});
   \node at (\k) {$\bullet$};
}
\end{tikzpicture}
}}$
\qquad\qquad
$\vcenter{\hbox{%
\text{(b)}}}$
$\vcenter{\hbox{%
\begin{tikzpicture}
\def\n{8}
\def\e{1.0}%2pi/n<epsilon<4pi/n gives polygon with n sides (circle) 
\foreach \k in {1,2,...,\n} {
   \coordinate (\k) at ({cos(deg(2*pi*\k/\n))},{sin(deg(2*pi*\k/\n))});
   \node at (\k) {$\bullet$};
%   \fill[opacity=0.1] (\k) circle ({\e/2});
%   \draw [line width=2mm,line cap=round,opacity=0.1,domain={\k-\e/2}:{\k+\e/2}] plot ({cos(deg(2*pi*\x/\n))},{sin(deg(2*pi*\x/\n))});
   \draw[thick] ({cos(deg(2*pi*\k/\n))},{sin(deg(2*pi*\k/\n))}) -- ({cos(deg(2*pi*(\k+1)/\n))},{sin(deg(2*pi*(\k+1)/\n))});
}
\end{tikzpicture}
}}$
\\
%%%%%%%%sphere and Mobius band%%%%%%%%%%%
$\vcenter{\hbox{%
\text{(c)}}}$
$\vcenter{\hbox{%
\begin{tikzpicture}
\def\n{4}
\def\e{3.14}%2pi/n<epsilon<3pi/n gives a two-dimensional surface (n=6 gives 2-sphere) 
\foreach \k in {1,2,...,\n} {
   \coordinate (\k) at ({cos(deg(2*pi*\k/\n))},{sin(deg(2*pi*\k/\n))});
   \node at (\k) {$\bullet$};
%   \fill[opacity=0.1] (\k) circle (\e);
   \draw[thick] ({cos(deg(2*pi*\k/\n))},{sin(deg(2*pi*\k/\n))}) -- ({cos(deg(2*pi*(\k+1)/\n))},{sin(deg(2*pi*(\k+1)/\n))});
%   \draw ({cos(deg(2*pi*\k/\n))},{sin(deg(2*pi*\k/\n))}) -- ({cos(deg(2*pi*(\k+2)/\n))},{sin(deg(2*pi*(\k+2)/\n))});
   \fill[opacity=0.1] ({cos(deg(2*pi*\k/\n))},{sin(deg(2*pi*\k/\n))}) -- ({cos(deg(2*pi*(\k+1)/\n))},{sin(deg(2*pi*(\k+1)/\n))}) -- ({cos(deg(2*pi*(\k+2)/\n))},{sin(deg(2*pi*(\k+2)/\n))}) -- cycle;
}
\fill[opacity=0.15] (0,1) -- (-1,0) -- (0,-1) -- cycle;
\draw[thick] ({cos(deg(2*pi*1/\n))},{sin(deg(2*pi*1/\n))}) -- ({cos(deg(2*pi*(1+2)/\n))},{sin(deg(2*pi*(1+2)/\n))});
\draw[dashed] ({cos(deg(2*pi*2/\n))},{sin(deg(2*pi*2/\n))}) -- ({cos(deg(2*pi*(2+2)/\n))},{sin(deg(2*pi*(2+2)/\n))});
\end{tikzpicture}
}}$
\qquad\qquad
%%%%%%%%%%%%%%%%%%%%%%sphere%%%%%%%%%%%%%%%%%
$\vcenter{\hbox{%
\text{(d)}}}$
$\vcenter{\hbox{%
\begin{tikzpicture}
\def\n{6}
\def\e{2.4}%2pi/n<epsilon<3pi/n gives a two-dimensional surface (n=6 gives 2-sphere) 
\foreach \k in {1,2,...,\n} {
   \coordinate (\k) at ({cos(deg(2*pi*\k/\n))},{sin(deg(2*pi*\k/\n))});
   \node at (\k) {$\bullet$};
%   \fill[opacity=0.1] (\k) circle (\e);
   \draw[thick] ({cos(deg(2*pi*\k/\n))},{sin(deg(2*pi*\k/\n))}) -- ({cos(deg(2*pi*(\k+1)/\n))},{sin(deg(2*pi*(\k+1)/\n))});
   \draw[dashed] ({cos(deg(2*pi*\k/\n))},{sin(deg(2*pi*\k/\n))}) -- ({cos(deg(2*pi*(\k+2)/\n))},{sin(deg(2*pi*(\k+2)/\n))});
}
\foreach \k in {1,3,5} {
   \draw[thick] ({cos(deg(2*pi*\k/\n))},{sin(deg(2*pi*\k/\n))}) -- ({cos(deg(2*pi*(\k+2)/\n))},{sin(deg(2*pi*(\k+2)/\n))});
}
\fill[opacity=0.1] ({cos(deg(2*pi*1/\n))},{sin(deg(2*pi*1/\n))}) -- ({cos(deg(2*pi*3/\n))},{sin(deg(2*pi*3/\n))}) -- ({cos(deg(2*pi*5/\n))},{sin(deg(2*pi*5/\n))}) -- cycle;
\fill[opacity=0.25] ({cos(deg(2*pi*1/\n))},{sin(deg(2*pi*1/\n))}) -- ({cos(deg(2*pi*2/\n))},{sin(deg(2*pi*2/\n))}) -- ({cos(deg(2*pi*3/\n))},{sin(deg(2*pi*3/\n))}) -- cycle;
\fill[opacity=0.15] ({cos(deg(2*pi*1/\n))},{sin(deg(2*pi*1/\n))}) -- ({cos(deg(2*pi*5/\n))},{sin(deg(2*pi*5/\n))}) -- ({cos(deg(2*pi*0/\n))},{sin(deg(2*pi*0/\n))}) -- cycle;
\fill[opacity=0.35] ({cos(deg(2*pi*3/\n))},{sin(deg(2*pi*3/\n))}) -- ({cos(deg(2*pi*4/\n))},{sin(deg(2*pi*4/\n))}) -- ({cos(deg(2*pi*5/\n))},{sin(deg(2*pi*5/\n))}) -- cycle;
\end{tikzpicture}
}}$
\\
$\vcenter{\hbox{%
\text{(e)}}}$
$\vcenter{\hbox{%
\begin{tikzpicture}
\def\n{7}
\def\e{3.14}%2pi/n<epsilon<3pi/n gives a two-dimensional surface (n=6 gives 2-sphere) 
\foreach \k in {1,2,...,\n} {
   \coordinate (\k) at ({cos(deg(2*pi*\k/\n))},{sin(deg(2*pi*\k/\n))});
   \node at (\k) {$\bullet$};
%   \fill[opacity=0.1] (\k) circle (\e);
   \draw[thick] ({cos(deg(2*pi*\k/\n))},{sin(deg(2*pi*\k/\n))}) -- ({cos(deg(2*pi*(\k+1)/\n))},{sin(deg(2*pi*(\k+1)/\n))});
    \draw[dashed] ({cos(deg(2*pi*\k/\n))},{sin(deg(2*pi*\k/\n))}) -- ({cos(deg(2*pi*(\k+2)/\n))},{sin(deg(2*pi*(\k+2)/\n))});
   \fill[opacity=0.1] ({cos(deg(2*pi*\k/\n))},{sin(deg(2*pi*\k/\n))}) -- ({cos(deg(2*pi*(\k+1)/\n))},{sin(deg(2*pi*(\k+1)/\n))}) -- ({cos(deg(2*pi*(\k+2)/\n))},{sin(deg(2*pi*(\k+2)/\n))}) -- cycle;
}
   \draw[thick] ({cos(deg(2*pi*4/\n))},{sin(deg(2*pi*4/\n))}) -- ({cos(deg(2*pi*6/\n))},{sin(deg(2*pi*6/\n))}) -- ({cos(deg(2*pi*1/\n))},{sin(deg(2*pi*1/\n))}) -- ({cos(deg(2*pi*3/\n))},{sin(deg(2*pi*3/\n))});
\draw ({cos(deg(2*pi*1/\n))},{sin(deg(2*pi*1/\n))}) -- ({cos(deg(2*pi*(1+2)/\n))},{sin(deg(2*pi*(1+2)/\n))});
%%%%%%point 2 to point 4 solid part%%%%%%%%%%
\draw[thick] ({(0.693)*cos(deg(2*pi*(2+0.5)/\n))},{(0.693)*sin(deg(2*pi*(2+0.5)/\n))}) -- ({(0.693)*cos(deg(2*pi*(2+1.5)/\n))},{(0.693)*sin(deg(2*pi*(2+1.5)/\n))});
%%%%%%%%point 3 to point 5%%%%%%%
\draw[thick] ({cos(deg(2*pi*3/\n))},{sin(deg(2*pi*3/\n))}) -- ({(0.693)*cos(deg(2*pi*(3+2-0.5)/\n))},{(0.693)*sin(deg(2*pi*(3+2-0.5)/\n))});
%%%%%point 5 to point 0%%%%%%
\draw[thick] ({(0.693)*cos(deg(2*pi*(5+0.5)/\n))},{(0.693)*sin(deg(2*pi*(5+0.5)/\n))}) -- ({(0.693)*cos(deg(2*pi*(5+1.5)/\n))},{(0.693)*sin(deg(2*pi*(5+1.5)/\n))});
%%%%%%%%point 0 to point 2%%%%%%%%%
\draw[thick] ({(0.693)*cos(deg(2*pi*(0+0.5)/\n))},{(0.693)*sin(deg(2*pi*(0+0.5)/\n))}) -- ({(0.693)*cos(deg(2*pi*(0+1.5)/\n))},{(0.693)*sin(deg(2*pi*(0+1.5)/\n))});
\end{tikzpicture}
}}$
\qquad\qquad
$\vcenter{\hbox{%
\text{(f)}}}$
$\vcenter{\hbox{%
\begin{tikzpicture}
\def\n{8}
\def\e{2.4}%2pi/n<epsilon<3pi/n gives a two-dimensional surface (n=6 gives 2-sphere) 
\foreach \k in {1,2,...,\n} {
   \coordinate (\k) at ({cos(deg(2*pi*\k/\n))},{sin(deg(2*pi*\k/\n))});
   \node at (\k) {$\bullet$};
%   \fill[opacity=0.1] (\k) circle (\e);
   \draw[thick] ({cos(deg(2*pi*\k/\n))},{sin(deg(2*pi*\k/\n))}) -- ({cos(deg(2*pi*(\k+1)/\n))},{sin(deg(2*pi*(\k+1)/\n))});
   \fill[opacity=0.1] ({cos(deg(2*pi*\k/\n))},{sin(deg(2*pi*\k/\n))}) -- ({cos(deg(2*pi*(\k+1)/\n))},{sin(deg(2*pi*(\k+1)/\n))}) -- ({cos(deg(2*pi*(\k+2)/\n))},{sin(deg(2*pi*(\k+2)/\n))}) -- cycle;
}
\foreach \k in {1,3,5,7} {
   \draw[thick] ({cos(deg(2*pi*\k/\n))},{sin(deg(2*pi*\k/\n))}) -- ({cos(deg(2*pi*(\k+2)/\n))},{sin(deg(2*pi*(\k+2)/\n))});
}
\foreach \k in {0,2,4,6} {
   \draw[dashed] ({cos(deg(2*pi*\k/\n))},{sin(deg(2*pi*\k/\n))}) -- ({cos(deg(2*pi*(\k+2)/\n))},{sin(deg(2*pi*(\k+2)/\n))});
   \draw[thick] ({(0.768)*cos(deg(2*pi*(\k+0.5)/\n))},{(0.768)*sin(deg(2*pi*(\k+0.5)/\n))}) -- ({(0.768)*cos(deg(2*pi*(\k+1.5)/\n))},{(0.768)*sin(deg(2*pi*(\k+1.5)/\n))});
}
\end{tikzpicture}
}}$
\caption{The $n^{\text{th}}$ roots of unity form a finite metric space when equipped with the geodesic distance on the unit circle. The Vietoris--Rips complex for various values of $n$ and $\epsilon$ are shown. A subset $\sigma\subseteq X$ is a $k$-simplex in the Vietoris--Rips complex at value $\epsilon$ when $\#\sigma=k+1$ and $\mathrm{diam}(\sigma)<\epsilon$. From top to bottom: (a) $n=8$ and $0<\epsilon<\frac{2\pi}{n}$, (b) $n=8$ and $\frac{2\pi}{n}<\epsilon<\frac{4\pi}{n}$ (a circle), (c) $n=4$  and $\frac{4\pi}{n}<\epsilon<\frac{6\pi}{n}$ (a solid tetrahedron), (d) $n=6$ and $\frac{4\pi}{n}<\epsilon<\frac{6\pi}{n}$ (a hollow octahedron, i.e., a sphere), (e) $n=7$  and $\frac{4\pi}{n}<\epsilon<\frac{6\pi}{n}$ (a M\"obius strip), (f) $n=8$  and $\frac{4\pi}{n}<\epsilon<\frac{6\pi}{n}$ (a cylinder). Dashed lines and shading are occasionally drawn to better visualize the shapes as if they were embedded in $\R^{3}$.}
\label{fig:Hausmannunity}
\end{figure}
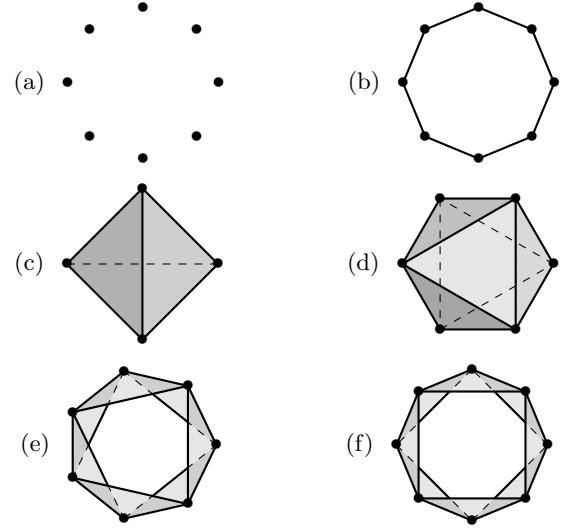

\begin{example}
\label{ex:Hausmannrootsofunity}
For each $n\in\N$, let 
\begin{equation}
X=\{x_{k}=e^{\frac{2\pi i k}{n}}\in\C\;:\;k=0,1,\dots,n-1\}
\end{equation}
be the set of $n^{\text{th}}$ roots of unity in the complex plane $\C$, so that $X$ is a subset of the unit circle $S^1$ in $\C$. Equip $X$ with the metric $d_{X}$ induced from the geodesic distance of $S^1$. 
Figure~\ref{fig:Hausmannunity} depicts the associated Vietoris--Rips complex of $(X,d_{X})$ for several values of $n$ and $\epsilon$. 
\begin{enumerate}
\item When $\epsilon\le\frac{2\pi}{n}$, then $\mathcal{R}^{\epsilon}(X,d_{X})=X$ 
consists only of the points of $X$, i.e., $\mathcal{R}^{\epsilon}(X,d_{X})$ only has $0$-simplices. 
Explicitly, 
\[
\mathcal{R}^{\epsilon}(X,d_{X})=\{x_0,x_1,\dots,x_{n-1}\}.
\]
\item When $n\ge4$ and $\frac{2\pi}{n}<\epsilon\le\frac{4\pi}{n}$, then $\mathcal{R}^{\epsilon}(X,d_{X})$ only contains $0$-simplices and $1$-simplices. 
The set of $0$-simplices equals $\{x_0,x_1,\dots,x_{n-1}\}=X$. 
Meanwhile, the set of $1$-simplices equals
\[
\mathcal{R}^{\epsilon}_{1}(X,d_{X})=\big\{ \{x_0,x_1\}, \{x_1,x_2\}, \dots, \{x_{n-1},x_0\}\big\}
\]
i.e., the set of pairs $\{x_i,x_{i+1}\}$ with $i=0,1,\dots,n-1$, where addition is defined modulo $n$. Hence, 
\[
\mathcal{R}^{\epsilon}(X,d_{X})=\big\{x_0,\dots,x_{n-1},\{x_0,x_1\},\dots,\{x_{n-1},x_0\}\big\}
\]
\item
When $n\ge 4$ and $\frac{4\pi}{n}<\epsilon\le\frac{6\pi}{n}$, then $\mathcal{R}^{\epsilon}(X,d_{X})$ takes on several different possibilities depending on $n$.
    \begin{enumerate}
    \item When $n=4$, $\mathcal{R}^{\epsilon}(X,d_{X})$ is a solid tetrahedron, and is therefore contractible. This is because $d_{X}(x_{i},x_{j})\le\pi<\epsilon$ for all $i,j$. 
    \item When $n=5$, $\mathcal{R}^{\epsilon}(X,d_{X})$ is 4-simplex, and is therefore contractible. This is also because $d_{X}(x_{i},x_{j})\le\frac{4\pi}{5}<\epsilon$ for all $i,j$. 
    \item When $n=6$, $\mathcal{R}^{\epsilon}(X,d_{X})$ is an octahedron, and is therefore homeomorphic to a sphere. This is because $\mathrm{diam}(\{x_{j},x_{j+2},x_{j+4}\})=\frac{2\pi}{3}<\epsilon$, while $d_{X}(x_{j},x_{j+3})=\pi\ge\epsilon$ for all $j$ (subscripts are calculated using mod $6$ arithmetic). In other words, opposing vertices are never contained in the same simplex for these values of $\epsilon$. 
    \item When $n=7$ (or any odd number greater than $6$), $\mathcal{R}^{\epsilon}(X,d_{X})$ is a M\"obius strip.
    \item When $n=8$ (or any even number great than $7$), $\mathcal{R}^{\epsilon}(X,d_{X})$ is a cylinder.
    \end{enumerate}
\item
When $n\ge6$ and $\epsilon>\frac{6\pi}{n}$, the simplicial complex $\mathcal{R}^{\epsilon}(X,d_{X})$ begins to admit 3-dimensional simplices, so we avoid attempting to draw these cases. However, we will revisit their topologies later. 
\end{enumerate}
\end{example}

Going back to the general case of arbitrary abstract simplicial complexes, we define morphisms between them.
\begin{definition}
Given two abstract simplicial complexes $\mathcal{S}$ and $\mathcal{T}$, a \define{simplicial map} $\phi$ from $\mathcal{T}$ to $\mathcal{S}$, denoted as $\phi:\mathcal{T}\to\mathcal{S}$, is a function for which there exists a unique function $\phi_{0}:\mathcal{T}_{0}\to\mathcal{S}_{0}$ of $0$-simplices satisfying the condition that whenever $\tau=\{x_0,x_1,\dots,x_{k}\}$ is a $k$-simplex in $\mathcal{T}$, then 
\begin{equation}
\phi(\tau):=\big\{\phi_0(x_0),\phi_0(x_1),\dots,\phi_0(x_k)\big\}
\end{equation}
is a simplex in $\mathcal{S}$ (the result need not be a $k$-simplex). 
\end{definition}

To keep notation simple, we will drop the $0$ subscript on $\phi$ when restricting to $0$-simplices. In particular, given an abstract simplicial complex $\mathcal{S}$, a subset $\mathcal{T}\subseteq\mathcal{S}$ is a \define{subcomplex} iff $\mathcal{T}$ is itself an abstract simplicial complex. Such a subcomplex may be expressed in terms of an inclusion map $\mathcal{T}\hookrightarrow\mathcal{S}$. The collection of abstract simplicial complexes and simplicial maps defines a category, which we denote by $\mathbf{SimpComp}$. We will look at simplicial maps coming from two examples. 

\begin{example}
First, if $f:(X,d_{X})\to (Y,d_{Y})$ is a distance nonincreasing function, then this induces a simplicial map $\phi:\mathcal{R}^{\epsilon}(X,d_{X})\to\mathcal{R}^{\epsilon}(Y,d_{Y})$, which is uniquely determined by the function $f$ since $X$ and $Y$ are identified with the vertices of $\mathcal{R}^{\epsilon}(X,d_{X})$ and $\mathcal{R}^{\epsilon}(Y,d_{Y})$, respectively. Thus, if $\tau=\{x_0,x_1,\dots,x_{k}\}$ is a $k$-simplex in $\mathcal{R}^{\epsilon}(X,d_{X})$, which means that $\mathrm{diam}(\tau)<\epsilon$, then $\{f(x_0),f(x_1),\dots,f(x_{k})\}$ is also a $k$-simplex in $\mathcal{R}^{\epsilon}(Y,d_{Y})$ because $\mathrm{diam}(f(\tau))\le \mathrm{diam}(\tau)<\epsilon$ due to the distance nonincreasing assumption on $f$. In this way, we arrive at a functor 
\begin{align}
\mathbf{FinMetSpace}&\xrightarrow{\mathcal{R}^{\epsilon}}\mathbf{SimpComp} \nonumber\\
(X,d_{X})&\xmapsto{\;\;\;}\mathcal{R}^{\epsilon}(X,d_{X})\nonumber\\
\left((X,d_{X})\xrightarrow{f}(Y,d_{Y})\right)&\xmapsto{\;\;\;}\left(\mathcal{R}^{\epsilon}(X,d_{X}) \xrightarrow{\phi}\mathcal{R}^{\epsilon}(Y,d_{Y})\right).
\end{align}
\end{example}

\begin{example}
\label{ex:simplicialfiltration}
As another example of a simplicial map, if $(X,d_{X})$ is a finite metric space and $0\le\epsilon\le\epsilon'$ is an ordered pair of positive numbers, $\mathcal{R}^{\epsilon}(X,d_{X})$ is a subcomplex of $\mathcal{R}^{\epsilon'}(X,d_{X})$. Such an inclusion will be written as $\mathcal{R}^{\epsilon}(X,d_{X})\hookrightarrow\mathcal{R}^{\epsilon'}(X,d_{X})$, and it defines a simplicial map. We can see such an inclusion in the examples in Figure~\ref{fig:Hausmannunity}, namely when $n=8$ and $\epsilon,\epsilon'$ satisfy $\frac{2\pi}{n}<\epsilon\le\frac{4\pi}{n}<\epsilon'\le\frac{6\pi}{n}$, specifically subfigures (b) and (f). 
\end{example}

The family of inclusions from Example~\ref{ex:simplicialfiltration} provides an example of a \define{filtered simplicial complex}, which consists of a family of simplicial complexes $\{\mathcal{S}^{\epsilon}\}$ indexed by $\epsilon\in[0,\infty)$ (or any partially ordered set) together with a family of simplicial inclusion maps $i^{(\epsilon',\epsilon)}:\mathcal{S}^{\epsilon}\hookrightarrow\mathcal{S}^{\epsilon'}$ whenever $\epsilon\le\epsilon'$ such that $i^{(\epsilon,\epsilon)}=\id$ for all $\epsilon$ and $i^{(\epsilon'',\epsilon')}\circ i^{(\epsilon',\epsilon)}=i^{(\epsilon'',\epsilon)}$ for all $\epsilon,\epsilon',\epsilon''$ satisfying $\epsilon\le\epsilon'\le\epsilon''$. It is fruitful to abstract this concept, as it will appear in many different contexts throughout~\cite{Carlsson2014}. 

\begin{definition}
\label{defn:persistenceobject}
Let $(P,\le)$ be a partially ordered set, called a \define{poset}. 
Let $\mathbf{P}$ be the partially ordered set $(P,\le)$ viewed as a category; namely, the objects of $\mathbf{P}$ are elements of $P$. The set of morphisms in $\mathbf{P}$ from an element $p\in P$ to $q\in P$ is a single element set if $p\le q$ and is empty otherwise. 
Let $\mathbf{C}$ be a category. A \define{persistence object} in $\mathbf{C}$ is a functor $F:\boldsymbol{[0,\infty)}\to \mathbf{C}$, where the category $\boldsymbol{[0,\infty)}$ is the poset $[0,\infty)$, with the usual ordering, viewed as a category. The value of $F$ on $\epsilon$ is often denoted by $F^{\epsilon}$, rather than by $F(\epsilon)$. 
\end{definition}

Thus, Example~\ref{ex:simplicialfiltration} illustrates that a filtered simplicial complex is an example of a persistence object $\mathcal{R}(X,d_{X}):\boldsymbol{[0,\infty)}\to \mathbf{C}$, where $\mathbf{C}=\mathbf{SimpComp}$. 

\begin{remark}
The collection $\{\mathcal{R}^{\epsilon}\}_{\epsilon\in[0,\infty)}$ of functors defines a single functor 
\begin{equation}
\mathbf{FinMetSpace}\times\boldsymbol{[0,\infty)}\xrightarrow{\mathcal{R}}\mathbf{SimpComp}.
\end{equation}
Thus, the collection of functors $\{\mathcal{R}^{\epsilon}\}_{\epsilon\in[0,\infty)}$ forms a persistence object in the functor category $\mathbf{SimpComp}^{\mathbf{FinMetSpace}}$~\cite{mac2013categories,riehl2016}. This technical, but unnecessary, remark highlights that the relationships between the different $\mathcal{R}^{\epsilon}$ are described as natural transformations.
\end{remark}

Abstract simplicial complexes are combinatorial objects that are often constructed from geometric objects, and simplicial homology turns these combinatorial objects into vector spaces. The latter are algebraic constructions that can be used to calculate invariants of the simplicial complexes from which they were constructed. The precedure of obtaining such algebraic invariants is achieved by taking certain quotients of vector spaces from a \emph{chain complex}, a sequence of vector spaces connected to each other by a \emph{boundary map} that lowers dimension, analogous to the boundary map appearing in Stokes' theorem from multivariable calculus~\cite{BottTu82,SpivakCalc65,Munkres91}. As such, we first review these concepts, referring to our main examples coming from topological data analysis~\cite{MemoliSinghal19,Halmos1958,hatcherbook2002,carlsson2021topological,Munkres84,MacLaneHomology}. 
In what follows, we fix a general ground field $\mathbb{F}$ over which all vector spaces are defined, unless otherwise stated.  
\begin{definition}
A \define{chain complex} $(C_{\bullet},\partial_{\bullet})$ is a collection of vector spaces $\{C_{k}\}_{k\in \Z}$ indexed by the integers together with a collection of linear transformations $\partial_{n}:C_{n}\to C_{n-1}$, called \define{boundary maps}, satisfying the condition that $\partial_{n}\circ\partial_{n+1}=0$ for all $n\in\Z$. Elements of $C_{n}$ are called \define{$n$-chains}.
\end{definition}

It is common to visualize a chain complex $(C_{\bullet},\partial_{\bullet})$ diagrammatically as
\[
\cdots\xrightarrow{\partial_{n+2}}C_{n+1}\xrightarrow{\partial_{n+1}}C_{n}\xrightarrow{\partial_{n}}C_{n-1}\xrightarrow{\partial_{n-1}}C_{n-2}\cdots.
\]

Our first general example will involve taking an abstract simplicial complex $\mathcal{S}$ and creating a chain complex from it. However, the boundary maps will involve signs, just as how Stokes' theorem involves signs coming from the different orientations that a manifold and its boundary can have. Therefore, we should first define orientations on simplices~\cite{DummitFoote,Munkres84}. 

\begin{definition}
Given an $n$-simplex $\sigma$ in $\mathcal{S}$ with $n\ge0$, an \define{ordering} of the elements of $\sigma$ is a bijective function $\upsilon:\{0,1,\dots,n\}\to\sigma$. Two such orderings $\upsilon,\vartheta$ are said to define the same \define{orientation equivalence class} iff there exists an \emph{even} permutation $\pi:\{0,1,\dots,n\}\to\{0,1,\dots,n\}$ such that $\upsilon=\vartheta\circ\pi$. An \define{orientation} on $\sigma$ is an orientation equivalence class of orderings on $\sigma$. 
\end{definition}

Each $n$-simplex has exactly two orientations when $n>0$ and one orientation when $n=0$. For example, if $\sigma=\{x_0,\dots,x_n\}$, one such ordering $\upsilon$ is given by $\upsilon(k)=x_{k}$ for all $k\in\{0,1,\dots,n\}$, and the associated pair $\big(\{x_0,\dots,x_n\},\upsilon\big)$ will be written more concisely as $[x_0,\dots,x_n]$. Meanwhile, the same simplex but with the opposite orientation will be denoted by $-[x_0,\dots,x_n]$. 
By abuse of notation, we use the notation $\sigma$ to represent an oriented simplex and $-\sigma$ to represent the same simplex with the opposite orientation. The set of oriented $n$-simplices will be denoted by $\mathcal{S}_{n}^{\circlearrowleft}$, so that $\mathcal{S}_{0}^{\circlearrowleft}\cong\mathcal{S}_{0}$ and $\mathcal{S}_{n}^{\circlearrowleft}\cong\mathcal{S}_{n}\coprod\mathcal{S}_{n}$, the disjoint union of $\mathcal{S}_{n}$ with itself, for $n>0$. Hence, given a simplicial complex $\mathcal{S}=\bigcup_{n\in\N\cup\{0\}}\mathcal{S}_{n}$, we write 
\begin{equation}
\mathcal{S}^{\circlearrowleft}=\bigcup_{n\in\N}\mathcal{S}_{n}^{\circlearrowleft}
\end{equation} 
for the set of \emph{oriented} simplices from $\mathcal{S}$. 

Now, given any abstract simplicial complex $\mathcal{S}$, for each $n\in\N\cup\{0\}$, let 
$
\mathbb{F}^{\mathcal{S}_{n}^{\circlearrowleft}}_{\mathrm{fs}}
$
denote the free vector space generated by the set of oriented $n$-simplices from $\mathcal{S}$. Explicitly, this is the vector space of finitely-supported functions $f:\mathcal{S}_{n}^{\circlearrowleft}\to\mathbb{F}$, i.e., $f(\sigma)=0$ for all but finitely many $\sigma\in\mathcal{S}_{n}^{\circlearrowleft}$. Let $C_{n}(\mathcal{S})$ be the subspace of $\mathbb{F}^{\mathcal{S}_{n}^{\circlearrowleft}}_{\mathrm{fs}}$ given by 
\begin{equation}
C_{n}(\mathcal{S})
=\Big\{f\in\mathbb{F}^{\mathcal{S}_{n}^{\circlearrowleft}}_{\mathrm{fs}}\,:\,f(-\sigma)=-f(\sigma)\;\;\forall\sigma\in\mathcal{S}_{n}^{\circlearrowleft}\Big\}.
\end{equation} 

Associated with any oriented $n$-simplex $\sigma$, there is a canonical function $f_{\sigma}:\mathcal{S}_{n}^{\circlearrowleft}\to\mathbb{F}$, called the \define{elementary chain} associated with $\sigma$, satisfying $f_{\sigma}(\sigma)=1$, $f_{\sigma}(-\sigma)=-1$, and $f_{\sigma}(\tau)=0$ for all $\tau\in\mathcal{S}_{n}^{\circlearrowleft}\setminus\{\sigma,-\sigma\}$. Notice, then, that $f_{-\sigma}=-f_{\sigma}$, which relates the two elementary chains $f_{-\sigma}$ and $f_{\sigma}$ as negatives of each other. Using Dirac bra-ket notation, we write the elementary chain $f_{\sigma}$ associated with the oriented simplex $\sigma$ as $\ket{\sigma}$, so that the elementary chain associated with $-\sigma$ satisfies $\ket{-\sigma}=-\ket{\sigma}$. Thus, every abstract $n$-simplex determines a one-dimensional subspace of $C_{n}(\mathcal{S})$, while the choice of an oriented $n$-simplex determines a basis for this subspace given by the elementary chain associated with that oriented simplex. Putting this together, $C_{n}(\mathcal{S})$ is the free vector space generated by the elements of $\mathcal{S}_{n}$, while a basis for $C_{n}(\mathcal{S})$ is obtained from $\mathcal{S}$ by choosing an orientation for each simplex. Note that $C_{0}(\mathcal{S})$ has a canonical basis, since each $0$-simplex has a canonical orientation. 

We next define the boundary maps $\partial_{n}:C_{n}(\mathcal{S})\to C_{n-1}(\mathcal{S})$ for $n\in\N$ as follows. For an oriented $n$-simplex $\sigma$, written as $\sigma=[x_0,x_1,\dots,x_{n}]$ for a distinct set of $0$-simplices $x_{0},x_{1},\dots,x_{n}\in\mathcal{S}_{0}$, the boundary of $\ket{\sigma}$ is defined as
\begin{equation}
\partial_{n}\big(\ket{\sigma}\big):=\sum_{k=0}^{n}(-1)^{k}\big|[x_0,\dots,\widehat{x_{k}},\dots,x_{n}]\big\rangle,
\end{equation}
where $\{x_0,\dots,\widehat{x_{k}},\dots,x_{n}\}$ denotes the $(n-1)$-simplex ${\sigma\setminus\{x_{k}\}}$, with associated oriented $(n-1)$-simplex represented as $[x_0,\dots,x_{k-1},x_{k+1},\dots,x_{n}]$. One then extends $\partial_{n}$ uniquely by linearity. One can show that $\big(C_{\bullet}(\mathcal{S}),\partial_{\bullet}\big)$ is a chain complex, where $C_{n}(\mathcal{S})=0$ for all integers $n<0$ and for all integers $n$ strictly greater than the dimension of the simplicial complex. 

As a special case, since $\mathcal{R}^{\epsilon}(X,d_{X})$ is an abstract simplicial complex for every $\epsilon>0$ and finite metric space $(X,d_X)$, we obtain the \define{Vietoris--Rips chain complex}, which we denote by $\big(C^{\epsilon}_{\bullet}(X,d_{X}),\partial_{\bullet}\big)$ for brevity.

\begin{definition}
Given two chain complexes $(C_{\bullet},\partial_{\bullet})$ and $(D_{\bullet},\partial_{\bullet})$ a \define{chain map} $f_{\bullet}$ from $(C_{\bullet},\partial_{\bullet})$ to $(D_{\bullet},\partial_{\bullet})$ is a collection of linear transformations $\{f_{n}:C_{n}\to D_{n}\}$ such that 
\begin{equation}
\label{eqn:chainmap}
\partial_{n}\circ f_{n+1}=f_{n}\circ\partial_{n}
\end{equation}
for all $n$. 
\end{definition}

We visualize~\eqref{eqn:chainmap} diagrammatically as the commutative diagram
\begin{equation}
\xy 0;/r.25pc/:
(-10,7.5)*+{C_{n}}="Cnp1";
(10,7.5)*+{C_{n-1}}="Cn";
(-10,-7.5)*+{D_{n}}="Dnp1";
(10,-7.5)*+{D_{n-1}}="Dn";
{\ar"Cnp1";"Cn"^{\partial_{n}}};
{\ar"Cnp1";"Dnp1"^{f_{n}}};
{\ar"Dnp1";"Dn"^{\partial_{n}}};
{\ar"Cn";"Dn"^{f_{n-1}}};
(1.5,1.5)*{}="tr";
(-1.5,-1.5)*{}="bl";
{\ar@{=} "tr";"bl"};
\endxy
.
\end{equation}
We avoided adding extra notation to the two sets of boundary maps and expect that the context will distinguish between the two. The collection of chain complexes over a field $\mathbb{F}$ and chain maps forms a category, which we denote by $\mathbf{ChainComp}$ (the field being understood from the context). 

One can show that if $\phi:\mathcal{T}\to\mathcal{S}$ is a simplicial map of abstract simplicial complexes, then $\phi$ induces a chain map $f_{\bullet}:\big(C_{\bullet}(\mathcal{T}),\partial_{\bullet}\big)\to\big(C_{\bullet}(\mathcal{S}),\partial_{\bullet}\big)$ between the associated chain complexes that is uniquely determined by linearity and the condition that for any oriented $n$-simplex $\tau=[x_0,x_1,\dots,x_n]$ from $\mathcal{T}$,  
\begin{equation}
f_{n}\big(\ket{\tau}\big)=
\big|[\phi(x_0),\phi(x_1),\dots,\phi(x_n)]\big\rangle
\end{equation}
if the $\phi(x_k)$ are all distinct, and $f_{n}\big(\ket{\tau}\big)=0$ otherwise. 
One can show that the assignment 
\begin{align}
\mathbf{SimpComp}&\xrightarrow{C_{\bullet}}\mathbf{ChainComp} \nonumber\\
\mathcal{S}&\xmapsto{\;\;\;}\big(C_{\bullet}(\mathcal{S}),\partial_{\bullet}\big) \nonumber\\
\left(\mathcal{T}\xrightarrow{\phi}\mathcal{S}\right)&\xmapsto{\;\;\;}\big(C_{\bullet}(\mathcal{T}),\partial_{\bullet}\big) \xrightarrow{f_{\bullet}}\big(C_{\bullet}(\mathcal{S}),\partial_{\bullet}\big) 
\end{align}
defines a functor. 

In the case of the Vietoris--Rips complex associated with a finite metric space $(X,d_{X})$, we saw that if $0\le\epsilon\le\epsilon'$, then we get an inclusion $\mathcal{R}^{\epsilon}(X,d_{X})\hookrightarrow\mathcal{R}^{\epsilon'}(X,d_{X})$ of abstract simplicial complexes, which is a simplicial map. As such, the above prescription leads to a chain map 
\begin{equation}
\label{eqn:persistencechainmaps}
f^{(\epsilon',\epsilon)}_{\bullet}:\big(C^{\epsilon}_{\bullet}(X,d_{X}),\partial_{\bullet}\big)\to\big(C^{\epsilon'}_{\bullet}(X,d_{X}),\partial_{\bullet}\big).
\end{equation}
The collection of such maps provides an example of a persistence chain complex~\cite{edelsbrunner2002topological,ZoCa05}. 

\begin{definition}
A \define{persistence chain complex} is a functor $\boldsymbol{[0,\infty)}\to\mathbf{ChainComp}$. In detail, a persistence chain complex is a collection of chain complexes $(C^{\epsilon}_{\bullet},\partial_{\bullet})$ for each $\epsilon\in[0,\infty)$ together with a family of chain maps $f^{(\epsilon',\epsilon)}_{\bullet}:(C^{\epsilon}_{\bullet},\partial_{\bullet})\to(C^{\epsilon'}_{\bullet},\partial_{\bullet})$ for all $\epsilon,\epsilon'\in[0,\infty)$ satisfying $\epsilon\le\epsilon'$. Moreover, these chain maps satisfy $f_{\bullet}^{(\epsilon,\epsilon)}=\id$ for all $\epsilon\in[0,\infty)$ and $f_{\bullet}^{(\epsilon'',\epsilon')}\circ f_{\bullet}^{(\epsilon',\epsilon)}=f_{\bullet}^{(\epsilon'',\epsilon)}$ for all $\epsilon,\epsilon',\epsilon''\in[0,\infty)$ satisfying $\epsilon\le\epsilon'\le\epsilon''$.
\end{definition}

\begin{example}
\label{ex:Hausmann4case}
We now revisit Example~\ref{ex:Hausmannrootsofunity} with $n=6$ in order to build the full Vietoris--Rips persistence chain complex. We first label the vertices and specify orientations on all the simplices in order to specify a basis for the vector spaces in the Vietoris--Rips chain complex. Since additional simplices appear for larger values of $\epsilon$, the ordered basis associated with these higher values of $\epsilon$ adjoin these new basis elements in the order that they appear. For $\epsilon=b$ in the range $\frac{2\pi}{n}<b\le\frac{4\pi}{n}$, we have 
\begin{center}
%%%%%%%%%%%%%%%%%epsilon=b%%%%%%%%%%%%%%
\scalebox{0.8}{
\begin{tikzpicture}
\def\n{6}
\def\e{2.4}
\foreach \k in {1,2,...,\n} {
   \coordinate (\k) at ({cos(deg(2*pi*\k/\n))},{sin(deg(2*pi*\k/\n))});
   \node at (\k) {$\bullet$};
   \draw[thick] ({cos(deg(2*pi*\k/\n))},{sin(deg(2*pi*\k/\n))}) -- ({cos(deg(2*pi*(\k+1)/\n))},{sin(deg(2*pi*(\k+1)/\n))});
}
%vertices
\node at (1.35,0) {$x_0$};
\node at (0.65,1.15) {$x_1$};
\node at (-0.65,1.15) {$x_2$};
\node at (-1.35,0) {$x_3$};
\node at (-0.65,-1.15) {$x_4$};
\node at (0.65,-1.15) {$x_5$};
%edges
\node at (0.95,0.5) {\rotatebox{-60}{$e_{01}$}};
\node at (0,1.05) {$e_{12}$};
\node at (-0.95,0.5) {\rotatebox{60}{$e_{23}$}};
\node at (-0.95,-0.5) {\rotatebox{120}{$e_{34}$}};
\node at (0,-1.05) {\rotatebox{180}{$e_{45}$}};
\node at (0.95,-0.5) {\rotatebox{240}{$e_{50}$}};
\end{tikzpicture}
}
\end{center}
With this labeling of vertices and edges, we write the ordered basis for $0$-simplices as 
\begin{equation}
\mathbf{e}^{(0)}_{1}=\ket{x_0}, \; \mathbf{e}^{(0)}_{2}=\ket{x_1}, \; 
\dots \; , \;
\mathbf{e}^{(0)}_{6}=\ket{x_5}
\end{equation}
and for $1$-simplices as
\begin{equation}
\mathbf{e}^{(1)}_{1}=\ket{e_{01}},\;
\mathbf{e}^{(1)}_{2}=\ket{e_{12}},\;
\dots\; , \;
\mathbf{e}^{(1)}_{6}=\ket{e_{50}}.\;
\end{equation}
The ordering of an edge $e_{23}$, for example, is from $x_2$ to $x_3$, where we use the shorthand notation $\ket{e_{23}}=\ket{[x_2,x_3]}$ and similarly for the other edges. In other words, if ${\partial^{b}_{1}:\mathbb{F}^{6}\to\mathbb{F}^{6}}$ denotes the boundary operator (from $2$-simplices to $1$-simplices at this value of $\epsilon$), then $\partial^{b}_{1}\ket{e_{23}}=\ket{x_3}-\ket{x_2}$. Based on the above ordering of the bases, $\partial^{b}_{1}$ is given in matrix form by 
\begin{equation}
\label{eq:Hauspb1}
\partial^{b}_{1}=\begin{bmatrix}-1&0&0&0&0&1\\1&-1&0&0&0&0\\0&1&-1&0&0&0\\0&0&1&-1&0&0\\0&0&0&1&-1&0\\0&0&0&0&1&-1\end{bmatrix}.
\end{equation}
Next, for $\epsilon=c$ in the range $\frac{4\pi}{n}<c\le\frac{6\pi}{n}$, we have 
\begin{center}
%%%%%%%%%%%%%%%%%epsilon=c%%%%%%%%%%%%%%
\scalebox{0.8}{
\begin{tikzpicture}[scale=2.0]
\def\n{6}
\def\e{2.4}
\foreach \k in {1,2,...,\n} {
   \coordinate (\k) at ({cos(deg(2*pi*\k/\n))},{sin(deg(2*pi*\k/\n))});
   \node at (\k) {$\bullet$};
   \draw[thick] ({cos(deg(2*pi*\k/\n))},{sin(deg(2*pi*\k/\n))}) -- ({cos(deg(2*pi*(\k+1)/\n))},{sin(deg(2*pi*(\k+1)/\n))});
   \draw[dashed] ({cos(deg(2*pi*\k/\n))},{sin(deg(2*pi*\k/\n))}) -- ({cos(deg(2*pi*(\k+2)/\n))},{sin(deg(2*pi*(\k+2)/\n))});
}
\foreach \k in {1,3,5} {
   \draw[thick] ({cos(deg(2*pi*\k/\n))},{sin(deg(2*pi*\k/\n))}) -- ({cos(deg(2*pi*(\k+2)/\n))},{sin(deg(2*pi*(\k+2)/\n))});
}
\fill[opacity=0.1] ({cos(deg(2*pi*1/\n))},{sin(deg(2*pi*1/\n))}) -- ({cos(deg(2*pi*3/\n))},{sin(deg(2*pi*3/\n))}) -- ({cos(deg(2*pi*5/\n))},{sin(deg(2*pi*5/\n))}) -- cycle;
\fill[opacity=0.25] ({cos(deg(2*pi*1/\n))},{sin(deg(2*pi*1/\n))}) -- ({cos(deg(2*pi*2/\n))},{sin(deg(2*pi*2/\n))}) -- ({cos(deg(2*pi*3/\n))},{sin(deg(2*pi*3/\n))}) -- cycle;
\fill[opacity=0.15] ({cos(deg(2*pi*1/\n))},{sin(deg(2*pi*1/\n))}) -- ({cos(deg(2*pi*5/\n))},{sin(deg(2*pi*5/\n))}) -- ({cos(deg(2*pi*0/\n))},{sin(deg(2*pi*0/\n))}) -- cycle;
\fill[opacity=0.35] ({cos(deg(2*pi*3/\n))},{sin(deg(2*pi*3/\n))}) -- ({cos(deg(2*pi*4/\n))},{sin(deg(2*pi*4/\n))}) -- ({cos(deg(2*pi*5/\n))},{sin(deg(2*pi*5/\n))}) -- cycle;
%nodes
\node at (1.20,0) {$x_0$};
\node at (0.60,1.025) {$x_1$};
\node at (-0.60,1.025) {$x_2$};
\node at (-1.20,0) {$x_3$};
\node at (-0.60,-1.025) {$x_4$};
\node at (0.60,-1.025) {$x_5$};
%new edges
\node at (0.21,0.33) {\rotatebox{-30}{$e_{02}$}};
\node at (-0.21,0.33) {\rotatebox{30}{$e_{13}$}};
\node at (-0.4,0) {\rotatebox{90}{$e_{24}$}};
\node at (-0.21,-0.33) {\rotatebox{150}{$e_{35}$}};
\node at (0.21,-0.33) {\rotatebox{210}{$e_{40}$}};
\node at (0.4,0) {\rotatebox{-90}{$e_{51}$}};
%faces (oriented so that outward normal)
\node at (0,0) {$\sigma_{135}$};
\node at (0.70,0) {\rotatebox{-90}{$\sigma_{015}$}};
\node at (0.34,0.59) {\rotatebox{-30}{$\sigma_{102}$}};
\node at (-0.34,0.59) {\rotatebox{30}{$\sigma_{123}$}};
\node at (-0.7,0) {\rotatebox{90}{$\sigma_{243}$}};
\node at (-0.34,-0.59) {\rotatebox{150}{$\sigma_{345}$}};
\node at (0.34,-0.59) {\rotatebox{210}{$\sigma_{054}$}};
\end{tikzpicture}
}
\end{center}
where $\sigma_{204}$ is not shown, as it corresponds to the dashed $2$-simplex. An oriented $2$-simplex $\sigma_{ijk}$ is shorthand notation for $\ket{\sigma_{ijk}}=\ket{[x_i,x_j,x_k]}$. 
In addition to the basis already provided for $1$-simplices, the vector space of $1$-simplices increases in dimension due to the addition of six additional edges. The orientation of these $1$-simplices is chosen to agree with the ordering of the subscripts as before. Hence, we adjoin vectors 
\begin{align}
\mathbf{e}^{(1)}_{7}&=\ket{e_{02}},\;
\mathbf{e}^{(1)}_{8}=\ket{e_{13}},\;
\mathbf{e}^{(1)}_{9}=\ket{e_{24}},\;\nonumber\\
\mathbf{e}^{(1)}_{10}&=\ket{e_{35}},\;
\mathbf{e}^{(1)}_{11}=\ket{e_{40}},\;
\mathbf{e}^{(1)}_{12}=\ket{e_{51}}\;
\end{align}
for $1$-simplices to get the vector space $\mathbb{F}^{12}$ for $\epsilon=c$. In this case, the boundary operator from $1$-simplices to $0$-simplices is of the form $\partial^{c}_{1}:\mathbb{F}^{12}\to\mathbb{F}^{6}$ and is given in matrix form by
\begin{equation}
\partial^{c}_{1}=
\begin{bmatrix}
\phantom{yes..}\partial^{b}_{1}\phantom{yes} &
    \begin{matrix}
    -1&0&0&0&-1&0\\
    0&-1&0&0&0&-1\\
    1&0&-1&0&0&0\\
    0&1&0&-1&0&0\\
    0&0&1&0&1&0\\
    0&0&0&1&0&1
    \end{matrix}
\end{bmatrix}.
\end{equation}
Additionally, there are now eight faces (oriented 2-simplices), written as $\sigma_{ijk}$ with the ordering also chosen based on the order of the indices. In this particular case, the ordering is chosen so as to agree with an outward pointing normal by the right-hand-rule from the above picture. Moreover, we write the ordered basis for $2$-simplices as  
\begin{align}
\mathbf{e}^{(2)}_{1}&=\ket{\sigma_{102}},\,
\mathbf{e}^{(2)}_{2}=\ket{\sigma_{123}},\,
\mathbf{e}^{(2)}_{3}=\ket{\sigma_{243}},\,
\mathbf{e}^{(2)}_{4}=\ket{\sigma_{345}}, \nonumber\\
\mathbf{e}^{(2)}_{5}&=\ket{\sigma_{054}},\,
\mathbf{e}^{(2)}_{6}=\ket{\sigma_{015}},\,
\mathbf{e}^{(2)}_{7}=\ket{\sigma_{135}},\,
\mathbf{e}^{(2)}_{8}=\ket{\sigma_{204}}.
\end{align}
The boundary operator from $2$-simplices to $1$-simplices at value $\epsilon=c$ is an operator of the form $\partial^{c}_{2}:\mathbb{F}^{8}\to\mathbb{F}^{12}$. For example, its value on $\ket{\sigma_{123}}$ is $\partial^{c}_{2}\ket{\sigma_{123}}=\ket{e_{12}}+\ket{e_{23}}-\ket{e_{13}}=\mathbf{e}^{(1)}_{2}+\mathbf{e}^{(1)}_{3}-\mathbf{e}^{(1)}_{8}$, where the minus sign comes from the fact that the orientation of the boundary is opposite on that edge to the orientation we chose for our basis element. The matrix representation for $\partial^{c}_{2}$ in our chosen bases is 
\begin{equation}
\partial^{c}_{2}=
\begin{bmatrix}
-1&0&0&0&0&1&0&0\\
-1&1&0&0&0&0&0&0\\
0&1&-1&0&0&0&0&0\\
0&0&-1&1&0&0&0&0\\
0&0&0&1&-1&0&0&0\\
0&0&0&0&1&1&0&0\\
1&0&0&0&0&0&0&-1\\
0&-1&0&0&0&0&1&0\\
0&0&1&0&0&0&0&-1\\
0&0&0&-1&0&0&1&0\\
0&0&0&0&1&0&0&-1\\
0&0&0&0&0&-1&1&0
\end{bmatrix}
.
\end{equation}
Finally, for $\epsilon=d$ in the range $\frac{6\pi}{n}<c\le\frac{8\pi}{n}$ (and technically greater than these values as well), we have 
\begin{center}
%%%%%%%%%%%%%%%%%epsilon=c%%%%%%%%%%%%%%
\scalebox{0.8}{
\begin{tikzpicture}[scale=1.75]
\def\n{6}
\def\e{2.4}
\foreach \k in {1,2,...,\n} {
   \coordinate (\k) at ({cos(deg(2*pi*\k/\n))},{sin(deg(2*pi*\k/\n))});
   \node at (\k) {$\bullet$};
       \draw[thick] ({cos(deg(2*pi*\k/\n))},{sin(deg(2*pi*\k/\n))}) -- ({cos(deg(2*pi*(\k+1)/\n))},{sin(deg(2*pi*(\k+1)/\n))});
       \draw[thick] ({cos(deg(2*pi*\k/\n))},{sin(deg(2*pi*\k/\n))}) -- ({cos(deg(2*pi*(\k+2)/\n))},{sin(deg(2*pi*(\k+2)/\n))});
       \draw[thick] ({cos(deg(2*pi*\k/\n))},{sin(deg(2*pi*\k/\n))}) -- ({cos(deg(2*pi*(\k+3)/\n))},{sin(deg(2*pi*(\k+3)/\n))});
}
\fill[opacity=0.1] ({cos(deg(2*pi*1/\n))},{sin(deg(2*pi*1/\n))}) -- ({cos(deg(2*pi*2/\n))},{sin(deg(2*pi*2/\n))}) -- ({cos(deg(2*pi*3/\n))},{sin(deg(2*pi*3/\n))}) -- ({cos(deg(2*pi*4/\n))},{sin(deg(2*pi*4/\n))}) -- ({cos(deg(2*pi*5/\n))},{sin(deg(2*pi*5/\n))}) -- (1,0) -- cycle;
%nodes
\node at (1.20,0) {$x_0$};
\node at (0.60,1.025) {$x_1$};
\node at (-0.60,1.025) {$x_2$};
\node at (-1.20,0) {$x_3$};
\node at (-0.60,-1.025) {$x_4$};
\node at (0.60,-1.025) {$x_5$};
%new edges
\node at (0.11,0.35) {\rotatebox{60}{$e_{14}$}};
\node at (-0.26,0.26) {\rotatebox{120}{$e_{13}$}};
\node at (0.31,0.08) {$e_{03}$};
%faces (oriented so that outward normal)
\end{tikzpicture}
}
\end{center}
We have no longer shaded different faces nor dashed any edges because the simplicial complex is that of a single $5$-simplex. The new edges are shown and this means we adjoin vectors 
\begin{equation}
\mathbf{e}^{(1)}_{13}=\ket{e_{03}},\;
\mathbf{e}^{(1)}_{14}=\ket{e_{14}},\;
\mathbf{e}^{(1)}_{15}=\ket{e_{25}}
\end{equation}
to the space of $1$-simplices. The number of faces added is $12$, and so we will not write them down. Moreover, we will not explicitly write the higher-dimensional simplices. This gives us enough information to later calculate persistent homology. Indeed, we now have a full picture of the persistence chain complexes from the Vietoris--Rips complex having selected parameter values $a,b,c,d$ with $0<a<\frac{2\pi}{n}<b<\frac{4\pi}{n}<c<\frac{6\pi}{n}<d<\frac{8\pi}{n}$. We can visualize this by drawing the Vietoris--Rips chain complex at each value of $\epsilon$ (including the boundary maps) in the vertical direction together with its chain maps as $\epsilon$ varies in the horizontal direction as follows. 
\begin{center}
\begin{tikzpicture}
\draw[->] (0,0) -- ({2*pi+0.5},0) node[right]{$\epsilon$}; 
% \foreach \n in {1,2,3,4}{
%     \draw ({2*pi*\n/4)},0.1) -- ({2*pi*\n/4)},-0.1) node[below]{$\frac{2\pi}{\n}$};
%     }
    \draw ({0},0.1) -- ({0},-0.1) node[below]{$0$};
    \draw ({2*pi/8},0.1) -- ({2*pi/8},-0.1) node[below]{$a$};
    \draw ({2*pi/4},0.1) -- ({2*pi/4},-0.1) node[below]{$\frac{2\pi}{6}$};
    \draw ({3*pi/4},0.1) -- ({3*pi/4},-0.1) node[below]{$b$};
    \draw ({4*pi/4},0.1) -- ({4*pi/4},-0.1) node[below]{$\frac{4\pi}{6}$};
    \draw ({5*pi/4},0.1) -- ({5*pi/4},-0.1) node[below]{$c$};
    \draw ({6*pi/4},0.1) -- ({6*pi/4},-0.1) node[below]{$\frac{6\pi}{6}$};
    \draw ({7*pi/4},0.1) -- ({7*pi/4},-0.1) node[below]{$d$};
    \draw ({8*pi/4},0.1) -- ({8*pi/4},-0.1) node[below]{$\frac{8\pi}{6}$};
\node at (-0.5,0.5) {$0$};
\node at (-0.5,1.5) {$1$};
\node at (-0.5,2.5) {$2$};
\node at (-0.5,3.5) {$3$};
\node at (-1.0,2.0) {\rotatebox{90}{degree}};
\node at ({2*pi/8},4.5) {%
    \begin{tikzpicture}[scale=0.425]
    \def\n{6}
    \foreach \k in {1,2,...,\n} {
       \coordinate (\k) at ({cos(deg(2*pi*\k/\n))},{sin(deg(2*pi*\k/\n))});
       \node at (\k) {$\bullet$};
    }
    \end{tikzpicture}
    };
\node (a3) at ({2*pi/8},3.5) {$0$};
\node (a2) at ({2*pi/8},2.5) {$0$};
\node (a1) at ({2*pi/8},1.5) {$0$};
\node (a0) at ({2*pi/8},0.5) {$\mathbb{F}^{6}$};
\draw[->] (a3) -- (a2);
\draw[->] (a2) -- (a1);
\draw[->] (a1) -- (a0);
\node at ({3*pi/4},4.5) {%
    \begin{tikzpicture}[scale=0.425]
    \def\n{6}
    \foreach \k in {1,2,...,\n} {
       \coordinate (\k) at ({cos(deg(2*pi*\k/\n))},{sin(deg(2*pi*\k/\n))});
       \node at (\k) {$\bullet$};
       \draw[thick] ({cos(deg(2*pi*\k/\n))},{sin(deg(2*pi*\k/\n))}) -- ({cos(deg(2*pi*(\k+1)/\n))},{sin(deg(2*pi*(\k+1)/\n))});
    }
    \end{tikzpicture}
    };
\node (b3) at ({3*pi/4},3.5) {$0$};
\node (b2) at ({3*pi/4},2.5) {$0$};
\node (b1) at ({3*pi/4},1.5) {$\mathbb{F}^{6}$};
\node (b0) at ({3*pi/4},0.5) {$\mathbb{F}^{6}$};
\draw[->] (b3) -- (b2);
\draw[->] (b2) -- (b1);
\draw[->] (b1) -- node[right]{\footnotesize $\partial^{b}_{1}$} (b0);
\node at ({5*pi/4},4.5) {%
    \begin{tikzpicture}[scale=0.425]
    \def\n{6}
    \def\e{2.4}%2pi/n<epsilon<3pi/n gives a two-dimensional surface (n=6 gives 2-sphere) 
    \foreach \k in {1,2,...,\n} {
       \coordinate (\k) at ({cos(deg(2*pi*\k/\n))},{sin(deg(2*pi*\k/\n))});
       \node at (\k) {$\bullet$};
    %   \fill[opacity=0.1] (\k) circle (\e);
       \draw[thick] ({cos(deg(2*pi*\k/\n))},{sin(deg(2*pi*\k/\n))}) -- ({cos(deg(2*pi*(\k+1)/\n))},{sin(deg(2*pi*(\k+1)/\n))});
       \draw[dashed] ({cos(deg(2*pi*\k/\n))},{sin(deg(2*pi*\k/\n))}) -- ({cos(deg(2*pi*(\k+2)/\n))},{sin(deg(2*pi*(\k+2)/\n))});
    }
    \foreach \k in {1,3,5} {
       \draw[thick] ({cos(deg(2*pi*\k/\n))},{sin(deg(2*pi*\k/\n))}) -- ({cos(deg(2*pi*(\k+2)/\n))},{sin(deg(2*pi*(\k+2)/\n))});
    }
    \fill[opacity=0.1] ({cos(deg(2*pi*1/\n))},{sin(deg(2*pi*1/\n))}) -- ({cos(deg(2*pi*3/\n))},{sin(deg(2*pi*3/\n))}) -- ({cos(deg(2*pi*5/\n))},{sin(deg(2*pi*5/\n))}) -- cycle;
    \fill[opacity=0.25] ({cos(deg(2*pi*1/\n))},{sin(deg(2*pi*1/\n))}) -- ({cos(deg(2*pi*2/\n))},{sin(deg(2*pi*2/\n))}) -- ({cos(deg(2*pi*3/\n))},{sin(deg(2*pi*3/\n))}) -- cycle;
    \fill[opacity=0.15] ({cos(deg(2*pi*1/\n))},{sin(deg(2*pi*1/\n))}) -- ({cos(deg(2*pi*5/\n))},{sin(deg(2*pi*5/\n))}) -- ({cos(deg(2*pi*0/\n))},{sin(deg(2*pi*0/\n))}) -- cycle;
    \fill[opacity=0.35] ({cos(deg(2*pi*3/\n))},{sin(deg(2*pi*3/\n))}) -- ({cos(deg(2*pi*4/\n))},{sin(deg(2*pi*4/\n))}) -- ({cos(deg(2*pi*5/\n))},{sin(deg(2*pi*5/\n))}) -- cycle;
    \end{tikzpicture}
    };
\node (c3) at ({5*pi/4},3.5) {$0$};
\node (c2) at ({5*pi/4},2.5) {$\mathbb{F}^{8}$};
\node (c1) at ({5*pi/4},1.5) {$\mathbb{F}^{12}$};
\node (c0) at ({5*pi/4},0.5) {$\mathbb{F}^{6}$};
\draw[->] (c3) -- (c2);
\draw[->] (c2) -- node[right]{\footnotesize $\partial^{c}_{2}$} (c1);
\draw[->] (c1) -- node[right]{\footnotesize $\partial^{c}_{1}$} (c0);
\node at ({7*pi/4},4.5) {%
    \begin{tikzpicture}[scale=0.425]
    \def\n{6}
    \def\e{3.14}%2pi/n<epsilon<3pi/n gives a two-dimensional surface (n=6 gives 2-sphere) 
    \foreach \k in {1,2,...,\n} {
       \coordinate (\k) at ({cos(deg(2*pi*\k/\n))},{sin(deg(2*pi*\k/\n))});
       \node at (\k) {$\bullet$};
    %   \fill[opacity=0.1] (\k) circle (\e);
       \draw[thick] ({cos(deg(2*pi*\k/\n))},{sin(deg(2*pi*\k/\n))}) -- ({cos(deg(2*pi*(\k+1)/\n))},{sin(deg(2*pi*(\k+1)/\n))});
       \draw[thick] ({cos(deg(2*pi*\k/\n))},{sin(deg(2*pi*\k/\n))}) -- ({cos(deg(2*pi*(\k+2)/\n))},{sin(deg(2*pi*(\k+2)/\n))});
       \draw[thick] ({cos(deg(2*pi*\k/\n))},{sin(deg(2*pi*\k/\n))}) -- ({cos(deg(2*pi*(\k+3)/\n))},{sin(deg(2*pi*(\k+3)/\n))});
    }
    \fill[opacity=0.1] ({cos(deg(2*pi*1/\n))},{sin(deg(2*pi*1/\n))}) -- ({cos(deg(2*pi*2/\n))},{sin(deg(2*pi*2/\n))}) -- ({cos(deg(2*pi*3/\n))},{sin(deg(2*pi*3/\n))}) -- ({cos(deg(2*pi*4/\n))},{sin(deg(2*pi*4/\n))}) -- ({cos(deg(2*pi*5/\n))},{sin(deg(2*pi*5/\n))}) -- (1,0) -- cycle;
    \end{tikzpicture}
    };
\node (d3) at ({7*pi/4},3.5) {$\mathbb{F}^{15}$};
\node (d2) at ({7*pi/4},2.5) {$\mathbb{F}^{20}$};
\node (d1) at ({7*pi/4},1.5) {$\mathbb{F}^{15}$};
\node (d0) at ({7*pi/4},0.5) {$\mathbb{F}^{6}$};
\draw[->] (d3) -- node[right]{\footnotesize $\partial^{d}_{3}$} (d2);
\draw[->] (d2) -- node[right]{\footnotesize $\partial^{d}_{2}$} (d1);
\draw[->] (d1) -- node[right]{\footnotesize $\partial^{d}_{1}$} (d0);
%chain maps
\foreach \n in {0,1,2,3} {
\draw[->] (a\n) -- (b\n);
\draw[->] (b\n) -- (c\n);
\draw[->] (c\n) -- (d\n);
}
\draw[->] (b1) -- node[above]{\footnotesize $f_{1}^{(c,b)}$} (c1);
\draw[->] (c1) -- node[above]{\footnotesize $f_{1}^{(d,c)}$} (d1);
\draw[->] (c2) -- node[above]{\footnotesize $f_{2}^{(d,c)}$} (d2);
\end{tikzpicture}
\end{center}
The degree of the chain complexes is shown on the vertical axis, while $\epsilon$ is shown horizontally. We have not included degrees 4 and 5 for $\epsilon>\frac{6\pi}{n}$ for simplicity.  
This is because the chain complex is that associated with a single $5$-simplex (and is therefore contractible) for those values of $\epsilon$. As such, we will not need them later when calculating persistent homology. We have labeled some of the inclusion maps as well. 
Namely, 
\begin{equation}
f_{1}^{(c,b)}=
\begin{bmatrix}
\mathds{1}_{6}\\
0
\end{bmatrix},
\quad
f_{1}^{(d,c)}=
\begin{bmatrix}
\mathds{1}_{12}\\
0
\end{bmatrix},
\quad
f_{2}^{(d,c)}=
\begin{bmatrix}
\mathds{1}_{8}\\
0
\end{bmatrix},
\end{equation}
where the bottom $0$ submatrices are $6\times 6$, $3\times12$, and $12\times 8$, respectively. 
Meanwhile, all of the unlabeled nontrivial chain maps are identity linear transformations when going between spaces of the same dimension. 
One can also obtain the Vietoris--Rips chain complexes for other values of $n$, though for brevity we leave them out. 
\end{example}

Associated with a chain complex, we will construct its associated homology vector spaces. This will require us to take quotient vector spaces. Moreover, since we have persistence chain complexes, this will induce persistence quotient vector spaces. The definition of the latter is a special case of that in Definition~\ref{defn:persistenceobject}. 

\begin{definition}
Let $\mathbf{Vect}$ be the category of vector spaces over a field $\mathbb{F}$. Namely, the objects are $\mathbb{F}$-vector spaces and the morphisms are $\mathbb{F}$-linear transformations. 
A \define{persistence vector space} is a functor $\mathfrak{V}:\boldsymbol{[0,\infty)}\to\mathbf{Vect}$.
In detail, a persistence vector space consists of a collection of vector spaces $V^{\epsilon}$ indexed by $\epsilon\in[0,\infty)$ together with a collection of linear transformations $f^{(\epsilon',\epsilon)}:V^{\epsilon}\to V^{\epsilon'}$ for all $\epsilon,\epsilon'$ satisfying $\epsilon\le\epsilon'$ such that $f^{(\epsilon,\epsilon)}=\id$ for all $\epsilon\in[0,\infty)$ and $f^{(\epsilon'',\epsilon')}\circ f^{(\epsilon',\epsilon)}=f^{(\epsilon'',\epsilon)}$ for all $\epsilon,\epsilon',\epsilon''$ satisfying $\epsilon\le\epsilon'\le\epsilon''$.
\end{definition}

\begin{example}
\label{ex:PabAtom}
As an example of a persistent vector space, which will be useful in characterizing them, is obtained from any pair of numbers $(a,b)$, with $a\in[0,\infty)$, $b\in[0,\infty]$, and $a<b$. Namely, let $\mathbb{F}(a,b)$ be the persistence vector space given by 
\begin{equation}
\mathbb{F}(a,b)^{\epsilon}:=\begin{cases}\;\mathbb{F}&\mbox{ if $\epsilon\in[a,b)$} \\ \{0\}&\mbox{ otherwise}.\end{cases}
\end{equation}
The linear transformations $f^{(\epsilon',\epsilon)}:\mathbb{F}(a,b)^{\epsilon}\to \mathbb{F}(a,b)^{\epsilon'}$ are given by  
\begin{equation}
f^{(\epsilon',\epsilon)}:=\begin{cases}\id_{\mathbb{F}}&\mbox{ if $\epsilon,\epsilon'\in[a,b)$} \\ 0&\mbox{ otherwise}.\end{cases}
\end{equation}
It is useful to visualize $\mathbb{F}(a,b)$ together with these linear transformations $f^{(\epsilon',\epsilon)}$ as a cartoon by drawing an interval $[a,b]\subset[0,\infty]$, and drawing arrows from some start point to an end point. If the arrow is fully contained in the interval, then the arrow depicts the identity transformation. Otherwise, it is the zero map, either because it maps \emph{to} the zero vector space or because it maps \emph{from} the zero vector space.
\[
\begin{tikzpicture}[scale=0.825]
\draw (-3,0) -- (7,0);
\draw[very thick] (0,0) -- (4,0); 
\draw (0,0.1) -- (0,-0.1) node[below]{$a$};
\draw (4,0.1) -- (4,-0.1) node[below]{$b$};
\draw[->] (-2.75,0.1) to[out=30,in=150] node[above]{$0$} (-1.25,0.1);
\draw[->] (-0.75,0.1) to[out=30,in=150] node[above]{$0$} (0.75,0.1);
\draw[->] (1.25,0.1) to[out=30,in=150] node[above]{$\id_{\mathbb{F}}$} (2.75,0.1);
\draw[->] (3.25,0.1) to[out=30,in=150] node[above]{$0$} (4.75,0.1);
\draw[->] (5.25,0.1) to[out=30,in=150] node[above]{$0$} (6.75,0.1);
\end{tikzpicture}
\]
\end{example}

Using this example, we can construct many examples of persistence vector spaces. We assume familiarity with the direct sum of vector spaces~\cite{Halmos1958}. 

\begin{definition}
Given two persistence vector spaces $\mathfrak{V}=(V^{\epsilon},f^{(\epsilon',\epsilon)})$ and $\mathfrak{W}=(W^{\epsilon},g^{(\epsilon',\epsilon)})$, their direct sum $\mathfrak{V}\oplus \mathfrak{W}$ is the persistence vector space defined as 
\begin{equation}
(V\oplus W)^{\epsilon}:=V^{\epsilon}\oplus W^{\epsilon}, 
\end{equation}
which is the direct sum of vector spaces, and
\begin{equation}
(f\oplus g)^{\epsilon}:=f^{\epsilon}\oplus g^{\epsilon}, 
\end{equation}
which is the direct sum of linear transformations. 
\end{definition}

Taking direct sums of persistence vector spaces $\mathbb{F}(a,b)$ will prove to be incredibly useful, since they form the objects from which many other persistence vector spaces can be described up to isomorphism. 

\begin{example}
\label{ex:PabsumPcd}
Let $(a,b),(c,d)\in[0,\infty)\times[0,\infty]$ with $a<b$ and $c<d$. The direct sum $\mathbb{F}(a,b)\oplus \mathbb{F}(c,d)$ is the persistence vector space 
given by
\begin{equation}
\big(\mathbb{F}(a,b)\oplus \mathbb{F}(c,d)\big)^{\epsilon}
=
\begin{cases}
\;\mathbb{F}^2&\mbox{ if $\epsilon\in [a,b)\cap[c,d)$} \\ 
\;\mathbb{F} &\mbox{ if $\epsilon\in [a,b)\Delta [c,d)$} \\ 
\{0\} 
&\mbox{ if $\big([a,b)\cup[c,d)\big)^{\mathrm{c}}$} , \\
\end{cases}
\end{equation}
where $A\Delta B:= (A\setminus B)\cup (B\setminus A)$ is the symmetric difference of two sets $A$ and $B$, and where $A^{\mathrm{c}}:=X\setminus A$ denotes the complement of a subset $A\subseteq X$. 
The linear transformations will be provided in some examples shortly. We first point out that it is often easier to visualize this persistence vector space in terms of horizontal bars. At each $\epsilon$, the number of intersections of $\epsilon$ with the intervals $[a,b)$ and $[c,d)$ indicates the number of copies of $\mathbb{F}$ that occur for that value of $\epsilon$. For example, if $a<c<b<d$, then we have
\[
\begin{tikzpicture}[scale=0.95]
\draw (-2.5,-0.5) -- (6.5,-0.5);
\draw (0,{0.1-0.5}) -- (0,{-0.1-0.5}) node[below]{$a$};
\draw (4,{0.1-0.5}) -- (4,{-0.1-0.5}) node[below]{$b$};
\draw (2,{0.1-0.5}) -- (2,{-0.1-0.5}) node[below]{$c$};
\draw (5,{0.1-0.5}) -- (5,{-0.1-0.5}) node[below]{$d$};
\draw[very thick] (0,0) -- (4,0); 
\draw (0,0.1) -- (0,-0.1);
\draw (4,0.1) -- (4,-0.1);
\draw[very thick] (2,0.4) -- (5,0.4); 
\draw (2,{0.1+0.4}) -- (2,{-0.1+0.4});
\draw (5,{0.1+0.4}) -- (5,{-0.1+0.4});
\node (0) at (-1.5,-1.2) {$\{0\}$};
\node (1) at (1,-1.2) {$\mathbb{F}$};
\node at (1,0.15) {\footnotesize$e_{1}$};
\node (2) at (3,-1.2) {$\mathbb{F}^2$};
\node at (3,0.15) {\footnotesize$e_{1}$};
\node at (3,0.55) {\footnotesize$e_{2}$};
\node (3) at (4.5,-1.2) {$\mathbb{F}$};
\node at (4.5,0.55) {\footnotesize$e_{1}$};
\node (4) at (6,-1.2) {$\{0\}$};
\draw[->] (0) to[out=-30,in=-150] node[below]{$\left[\begin{smallmatrix}\end{smallmatrix}\right]$} (1);
\draw[->] (1) to[out=-30,in=-150] node[below]{$\left[\begin{smallmatrix}1\\0\end{smallmatrix}\right]$} (2);
\draw[->] (2) to[out=-30,in=-150] node[below]{$\left[\begin{smallmatrix}0 &1\end{smallmatrix}\right]$} (3);
\draw[->] (3) to[out=-30,in=-150] node[below]{$\left[\begin{smallmatrix}\end{smallmatrix}\right]$} (4);
\end{tikzpicture}
\]
where the vector space for each subinterval is indicated underneath. Moreover, by choosing the ordering for the direct sum as presented, the linear transformations going from one vector space to the next have been included in the above figure as matrices. The standard basis vectors $e_{i}$ are drawn to illustrate the corresponding generator at that particular value of $\epsilon$. An empty matrix is drawn to represent the trivial mapping. 
If $a<b<c<d$, then the direct sum persistence vector space looks like 
\[
\begin{tikzpicture}[scale=0.75]
\draw (-2.5,-0.5) -- (6.5,-0.5);
\draw (0,{0.1-0.5}) -- (0,{-0.1-0.5}) node[below]{$a$};
\draw (4,{0.1-0.5}) -- (4,{-0.1-0.5}) node[below]{$c$};
\draw (2,{0.1-0.5}) -- (2,{-0.1-0.5}) node[below]{$b$};
\draw (5,{0.1-0.5}) -- (5,{-0.1-0.5}) node[below]{$d$};
\draw[very thick] (0,0) -- (2,0); 
\draw (0,0.1) -- (0,-0.1);
\draw (2,0.1) -- (2,-0.1);
\draw[very thick] (4,0.4) -- (5,0.4); 
\draw (4,{0.1+0.4}) -- (4,{-0.1+0.4});
\draw (5,{0.1+0.4}) -- (5,{-0.1+0.4});
\node (0) at (-1.5,-1.2) {$\{0\}$};
\node (1) at (1,-1.2) {$\mathbb{F}$};
\node (2) at (3,-1.2) {$\{0\}$};
\node (3) at (4.5,-1.2) {$\mathbb{F}$};
\node (4) at (6,-1.2) {$\{0\}$};
\end{tikzpicture}
\]
We will look at generalizations of this example to allow for more direct sums later. Such examples motivate the definitions of barcodes, which we will see classify certain persistence vector spaces up to isomorphism.
\end{example}

\begin{definition}
Given a vector space $V$ and a vector subspace $U$ of $V$, the \define{quotient vector space} $V/U$ is the vector space
\begin{equation}
V/U:=\{v+U\;:\;v\in V\},
\end{equation}
where $v+U:=\{v+u\;:\;u\in U\}$ denotes the set of elements of the form $v+u$. The additive structure on $V/U$ is defined by $(v+U)+(v'+U):=(v+v')+U$ for $v,v'\in V$, and scalar multiplication is defined by $c(v+U):=(cv)+U$ for $c\in\mathbb{F}$ and $v\in V$.
\end{definition}

\begin{definition}
\label{defn:cyclesboundaries}
Given a chain complex $(C_{\bullet},\partial_{\bullet})$, for each $n\in\Z$, let $\mathrm{ker}(\partial_{n})\subseteq C_{n}$ denote the kernel of $\partial_{n}$, elements of which are called \define{$n$-cycles}, and let $\mathrm{im}(\partial_{n+1})\subseteq C_{n}$ denote the image of $\partial_{n+1}$, elements of which are called \define{$n$-boundaries}. 
\end{definition}

By the condition $\partial_{n}\circ\partial_{n+1}=0$, it follows that $\mathrm{im}(\partial_{n+1})$ is a vector subspace of $\mathrm{ker}(\partial_{n})$. The associated quotient is due to Emmy Noether~\cite{MacLane86,Dieudonne84}. 

\begin{definition}
Using the notation from Definition~\ref{defn:cyclesboundaries}, the quotient space 
\begin{equation}
H_{n}:=\mathrm{ker}(\partial_{n})/\mathrm{im}(\partial_{n+1})
\end{equation}
is called the \define{$n^{\text{th}}$ homology} of $(C_{\bullet},\partial_{\bullet})$. 
We may write $H_{n}(C)$ instead of $H_{n}$ to clarify with respect to which chain complex the homology is defined.
An element of $H_{n}$ is called a \define{homology class}. 
The dimension of the $n^{\text{th}}$ homology is called the \define{$n^{\text{th}}$ Betti number} associated with the chain complex $(C_{\bullet},\partial_{\bullet})$ and is written as 
\begin{equation}
\beta_{n}:=\mathrm{dim}(H_{n}). 
\end{equation}
\end{definition}

We note that the homology is normally defined not as a vector space but as an abelian group, which can be viewed as a $\Z$-module, a generalization of a vector space~\cite{DummitFoote}. Tensoring this $\Z$-module with the field $\mathbb{F}$ gives a vector space, and there is in general some information loss due to such an operation. This information loss is captured by the universal coefficient theorem~\cite{Munkres84,hatcherbook2002,MacLaneHomology}. More generally, homology does not provide a full invariant of a topological space, though these points will not concern us here. 

\begin{example}
\label{ex:persistencehomology}
As a relevant example of homology, given a metric space $(X,d_{X})$ and a fixed $\epsilon\in[0,\infty)$, the $n^{\text{th}}$ homology of $(C^{\epsilon}_{\bullet}(X,d_{X}),\partial_{\bullet})$ is denoted by $H_{n}^{\epsilon}(X,d_{X})$. We note that the superscript denotes the filtration index, and not a cohomology index. In practice, the Vietoris--Rips chain complex constructed from a finite metric space $(X,d_{X})$ has constituent vector spaces $C^{\epsilon}_{n}(X,d_{X})$ that are often of large dimension. Among exhibiting other important features, the homology occasionally significantly reduces the dimensionality of these constructions. For example, if $(X,d_{X})$ is a metric space with $n$ elements, the dimension of $C^{\epsilon}_{k}(X,d_{X})$ is always less than or equal to $\binom{n}{k+1}:=\frac{n!}{(k+1)!(n-k-1)!}$, %the binomial coefficient $n$ choose $k+1$, 
where the $k+1$ should be read as $k+1\mod n$. In particular, when $X$ is the set consisting of the $n^\text{th}$ roots of unity as in Example~\ref{ex:Hausmannrootsofunity}, and when $k=1$, the dimension $\mathrm{dim}(n,\epsilon)$ of $C^{\epsilon}_{1}(X,d_{X})$ 
is given by 
\begin{equation}
\mathrm{dim}(n,\epsilon)=\begin{cases}kn&\mbox{if $\frac{2\pi k}{n}\le \epsilon<\frac{2\pi (k+1)}{n}$ and $k<m$}\\\binom{n}{2}&\mbox{if $\frac{2(k+1)\pi}{n}\le\epsilon$ and $k\ge m$}\end{cases}
\end{equation}
assuming that $k$ is a nonnegative integer and $m$ is the unique integer satisfying $n=2m$ or $n=2m+1$ depending on whether $n$ is even or odd. Note that when $n=2m+1$ is odd, then $\binom{n}{2}=mn$ so that the edge case is well-defined. 
For example, when $n=8=2(4)$, we can visualize the 1-simplices (edges) over different choices of $\epsilon$ in the given ranges as 
\begin{center}
\begin{tikzpicture}
\draw[->] ({2*pi/4},0) -- ({2*pi+2*pi/4+0.75},0) node[right]{$\epsilon$}; 
    \draw ({2*pi/4},0.1) -- ({2*pi/4},-0.1) node[below]{$\frac{2\pi}{8}$};
    % \draw ({3*pi/4},0.1) -- ({3*pi/4},-0.1) node[below]{$a$};
    \node (a) at ({3*pi/4},1.3) {%one connection
        \begin{tikzpicture}[scale=0.6]
        \def\n{8}
        \foreach \k in {1,2,...,\n} {
       \coordinate (\k) at ({cos(deg(2*pi*\k/\n))},{sin(deg(2*pi*\k/\n))});
       \node at (\k) {$\bullet$};
      \draw[thick] ({cos(deg(2*pi*\k/\n))},    {sin(deg(2*pi*\k/\n))}) -- ({cos(deg(2*pi*(\k+1)/\n))},{sin(deg(2*pi*(\k+1)/\n))});
        }
        \end{tikzpicture}
        };
    \node (adim) at ({3*pi/4},2.5) {$8$};
    \draw[dashed] ({3*pi/4},-0.1) -- (a);
    \draw ({4*pi/4},0.1) -- ({4*pi/4},-0.1) node[below]{$\frac{4\pi}{8}$};
    \node (b) at ({5*pi/4},1.3) {%two connections
        \begin{tikzpicture}[scale=0.6]
        \def\n{8}
        \foreach \k in {1,2,...,\n} {
           \coordinate (\k) at ({cos(deg(2*pi*\k/\n))},{sin(deg(2*pi*\k/\n))});
           \node at (\k) {$\bullet$};
           \draw[thick] ({cos(deg(2*pi*\k/\n))},{sin(deg(2*pi*\k/\n))}) -- ({cos(deg(2*pi*(\k+1)/\n))},{sin(deg(2*pi*(\k+1)/\n))});
           \draw[thick] ({cos(deg(2*pi*\k/\n))},{sin(deg(2*pi*\k/\n))}) -- ({cos(deg(2*pi*(\k+2)/\n))},{sin(deg(2*pi*(\k+2)/\n))});
        }
        \end{tikzpicture}
        };
        \node (bdim) at ({5*pi/4},2.5) {$16$};
    % \draw ({5*pi/4},0.1) -- ({5*pi/4},-0.1) node[below]{$b$};
    \draw[dashed] ({5*pi/4},-0.1) -- (b); 
    \draw ({6*pi/4},0.1) -- ({6*pi/4},-0.1) node[below]{$\frac{6\pi}{8}$};
    \node (c) at ({7*pi/4},1.3) {%three connections
            \begin{tikzpicture}[scale=0.6]
        \def\n{8}
        \foreach \k in {1,2,...,\n} {
           \coordinate (\k) at ({cos(deg(2*pi*\k/\n))},{sin(deg(2*pi*\k/\n))});
           \node at (\k) {$\bullet$};
           \draw[thick] ({cos(deg(2*pi*\k/\n))},{sin(deg(2*pi*\k/\n))}) -- ({cos(deg(2*pi*(\k+1)/\n))},{sin(deg(2*pi*(\k+1)/\n))});
           \draw[thick] ({cos(deg(2*pi*\k/\n))},{sin(deg(2*pi*\k/\n))}) -- ({cos(deg(2*pi*(\k+2)/\n))},{sin(deg(2*pi*(\k+2)/\n))});
           \draw[thick] ({cos(deg(2*pi*\k/\n))},{sin(deg(2*pi*\k/\n))}) -- ({cos(deg(2*pi*(\k+3)/\n))},{sin(deg(2*pi*(\k+3)/\n))});
        }
        \end{tikzpicture}
        };
        \node (cdim) at ({7*pi/4},2.5) {$24$};
    % \draw ({7*pi/4},0.1) -- ({7*pi/4},-0.1) node[below]{$c$};
    \draw[dashed] ({7*pi/4},-0.1) -- (c);
    \draw ({8*pi/4},0.1) -- ({8*pi/4},-0.1) node[below]{$\frac{8\pi}{8}$};
    \node (d) at ({9*pi/4},1.3) {%four connections
            \begin{tikzpicture}[scale=0.6]
        \def\n{8}
        \foreach \k in {1,2,...,\n} {
           \coordinate (\k) at ({cos(deg(2*pi*\k/\n))},{sin(deg(2*pi*\k/\n))});
           \node at (\k) {$\bullet$};
           \draw[thick] ({cos(deg(2*pi*\k/\n))},{sin(deg(2*pi*\k/\n))}) -- ({cos(deg(2*pi*(\k+1)/\n))},{sin(deg(2*pi*(\k+1)/\n))});
           \draw[thick] ({cos(deg(2*pi*\k/\n))},{sin(deg(2*pi*\k/\n))}) -- ({cos(deg(2*pi*(\k+2)/\n))},{sin(deg(2*pi*(\k+2)/\n))});
           \draw[thick] ({cos(deg(2*pi*\k/\n))},{sin(deg(2*pi*\k/\n))}) -- ({cos(deg(2*pi*(\k+3)/\n))},{sin(deg(2*pi*(\k+3)/\n))});
           \draw[thick] ({cos(deg(2*pi*\k/\n))},{sin(deg(2*pi*\k/\n))}) -- ({cos(deg(2*pi*(\k+4)/\n))},{sin(deg(2*pi*(\k+4)/\n))});
        }
        \end{tikzpicture}
        };
        \node (ddim) at ({9*pi/4},2.5) {$28$};
    % \draw ({9*pi/4},0.1) -- ({9*pi/4},-0.1) node[below]{$d$};
    \draw[dashed] ({9*pi/4},-0.1) -- (d);
    \draw ({10*pi/4},0.1) -- ({10*pi/4},-0.1) node[below]{$\frac{10\pi}{8}$};
\end{tikzpicture}
\end{center}
The number above each graph denotes the number of edges, i.e., the dimension of $C^{\epsilon}_{1}(X,d_{X})$. As another example, when $n=9=2(4)+1$ we have
\begin{center}
\begin{tikzpicture}
\draw[->] ({2*pi/4},0) -- ({2*pi+2*pi/4+0.75},0) node[right]{$\epsilon$}; 
    \draw ({2*pi/4},0.1) -- ({2*pi/4},-0.1) node[below]{$\frac{2\pi}{9}$};
    % \draw ({3*pi/4},0.1) -- ({3*pi/4},-0.1) node[below]{$a$};
    \node (a) at ({3*pi/4},1.3) {%one connection
        \begin{tikzpicture}[scale=0.6]
        \def\n{9}
        \foreach \k in {1,2,...,\n} {
       \coordinate (\k) at ({cos(deg(2*pi*\k/\n))},{sin(deg(2*pi*\k/\n))});
       \node at (\k) {$\bullet$};
      \draw[thick] ({cos(deg(2*pi*\k/\n))},    {sin(deg(2*pi*\k/\n))}) -- ({cos(deg(2*pi*(\k+1)/\n))},{sin(deg(2*pi*(\k+1)/\n))});
        }
        \end{tikzpicture}
        };
    \node (adim) at ({3*pi/4},2.5) {$9$};
    \draw[dashed] ({3*pi/4},-0.1) -- (a);
    \draw ({4*pi/4},0.1) -- ({4*pi/4},-0.1) node[below]{$\frac{4\pi}{9}$};
    \node (b) at ({5*pi/4},1.3) {%two connections
        \begin{tikzpicture}[scale=0.6]
        \def\n{9}
        \foreach \k in {1,2,...,\n} {
           \coordinate (\k) at ({cos(deg(2*pi*\k/\n))},{sin(deg(2*pi*\k/\n))});
           \node at (\k) {$\bullet$};
           \draw[thick] ({cos(deg(2*pi*\k/\n))},{sin(deg(2*pi*\k/\n))}) -- ({cos(deg(2*pi*(\k+1)/\n))},{sin(deg(2*pi*(\k+1)/\n))});
           \draw[thick] ({cos(deg(2*pi*\k/\n))},{sin(deg(2*pi*\k/\n))}) -- ({cos(deg(2*pi*(\k+2)/\n))},{sin(deg(2*pi*(\k+2)/\n))});
        }
        \end{tikzpicture}
        };
        \node (bdim) at ({5*pi/4},2.5) {$18$};
    % \draw ({5*pi/4},0.1) -- ({5*pi/4},-0.1) node[below]{$b$};
    \draw[dashed] ({5*pi/4},-0.1) -- (b); 
    \draw ({6*pi/4},0.1) -- ({6*pi/4},-0.1) node[below]{$\frac{6\pi}{9}$};
    \node (c) at ({7*pi/4},1.3) {%three connections
            \begin{tikzpicture}[scale=0.6]
        \def\n{9}
        \foreach \k in {1,2,...,\n} {
           \coordinate (\k) at ({cos(deg(2*pi*\k/\n))},{sin(deg(2*pi*\k/\n))});
           \node at (\k) {$\bullet$};
           \draw[thick] ({cos(deg(2*pi*\k/\n))},{sin(deg(2*pi*\k/\n))}) -- ({cos(deg(2*pi*(\k+1)/\n))},{sin(deg(2*pi*(\k+1)/\n))});
           \draw[thick] ({cos(deg(2*pi*\k/\n))},{sin(deg(2*pi*\k/\n))}) -- ({cos(deg(2*pi*(\k+2)/\n))},{sin(deg(2*pi*(\k+2)/\n))});
           \draw[thick] ({cos(deg(2*pi*\k/\n))},{sin(deg(2*pi*\k/\n))}) -- ({cos(deg(2*pi*(\k+3)/\n))},{sin(deg(2*pi*(\k+3)/\n))});
        }
        \end{tikzpicture}
        };
        \node (cdim) at ({7*pi/4},2.5) {$27$};
    % \draw ({7*pi/4},0.1) -- ({7*pi/4},-0.1) node[below]{$c$};
    \draw[dashed] ({7*pi/4},-0.1) -- (c);
    \draw ({8*pi/4},0.1) -- ({8*pi/4},-0.1) node[below]{$\frac{8\pi}{9}$};
    \node (d) at ({9*pi/4},1.3) {%four connections
            \begin{tikzpicture}[scale=0.6]
        \def\n{9}
        \foreach \k in {1,2,...,\n} {
           \coordinate (\k) at ({cos(deg(2*pi*\k/\n))},{sin(deg(2*pi*\k/\n))});
           \node at (\k) {$\bullet$};
           \draw[thick] ({cos(deg(2*pi*\k/\n))},{sin(deg(2*pi*\k/\n))}) -- ({cos(deg(2*pi*(\k+1)/\n))},{sin(deg(2*pi*(\k+1)/\n))});
           \draw[thick] ({cos(deg(2*pi*\k/\n))},{sin(deg(2*pi*\k/\n))}) -- ({cos(deg(2*pi*(\k+2)/\n))},{sin(deg(2*pi*(\k+2)/\n))});
           \draw[thick] ({cos(deg(2*pi*\k/\n))},{sin(deg(2*pi*\k/\n))}) -- ({cos(deg(2*pi*(\k+3)/\n))},{sin(deg(2*pi*(\k+3)/\n))});
           \draw[thick] ({cos(deg(2*pi*\k/\n))},{sin(deg(2*pi*\k/\n))}) -- ({cos(deg(2*pi*(\k+4)/\n))},{sin(deg(2*pi*(\k+4)/\n))});
        }
        \end{tikzpicture}
        };
        \node (ddim) at ({9*pi/4},2.5) {$36$};
    % \draw ({9*pi/4},0.1) -- ({9*pi/4},-0.1) node[below]{$d$};
    \draw[dashed] ({9*pi/4},-0.1) -- (d);
    \draw ({10*pi/4},0.1) -- ({10*pi/4},-0.1) node[below]{$\frac{10\pi}{9}$};
\end{tikzpicture}
\end{center}
It is an interesting exercise in combinatorics to identify the dimension of $C^{\epsilon}_{k}(X,d_{X})$ for $k>1$. Meanwhile, as we will see later in Example~\ref{ex:Hausmannrootshomology}, the homology groups drastically reduce these dimensions in this example. 
\end{example}

The family of quotient spaces in Example~\ref{ex:persistencehomology} indexed by $\epsilon$ is an example of a quotient persistence vector space, which we now formally define. This notion will also be important in the classification of persistence vector spaces.  

\begin{definition}
Let $\mathfrak{V}=\{V^{\epsilon},f^{(\epsilon',\epsilon)}\}$ be a persistence vector space and let $\mathfrak{U}=\{U^{\epsilon},f_{|\mathfrak{U}}^{(\epsilon',\epsilon)}\}$ be a \define{sub-persistence vector space}, i.e., $U^{\epsilon}\subseteq V^{\epsilon}$ for all $\epsilon\in[0,\infty)$ and $f_{|\mathfrak{U}}^{(\epsilon',\epsilon)}(u)=f^{(\epsilon',\epsilon)}(u)$ for all $u\in U^{\epsilon}$ and for all $\epsilon,\epsilon'\in[0,\infty)$ with $\epsilon\le\epsilon'$. 
The \define{quotient persistence vector space} $\mathfrak{V}/\mathfrak{U}=\{(V/U)^{\epsilon},\mathsf{f}^{(\epsilon',\epsilon)}\}$ has $(V/U)^{\epsilon}=V^{\epsilon}/U^{\epsilon}$ being the quotient vector space for each $\epsilon\in[0,\infty)$ and $\mathsf{f}^{(\epsilon',\epsilon)}(v+U^{\epsilon}):=f^{(\epsilon',\epsilon)}+U^{\epsilon'}$, which defines a linear transformation, for all $v+U^{\epsilon}\in V^{\epsilon}/U^{\epsilon}$ and for all $\epsilon,\epsilon'\in[0,\infty)$ with $\epsilon\le\epsilon'$.
\end{definition}

To make sense of the kernel and image of the boundary operator as persistence vector spaces used to define homology, we first provide a general definition for morphisms of persistence objects, and then specialize to the case of persistence vector spaces momentarily. 

\begin{definition}
Let $\mathbf{C}$ be any category. Given two persistence objects $F,G:\boldsymbol{[0,\infty)}\to \mathbf{C}$, a \define{morphism of persistence objects} from $F$ to $G$ is a natural transformation $\eta:F\Rightarrow G$. 
\end{definition}

Rather than spelling out what this means abstractly, which follows from Definitions~\ref{defn:nattrans} and~\ref{defn:persistenceobject}, we write it out explicitly for a morphism between persistence vector spaces. 

\begin{definition}
\label{defn:morphismofpersvect}
Given two persistence vector spaces $\mathfrak{V}=\{V^{\epsilon}, f^{(\epsilon',\epsilon)}\}$ and $\mathfrak{W}=\{W^{\epsilon}, g^{(\epsilon',\epsilon)}\}$, a \define{morphism of persistence vector spaces} (sometimes also called a \define{linear transformation}) $\eta:\mathfrak{V}\to\mathfrak{W}$ consists of a collection of linear transformations $\eta^{\epsilon}:V^{\epsilon}\to W^{\epsilon}$ such that 
$g^{(\epsilon',\epsilon)}\circ \eta^{\epsilon}=\eta^{\epsilon'}\circ f^{(\epsilon',\epsilon)}$ for all $\epsilon\le\epsilon'$. 
\end{definition}

\begin{example}
\label{ex:VRkernimage}
For the relevant example involving the Vietoris--Rips complex associated with a finite metric space $(X,d_{X})$, the boundary maps of the chain complexes 
\begin{equation}
\partial^{\epsilon}_{n}:C^{\epsilon}_{n}(X,d_{X})\to C^{\epsilon}_{n-1}(X,d_{X})
\end{equation}
define a morphism $\partial_{n}:C_{n}(X,d_{X})\to C_{n-1}(X,d_{X})$ of persistence vector spaces for each $n\in\N$. In fact, we have already seen this in Example~\ref{ex:Hausmann4case} for the $n^{\text{th}}$ roots of unity. Hence, $\ker(\partial_{n})$ becomes a sub-persistence vector space of $C_{n}(X,d_{X})$, while $\mathrm{im}(\partial_{n+1})$ becomes a sub-persistence vector space of $\ker(\partial_{n})$. 
\end{example}

\begin{definition}
Let $\eta:\mathfrak{V}\to\mathfrak{W}$ be a morphism of persistence vector spaces, as in Definition~\ref{defn:morphismofpersvect}. 
The \define{image persistence vector space} $\mathrm{im}(\eta)$ is the sub-persistence vector space of $\mathfrak{W}$ whose  constituent vector spaces are given by $\eta^{\epsilon}(V^{\epsilon})\subseteq W^{\epsilon}$. The \define{kernel persistence vector space} $\mathrm{ker}(\eta)$ is the sub-persistence vector space of $\mathfrak{V}$ whose consituent vector spaces are given by $\mathrm{ker}(\eta^{\epsilon})\subseteq V^{\epsilon}$. 
\end{definition}

\begin{example}
Continuing from Example~\ref{ex:VRkernimage}, the associated quotient persistence vector space $\ker(\partial_{n})/\mathrm{im}(\partial_{n+1})$ defines the $n^{\text{th}}$ degree homology persistence vector space $H_{n}(X,d_{X})$.
\end{example}

Moving on, any chain map $f_{\bullet}:(C_{\bullet},\partial_{\bullet})\to(D_{\bullet},\partial_{\bullet})$ induces a canonical linear transformation on the homology $\mathsf{f}_{n}:H_{n}(C)\to H_{n}(D)$ defined as follows. 
Given an $n$-cycle $z\in\ker(\partial_{n})$ with its associated equivalence class in $H_{n}(C)$ denoted by $[z]$, the application of $\mathsf{f}_{n}$ to $[z]$ is defined by 
$\mathsf{f}_{n}\big([z]\big):=\big[f_{n}(z)\big]$. It follows that $\mathsf{f}_{n}$ is well-defined and a linear transformation from the fact that $f_{\bullet}$ is a chain map. It also follows from this definition that $\mathsf{id}_{n}:H_{n}(C)\to H_{n}(C)$ is the identity map and $(\mathsf{g}\circ \mathsf{f})_{n}=\mathsf{g}_{n}\circ \mathsf{f}_{n}$ as linear transformations from $H_{n}(C)$ to $H_{n}(E)$ for a pair of chain maps  $f_{\bullet}:(C_{\bullet},\partial_{\bullet})\to(D_{\bullet},\partial_{\bullet})$  and $g_{\bullet}:(D_{\bullet},\partial_{\bullet})\to(E_{\bullet},\partial_{\bullet})$. 
One can show that the assignment 
\begin{align}
\mathbf{ChainComp}&\xrightarrow{H_{n}}\mathbf{Vect} \nonumber\\
(C_{\bullet},\partial_{\bullet})&\xmapsto{\;\;\;}H_{n}(C) \nonumber\\
\left((C_{\bullet},\partial_{\bullet})\xrightarrow{f_{\bullet}}(D_{\bullet},\partial_{\bullet})\right)&\xmapsto{\;\;\;} \left(H_{n}(C) \xrightarrow{\mathsf{f}_n} H_{n}(D) \right)
\end{align}
defines a functor. 

\begin{example}
\label{ex:persistenthomologyofmetricspace}
As an example, given a finite metric space $(X,d_{X})$, we have already seen how we first obtain a filtered simplicial complex and then a chain complex together with its associated chain maps $f^{(\epsilon',\epsilon)}_{\bullet}$ in~\eqref{eqn:persistencechainmaps}. Each of these maps therefore induces a linear transformation $\mathsf{f}^{(\epsilon',\epsilon)}_{n}:H_{n}^{\epsilon}(X,d_{X})\to H_{n}^{\epsilon'}(X,d_{X})$ at the level of homology. The associated persistence vector space $\mathfrak{H}_{n}=\big\{H_{n}^{\epsilon}(X,d_{X}),\mathsf{f}^{(\epsilon',\epsilon)}_{n}\big\}$ is called the \define{$n^{\text{th}}$ degree persistence homology} associated with $(X,d_{X})$. 
For each $n\in\N\cup\{0\}$, the \define{persistent homology groups} are the images of the maps $\mathsf{f}^{(\epsilon',\epsilon)}_{n}: H_{n}^{\epsilon}(X,d_{X})\to H_{n}^{\epsilon'}(X,d_{X})$ for each pair of real numbers $\epsilon,\epsilon'$ with $0\le \epsilon\le \epsilon'$. The ranks of the persistent homology groups are called \define{persistent Betti numbers}. 

Thus, the persistence vector space $\mathfrak{H}_{n}$ contains not only the information about the topology of the Vietoris--Rips complex for each value of $\epsilon$ through the vector space $H_{n}^{\epsilon}(X,d_{X})$, it also contains information about the relationships between the topologies for different values of $\epsilon$ through the maps $\mathsf{f}^{(\epsilon',\epsilon)}_{n}$~\cite{ZoCa05,edelsbrunner2002topological,Virk2022,CCGGO09}. Alternatively, the $n^{\text{th}}$ persistent homology groups between value $\epsilon$ and $\epsilon'$ may equivalently be defined as
\begin{equation}
H^{(\epsilon',\epsilon)}_{n}(X,d_{X})=\ker(\partial^{\epsilon}_{n})/(\mathrm{im}(\partial^{\epsilon'}_{n})\cap\ker(\partial^{\epsilon}_{n}))
\end{equation}
and the associated persistent Betti number is the dimension of $H^{(\epsilon',\epsilon)}_{n}(X,d_{X})$~\cite{edelsbrunner2002topological,ZoCa05}. 
\end{example}

\begin{example}
\label{ex:Hausmannrootshomology}
Here, we will calculate the homology groups and persistent homology vector spaces associated with Example~\ref{ex:Hausmann4case}. For $\epsilon=a$ satisfying $0<a\le\frac{2\pi}{6}$, all homology groups are trivial \emph{except} for the degree $0$ homology, which is a vector space of dimension $6$. The reason is because every basis element gets sent to $0$, and so each basis element represents a nontrivial homology class. 
Visually, this says that the simplicial complex has $6$ connected components. 
For $\epsilon=b$ satisfying $\frac{2\pi}{6}<b\le\frac{4\pi}{6}$, there is one nontrivial homology vector space in degrees $0$ and $1$, and each is one-dimensional. The degree $1$ space comes from the kernel of $\partial^{b}_{1}$ from~\eqref{eq:Hauspb1}. A quick verification shows that 
\begin{equation}
\partial^{b}_{1}\big(\mathbf{e}^{(1)}_{1}+\cdots+\mathbf{e}^{(1)}_{6}\big)=0
\end{equation}
so that $\mathbf{e}^{(1)}_{1}+\cdots+\mathbf{e}^{(1)}_{6}=\ket{e_{01}}+\cdots+\ket{e_{50}}$ is a vector representing the homology class in degree $1$. This vector represents the cycle obtained by concatenating all edges going around counter-clockwise. 
The generator for degree $0$ is a bit more subtle. The image of $\partial^{b}_{1}$ is a $5$-dimensional subspace of $\mathbb{F}^{6}$ spanned by the columns of the associated matrix~\eqref{eq:Hauspb1}. A vector perpendicular to these columns is given by $\mathbf{e}^{(0)}_{1}+\cdots+\mathbf{e}^{(0)}_{6}=\ket{x_0}+\cdots\ket{x_5}$, and this therefore represents the cohomology class since this vector is not in the image of $\partial^{b}_{1}$ (technically, since the homology group is a quotient vector space, any of the vectors $\mathbf{e}^{(0)}_{j}$ can be used as a representative).
Visually, this class represents the fact that the simplicial complex now has a single connected component. 
For $\epsilon=c$ satisfying $\frac{4\pi}{6}<c\le\frac{6\pi}{6}$, the degree $1$ homology vector space vanishes, but now there is a nontrivial degree $2$ homology vector space. This comes from the sum of the eight faces making up the octahedron, which, using our convention where the orientation satisfies the right-hand-rule, is represented by the vector $\mathbf{e}^{(2)}_{1}+\cdots+\mathbf{e}^{(2)}_{8}$. Finally, when $\epsilon=d$ satisfies $\frac{6\pi}{6}<d$, we have no nontrivial homology vector spaces since all cycles are boundaries. Putting this all together, we have the following picture of the associated homology vector spaces. 
\begin{center}
\begin{tikzpicture}
\draw[->] (0,0) -- ({2*pi+0.5},0) node[right]{$\epsilon$}; 
    \draw ({0},0.1) -- ({0},-0.1) node[below]{$0$};
    \draw ({2*pi/8},0.1) -- ({2*pi/8},-0.1) node[below]{$a$};
    \draw ({2*pi/4},0.1) -- ({2*pi/4},-0.1) node[below]{$\frac{2\pi}{6}$};
    \draw ({3*pi/4},0.1) -- ({3*pi/4},-0.1) node[below]{$b$};
    \draw ({4*pi/4},0.1) -- ({4*pi/4},-0.1) node[below]{$\frac{4\pi}{6}$};
    \draw ({5*pi/4},0.1) -- ({5*pi/4},-0.1) node[below]{$c$};
    \draw ({6*pi/4},0.1) -- ({6*pi/4},-0.1) node[below]{$\frac{6\pi}{6}$};
    \draw ({7*pi/4},0.1) -- ({7*pi/4},-0.1) node[below]{$d$};
    \draw ({8*pi/4},0.1) -- ({8*pi/4},-0.1) node[below]{$\frac{8\pi}{6}$};
\node at (-0.4,0.5) {$H^{\epsilon}_0$};
\node at (-0.4,1.5) {$H^{\epsilon}_1$};
\node at (-0.4,2.5) {$H^{\epsilon}_2$};
\node at (-0.4,3.5) {$H^{\epsilon}_3$};
\node at (-1.0,2.0) {\rotatebox{90}{homology degree}};
\node at ({2*pi/8},4.5) {%
    \begin{tikzpicture}[scale=0.425]
    \def\n{6}
    \foreach \k in {1,2,...,\n} {
       \coordinate (\k) at ({cos(deg(2*pi*\k/\n))},{sin(deg(2*pi*\k/\n))});
       \node at (\k) {$\bullet$};
    }
    \end{tikzpicture}
    };
\node (a3) at ({2*pi/8},3.5) {$0$};
\node (a2) at ({2*pi/8},2.5) {$0$};
\node (a1) at ({2*pi/8},1.5) {$0$};
\node (a0) at ({2*pi/8},0.5) {$\mathbb{F}^{6}$};
\node at ({3*pi/4},4.5) {%
    \begin{tikzpicture}[scale=0.425]
    \def\n{6}
    \foreach \k in {1,2,...,\n} {
       \coordinate (\k) at ({cos(deg(2*pi*\k/\n))},{sin(deg(2*pi*\k/\n))});
       \node at (\k) {$\bullet$};
       \draw[thick] ({cos(deg(2*pi*\k/\n))},{sin(deg(2*pi*\k/\n))}) -- ({cos(deg(2*pi*(\k+1)/\n))},{sin(deg(2*pi*(\k+1)/\n))});
    }
    \end{tikzpicture}
    };
\node (b3) at ({3*pi/4},3.5) {$0$};
\node (b2) at ({3*pi/4},2.5) {$0$};
\node (b1) at ({3*pi/4},1.5) {$\mathbb{F}$};
\node (b0) at ({3*pi/4},0.5) {$\mathbb{F}$};
\node at ({5*pi/4},4.5) {%
    \begin{tikzpicture}[scale=0.425]
    \def\n{6}
    \def\e{2.4}%2pi/n<epsilon<3pi/n gives a two-dimensional surface (n=6 gives 2-sphere) 
    \foreach \k in {1,2,...,\n} {
       \coordinate (\k) at ({cos(deg(2*pi*\k/\n))},{sin(deg(2*pi*\k/\n))});
       \node at (\k) {$\bullet$};
    %   \fill[opacity=0.1] (\k) circle (\e);
       \draw[thick] ({cos(deg(2*pi*\k/\n))},{sin(deg(2*pi*\k/\n))}) -- ({cos(deg(2*pi*(\k+1)/\n))},{sin(deg(2*pi*(\k+1)/\n))});
       \draw[dashed] ({cos(deg(2*pi*\k/\n))},{sin(deg(2*pi*\k/\n))}) -- ({cos(deg(2*pi*(\k+2)/\n))},{sin(deg(2*pi*(\k+2)/\n))});
    }
    \foreach \k in {1,3,5} {
       \draw[thick] ({cos(deg(2*pi*\k/\n))},{sin(deg(2*pi*\k/\n))}) -- ({cos(deg(2*pi*(\k+2)/\n))},{sin(deg(2*pi*(\k+2)/\n))});
    }
    \fill[opacity=0.1] ({cos(deg(2*pi*1/\n))},{sin(deg(2*pi*1/\n))}) -- ({cos(deg(2*pi*3/\n))},{sin(deg(2*pi*3/\n))}) -- ({cos(deg(2*pi*5/\n))},{sin(deg(2*pi*5/\n))}) -- cycle;
    \fill[opacity=0.25] ({cos(deg(2*pi*1/\n))},{sin(deg(2*pi*1/\n))}) -- ({cos(deg(2*pi*2/\n))},{sin(deg(2*pi*2/\n))}) -- ({cos(deg(2*pi*3/\n))},{sin(deg(2*pi*3/\n))}) -- cycle;
    \fill[opacity=0.15] ({cos(deg(2*pi*1/\n))},{sin(deg(2*pi*1/\n))}) -- ({cos(deg(2*pi*5/\n))},{sin(deg(2*pi*5/\n))}) -- ({cos(deg(2*pi*0/\n))},{sin(deg(2*pi*0/\n))}) -- cycle;
    \fill[opacity=0.35] ({cos(deg(2*pi*3/\n))},{sin(deg(2*pi*3/\n))}) -- ({cos(deg(2*pi*4/\n))},{sin(deg(2*pi*4/\n))}) -- ({cos(deg(2*pi*5/\n))},{sin(deg(2*pi*5/\n))}) -- cycle;
    \end{tikzpicture}
    };
\node (c3) at ({5*pi/4},3.5) {$0$};
\node (c2) at ({5*pi/4},2.5) {$\mathbb{F}$};
\node (c1) at ({5*pi/4},1.5) {$0$};
\node (c0) at ({5*pi/4},0.5) {$\mathbb{F}$};
\node at ({7*pi/4},4.5) {%
    \begin{tikzpicture}[scale=0.425]
    \def\n{6}
    \def\e{3.14}%2pi/n<epsilon<3pi/n gives a two-dimensional surface (n=6 gives 2-sphere) 
    \foreach \k in {1,2,...,\n} {
       \coordinate (\k) at ({cos(deg(2*pi*\k/\n))},{sin(deg(2*pi*\k/\n))});
       \node at (\k) {$\bullet$};
    %   \fill[opacity=0.1] (\k) circle (\e);
       \draw[thick] ({cos(deg(2*pi*\k/\n))},{sin(deg(2*pi*\k/\n))}) -- ({cos(deg(2*pi*(\k+1)/\n))},{sin(deg(2*pi*(\k+1)/\n))});
       \draw[thick] ({cos(deg(2*pi*\k/\n))},{sin(deg(2*pi*\k/\n))}) -- ({cos(deg(2*pi*(\k+2)/\n))},{sin(deg(2*pi*(\k+2)/\n))});
       \draw[thick] ({cos(deg(2*pi*\k/\n))},{sin(deg(2*pi*\k/\n))}) -- ({cos(deg(2*pi*(\k+3)/\n))},{sin(deg(2*pi*(\k+3)/\n))});
    }
    \fill[opacity=0.1] ({cos(deg(2*pi*1/\n))},{sin(deg(2*pi*1/\n))}) -- ({cos(deg(2*pi*2/\n))},{sin(deg(2*pi*2/\n))}) -- ({cos(deg(2*pi*3/\n))},{sin(deg(2*pi*3/\n))}) -- ({cos(deg(2*pi*4/\n))},{sin(deg(2*pi*4/\n))}) -- ({cos(deg(2*pi*5/\n))},{sin(deg(2*pi*5/\n))}) -- (1,0) -- cycle;
    \end{tikzpicture}
    };
\node (d3) at ({7*pi/4},3.5) {$0$};
\node (d2) at ({7*pi/4},2.5) {$0$};
\node (d1) at ({7*pi/4},1.5) {$0$};
\node (d0) at ({7*pi/4},0.5) {$\mathbb{F}$};
%chain maps
\foreach \n in {0,1,2,3} {
\draw[->] (a\n) -- (b\n);
\draw[->] (b\n) -- (c\n);
\draw[->] (c\n) -- (d\n);
}
\draw[->] (a0) -- node[above]{\small $\mathsf{f}_{0}^{(b,a)}$} (b0);
\draw[->] (b1) -- (c1);
\draw[->] (c1) -- (d1);
\draw[->] (c2) -- (d2);
\end{tikzpicture}
\end{center}
Here, $\mathsf{f}^{(b,a)}_{0}$ is given by the unique linear extension sending $\mathbf{e}^{(0)}_{j}$ to $1$, the generator of $\mathbb{F}$, for all $j\in\{1,\dots,6\}$. The other unlabelled maps are either all the zero map or the identity map if going from $\mathbb{F}$ to itself. If $\mathsf{f}^{(\epsilon',\epsilon)}_{k}:H^{\epsilon}_{k}\to H^{\epsilon'}_{k}$ denotes the associated map on the degree $k$ persistent homology spaces, then $\mathsf{f}^{(\epsilon',\epsilon)}_{k}$ is the identity whenever $\epsilon'-\epsilon\le\frac{2\pi}{6}$. Thus, for example, we say that the single nontrivial degree $1$ homology class persists on the interval $\left(\frac{2\pi}{6},\frac{4\pi}{6}\right]$. Similarly, the single nontrivial degree $2$ homology class persists on the interval $\left(\frac{4\pi}{6},\frac{6\pi}{6}\right]$.
\end{example}

Thus, combining our functors, we may summarize what we have done. First, for each $\epsilon\in[0,\infty)$, we have a sequence of functors 
\[
\xy 0;/r.25pc/:
   (-30,-7.5)*+{\mathbf{FinMetSpace}}="FMS";
   (-15,7.5)*+{\mathbf{SimpComp}}="SC";
   (15,7.5)*+{\mathbf{ChainComp}}="CC";
   (30,-7.5)*+{\mathbf{Vect}}="V";
   {\ar"FMS";"SC"^{\mathcal{R}^{\epsilon}}};
   {\ar"SC";"CC"^{C_{\bullet}}};
   {\ar"CC";"V"^{H_{n}}};
   {\ar"FMS";"V"^{H^{\epsilon}_{n}}};
\endxy
\]
where $H_{n}^{\epsilon}:=H_{n}\circ C_{\bullet}\circ \mathcal{R}^{\epsilon}$, 
which takes a finite metric space $(X,d_{X})$, produces a simplicial complex $\mathcal{R}^{\epsilon}(X,d_{X})$, which then is used to construct a chain complex $\big(C_{\bullet}^{\epsilon}(X,d_{X}),\partial_{\bullet}\big)$, which is then used to construct the $n^{\text{th}}$ homology vector space $H_{n}^{\epsilon}(X,d_{X})$. 

Moreover, given a pair of nonnegative numbers $\epsilon\le\epsilon'$, we obtain a natural transformation $H^{\epsilon}_{n}\Rightarrow H^{\epsilon'}_{n}$, which itself stems from the natural transformation $\mathcal{R}^{\epsilon}\Rightarrow\mathcal{R}^{\epsilon'}$. The meaning of the natural transformation $H^{\epsilon}_{n}\Rightarrow H^{\epsilon'}_{n}$ is that it assigns to every finite metric space $(X,d_{X})$ a linear transformation $H^{\epsilon}_{n}(X,d_{X})\to H^{\epsilon'}_{n}(X,d_{X})$, which in effect tracks the presence of homology classes at various scales. These maps detect persistent homology classes as defined in Example~\ref{ex:persistenthomologyofmetricspace}.

%%%%%%%%%%%%%%%%%%%%%%%%%%%%%%%%%%%%%%%%%%%%%%%%
\subsection{Persistence diagrams and the Bottleneck distance}
\label{subsec:bottleneck-distance}
%%%%%%%%%%%%%%%%%%%%%%%%%%%%%%%%%%%%%%%%%%%%%%%%

The bottleneck distance is a metric defined on the collection of persistence diagrams in topological data analysis, the latter of which are motivated by the examples of persistence vector spaces from Examples~\ref{ex:PabAtom} and~\ref{ex:PabsumPcd}. We first provide an intuitive idea before reviewing the rigorous definition. Firstly, a persistence diagram is a \emph{multiset} of points 
$(a_i,b_i)\in[0,\infty)\times[0,\infty]$, with $a_i<b_i$, indexed by a finite set $I$. 
Note that a multiset allows the possibility for repetitions, unlike in ordinary sets. Thus, a persistence diagram can be viewed as a function $I\to[0,\infty)\times[0,\infty]$ such that the second coordinate is always greater than the first coordinate. The purpose of a persistence diagram is that it provides a characterizing invariant of a persistence vector space up to isomorphism, much like the dimension of a vector space characterizes it up to isomorphism. Interestingly, unlike the case of vector spaces, where the difference between finite-dimensional vector spaces comes in discrete steps, the difference between two persistence vector spaces could be, in principal, any real number. 

The theorem characterizing persistence vector spaces in terms of such invariants utilizes the notion of an isomorphism between persistence vector spaces. 

\begin{definition} 
An \define{isomorphism of persistence vector spaces} is a morphism $\eta:\mathfrak{V}\to\mathfrak{W}$ of persistence vector spaces that satisfies the condition that there exists another morphism $\chi:\mathfrak{W}\to\mathfrak{V}$ of persistence vector spaces such that $\chi^{\epsilon}\circ\eta^{\epsilon}=\id_{V^{\epsilon}}$ and $\eta^{\epsilon}\circ \chi^{\epsilon}=\id_{W^{\epsilon}}$ for all $\epsilon\in[0,\infty)$. In such a case, $\mathfrak{V}$ and $\mathfrak{W}$ are said to be \define{isomorphic}. 
\end{definition}

The following notion of pointwise finite dimensionality will be used in a way analogous to finite dimensionality of vector spaces~\cite{CrawleyBoevey2015}. 

\begin{definition}
A persistence vector space $\mathfrak{V}=\{V^{\epsilon},f^{(\epsilon',\epsilon)}\}$ is \define{pointwise finite dimensional} iff $V^{\epsilon}$ is finite-dimensional for all $\epsilon\in[0,\infty)$. 
\end{definition}

\begin{example}
Let $(X,d_{X})$ be a finite metric space. Then the $n^\text{th}$ degree persistence homology  $H^{\epsilon}_{n}(X,d_{X})$ is a pointwise finite-dimensional persistence vector space. This follows immediately from the fact that the Vietoris--Rips simplicial complex $\mathcal{R}^{\epsilon}(X,d_{X})$ has a finite number of simplices for all $\epsilon\in[0,\infty)$ so that the associated chain complex $C_{n}^{\epsilon}(X,d_{X})$ is finite-dimensional for all $\epsilon\in[0,\infty)$ and for all $n\in\N\cup\{0\}$. 
\end{example}

We will now extend Example~\ref{ex:PabsumPcd} to allow for more direct sums. For this, we first recall that a \define{multiset} $\{\!\!\{z_{i}\}\!\!\}_{i\in I}$ of points in any set $Z$, indexed by a set $I$, is a function $I\to Z$ %\footnote
\footnote{A multiset of points in $Z$ is often defined as a function $Z\to\N$, where each element $z$ is associated with a natural number indicating its multiplicity. Since our index sets will always be finite, our present definition is equivalent to this one.}.
%end footnote
Such a multiset will either be written as $I\to Z$ or $\{\!\!\{z_{i}\}\!\!\}_{i\in I}$. 

\begin{example}
\label{ex:dsumPab}
Let $I\to [0,\infty)\times[0,\infty]$ be any multiset of points $(a_i,b_i)$ with $a_i<b_i$ for all $i\in I$. Then this defines a pointwise finite-dimensional persistence vector space given by 
\begin{equation}
\mathbb{F}(a_1,b_1)\oplus \mathbb{F}(a_2,b_2)\oplus\cdots\oplus \mathbb{F}(a_n,b_n),
\end{equation}
where $n=\# I$. 
As a concrete example, we can visually represent the persistence vector space 
\begin{equation}
\mathbb{F}(1,6)\oplus \mathbb{F}(2,5)\oplus \mathbb{F}(0,4) \oplus \mathbb{F}(3,9) \oplus \mathbb{F}(7,8) \oplus \mathbb{F}(2,5)
\end{equation}
as
\[
\begin{tikzpicture}[scale=0.8]
\draw (0,0) -- (10,0);
\foreach \x in {0,1,...,10} {
\draw (\x,0.1) -- (\x,-0.1) node[below]{$\x$};
}
\draw[very thick] (1,0.5) -- (6,0.5);
\draw[very thick] (2,1.0) -- (5,1.0);
\draw[very thick] (0,1.5) -- (4,1.5);
\draw[very thick] (3,2.0) -- (9,2.0);
\draw[very thick] (7,2.5) -- (8,2.5);
\draw[very thick] (2,3.0) -- (5,3.0);
\node (01) at (0.5,-0.8) {$\mathbb{F}$};
\node at (0.5,1.7) {\footnotesize $\mathbf{e}_1$};
\coordinate (01mid) at (1.0,-1.35);
\node (12) at (1.5,-0.8) {$\mathbb{F}^2$};
\node at (1.5,1.7) {\footnotesize $\mathbf{e}_1$};
\node at (1.5,0.7) {\footnotesize $\mathbf{e}_2$};
\coordinate (12mid) at (2.0,-1.35);
\node (23) at (2.5,-0.8) {$\mathbb{F}^4$};
\node at (2.5,1.7) {\footnotesize $\mathbf{e}_1$};
\node at (2.5,0.7) {\footnotesize $\mathbf{e}_2$};
\node at (2.5,1.2) {\footnotesize $\mathbf{e}_3$};
\node at (2.5,3.2) {\footnotesize $\mathbf{e}_4$};
\coordinate (23mid) at (3.0,-1.35);
\node (34) at (3.5,-0.8) {$\mathbb{F}^5$};
\node at (3.5,1.7) {\footnotesize $\mathbf{e}_1$};
\node at (3.5,0.7) {\footnotesize $\mathbf{e}_2$};
\node at (3.5,1.2) {\footnotesize $\mathbf{e}_3$};
\node at (3.5,3.2) {\footnotesize $\mathbf{e}_4$};
\node at (3.5,2.2) {\footnotesize $\mathbf{e}_5$};
\coordinate (34mid) at (4.0,-1.35);
\node (45) at (4.5,-0.8) {$\mathbb{F}^4$};
\node at (4.5,0.7) {\footnotesize $\mathbf{e}_1$};
\node at (4.5,1.2) {\footnotesize $\mathbf{e}_2$};
\node at (4.5,3.2) {\footnotesize $\mathbf{e}_3$};
\node at (4.5,2.2) {\footnotesize $\mathbf{e}_4$};
\coordinate (45mid) at (5.0,-1.35);
\node (56) at (5.5,-0.8) {$\mathbb{F}^2$};
\node at (5.5,0.7) {\footnotesize $\mathbf{e}_1$};
\node at (5.5,2.2) {\footnotesize $\mathbf{e}_2$};
\coordinate (56mid) at (6.0,-1.35);
\node (67) at (6.5,-0.8) {$\mathbb{F}$};
\node at (6.5,2.2) {\footnotesize $\mathbf{e}_1$};
\coordinate (67mid) at (7.0,-1.35);
\node (78) at (7.5,-0.8) {$\mathbb{F}^2$};
\node at (7.5,2.2) {\footnotesize $\mathbf{e}_1$};
\node at (7.5,2.7) {\footnotesize $\mathbf{e}_2$};
\coordinate (78mid) at (8.0,-1.35);
\node (89) at (8.5,-0.8) {$\mathbb{F}$};
\node at (8.5,2.2) {\footnotesize $\mathbf{e}_1$};
\coordinate (89mid) at (9.0,-1.35);
\node (910) at (9.5,-0.8) {$\{0\}$};
\draw[->] (01) to[out=-70,in=180] (01mid) node[below]{\tiny\scalebox{0.8}{$\left[\!\begin{smallmatrix}1\\0\end{smallmatrix}\!\right]$}} to[out=0,in=250]  (12);
\draw[->] (12) to[out=-70,in=180] (12mid) node[below]{\tiny\scalebox{0.8}{$\left[\!\begin{smallmatrix}1&0\\0&1\\0&0\\0&0\end{smallmatrix}\!\right]$}} to[out=0,in=250]  (23);
\draw[->] (23) to[out=-70,in=180] (23mid) node[below]{\tiny\scalebox{0.8}{$\left[\!\begin{smallmatrix}1&0&0&0\\0&1&0&0\\0&0&1&0\\0&0&0&1\\0&0&0&0\end{smallmatrix}\!\right]$}} to[out=0,in=250] (34);
\draw[->] (34) to[out=-70,in=180] (34mid) node[below]{\tiny\scalebox{0.8}{$\left[\!\begin{smallmatrix}0&1&0&0&0\\0&0&1&0&0\\0&0&0&1&0\\0&0&0&0&1\end{smallmatrix}\!\right]$}} to[out=0,in=250] (45);
\draw[->] (45) to[out=-70,in=180] (45mid) node[below]{\tiny\scalebox{0.8}{$\left[\!\begin{smallmatrix}1&0&0&0\\0&0&0&1\end{smallmatrix}\!\right]$}} to[out=0,in=250] (56);
\draw[->] (56) to[out=-70,in=180] (56mid) node[below]{\tiny\scalebox{0.8}{$\left[\!\begin{smallmatrix}0&1\end{smallmatrix}\!\right]$}} to[out=0,in=250] (67);
\draw[->] (67) to[out=-70,in=180] (67mid) node[below]{\tiny\scalebox{0.8}{$\left[\!\begin{smallmatrix}1\\0\end{smallmatrix}\!\right]$}} to[out=0,in=250] (78);
\draw[->] (78) to[out=-70,in=180] (78mid) node[below]{\tiny\scalebox{0.8}{$\left[\!\begin{smallmatrix}1&0\end{smallmatrix}\!\right]$}} to[out=0,in=250] (89);
\draw[->] (89) to[out=-70,in=180] (89mid) node[below]{\tiny\scalebox{0.8}{$\left[\!\begin{smallmatrix}\end{smallmatrix}\!\right]$}} to[out=0,in=250] (910);
\end{tikzpicture}
\]
The matrices shown here are defined with respect to a particular choice and ordering of the basis elements. That basis choice is illustrated as a unit vector $\mathbf{e}_{i}$ above each vector space on each corresponding chunk of the interval. 
\end{example}

The importance of this example is that it provides the general structure of any pointwise finite-dimensional persistence vector space (cf.\ Theorem~1.1 in Ref.~\cite{CrawleyBoevey2015}). 

\begin{proposition}
\label{prop:pfdpersdecomp}
Every pointwise finite-dimensional persistence vector space $\mathfrak{V}$ over a field $\mathbb{F}$ is isomorphic to a persistence vector space of the form 
\begin{equation}
\mathbb{F}(a_1,b_1)\oplus \cdots\oplus \mathbb{F}(a_n,b_n)
\end{equation}
for some multiset $\{\!\!\{(a_i,b_i)\}\!\!\}_{i\in I}$ with $n=\# I$ and 
$a_{i}\in[0,\infty)$, $b_{i}\in[0,\infty]$ satisfying $a_{i}<b_{i}$ for all $i\in I$ (if $n=0$, then the persistence vector space is $\{0\}$). 
Moreover, if $\mathfrak{V}$ is also isomorphic to
\begin{equation}
\mathbb{F}(c_1,d_1)\oplus \cdots\oplus \mathbb{F}(c_m,d_m)
\end{equation}
for some multiset $\{\!\!\{(c_j,d_j)\}\!\!\}_{j\in J}$ with $m=\# J$ and $c_{j}\in[0,\infty)$, $d_{j}\in[0,\infty]$ satisfying $c_{j}<d_{j}$ for all $j\in J$, then there exists a bijection $\gamma : I \to J$ such that $c_{\gamma(i)}=a_{i}$ and $d_{\gamma(i)}=b_{i}$ for all $i \in I$. 
\end{proposition}

Due to such a classification of persistence vector spaces, it is useful to introduce the persistence analogues of dimension, which are barcode diagrams~\cite{Carlsson2004barcodes} and persistence diagrams~\cite{edelsbrunner2002topological} (see also Refs.~\cite{ghrist2008barcodes,tinarrage2021notes} for a detailed review). 

\begin{definition}
A \define{birth-death diagram} is a multiset $\{\!\!\{(a_{i},b_{i})\}\!\!\}_{i\in I}$ of points in $[0,\infty)\times[0,\infty]$, indexed by a finite set $I$, with $a_{i}< b_{i}$ for all $i\in I$. 
For each $i\in I$, the elements $a_{i}$ and $b_{i}$ are referred to as the \define{birth time} and \define{death time} associated with $i$, respectively. 
Associated with a birth-death diagram is its \define{barcode diagram}, which is the subset 
\begin{equation}
\bigcup_{i\in I}\big\{(x,i)\in\R\times I\;:\;b_{i}\le x<d_{i}\big\}.
\end{equation}
Let $\mathfrak{V}$ be a pointwise finite-dimensional persistence vector space. The \define{persistence diagram} associated with $\mathfrak{V}$ (and similarly its associated \define{barcode diagram}) is the birth-death diagram (resp.\ barcode diagram) of the multiset obtained from some choice of isomorphism $\mathfrak{V}\cong\bigoplus_{i\in I}\mathbb{F}(a_i,b_i)$. 
\end{definition}

A simple example for why we allow the second coordinate in a birth-death diagram to be infinity is degree $0$ homology (the Vietoris--Rips complex of a finite metric space becomes connected for all $\epsilon$ past some large enough value). Also, note that the persistence diagram and barcode diagram of a pointwise finite-dimensional persistence vector space are well-defined by Proposition~\ref{prop:pfdpersdecomp}. If a total ordering of $I$ is given $I=\{1,\dots,n\}$, then the barcode diagram can be viewed as a subset of horizontal intervals in $\R^{2}$. 

\begin{example}
Example~\ref{ex:dsumPab} has already shown an example of a barcode diagram. The associated birth-death diagram is given by 
\[
\begin{tikzpicture}[scale=0.325]
\draw (0,0) -- (10,0);
\draw (0,0) -- (0,10);
\foreach \x in {0,1,...,10} {
\draw (\x,0.1) -- (\x,-0.1) node[below]{$\x$};
\draw (0.1,\x) -- (-0.1,\x) node[left]{$\x$};
}
\fill[gray,opacity=0.3] (0,0) -- (10,10) -- (10,0) -- cycle;
\node at (1,6) {$\bullet$};
\node at (2,5) {$\bullet$};
\node at (0,4) {$\bullet$};
\node at (3,9) {$\bullet$};
\node at (7,8) {$\bullet$};
\node at (5,-1.75) {birth};
\node at (-1.75,5) {\rotatebox{90}{death}};
\end{tikzpicture}
\]
Notice that visually it is not possible to see the multiplicity in a plot of the birth-death diagram as drawn above, though the multiplicity is encoded in the mathematical definition through the usage of multisets. Also note that the bottom half of the first quadrant has been grayed out because birth-death diagram points $(a_i,b_i)$ are not allowed to appear in this subset due to the condition $a_i<b_i$. 
\end{example}

Proposition~\ref{prop:pfdpersdecomp} also allows us to define the persistence diagram associated with finite metric spaces. 

\begin{definition}
\label{defn:pbarcode}
Let $(X,d_{X})$ be a finite metric space. The \define{$n^{\text{th}}$ degree persistence diagram} of $(X,d_{X})$ is the birth-death diagram of the $n^{\text{th}}$ degree homology persistence vector space $H_{n}^{\epsilon}(X,d_{X})$. This persistence diagram is denoted by $\mathrm{dgm}_{n}(X,d_{X})$. The \define{$n^{\text{th}}$ degree persistence barcode} of $(X,d_{X})$ is the associated barcode diagram. The \define{persistence barcode} of $(X,d_{X})$ is the disjoint union of all $n^{\text{th}}$ degree persistence barcodes of $(X,d_{X})$ for $n\in\N$. 
\end{definition}

\begin{example}
Revisiting the $n^{\text{th}}$ roots of unit from Example~\ref{ex:Hausmannrootshomology} with $n=6$, we can draw the associated birth-death and persistence barcode diagrams. Since there are multiple degrees, we label the different degrees with different symbols and shadings so that all homology degrees can be visualized on the same diagrams. 
\begin{center}
\begin{tikzpicture}[scale=0.625]
\draw[->] (0,0) -- ({10*pi/6},0);
\draw[->] (0,0) -- (0,{10*pi/6});
\foreach \x in {2,4,6,8} {
\draw ({\x*pi/6},0.1) -- ({\x*pi/6},-0.1) node[below]{$\frac{\x\pi}{n}$};
\draw (0.1,{\x*pi/6}) -- (-0.1,{\x*pi/6}) node[left]{$\frac{\x\pi}{n}$};
}
\fill[gray,opacity=0.3] (0,0) -- ({10*pi/6},{10*pi/6}) -- ({10*pi/6},0) -- cycle;
\node at (0,{2*pi/6}) {\Large $\bullet$};
\node at (0,{10*pi/6}) {\Large $\bullet$};
\node at ({2*pi/6},{4*pi/6}) {\large $\blacklozenge$};
\node at ({4*pi/6},{6*pi/6}) {$\blacksquare$};
\draw (0.1,{10*pi/6}) -- (-0.1,{10*pi/6}) node[left]{$\infty$};
\draw (0.1,0) -- (-0.1,0);
\draw (0,0.1) -- (0,-0.1) node[below]{$0$};
\node at (4.0,1.75) {\small \begin{tabular}{|ll|}\multicolumn{2}{c}{Legend} \\ 
\hline 
{\Large $\bullet$} & $H_{0}$\\
{\large $\blacklozenge$} & $H_{1}$ \\
{$\blacksquare$} & $H_{2}$\\
\hline
\end{tabular}};
\end{tikzpicture}
\quad
\begin{tikzpicture}[scale=0.625]
\def\sep{0.3}
\draw[->] (0,0) -- ({10*pi/6},0) node[right]{$\epsilon$};
\draw (0,0.1) -- (0,-0.1) node[below]{$0$};
\foreach \x in {2,4,6,8} {
\draw ({\x*pi/6},0.1) -- ({\x*pi/6},-0.1) node[below]{$\frac{\x\pi}{n}$};
}
\foreach \k in {1,3,5,7,9} {
\draw[pattern=dots] (0,{\k*\sep}) rectangle ({2*pi/6},{(\k+1)*\sep});
}
\draw[pattern=dots] (0,{11*\sep}) rectangle ({10*pi/6},{12*\sep});
\draw[pattern=north east lines] ({2*pi/6},{13*\sep}) rectangle ({4*pi/6},{14*\sep});
\draw[pattern=north west lines] ({4*pi/6},{15*\sep}) rectangle ({6*pi/6},{16*\sep});
\node at (4.0,1.75) {\small \begin{tabular}{|ll|}\multicolumn{2}{c}{Legend} \\ 
\hline 
\begin{tikzpicture}\draw[pattern=dots] (0,0) rectangle (0.75,0.2);\end{tikzpicture}& $H_{0}$\\
\begin{tikzpicture}\draw[pattern=north east lines] (0,0) rectangle (0.75,0.2);\end{tikzpicture}& $H_{1}$ \\
\begin{tikzpicture}\draw[pattern=north west lines] (0,0) rectangle (0.75,0.2);\end{tikzpicture}& $H_{2}$\\
\hline
\end{tabular}};
\end{tikzpicture}
\end{center}
Notice that the multiplicity of the degree $0$ classes is not visible in the birth-death diagram on the left. 
\end{example}

The bottleneck distance calculates the difference between two persistence diagrams. To define the bottleneck distance, we review the concept of matchings and cost of a matching~\cite{oudot2017persistence,tinarrage2021notes}.  

\begin{definition}
\label{defn:pmatchingsets}
Given two sets $I$ and $J$, a \define{partial matching} between $I$ and $J$ is a subset $M\subseteq I\times J$ such that 
\begin{itemize}
\item for every $j\in J$, there exists at most one $i\in I$ such that $(i,j)\in M$ and 
\item for every $i\in J$, there exists at most one $j\in J$ such that $(i,j)\in M$. 
\end{itemize}
An element $i\in I$ is said to be \define{matched} by $M$ when $(i,j)\in M$ for some (necessarily unique) $j\in J$. Similarly, $j\in J$ is said to be \define{matched} by $M$ when ${(i,j)\in M}$ for some (necessarily unique) $i\in I$. An element ${(i,j)\in M}$ is called a \define{matched pair}. An element $i\in I$ or $j\in J$ is \define{unmatched} iff $i\in I\setminus\pi_{I}(M)$ or $j\in J\setminus\pi_{J}(M)$, respectively, where $\pi_{I}:I\times J\to I$ and $\pi_{J}:I\times J\to J$ are the projection maps. The set of unmatched points is denoted by
\begin{equation}
M^{u}:=\big(I\setminus\pi_{I}(M)\big){\textstyle{\coprod}}\big(J\setminus\pi_{J}(M)\big).
\end{equation}
A partial matching $M$ as above is denoted by ${I\xleftrightarrow{M}J}$.
\end{definition}

The above definition was for partial matchings between \emph{sets}. The following definition is for \emph{multisets}. 

\begin{definition}
\label{defn:pmatchingmultisets}
Given two multisets $P:I\to A$ and $Q:J\to B$, a \define{partial matching} between $P$ and $Q$ is a multiset $M\to A\times B$ with $M$ a partial matching between $I$ and $J$. A \define{matched pair} in $M$ is an element $(a_i,b_j)\in A\times B$ with $(i,j)\in M$. An element $a_i\in A$ in the image of $I\to A$ is \define{unmatched} iff $i$ is unmatched, and similarly for $b_{j}\in B$.
A matched pair is written as $(a_i,b_j)\in M$, while unmatched points $p$ are written as $p\in M^{u}$, by an abuse of notation. A partial matching as above is denoted by ${P\xleftrightarrow{M}Q}$. 
\end{definition}

\begin{definition}
\label{defn:costofpmatching}
Let $P:I\to[0,\infty)\times[0,\infty]$ and $Q:J\to[0,\infty)\times[0,\infty]$ be two birth-death diagrams, and let $P\xleftrightarrow{M}Q$ be a partial matching. The \define{cost} of a matched pair $(p,q):=\big((a_{i},b_{i}),(a_{j},b_{j})\big)\in M$ is the $\ell^{\infty}$ distance between $p$ and $q$, i.e., 
\begin{equation}
\lVert p-q\rVert_{\infty}=\max\big\{|a_i-a_j|,|b_i-b_j|\big\}.
\end{equation}
The \define{cost} of an unmatched point $s=(a,b)\in M^{u}$ is $\frac{|b-a|}{2}$, which is the $\ell^{\infty}$ distance between $s$ and its closest point to the diagonal in $\R^{2}$. The \define{bottleneck cost} of the partial matching $M$ is 
\begin{equation}
C(M)=\max\left\{\sup_{(p,q)\in M}\lVert p-q\rVert_{\infty},\sup_{(a,b)\in M^{u}}\frac{|a-b|}{2}\right\}.
\end{equation}
The \define{bottleneck distance} between the birth-death diagrams $P$ and $Q$ is the smallest bottleneck cost over all partial matchings, i.e., 
\begin{equation}
d_{B}(P,Q):=\inf\left\{C(M)\;:\;P\xleftrightarrow{M}Q\right\}.
\end{equation}
\end{definition}

The concept of a bottleneck distance can be transferred to a distance between pointwise finite-dimensional persistence vector spaces due to Proposition~\ref{prop:pfdpersdecomp}. 

\begin{definition}
If $\mathfrak{V}$ and $\mathfrak{W}$ are pointwise finite-dimensional persistence vector spaces with 
\begin{equation}
\mathfrak{V}\cong\bigoplus_{i\in I}\mathbb{F}(a_i,b_i)
\quad\text{and}\quad
\mathfrak{W}\cong\bigoplus_{j\in J}\mathbb{F}(c_j,d_j),
\end{equation}
then the \define{bottleneck distance} between $\mathfrak{V}$ and $\mathfrak{W}$ is 
\begin{equation}
d_{B}(\mathfrak{V},\mathfrak{W}):=d_{B}(P,Q),
\end{equation}
where $P:I\to[0,\infty)\times[0,\infty]$ and $Q:J\to[0,\infty)\times[0,\infty]$ are the birth-death diagrams 
\begin{equation}
P(i):=(a_i,b_i)
\quad\text{ and }\quad
Q(j):=(c_j,d_j)
\end{equation}
for all $i\in I$ and $j\in J$. 
\end{definition}

Although the bottleneck distance does not define a distance in the sense of a metric, it does define a semimetric (see Section 3.1.1.\ in Ref.~\cite{oudot2017persistence}), which is defined as follows (it is called an \emph{uber metric} in Ref.~\cite{Spivak09}). 

\begin{definition}
\label{defn:semimetricspace}
A \define{semi-metric space} is a pair $(X,d_{X})$, where $X$ is a set and $d_{X}:X\times X\to[0,\infty]$ is a function that satisfies 
\begin{enumerate}[i.]
\item $d_{X}(x,x)=0$ for all $x\in X$
\item $d_{X}(x_1,x_2)=d_{X}(x_2,x_1)$ for all $x_1,x_2\in X$
\item $d_{X}(x_{1},x_{2})\le d_{X}(x_{1},x_{3})+d_{X}(x_2,x_3)$ for all $x_1,x_2,x_3\in X$ (the \emph{triangle inequality}).
\end{enumerate}
Such a function $d_{X}$ is called a \define{semimetric} (or an \define{extended pseudometric}). 
\end{definition}

There are two differences between a metric and a semimetric. A metric additionally requires that $d_{X}(x,y)=0$ implies $x=y$. Moreover, a metric does not allow the distance to be infinite. Thus, a semimetric is slightly more general than a metric. 

\begin{proposition}
The bottleneck distance defines a semimetric on the set of partial matchings $P:\{1,\dots,n\}\to[0,\infty)\times[0,\infty]$, with $n\in\N$ arbitrary. 
\end{proposition}

We can apply the notion of the bottleneck distance to the barcode diagram associated with persistence homologies of two metric spaces $(X,d_{X})$ and $(Y,d_{Y})$. We will do this momentarily, but first we provide some basic examples computing the bottleneck distance. 

\begin{example}
As a first example, we will follow Ref.~\cite{tinarrage2021notes} and compute the bottleneck distance between the barcodes associated with the persistence vector spaces $\mathbb{F}(a,b)$ and $\mathbb{F}(c,d)$. Here, there are only two partial matchings, which are $\varnothing$ (the empty set) and $\big\{\big((a,b),(c,d)\big)\big\}$. If $M_{0}:=\varnothing$, then the bottleneck cost is 
\begin{equation}
C(M_{0})=\max\left\{\frac{|b-a|}{2},\frac{|d-c|}{2}\right\}.
\end{equation}
If $M_{1}:=\big\{\big((a,b),(c,d)\big)\big\}$, then the bottleneck cost is 
\begin{equation}
C(M_{1})=\max\big\{|a-c|,|b-d|\big\}.
\end{equation}
Therefore, the bottleneck distance between the barcodes associated with $\mathbb{F}(a,b)$ and $\mathbb{F}(c,d)$ is 
\begin{equation}
d_{B}\Big(\big\{(a,b)\big\},\big\{(c,d)\big\}\Big)=\min\big\{C(M_{0}),C(M_{1})\big\}.
\end{equation}
\end{example}

\begin{example}
Revisiting the $n^\text{th}$ roots of unity from Example~\ref{ex:Hausmannrootsofunity}, we can look at the persistence barcodes when $n=20$ and $n=50$. We can then compare these results to the ordinary homology of the standard unit circle. For $n=20$, the following barcode diagram is obtained. 
\begin{center}
\begin{tikzpicture}[scale=0.625]
\def\sep{0.1}
\def\wid{0.3}
\draw[->] (0,0) -- ({24*pi/6},0) node[right]{$\epsilon$};
\draw (0,0.1) -- (0,-0.1) node[below]{$0$};
\foreach \x in {2,4,6,8,10,12,14,16,18,20,22} {
\draw ({\x*pi/6},0.1) -- ({\x*pi/6},-0.1) node[below]{$\frac{\x\pi}{n}$};
}
\foreach \k in {1,2,3} {
    \draw[pattern=dots] (0,{\k*(\wid+\sep))}) rectangle ({2*pi/6},{(\k+1)*\wid+\k*\sep});
    }
\foreach \k in {6,7,8} {
    \draw[pattern=dots] (0,{\k*(\wid+\sep))}) rectangle ({2*pi/6},{(\k+1)*\wid+\k*\sep});
    }
    \node at ({pi/6},{5.25*(\wid+\sep)}) {$\vdots$};
    \node at ({9*pi/12},{5.5*\wid+3*\sep}) {$\left.\begin{matrix}\phantom{x}\\\phantom{x}\\\phantom{x}\\\phantom{x}\\\phantom{x}\end{matrix}\right\}$ 19 bars};
\draw[pattern=dots] (0,{9*(\wid+\sep)}) rectangle ({24*pi/6},{10*\wid+9*\sep});
\draw[pattern=north east lines] ({2*pi/6},{10*(\wid+\sep)}) rectangle ({14*pi/6},{11*\wid+10*\sep});
\draw[pattern=vertical lines] ({14*pi/6},{11*(\wid+\sep)}) rectangle ({16*pi/6},{12*\wid+11*\sep});
\foreach \k in {12,13,14} {
    \draw[pattern=checkerboard] ({16*pi/6},{\k*(\wid+\sep))}) rectangle ({18*pi/6},{(\k+1)*\wid+\k*\sep});
    }
\node at (11.0,1.875) {\small \begin{tabular}{|ll|}\multicolumn{2}{c}{Legend} \\ 
\hline 
\begin{tikzpicture}\draw[pattern=dots] (0,0) rectangle (0.75,0.2);\end{tikzpicture}& $H_{0}$\\
\begin{tikzpicture}\draw[pattern=north east lines] (0,0) rectangle (0.75,0.2);\end{tikzpicture}& $H_{1}$ \\
\begin{tikzpicture}\draw[pattern=vertical lines] (0,0) rectangle (0.75,0.2);\end{tikzpicture}& $H_{3}$\\
\begin{tikzpicture}\draw[pattern=checkerboard] (0,0) rectangle (0.75,0.2);\end{tikzpicture}& $H_{4}$\\
\hline
\end{tabular}};
\end{tikzpicture}
\end{center}
As for $n=50$, we get the following barcode diagram 
\begin{center}
\begin{tikzpicture}[scale=0.5625]
\def\sep{0.1}
\def\wid{0.3}
\draw[->] (0,0) -- ({54*pi/12},0) node[right]{$\epsilon$};
\draw (0,0.1) -- (0,-0.1) node[below]{$0$};
\foreach \x in {4,8,12,...,52} {
\draw ({\x*pi/12},0.1) -- ({\x*pi/12},-0.1) node[below]{$\frac{\x\pi}{n}$};
}
\foreach \k in {1,2,3,7,8,9} {
    \draw[pattern=dots] (0,{\k*(\wid+\sep))}) rectangle ({2*pi/12},{(\k+1)*\wid+\k*\sep});
    }
    \node at ({pi/12},{6*\wid+5*\sep}) {$\vdots$};
    \node at ({7*pi/12},{5.5*\wid+5*\sep}) {$\left.\begin{matrix}\phantom{x}\\\phantom{x}\\\phantom{x}\\\phantom{x}\\\phantom{x}\end{matrix}\right\}$ 49 bars};
\draw[pattern=dots] (0,{10*(\wid+\sep)}) rectangle ({54*pi/12},{11*\wid+10*\sep});
\draw[pattern=north east lines] ({2*pi/12},{11*(\wid+\sep)}) rectangle ({34*pi/12},{12*\wid+11*\sep});
\draw[pattern=vertical lines] ({34*pi/12},{12*(\wid+\sep)}) rectangle ({40*pi/12},{13*\wid+12*\sep});
\foreach \k in {13,14,...,21} {
    \draw[pattern=checkerboard] ({40*pi/12},{\k*(\wid+\sep))}) rectangle ({42*pi/12},{(\k+1)*\wid+\k*\sep});
    }
\draw[pattern=horizontal lines gray] ({42*pi/12},{22*(\wid+\sep)}) rectangle ({44*pi/12},{23*\wid+22*\sep});
\node at (13.0,6.725) {\small \begin{tabular}{|ll|}\multicolumn{2}{c}{Legend} \\ 
\hline 
\begin{tikzpicture}\draw[pattern=dots] (0,0) rectangle (0.75,0.2);\end{tikzpicture}& $H_{0}$\\
\begin{tikzpicture}\draw[pattern=north east lines] (0,0) rectangle (0.75,0.2);\end{tikzpicture}& $H_{1}$ \\
\begin{tikzpicture}\draw[pattern=vertical lines] (0,0) rectangle (0.75,0.2);\end{tikzpicture}& $H_{3}$\\
\begin{tikzpicture}\draw[pattern=checkerboard] (0,0) rectangle (0.75,0.2);\end{tikzpicture}& $H_{4}$\\
\begin{tikzpicture}\draw[pattern=horizontal lines gray] (0,0) rectangle (0.75,0.2);\end{tikzpicture}& $H_{5}$\\
\hline
\end{tabular}};
\end{tikzpicture}
\end{center}
As we can see, for a wide range of low values of $\epsilon$, the homology is exactly the same as the homology of the standard circle. 
The precise statements about convergence of such approximations is given by Hausmann's theorem, Latschev's theorem, and their extensions~\cite{Hausmann96,Latschev2001,Majhi2025}. 
\end{example}

%%%%%%%%%%%%%%%%%%%%%%%%%%%%%%%%%%%%%%%%%%%%%%%%
\subsection{The stability theorem in persistent homology}
\label{subsec:stability}
%%%%%%%%%%%%%%%%%%%%%%%%%%%%%%%%%%%%%%%%%%%%%%%%

Because persistent homology is usually constructed from a sampled dataset, it would be convenient if these topological invariants did not vary drastically for different samples or if there was some slight noise. The notion of \emph{stability} aims to address questions of robustness of persistent homology to such perturbations. A version of the stability theorem suitable for our purposes is the following~\cite{Chazal2009Stable} (see also Refs.~\cite{ChoietalTDA2025,memoli2011metric,MemoliSinghal19} for reviews and Refs.~\cite{CCGGO09,BauerLesnick2015Stability,gardner2017stability} for alternative notions of stability in the context of TDA). 

\begin{proposition}
\label{prop:stabilityV1}
Let $(X,d_X)$ and $(Y,d_Y)$ be finite metric spaces. Then 
\begin{equation}
\frac{1}{2}d_{B}\big(\mathrm{dgm}_{n}(X),\mathrm{dgm}_{n}(Y))\le d_{GH}\big((X,d_X),(Y,d_Y)\big)
\end{equation}
for all $n\ge0$. 
\end{proposition}

In the statement of Proposition~\ref{prop:stabilityV1}, $\mathrm{dgm}_{n}(X):=\mathrm{dgm}_{n}(X,d_{X})$ is the $n^{\text{th}}$ degree persistence diagram from Definition~\ref{defn:pbarcode} (and similarly for $\mathrm{dgm}_{n}(Y)$), $d_{B}$ is the bottleneck distance from Definition~\ref{defn:costofpmatching}, and $d_{GH}$ is the Gromov--Hausdorff distance from Definition~\ref{defn:GHdistance}. The factor of $\frac{1}{2}$ comes from the fact that we have been using the Vietoris--Rips complex to construct the persistent homology, and there would be no factor of $\frac{1}{2}$ if we used the {\v C}ech complex~\cite{Virk2022}. In more practical terms, Proposition~\ref{prop:stabilityV1} says that if the bottleneck distances between the barcodes of two point clouds $X$ and $Y$ are large, then the underlying spaces from which $X$ and $Y$ are sampled from are probably different. 

For our purposes, we will apply Proposition~\ref{prop:stabilityV1} to the following metric spaces. Given a point cloud $X\subset\R^{m}$, which represents our classical dataset, let $d_{X}$ denote the induced Euclidean metric on $X$ from $\R^{m}$. Let $\rho:X\to\states(\mathcal{H})$ be some injective quantum encoding for some Hilbert space $\mathcal{H}$, and equip $\states(\mathcal{H})$ with some metric, such as the quantum Bures fidelity distance $d_{F}$ from Definition~\ref{defn:quantumdistances}. Now let $Y=\rho(X)$ be the image of $X$ under the encoding $\rho$ and equip $Y$ with the metric $d_{Y}$ given by the induced metric from $d_{F}$ (equivalently, we can set $Y$ to be $X$ itself and equip $Y$ with the pull-back metric $d_{\rho}$ from Definition~\ref{defn:pullbackmetric}). This gives us two finite metric spaces $(X,d_{X})$ and $(Y,d_{Y})$. Because the quantum encoding might distort distances, we can use Proposition~\ref{prop:stabilityV1} as a test to see whether or not a given encoding distorts the distances in a way that changes the associated persistent homologies thereby altering the homological inference on the underlying manifold from which the point cloud $X$ was sampled from. 

Explicitly, we can deduce two things. First, if the Gromov--Hausdorff distance between the classical dataset $(X,d_{X})$ and the quantum encoded dataset $(Y,d_{Y})$ is small, then we can conclude that the quantum encoding $\rho$ will preserve the persistent homology to good accuracy. Conversely, if the persistent homologies of $(X,d_{X})$ and $(Y,d_{Y})$ are substantially different from each other, then the quantum encoding $\rho$ is not adequately preserving the pairwise distances in the original data set.

%%%%%%%%%%%%%%%%%%%%%%%%%%%%%%%%%%%%%%%%%%%%%%%%
\section{Quantum encodings for persistent homology}
\label{sec:main}
%%%%%%%%%%%%%%%%%%%%%%%%%%%%%%%%%%%%%%%%%%%%%%%%

Proposing quantum encodings that preserve persistent homology is not really a well-defined question if one is simply given a dataset. The question depends highly on the metric used to describe the distances between the data points. Once a metric is specified, one arrives at the associated persistent homology, and only then can one ask for a quantum encoding that accurately preserves that persistent homology. In other words, preservation of the persistent homology implicitly assumes a metric has been chosen, whether or not that metric is the most accurate one that describes the data. With this in mind, this section aims to propose some quantum encodings that preserve the persistent homology that will depend on the choice of such a metric. 

Before describing these encodings, we define what it means to scale a metric. The purpose of doing this is because if the persistent homologies are related by an overall scaling factor, then they should be deemed equivalent. This observation has been made elsewhere in the literature~\cite{MayKrishnamoorthyGambill24} (Ref.~\cite{indyk20178} used similar concepts under the name \emph{distortion}). Therefore, we will define a \emph{scale-free distortion} that takes this into account. 

\begin{definition}
Given a metric space $(X,d_{X})$ and a positive number $\lambda$, the metric $\lambda\, d_{X}$ on $X$ is defined by $(\lambda\, d_{X})(x_1,x_2):=\lambda\, d_{X}(x_1,x_2)$ for all $x_1,x_2\in X$. The new metric space $(X,\lambda \, d_{X})$ is said to be \define{rescaled}. 
\end{definition}

\begin{definition}
Let $(X,d_{X})$ and $(Y,d_{Y})$ be finite metric spaces, and let $f:X\to Y$ be a function (not necessarily distance-preserving). The \define{scale-free distortion} of $f$ is the number given by 
\begin{equation}
\mathrm{sfdis}(f):=\inf_{\lambda>0} \sup_{x,x'\in X}\left| d_{X}(x,x')-\frac{1}{\lambda}d_{Y}\big(f(x),f(x')\big)\right|.
\end{equation}
\end{definition}
As a simple example, if $(X,d_{X})$ is a metric space with just two elements, and if $f$ is an injective function, then the scale-free distortion is zero. The reason for the $\frac{1}{\lambda}$ factor in front of $d_{Y}$ (instead of a $\lambda$ factor in front of $d_{X}$) is because $f$ need not be injective. As an extreme case, if $X$ consists of two points $x$ and $x'$ and $f(x)=f(x')$, then $\mathrm{sfdis}(f)=d_{X}(x,x')$, which is the same as its distortion. 
The idea of a scale-free distortion will be important for quantum encodings because quantum encodings take finite datasets from ambient spaces that may have unbounded distance (such as Euclidean space) to a quantum state spaces, which have finite distance with respect to the metrics we introduced on quantum states. However, the inferred topology from persistence homology only changes by an overall scaling factor when the metric is rescaled by a number. Because we will need to often rescale classical data before encoding it into a quantum system, we should therefore infer the topology of a dataset in a way that is agnostic to these overall scalings. The scale-free distortion achieves this, as the following proposition shows. %Burago defines homothety in proposition 1.7.8. 

\begin{proposition}
Let $(X,d_{X})$ and $(Y,d_{Y})$ be finite metric spaces, and let $f:X\to Y$ be a function (not necessarily distance-preserving). Suppose that there exists a $\lambda\ge0$ such that $\lambda \,d_{X}(x,x')= d_{Y}(f(x),f(x'))$ for all $x,x'\in X$ (such an $f$ is called a \define{$\lambda$-homothety}). Then $\mathrm{sfdis}(f)=0$ and the persistence barcode associated with $(X,\lambda\, d_X)$ is the same as the persistence barcode associated with $(X,d_{f})$. 
\end{proposition}

Let us now briefly revisit some of the quantum encodings we have already seen in Section~\ref{subsec:feature-maps}. We first note that it is rarely possible to preserve persistent homology with amplitude encoding as in Example~\ref{ex:ampencoding}, which can drastically distort the distances between the original data set, even if we use variations on amplitude encoding as described in Ref.~\cite{PBVP24}. One can make a precise statement in this direction. In what follows, given a metric space $(X,d_{X})$ the number 
\begin{equation}
\mathrm{diam}(X,d_{X}):=\sup_{x,x'\in X}d_{X}(x,x')
\end{equation}
denotes the \define{diameter} of $(X,d_{X})$. When $d_{X}$ is understood, it is sometimes denoted by $\mathrm{diam}(X)$. 

\begin{theorem}
\label{thm:amplitudeencodingsucksforTDA}
Let $(X,d_{X})$ be a point cloud in $\R^{2^{d}}\setminus\{0\}$, with $d\in\N$. Let $\rho:\R^{2^{d}}\setminus\{0\}\to\states(\C^{2^d})$ be the amplitude encoding map from~\eqref{eqn:amplitudeencoding}, where $\states(\C^{2^{d}})$ is equipped with 
the Bures fidelity metric from Example~\ref{defn:quantumdistances}, and let $\rho|_{X}$ be the restriction of $\rho$ to the point cloud $X$. Then the following facts hold. 
\begin{enumerate}[i.]
\item 
The distortion of $\rho|_{X}$ satisfies
\begin{equation}
\label{eqn:ampencinequality}
0\le\mathrm{dis}(\rho|_{X})\le
\max\big\{\mathrm{diam}(X,d_{X}),\mathrm{diam}(Y,d_{Y})\big\},
\end{equation}
where $Y=\rho(X)$ and $d_{Y}$ is the induced metric from $\states(\C^{2^d})$.
Moreover, there is an example of a dataset $X$ satisfying $\mathrm{dis}(\rho|_{X})=0$, another dataset $X$ satisfying $\mathrm{dis}(\rho|_{X})=\mathrm{diam}(X,d_{X})$, and finally for every $\epsilon$ with $0<\epsilon<1$, there exists a dataset $X$ such that $\#X$ is constant (independent of $\epsilon$) and $\mathrm{dis}(\rho|_{X})=\mathrm{diam}(Y,d_{Y})-\epsilon$, respectively. 
\item
The scale-free distortion of $\rho|_{X}$ satisfies 
\begin{equation}
\label{eqn:ampencinequalitysfdis}
0\le\mathrm{sfdis}(\rho|_{X})\le
\max\big\{\mathrm{diam}(X,d_{X}),\mathrm{diam}(Y,d_{Y})\big\}.
\end{equation}
Moreover, there is an example of a dataset $X$ satisfying $\mathrm{sfdis}(\rho|_{X})=0$ and another dataset $X$ satisfying $\mathrm{sfdis}(\rho|_{X})=\mathrm{diam}(X,d_{X})$.  
\item
There exists a sequence of finite metric spaces $(X_{n},d_{X_n})$ in $\R^{2^{d}}\setminus\{0\}$ such that $\#X_{n}$ is constant (independent of $n$) and $\lim_{n\to\infty}\mathrm{sfdis}(\rho|_{X_n})=\infty$.
\end{enumerate}
\end{theorem}

We note that the fact that $\mathrm{dis}(\rho)=\infty$ when $\rho$ is defined on all of Euclidean space minus the origin follows from the fact that the Euclidean distance is unbounded, whereas every metric from Example~\ref{defn:quantumdistances} is bounded. The last claim about the scale-free distortion is the key point of this theorem, because it states that preserving the topology of all datasets using amplitude encoding is generally not possible (even up to an overall scale). 

\begin{proof}
[Proof of Theorem~\ref{thm:amplitudeencodingsucksforTDA}]
Equation~\eqref{eqn:ampencinequality} follows from the simple inequalities
\begin{equation}
|u-v|\le|u|
\quad\text{ and }\quad
|u-v|\le|v|, 
\end{equation}
which hold for all $u,v\in[0,\infty)$. As for the examples, we mention the three cases. 
If $X\subset S^{2^{d}-1}\cap\R^{2^{d}}_{\ge0}$, so that each element of $X$ is a vector with non-negative entries that also lies on the unit sphere in Euclidean space, then the distortion of $\rho|_{X}$ is $0$ if the metric $d_{Y}$ is $\sqrt{2}\,d_{F}$, where $d_{F}$ is the Bures fidelity metric. This is because the Euclidean distance between such points satisfies $d_{X}(\mathbf{x}_j,\mathbf{x}_k)=\sqrt{2}\sqrt{1-\mathbf{x}_j^{T}\mathbf{x}_k}$, which is equal to the Bures fidelity distance up to an overall scalar factor. Thus, amplitude encoding is a homothety in this case, which implies that the scale-free distortion is zero. The distortion is exactly zero when the data have all non-negative entries and lie on the sphere of radius $\frac{1}{\sqrt{2}}$. 
An example where $\dis(\rho|_{X})=\mathrm{diam}(X,d_{X})$ is $X=\{ \mathbf{e}_{1}, c\mathbf{e}_{1} \}$, where $\mathbf{e}_{1}$ is the first standard unit vector and $c>1$. In this case, $\mathrm{diam}(Y,d_{Y})=0$, while $\mathrm{diam}(X,d_{X})=c$, the latter of which could be arbitrarily large. This example also satisfies $\mathrm{sfdis}(\rho|_{X})=\mathrm{diam}(X,d_{X})$. 
Now, given an $0<\epsilon<1$, an example where $\mathrm{dis}(\rho|_{X})=\mathrm{diam}(Y,d_{Y})-\epsilon$ is $X=\big\{ \frac{\epsilon}{2}(\mathbf{e}_{1}+\mathbf{e}_{2}),\frac{\epsilon}{2}(\mathbf{e}_{1}-\mathbf{e}_{2})\big\}$. In this case, $\mathrm{diam}(X,d_{X})=\epsilon$, while $\mathrm{diam}(Y,d_Y)=1$. However, the scale-free distortion in this case is $0$ (because the scale-free distortion is zero for a metric space of two points and an injective encoding). 
As for the final claim, for each $n\in\N$, let 
\begin{equation}
X_{n}:=\left\{\mathbf{e}_{1},n\mathbf{e}_{1},\frac{1}{n}(\mathbf{e}_{1}+\mathbf{e}_{2}),\frac{1}{n}(\mathbf{e}_{1}-\mathbf{e}_{2})\right\}
\end{equation}
(see Figure~\ref{fig:ampencodscale}).
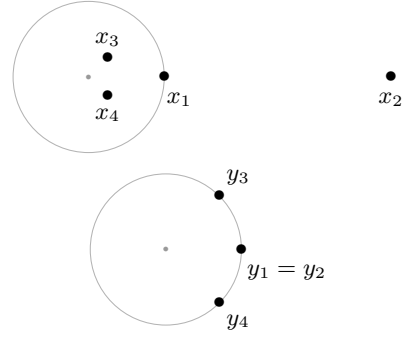
\begin{figure}
\begin{tabular}{cc}
\begin{tikzpicture}
\draw[opacity=0.3] (0,0) circle (1cm);
\node at (1,0) {$\bullet$};
\node[opacity=0.4] at (0,0) {\scalebox{0.5}{$\bullet$}};
\node at (1.2,-0.3) {$x_1$};
\node at (0.25,0.25) {$\bullet$};
\node at (0.25,0.5) {$x_3$};
\node at (0.25,-0.25) {$\bullet$};
\node at (0.25,-0.5) {$x_4$};
\node at (4,0) {$\bullet$};
\node at (4,-0.3) {$x_2$};
\end{tikzpicture} 
\\
\begin{tikzpicture}
\draw[opacity=0.3] (0,0) circle (1cm);
\node at (1,0) {$\bullet$};
\node[opacity=0.4] at (0,0) {\scalebox{0.5}{$\bullet$}};
\node at (1.6,-0.3) {$y_1=y_2$};
\node at ({1/sqrt(2)},{1/sqrt(2)}) {$\bullet$};
\node at ({0.25+1/sqrt(2)},{0.25+1/sqrt(2)}) {$y_3$};
\node at ({1/sqrt(2)},{-1/sqrt(2)}) {$\bullet$};
\node at ({0.25+1/sqrt(2)},{-0.25-1/sqrt(2)}) {$y_4$};
\end{tikzpicture} 
\end{tabular}
\caption{The top image shows a finite metric space $X$ consisting of four points $x_1,x_2,x_3,x_4$ in Euclidean space, together with the standard unit circle drawn for reference. When $n$ is large, the distance between $x_3$ and $x_4$ is small, but their distance to $x_2$ is large. The bottom image shows amplitude encoding applied to $X$ so that the image lies in the real subspace of the Hilbert space $\C^{2}$ consisting of vectors with real entries. The vectors $y_3$ and $y_4$ are now perpendicular, so their distance is the maximum possible value, which is $1$. Meanwhile $y_1=y_2$ due to the lack of injectivity of amplitude encoding so that their distance is now $0$. This difference shows that the scale-free distortion of amplitude encoding when applied to such a dataset can be arbitrarily large.}
\label{fig:ampencodscale}
\end{figure}
Using the Bures fidelity distance, the scale-free distortion for amplitude encoding is $n-1$ for $n\ge2$. Hence, the scale-free distortion tends to infinity as $n$ increases. Although the actual value of the scale-free distortion is slightly different for the Hilbert--Schmidt or trace distances, the same conclusion holds as $n$ increases. 
\end{proof}

We have also already seen in Figure~\ref{fig:ABCqenc} that dense angle encoding also distorts distances, bringing distant points closer to each other in quantum state space due to the periodicity of the encoding. A similar statement can be made for ordinary angle encoding. IQP encoding also distorts distances, and as of now the square-root encoding is only valid for data that already lie on the standard simplices. We will actually illustrate how several of these encodings distort the persistent homology in Section~\ref{sec:experiments} in explicit examples. We are thus left wondering if there are \emph{any} quantum encodings that could preserve persistent homology. 

Indeed, note that it is generally \emph{impossible} to preserve the Euclidean distance in quantum state space, simply due to the fact that $\states(\mathcal{H})$ is compact and all of the metrics introduced for it in Section~\ref{sec:metricspaces} are bounded, while the Euclidean distance is always unbounded. A second reason is that even if $X$ is a compact subset of Euclidean space with diameter bounded by $1$, it is not generally possible to isometrically embed this set into the space of pure states equipped with one of the metrics introduced in Section~\ref{sec:metricspaces} due to its intrinsic curvature and non-Euclidean geometry since the set of pure states is a complex projective space. 

We provide a few resolutions to these problems. One approach, which we describe in detail in Section~\ref{sec:MAE}, involves first \emph{globally} (i.e., uniformly) rescaling the metric proportional to the size of the data set, which is assumed finite, and then applying a sequence of isometries followed by diagonal encoding. The scaling and isometry component of such an approach has the benefit that the persistent homology and associated barcodes will all scale uniformly by the same constant factor, so that the same topological inference can be made before and after encoding. We note that nonuniform scaling, such as what occurs when standardizing in multivariate statistics (i.e., rescaling a random variable by its standard deviation)~\cite{Strang2019Data}, may distort the persistent homology and barcodes, as shown in Figure~\ref{fig:figure8}. 

\begin{figure}
\includegraphics{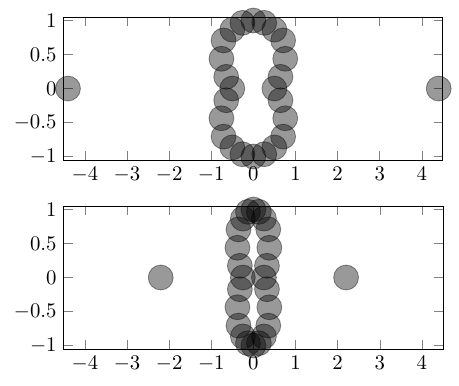}
\caption{The top image is one of a centered data set with disks of a fixed radius and each data point at the center of one disk. The bottom image is the same data set after being rescaled along the horizontal direction so that the standard deviations in each direction are identical. The sizes of the discs are the same in each figure. The inferred topology of the data set is different due to the rescaling. More convincing examples can be constructed.}
\label{fig:figure8}
\end{figure}

One downside to uniformly scaling is that it depends on the given data. However, it seems that there are not many alternative options because rescaling methods that do not depend on the data, such as a sigmoid function $x\mapsto \frac{1}{1+e^{-x}}$ that can be used to smoothly map $\R$ onto $(0,1)$, drastically distort distances. As such, the rescaling methods that we will utilize will depend on the data to avoid such distortions. Intuitively, if we take the manifold hypothesis seriously~\cite{FeMiNa2016}, so that our data are sampled from a compact submanifold of Euclidean space (or compact subspace in general), then a uniform rescaling depending on the data set would allow for as much preservation of persistent homology as possible.  

Another approach, which we discuss in Section~\ref{sec:QMDS}, is to theoretically minimize the distortion of a mapping. Namely, given a finite metric space $(X,d_{X})$ and a fixed Hilbert space $\mathcal{H}$, with some metric $d_{Y}$ chosen on the quantum state space $Y=\states(\mathcal{H})$, which mapping $\rho: X\to Y$ minimizes the distortion of $f$, as defined in Definition~\ref{defn:distortionf}? A related problem appears in the context of classical multidimensional scaling~\cite{DeLeeuwHeiser1982}, where $Y$ is replaced with $\R^{k}$ for some $k\in\N$ and $d_{Y}$ the Euclidean metric. Therefore, since the space of pure states in a Hilbert space is identified with complex projective space $\CP^{n}$ for some $n\in\N$, the question we are asking is to find a mapping $\rho:X\to\CP^{n}$ that minimizes the distortion, which is a generalization of multidimensional scaling to complex projective spaces. We will develop such an approach as a possible answer to what quantum encodings best preserve persistent homology. Such a distortion-minimizing encoding will provide a theoretical lower bound on the achievable distortion for \emph{any} quantum encoding. In other words, every other quantum encoding will have distortion greater than this optimal one. Such a result will then aid in the search for quantum circuit design for encodings that preserve persistent homology. Namely, encodings that come close to this theoretical limit can be considered as close-to-optimal encodings for purposes of persistent homology. 

%%%%%%%%%%%%%%%%%%%%%%%%%%%%%%%%%%%%%%%%%%%%%%%
\subsection{Uniformly transformed diagonal encoding}
\label{sec:MAE}
%%%%%%%%%%%%%%%%%%%%%%%%%%%%%%%%%%%%%%%%%%%%%%%

Motivated by the distance-preserving properties of diagonal encoding from Example~\ref{ex:diagonalencoding}, this section will extend the applicability of diagonal encoding for datasets in $\R^{m}$, rather than just data already coming from a simplex. Our main result in this section is an \emph{exact} preservation of persistent homology for Euclidean datasets (point clouds).

Let $X\subset\R^{m}$ be a finite set, interpreted as a dataset or point cloud, inside $m$-dimensional Euclidean space equipped with the Euclidean metric. In what follows, we will apply a sequence of metric embeddings and isometries in order to embed $X$ into the $m$-dimensional plane $\Pi_{m}$ described by the equation $x_1+\cdots+x_m+x_{m+1}=1$ inside $\R^{m+1}$, i.e.,
\[
\Pi_{m}=\big\{(x_1,\dots,x_{m+1})\in\R^{m+1}\,:\,x_1+\cdots+x_{m+1}=1\big\}.
\]
Thus, the distances between elements in the dataset will be preserved and therefore any calculations of persistent homology will be invariant under such a transformation. This sequence of steps are illustrated in Figure~\ref{fig:shiftrescaledata}. 

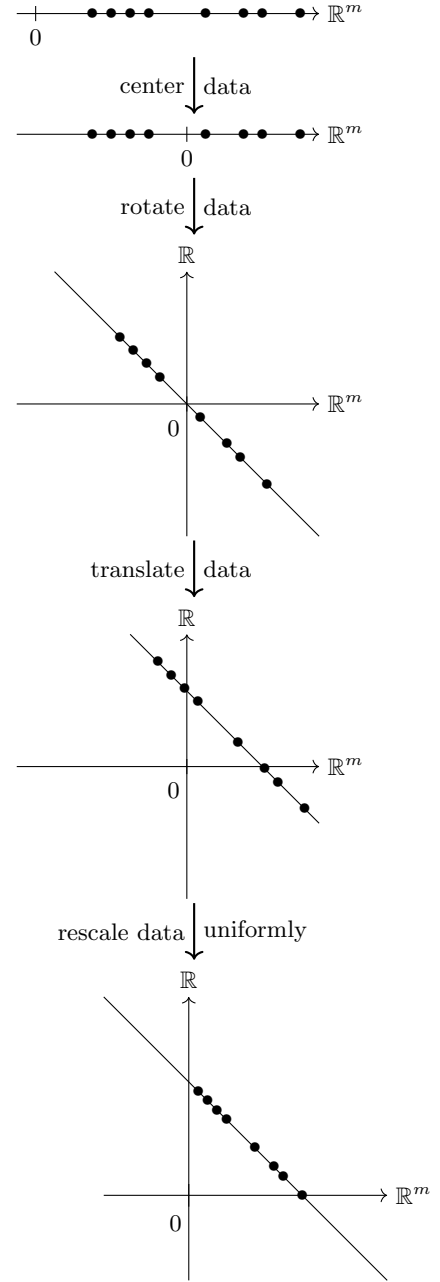
\begin{figure}
\begin{tikzpicture}
\draw[->] (-0.25,0) -- (3.75,0) node[right]{$\R^{m}$};
\draw (0,0.1) -- (0,-0.1) node[below]{$0$};
\node at (0.75,0) {$\bullet$};
\node at (1.0,0) {$\bullet$};
\node at (1.25,0) {$\bullet$};
\node at (1.5,0) {$\bullet$};
\node at (2.25,0) {$\bullet$};
\node at (3,0) {$\bullet$};
\node at (3.5,0) {$\bullet$};
\node at (2.75,0) {$\bullet$};
%sum=16
%average=2
%largest radius=3.5-2=1.5
\end{tikzpicture}
\begin{tikzpicture}
\node at (-2,0) {};
\node at (2,0) {};
\draw[thick,->] (0,0.25) -- node[left]{center} node[right]{data} (0,-0.5);
\end{tikzpicture}
\begin{tikzpicture}
\def\c{2};
\draw[->] ({-0.25-\c},0) -- ({3.75-\c},0) node[right]{$\R^{m}$};
\draw (0,0.1) -- (0,-0.1) node[below]{$0$};
\node at ({0.75-\c},0) {$\bullet$};
\node at ({1.0-\c},0) {$\bullet$};
\node at ({1.25-\c},0) {$\bullet$};
\node at ({1.5-\c},0) {$\bullet$};
\node at ({2.25-\c},0) {$\bullet$};
\node at ({3-\c},0) {$\bullet$};
\node at ({3.5-\c},0) {$\bullet$};
\node at ({2.75-\c},0) {$\bullet$};
\end{tikzpicture}
\begin{tikzpicture}
\node at (-2,0) {};
\node at (2,0) {};
\draw[thick,->] (0,0.25) -- node[left]{rotate} node[right]{data} (0,-0.5);
\end{tikzpicture}
\begin{tikzpicture}
\def\c{2};
\draw[->] ({-0.25-\c},0) -- ({3.75-\c},0) node[right]{$\R^{m}$};
\draw[->] (0,-1.75) -- (0,1.75) node[above]{$\R$};
\draw (0,0.1) -- (0,-0.1) node[below,xshift=-5pt]{$0$};
\draw (-1.75,1.75) -- (1.75,-1.75);
\node at ({(0.75-\c)/sqrt(2)},{-(0.75-\c)/sqrt(2)}) {$\bullet$};
\node at ({(1.0-\c)/sqrt(2)},{-(1.0-\c)/sqrt(2)}) {$\bullet$};
\node at ({(1.25-\c)/sqrt(2)},{-(1.25-\c)/sqrt(2)}) {$\bullet$};
\node at ({(1.5-\c)/sqrt(2)},{-(1.5-\c)/sqrt(2)}) {$\bullet$};
\node at ({(2.25-\c)/sqrt(2)},{-(2.25-\c)/sqrt(2)}) {$\bullet$};
\node at ({(3-\c)/sqrt(2)},{-(3-\c)/sqrt(2)}) {$\bullet$};
\node at ({(3.5-\c)/sqrt(2)},{-(3.5-\c)/sqrt(2)}) {$\bullet$};
\node at ({(2.75-\c)/sqrt(2)},{-(2.75-\c)/sqrt(2)}) {$\bullet$};
\end{tikzpicture}
\begin{tikzpicture}
\node at (-2,0) {};
\node at (2,0) {};
\draw[thick,->] (0,0.25) -- node[left]{translate} node[right]{data} (0,-0.5);
\end{tikzpicture}
\begin{tikzpicture}
\def\c{2};
\draw[->] ({-0.25-\c},0) -- ({3.75-\c},0) node[right]{$\R^{m}$};
\draw[->] (0,-1.75) -- (0,1.75) node[above]{$\R$};
\draw (0,0.1) -- (0,-0.1) node[below,xshift=-5pt]{$0$};
\draw ({-1.75+1},1.75) -- (1.75,{-1.75+1});
\node at ({(0.75-\c)/sqrt(2)+1/2},{-(0.75-\c)/sqrt(2)+1/2}) {$\bullet$};
\node at ({(1.0-\c)/sqrt(2)+1/2},{-(1.0-\c)/sqrt(2)+1/2}) {$\bullet$};
\node at ({(1.25-\c)/sqrt(2)+1/2},{-(1.25-\c)/sqrt(2)+1/2}) {$\bullet$};
\node at ({(1.5-\c)/sqrt(2)+1/2},{-(1.5-\c)/sqrt(2)+1/2}) {$\bullet$};
\node at ({(2.25-\c)/sqrt(2)+1/2},{-(2.25-\c)/sqrt(2)+1/2}) {$\bullet$};
\node at ({(3-\c)/sqrt(2)+1/2},{-(3-\c)/sqrt(2)+1/2}) {$\bullet$};
\node at ({(3.5-\c)/sqrt(2)+1/2},{-(3.5-\c)/sqrt(2)+1/2}) {$\bullet$};
\node at ({(2.75-\c)/sqrt(2)+1/2},{-(2.75-\c)/sqrt(2)+1/2}) {$\bullet$};
\end{tikzpicture}
\begin{tikzpicture}
\node at (-2,0) {};
\node at (2,0) {};
\draw[thick,->] (0,0.25) -- node[left]{rescale data} node[right]{uniformly} (0,-0.5);
\end{tikzpicture}
\begin{tikzpicture}[scale=1.5]
\def\c{2};
\def\d{1.5*sqrt(2)}
\node at ({-0.25-\c+0.2},0) {};
\draw[->] ({-0.25+1.5-\c},0) -- ({3.75-\c},0) node[right]{$\R^{m}$};
\draw[->] (0,{-1.75+1}) -- (0,1.75) node[above]{$\R$};
\draw (0,0.1) -- (0,-0.1) node[below,xshift=-5pt]{$0$};
\draw ({-1.75+1},1.75) -- (1.75,{-1.75+1});
\node at ({(0.75-\c)/(sqrt(2)*\d)+1/2},{-(0.75-\c)/(sqrt(2)*\d)+1/2}) {$\bullet$};
\node at ({(1.0-\c)/(sqrt(2)*\d)+1/2},{-(1.0-\c)/(sqrt(2)*\d)+1/2}) {$\bullet$};
\node at ({(1.25-\c)/(sqrt(2)*\d)+1/2},{-(1.25-\c)/(sqrt(2)*\d)+1/2}) {$\bullet$};
\node at ({(1.5-\c)/(sqrt(2)*\d)+1/2},{-(1.5-\c)/(sqrt(2)*\d)+1/2}) {$\bullet$};
\node at ({(2.25-\c)/(sqrt(2)*\d)+1/2},{-(2.25-\c)/(sqrt(2)*\d)+1/2}) {$\bullet$};
\node at ({(3-\c)/(sqrt(2)*\d)+1/2},{-(3-\c)/(sqrt(2)*\d)+1/2}) {$\bullet$};
\node at ({(3.5-\c)/(sqrt(2)*\d)+1/2},{-(3.5-\c)/(sqrt(2)*\d)+1/2}) {$\bullet$};
\node at ({(2.75-\c)/(sqrt(2)*\d)+1/2},{-(2.75-\c)/(sqrt(2)*\d)+1/2}) {$\bullet$};
\end{tikzpicture}
\caption{A sequence of three affine transformations, each of which leaves the Euclidean distances between all points invariant, followed by a uniform rescaling. The first step centers the data. The second step applies a rotation inside $\R^{m+1}$. The third step translates this subspace to the hyperplane $\Pi_{m}$. In the last step, a uniform rescaling is applied so that the data fit onto the simplex, so that the data vectors can be interpreted as probability vectors and where the average of the data is at the uniform probability. Although the last step does not define an isometry, the distances are all uniformly scaled proportional to the diameter of the data set. Note that the first and last steps depend on the given data set. Additionally, it is convenient to implement the rescaling immediately after centering the data, which we do when summarizing this procedure in~\eqref{eqn:TXmapping}.}
\label{fig:shiftrescaledata}
\end{figure}

By choosing an ordering of the $n$ elements of $X$, the dataset $X$ can be represented as an $m\times n$ matrix of column vectors 
\begin{equation}
A_{0}=\begin{bmatrix}|&&|\\\mathbf{a}_{1}&\cdots&\mathbf{a}_{n}\\|&&|\end{bmatrix}.
\end{equation}
Let $\mathbf{1}_{n}\in\R^{n}$ be the vector of all $1$'s, which is not to be confused with the $n\times n$ identity matrix denoted by $\mathds{1}_{n}$. Centering the data matrix $A_{0}$ by subtracting the averages of each row from each entry in that row can be expressed by the matrix equation $A_{0}C$, where
\begin{equation}
\label{eq:centeringmatrix}
C:=\mathds{1}_{n}-\frac{1}{n}\mathbf{1}_{n}\mathbf{1}_{n}^{T}
\end{equation}
is the \define{centering matrix}. 
Next, let $A'$ be the $(m+1)\times n$ matrix obtained by including an additional row of $0$'s at the bottom of $A_{0}C$, which can also be described by the matrix operation 
\begin{equation}
A'=
\begin{bmatrix}\mathds{1}_{m}\\0\end{bmatrix}A_{0}C.
\end{equation}

Next, we will construct a particular $(m+1)\times(m+1)$ orthogonal (rotation) matrix $R_{m}$ that sends the plane 
\begin{equation}
P_{m}=\big\{(x_1,\dots,x_m,0)\in\R^{m+1}\big\}
\end{equation}
bijectively and isometrically to the plane 
\begin{equation}
\Pi_{m}^{(0)}=\left\{(x_1,\dots,x_{m+1})\in\R^{m+1}\,:\,\sum_{j=1}^{m+1}x_{j}=0\right\}.
\end{equation}
To set up the construction, let 
\begin{equation}
\mathbf{s}_{m}=\frac{1}{\sqrt{m+1}}(1,\dots,1)
\end{equation}
be the positive unit vector in $\R^{m+1}$ that is normal to the plane $\Pi_{m}$ (and hence also $\Pi_{m}^{(0)}$). Let $\mathbf{e}_{m+1}=(0,\dots,0,1)$ be the last standard unit vector in $\R^{m+1}$. Let 
\begin{align}
\mathbf{u}_{m}&=\frac{\mathbf{e}_{m+1}-\mathbf{e}_{m+1}^{T}\mathbf{s}_{m}}{\lVert \mathbf{e}_{m+1}-\mathbf{e}_{m+1}^{T}\mathbf{s}_{m}\rVert} \nonumber \\
&=\frac{1}{\sqrt{m(m+1)}}(-1,\dots,-1,m)
\end{align}
be the vector perpendicular to $\mathbf{s}_{m}$ and in $\mathrm{span}(\mathbf{s}_{m},\mathbf{e}_{m+1})$. The vector $\mathbf{u}_{m}$ is obtained by the standard Gram--Schmidt orthogonalization procedure~\cite{Strang2022} applied to the ordered pair of vectors $(\mathbf{s}_{m},\mathbf{e}_{m+1})$, as illustrated 
in the following image ($\theta_{m}$ is the angle from $\mathbf{s}_{m}$ to $\mathbf{e}_{m+1}$):

\begin{center}
\begin{tikzpicture}[scale=0.45]
\draw[ultra thick,->](0,0)--(0,5.0) node[above]{$\mathbf{e}_{m+1}$};
\draw[ultra thick,->] (0,0) -- (4.0,3.0) node[right]{$\mathbf{s}_{m}$};
\draw[ultra thick,->] (0,0) -- (-3.0,4.0) node[left]{$\mathbf{u}_{m}$};
\draw[->] (0,0) +(36.87:1cm) arc (36.87:90:1cm) node[above,xshift=10pt]{$\theta_{m}$} ;
\end{tikzpicture}
\end{center}

The rotation matrix $R_{m}$ that we want is a rotation by angle $-\theta_{m}$ that sends the vector $\mathbf{e}_{m+1}$ to $\mathbf{s}_{m}$ and fixes the subspace perpendicular to $\mathrm{span}(\mathbf{s}_{m},\mathbf{e}_{m+1})$. Viewing the vectors $\mathbf{s}_{m}$ and $\mathbf{u}_{m}$ as column vectors, such a rotation matrix is given by 
\begin{align}
R_{m}&=\cos(\theta_{m})\mathbf{s}_{m}\mathbf{s}_{m}^{T}+\sin(\theta_{m}) \mathbf{s}_{m}\mathbf{u}_{m}^{T} \nonumber \\
&-\sin(\theta_{m})\mathbf{u}_{m}\mathbf{s}_{m}^{T}+\cos(\theta_{m})\mathbf{u}_{m}\mathbf{u}_{m}^{T} \nonumber \\
&+\mathds{1}_{m+1}-\mathbf{s}_{m}\mathbf{s}_{m}^{T}-\mathbf{u}_{m}\mathbf{u}_{m}^{T}. \label{eqn:Rmmatrixa}
\end{align}
The first four terms describe a rotation in the plane $\mathrm{span}(\mathbf{s}_{m},\mathbf{e}_{m+1})$ by angle $-\theta_{m}$ (compare with the standard $2\times 2$ matrix for a rotation). The last three terms say that the orthogonal complement $\mathrm{span}(\mathbf{s}_{m},\mathbf{e}_{m+1})^{\perp}$ remains fixed. The matrix $R_{m}$ is indeed an orthogonal matrix, as one can check that $R_{m}^{T}R_{m}=\mathds{1}_{m+1}$. Before checking that this rotation matrix sends the plane $P_{m}\subset\R^{m+1}$ (given by $x_{m+1}=0$) to $\Pi_{m}^{(0)}$, we will first simplify expression~\eqref{eqn:Rmmatrixa} as follows. First, notice that 
\begin{equation}
\cos(\theta_{m})=\mathbf{e}_{m+1}^{T}\mathbf{s}_{m}=\frac{1}{\sqrt{m+1}},
\end{equation}
which then implies
\begin{equation}
\sin(\theta_{m})=\frac{\sqrt{m}}{\sqrt{m+1}}.
\end{equation}
Therefore, 
\begin{align}
R_{m}&=\mathds{1}_{m+1}+\frac{1-\sqrt{m+1}}{\sqrt{m+1}}\left(\mathbf{s}_{m}\mathbf{s}_{m}^{T}+\mathbf{u}_{m}\mathbf{u}_{m}^{T}\right) \nonumber \\
&+\sqrt{\frac{m}{m+1}}\left(\mathbf{s}_{m}\mathbf{u}_{m}^{T}-\mathbf{u}_{m}\mathbf{s}_{m}^{T}\right).
\end{align}
With this expression, we next check that $R_{m}$ sends the plane $P_{m}$ to the plane $\Pi_{m}^{(0)}$ by showing that $R_{m}^{T}$, the inverse of $R_{m}$, takes the plane $\Pi_{m}^{(0)}$ to $P_{m}$. If for each $i\in\{1,\dots,m+1\}$, we set 
\begin{equation}
\mathbf{v}_{i}:=\mathbf{e}_{i}-\mathbf{c}_{m}
\end{equation}
to be the vertices of the standard $m$-simplex $\Delta_{m}\subset\R^{m+1}$ after it has been shifted by the negative of the vector 
\begin{equation}
\mathbf{c}_{m}=\frac{1}{m+1}(1,\dots,1)=\frac{1}{\sqrt{m+1}}\mathbf{s}_{m}
\end{equation}
in $\R^{m+1}$, then the inverse of $R_{m}$, which is $R_{m}^{T}$, should send each of the vectors $\mathbf{v}_{i}$ to vectors in the plane $P_{m}$. This can indeed be verified by showing that $\mathbf{e}_{m+1}^{T}R_{m}^{T}\mathbf{v}_{i}=0$ for all $i\in\{1,\dots,m+1\}$.

For example, in dimensions $m=1$ and $m=2$, the rotation matrix $R_{m}$ is given by 
\begin{equation}
R_{1}=\frac{1}{\sqrt{2}}\begin{bmatrix}1&1\\-1&1\end{bmatrix},
\end{equation}
which is illustrated in Figure~\ref{fig:shiftrescaledata}, 
and
\begin{equation}
\label{eq:rotation2to3}
R_{2}=\frac{1}{\sqrt{3}}\begin{bmatrix}\frac{1+\sqrt{3}}{2}&\frac{1-\sqrt{3}}{2}&1\\\frac{1-\sqrt{3}}{2}&\frac{1+\sqrt{3}}{2}&1\\-1&-1&1\end{bmatrix},
\end{equation}
which is illustrated in Figure~\ref{fig:2drescaling}, 
respectively. 
For general $m$, it is given by 
\begin{equation}
R_{m}=
\frac{1}{\sqrt{m+1}}
\begin{bmatrix}
\frac{1+(m-1)\sqrt{m+1}}{m}&\cdots&\frac{1-\sqrt{m+1}}{m}&1\\
\vdots&&\vdots&\vdots\\
\frac{1-\sqrt{m+1}}{m}&\cdots&\frac{1+(m-1)\sqrt{m+1}}{m}&1\\
-1&\cdots&-1&1
\end{bmatrix},
\end{equation}
where the top-left $m\times m$ submatrix has all diagonal entries given as shown, and all the other entries take the same value as the off-diagonal entries shown. More precisely, the $ij$ entry of $R_{m}$ is given by 
\begin{equation}
\mathbf{e}_{i}^{T}R_{m}\mathbf{e}_{j}=%\frac{1}{\sqrt{m+1}}
\begin{cases}
\frac{1+(m-1)\sqrt{m+1}}{m\sqrt{m+1}} & \mbox{if $i,j\in[m]$, $i=j$}\\
\frac{1-\sqrt{m+1}}{m\sqrt{m+1}} & \mbox{if $i,j\in[m]$, $i\ne j$}\\
-\frac{1}{\sqrt{m+1}}&\mbox{if $i=m+1$, $j\in[m]$}\\
\frac{1}{\sqrt{m+1}}&\mbox{if $i\in[m+1]$, $j=m+1$},
\end{cases}
\end{equation}
where we've used the shorthand notation $[k]:=\{1,\dots,k\}$. 

Next, consider the affine transformation that translates all vectors in $\R^{m+1}$ by the vector $\mathbf{c}_{m}$, which, in particular, translates the origin $0$ to the center of the $m$-simplex $\Delta_{m}$. If we have an $(m+1)\times n$ matrix $A''=R_{m}A'$, then such an affine transformation can be described using matrix operations as $A''+S$, where 
\begin{equation}
S:=\frac{1}{m+1}\mathbf{1}_{m+1}\mathbf{1}_{n}^{T}.
\end{equation}

Next, in order to rescale the data so that it lies in the $m$-simplex $\Delta_{m}\subset\R^{m+1}$, we need to calculate the scalar factor. To obtain this, one can uniformly scale the centered dataset so that it lies within the largest $(m-1)$-dimensional sphere inside of $\R^{m}$ that fits inside a regular $m$-dimensional tetrahedron whose side lengths are $\sqrt{2}$ and whose center is at the origin. Such a tetrahedron will be isometric to the $m$-simplex $\Delta^{m}$ inside $\R^{m+1}$, and the radius of the largest inscribed $(m-1)$-dimensional sphere that fits inside this tetrahedron is $\frac{1}{\sqrt{m(m+1)}}$. Hence, if the largest radial distance of the centered dataset is denoted by 
\begin{equation}
r(X):=\frac{\mathrm{diam}(X)}{2},
\end{equation}
then the data should be scaled by dividing the data vectors by $r(X)\sqrt{m(m+1)}$ in order to guarantee that the transformed data lie inside the $m$-simplex $\Delta^{m}$ inside $\R^{m+1}$ after the above sequence of affine transformations (see Figure~\ref{fig:2drescaling} for the visual proof of this when $m=2$).

\begin{figure}
    \begin{tikzpicture}[scale=2.5]
    \node at (1,1) {$\R^{2}$};
    \draw[->] (-0.75,0) -- (1,0) node[right]{$x_1$};
    \draw[->] (0,-0.75) -- (0,1) node[above]{$x_2$};
    \draw[thick] ({(3+sqrt(3))/6},{(-3+sqrt(3))/6}) node[right]{$\left(\frac{\sqrt{3}+3}{6},\frac{\sqrt{3}-3}{6}\right)$} -- node[right]{$\sqrt{2}$} ({(-3+sqrt(3))/6},{(3+sqrt(3))/6}) node[left]{$\left(\frac{\sqrt{3}-3}{6},\frac{\sqrt{3}+3}{6}\right)$} -- node[left]{$\sqrt{2}$} ({-1/sqrt(3)},{-1/sqrt(3)}) node[below]{$\left(-\frac{1}{\sqrt{3}},-\frac{1}{\sqrt{3}}\right)$} -- node[below]{$\sqrt{2}$} cycle;
    \fill[opacity=0.1] ({(3+sqrt(3))/6},{(-3+sqrt(3))/6}) -- ({(-3+sqrt(3))/6},{(3+sqrt(3))/6}) -- ({-1/sqrt(3)},{-1/sqrt(3)}) -- cycle;
    \fill[opacity=0.2] (0,0) circle ({1/sqrt(6)});
    \draw[dashed,thick] ({-1/sqrt(3)},{-1/sqrt(3)}) -- (0,0);
    \draw[dashed] (0,0) -- ({sqrt(3)/6},{sqrt(3)/6});
    \draw[dashed] ({(-3+sqrt(3))/6},{(3+sqrt(3))/6}) -- (0,0); 
    \draw[dashed,thick] (0,0) -- node[right]{$r$} ({(1/4)-sqrt(3)/12},{(-1/4)-sqrt(3)/12});
    \node at (0.7,0.7) {$\tan\left(\frac{\pi}{6}\right)=\sqrt{2}r$};
    \node at (0.7,0.5) {$\implies r=\frac{1}{\sqrt{6}}$};
    \end{tikzpicture}
    %%%%%%
    \begin{tikzpicture}
    \node at (-2,0) {};
    \node at (2,0) {};
    \draw[thick,->] (0,0.25) -- node[left]{embed, rotate,} node[right]{and translate data} (0,-0.5);
    \end{tikzpicture}
    %%%%%%
	\begin{tikzpicture}
        \node at (3.0,2.25) {$\R^{3}$};
        \node at (-1.6,1.5) {$x_1+x_2+x_3=1$};
	\draw[dotted] (-2,1) -- (2,-1.0);
	\draw (-2,1) -- ({-2+5*(0.26)},{1-2.5*(0.26)});
	\draw[->] (2,-1) -- ({2+5*(0.22)},{-1-2.5*(0.22)}) node[right]{$x_2$};
	\draw[dotted] (1.5,1.5) -- (-1.5,-1.5);
	\draw (1.5,1.5) -- ({1.5-3*(0.23)},{1.5-3*(0.23)});
	\draw[->] (-1.5,-1.5) -- ({-1.5-3*(0.2)},{-1.5-3*(0.2)}) node[left]{$x_1$};
	\draw[dotted] (0,2) -- (0,-2);
	\draw[->] (0,2) -- (0,2.5) node[above]{$x_3$};
	\draw (0,-2) -- (0,{-2+0.75});
	\draw[thick] (0,2) node[right,xshift=1pt]{$(0,0,1)$} -- node[right]{$\sqrt{2}$} (2,-1) node[right,yshift=5pt]{$(0,1,0)$}  -- node[below]{$\sqrt{2}$} (-1.5,-1.5) node[right,yshift=-6pt,xshift=-4pt]{$(1,0,0)$} -- node[left]{$\sqrt{2}$} cycle;
        \fill[fill opacity=0.1] (0,2) -- (2,-1) -- (-1.5,-1.5) -- cycle;
        \draw[thick,->] (0,0) -- (0.20,-0.21);
        \fill[opacity=0.2] (0.20,-0.21) circle (1.05cm);
	\end{tikzpicture}
\caption{A centered data set in $\R^{2}$ must be rescaled to a small enough disc around the origin and inscribed within an equilateral triangle with side lengths equal to $\sqrt{2}$. Trigonometry shows that the radius of this circle must be $r=\frac{1}{\sqrt{6}}$. After embedding $\R^2$ into $\R^3$, the data are then rotated by the transformation $R_{2}$ in~\eqref{eq:rotation2to3}. After the rotation, a translation by the vector $\frac{1}{3}(1,1,1)$ shifts all the data points so that they lie inside the standard 2-simplex inside of $\R^{3}$. The vertices of the particular equilateral triangle drawn in $\R^{2}$ are those that correspond to the vertices of the standard simplex in $\R^{3}$ before applying the rotation $R_{2}$ followed by the translation.}
\label{fig:2drescaling}
\end{figure}
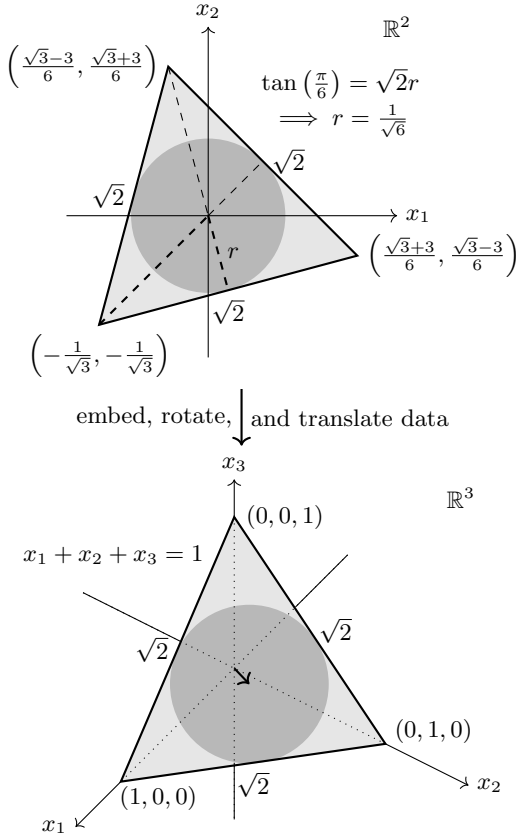

It is easier to apply the rescaling operation \emph{before} the affine transformation shifting the plane. Putting these transformations all together, the transformation from the data $X\subset\R^{m}$, written as the columns of a matrix $A_{0}$, so that it lies inside the $m$-simplex $\Delta^{m}\subset\R^{m+1}$ is given by 
\begin{equation}
\label{eqn:TXmapping}
A_{0}\xmapsto{\mathcal{T}_{X}}\frac{1}{r(X)\sqrt{m(m+1)}}
R_{m}
\begin{bmatrix}\mathds{1}_{m}\\0\end{bmatrix}A_{0}C+S.
\end{equation}
This assignment defines, for a given dataset $X\subset\R^{m}$, a transformation $\mathcal{T}_{X}:\R^{m}\to\R^{m+1}$ that satisfies the conditions
\begin{equation}
\label{eqn:TXscalingmetric}
\frac{d(x,y)}{r(X)\sqrt{m(m+1)}}=d\big(\mathcal{T}_{X}(x),\mathcal{T}_{X}(y)\big)
\end{equation}
for all $x,y\in\R^{m}$ and 
\begin{equation}
\mathcal{T}_{X}(X)\subset \Delta^{m}.
\end{equation}

We call the map $\mathcal{T}_{X}$ the \define{uniform transformation} associated with the point cloud $X$ in $\R^{m}$.
After this classical preprocessing of the data, we can then apply a suitable quantum encoding that is defined on the simplex. If we apply diagonal encoding, then we call the whole encoding 
\begin{equation}
\label{eqn:rhoUTD}
\rho_{\mathrm{UTD}}:X\xrightarrow{\mathcal{T}_{X}}\Delta^{m}\xrightarrow{\rho_{\mathrm{diag}}}\states(\C^{m+1})
\end{equation}
the \define{uniformly transformed diagonal (UTD) encoding} associated with the dataset $X\subset\R^{m}$.  

\begin{theorem}
\label{thm:UTDpreservesmetric}
Given a point cloud $(X,d_{X})$ in $\R^{m}$, so that $d_{X}$ is the induced Euclidean metric, the UTD encoding~\eqref{eqn:rhoUTD} satisfies 
\begin{equation}
\label{eqn:UTDencodingpreservesEuclidean}
d_{\mathrm{HS}}\big(\rho_{\mathrm{UTD}}(x),\rho_{\mathrm{UTD}}(y)\big)=\frac{d_{X}(x,y)}{\sqrt{m(m+1)r(X)}}
\end{equation}
for all $x,y\in X$. 
\end{theorem}

\begin{proof}
This follows from the fact that $\mathcal{T}_{X}$ as defined by~\eqref{eqn:TXmapping} satisfies~\eqref{eqn:TXscalingmetric} and because the diagonal encoding $\rho_{\mathrm{diag}}$ is an embedding as shown in Example~\ref{ex:diagonalencoding}. 
\end{proof}

\begin{corollary}
\label{cor:UTDpreservespersistentH}
Given a point cloud $(X,d_{X})$ in $\R^{m}$, so that $d_{X}$ is the induced Euclidean metric. Then, the persistent homology of $(X,d_{X})$ equals the persistent homology of $(Y,\lambda\, d_{Y})$, where $Y=\rho_{\mathrm{UTD}}(X)$, $d_{Y}$ is the Hilbert--Schmidt metric, and $\lambda=\sqrt{m(m+1)r(X)}$. 
\end{corollary}

Despite this positive result showing how there exists a quantum encoding that perfectly preserves the persistent homology of a point cloud in Euclidean space, as we mentioned in Example~\ref{ex:diagencoding}, diagonal encoding is neither efficiently implementable nor is it typically ideal for quantum algorithms since it only utilizes classical randomness. 
Moreover, due to the $\sqrt{m(m+1)}$ factor in the denominator in~\eqref{eqn:UTDencodingpreservesEuclidean}, the data become more difficult to distinguish as $m$ increases. As such, one might want to instead use a more efficient encoding, such as square-root encoding $\rho_{\sqrt{}}$, in place of $\rho_{\mathrm{diag}}$, in the composition~\eqref{eqn:rhoUTD}. This of course comes at the price of distorting the distances slightly, which may alter the persistent homology. Nevertheless, we illustrate how it performs in some example datasets in Section~\ref{sec:experiments}. 

%%%%%%%%%%%%%%%%%%%%%%%%%%%%%%%%%%%%%%%%%%%%%%%
\subsection{Quantum Multidimensional Scaling}
\label{sec:QMDS}
%%%%%%%%%%%%%%%%%%%%%%%%%%%%%%%%%%%%%%%%%%%%%%%

A potentially more promising approach than UTD encoding for preserving distances would be to obtain provably optimal encodings, at least theoretically. This is what we will achieve in the present section. 
Namely, we attempt to utilize multidimensional scaling (MDS) to propose quantum encodings of data into state spaces of quantum systems. The idea of classical MDS is reviewed first, followed by a review of metric MDS. Afterwards, we propose a procedure to apply the ideas of MDS in order to construct approximate distance-preserving embeddings of classical data into quantum systems. 

For classical MDS, we refer to Refs.~\cite{BoGr2005MDS,Sorzano2014surveydimred,MaKeBi1979Multivariate} for more details (though our matrix conventions follow Ref.~\cite{Strang2022}). 
Let $A_{0}=\begin{bmatrix}\mathbf{a}_{1}&\cdots&\mathbf{a}_{n}\end{bmatrix}$ be an $m\times n$ matrix of column vectors representing $n$ distinct samples of data with $m$ features. We therefore have a finite metric space $(X, d_{X})$, where $X=\{\mathbf{a}_{1},\dots,\mathbf{a}_{n}\}$ and $d_{X}$ is the induced Euclidean metric from $\R^{m}$. Let 
\begin{equation}
G_{0}:=A_{0}^{T}A_{0}
\end{equation}
be the \define{Gram matrix} associated with the data matrix $A_{0}$. Note that the $ij$ entry of the Gram matrix is the inner product $\mathbf{a}_{i}^{T}\mathbf{a}_{j}=\mathbf{a}_{j}^{T}\mathbf{a}_{i}$. Meanwhile, let $D^{(2)}_{0}$ be the squared distance matrix of the data set, i.e., the $ij$ entry of $D^{(2)}_{0}$ is 
\begin{align}
\left(D^{(2)}_{0}\right)_{ij}&=\lVert\mathbf{a}_{i}-\mathbf{a}_{j}\rVert^{2}= (\mathbf{a}_{i}-\mathbf{a}_{j})^{T}(\mathbf{a}_{i}-\mathbf{a}_{j}) \nonumber \\
&=\lVert\mathbf{a}_{i}\rVert^{2}+\lVert\mathbf{a}_{j}\rVert^{2}-2\mathbf{a}_{i}^{T}\mathbf{a}_{j}.
\end{align}
This calculation shows that $D^{(2)}_{0}$ may also be written as 
\begin{equation}
\label{eq:D0squared}
D^{(2)}_{0}=\boldsymbol{1}_{n}\boldsymbol{\ell}^{T}+\boldsymbol{\ell}\boldsymbol{1}_{n}^{T}-2G_{0},
\end{equation}
where 
\begin{equation}
\label{eqn:ellvector}
\boldsymbol{\ell}:=\begin{bmatrix}\lVert \mathbf{a}_{1}\rVert^{2} & \cdots & \lVert \mathbf{a}_{n}\rVert^{2}\end{bmatrix}^{T}
\end{equation}
is the vector of the squared lengths of the data vectors. 
It will also occasionally be useful to have an expression for $D_{0}^{(2)}$ purely in terms of the Gram matrix as 
\begin{equation}
\label{eqn:D2intermsofGonly}
D_{0}^{(2)}=\big\{G_{0}\odot\mathds{1}_{n},\boldsymbol{1}_{n}\boldsymbol{1}_{n}^{T}\big\}-2G_{0},
\end{equation}
where $\odot$ denotes the Hadamard--Schur product of matrices of the same size $(B\odot B')_{ij}:=B_{ij}B'_{ij}$ for the $ij$ entry, and where $\{\;\cdot\;,\;\cdot\;\}$ denotes the anticommutator, i.e., $\{B,B'\}:=BB'+B'B$. 
Thus, Equation~\eqref{eq:D0squared} shows that the Gram matrix can be solved for in terms of the distances between the data vectors and the lengths of the data vectors through
\begin{equation}
\label{eqn:G0}
G_{0}=\frac{1}{2}\left(\boldsymbol{1}_{n}\boldsymbol{\ell}^{T}+\boldsymbol{\ell}\boldsymbol{1}_{n}^{T}-D^{(2)}_{0}\right).
\end{equation}
This formula can be simplified if we center our data. Namely, the centered data matrix is $A:=A_{0}C$, where $C:=\mathds{1}_{n}-\frac{1}{n}\mathbf{1}_{n}\mathbf{1}_{n}^{T}$ is the centering matrix. Centering the data changes the Gram matrix to 
\begin{equation}
G:=A^{T}A=C A_{0}^{T} A_{0} C=C G_{0} C=-\frac{1}{2}CD^{(2)}_{0} C,
\end{equation}
where the second equality follows from the fact that $C^T=C$,  while the last equality follows from this and also~\eqref{eqn:G0} due to the fact that $C\boldsymbol{1}_{n}=\boldsymbol{0}$. However, centering the data does not change the squared distance matrix, i.e., if $D^{(2)}$ denotes the squared distance matrix associated with the centered data matrix $A$, then $D^{(2)}=D^{(2)}_{0}$ so that 
\begin{equation}
\label{eq:GramfromDistance}
G=-\frac{1}{2} C D^{(2)} C.
\end{equation}
This shows that if we know the distances between the data points, but not necessarily their positions in Euclidean space, we can obtain their associated (centered) Gram matrix using~\eqref{eq:GramfromDistance}. Next, obtain a reduced singular value decomposition (RSVD) of the Gram matrix (which is symmetric and positive semidefinite), i.e.,  
\begin{equation}
G=Q\Lambda Q^{T},
\end{equation}
where $Q$ is an $n\times r$ matrix consisting of $r$ orthonormal vectors, $\Lambda$ is a diagonal $r\times r$ matrix of strictly positive numbers aligned in decreasing order, and $r$ is the rank of $G$~\cite{Strang2022}. The matrix 
\begin{equation}
\widetilde{A}:=\sqrt{\Lambda}\,Q^{T}
\end{equation}
is an $r\times n$ matrix, whose $n$ column vectors $\widetilde{\mathbf{a}}_{1},\dots,\widetilde{\mathbf{a}}_{n}$ are $r$-component vectors. By construction, the associated pairwise distances among the $\widetilde{\mathbf{a}}_{1},\dots,\widetilde{\mathbf{a}}_{n}$ are the same as those of $\mathbf{a}_{1},\dots,\mathbf{a}_{n}$. 
To see this, let $\widetilde{G}:=\widetilde{A}^{T}\widetilde{A}$ denote the Gram matrix of $\widetilde{A}$ and let $\tilde{D}^{(2)}$ denote the squared-distance matrix associated with the columns of $\widetilde{A}$. Then, by the definition of $\widetilde{A}$, we immediately have $\widetilde{G}=G$. Combining this with~\eqref{eqn:D2intermsofGonly} for $D^{(2)}$ and $\widetilde{D}^{(2)}$ yields $\widetilde{D}^{(2)}=D^{(2)}$. 
Using our notation of finite metric spaces, let $Y=\{\widetilde{\mathbf{a}}_{1},\dots,\widetilde{\mathbf{a}}_{n}\}$ and $d_{Y}$ be the induced Euclidean metric from $\R^{r}$. Then the distances satisfy 
\begin{equation}
\label{eq:MDSdistancePreserved}
d_{Y}\left(\widetilde{\mathbf{a}}_{i},\widetilde{\mathbf{a}}_{j}\right)=d_{X}\left(\mathbf{a}_{i},\mathbf{a}_{j}\right)
\end{equation}
for all $i,j$. 
It is guaranteed that $r\le m$, so that this provides a dimensionality reduction technique. Moreover, the construction of $\widetilde{A}$ defines an isometry of metric spaces $\Phi:(X,d_{X})\to (Y,d_{Y})$, where $\Phi$ is defined by $\Phi(\mathbf{a}_{i})=\widetilde{\mathbf{a}}_{i}$ for all $i$. 

The above situation corresponds to the ideal case when a Euclidean dataset $(X,d_{X})$ is embedded into a Euclidean space of dimension $r\le m$, which is typically much smaller than $m$. However, if $r$ is close to $m$, and one wants to embed the data points into $\R^{k}$ with $k$ smaller than $r$, one can truncate SVD and use the largest $k$ singular values (multiplicity included) of the centered data matrix $A=A_{0}C$ in order to transfer the data into $\R^{k}$ in a way that best preserves the inner products between the centered data from $A$ and the data from $\widetilde{A}$ (this will be made precise momentarily, though the central idea follows from the Eckart--Young theorem~\cite{Strang2019Data,Dax2013}). This means that if we demand to encode our $n$ samples as vectors in $\R^{k}$ with $k<r$, then the distances will not all be preserved. This is fine for some purposes, especially if the reduction in dimension from $d$ to $k$ is significant and the distances between the encoded data points remains close. The notion of closeness is flexible, and we will outline one such notion soon. 

In a more general setting, imagine that one is instead given a finite metric space $(X,d_{X})$, that is not necessarily a finite subset of a Euclidean space equipped with the Euclidean metric (an example is a dataset that is obtained from sampling points on a Riemannian submanifold of Euclidean space equipped with the geodesic distance). Suppose $X$ has $n$ elements, which we denote by $X=\{x_1,\dots,x_n\}$. Let $D^{(2)}_{X}$ be its associated distance-squared matrix, so that its $ij^{\text{th}}$ entry is $d_{X}(x_{i},x_{j})^{2}$. In such a general case, it is not necessarily the case that the points can be embedded into a Euclidean space of some finite dimension. The necessary and sufficient conditions for the existence of such a Euclidean embedding are formalized by Schoenberg's theorem~\cite{LimMemoli2024,Schoenberg1935,YoungHouseholder1938}, which says the following. 

\begin{proposition}
\label{thm:Schoenberg}
Let $(X,d_{X})$ be a finite metric space, and let $D^{(2)}_{X}$ be its associated distance-squared matrix, as in the previous paragraph. Then $(X,d_{X})$ is isometric to a set of $n$ points $\{\mathbf{y}_{1},\dots,\mathbf{y}_{n}\}\subset\R^{k}$ for some nonnegative integer $k$ if and only if $G_{X}:=-\frac{1}{2}C D^{(2)}_{X} C$ is positive semidefinite, where $C=\mathds{1}_{n}-\frac{1}{n}\boldsymbol{1}_{n}\boldsymbol{1}_{n}^{T}$. Moreover, when $G_{X}$ is positive semidefinite, the smallest integer $k$ for which such an embedding exists is the rank of $G_{X}$.  
\end{proposition}

\begin{remark}
\label{rmk:Schoenberg}
We will not go over the proof, which can be found in the references. Instead, we make some comments. In Ref.~\cite{YoungHouseholder1938}, Young and Householder remark that when $n=3$, the condition that $G_{X}$ be positive semidefinite is equivalent to the triangle inequality. It is interesting to note that in Ref.~\cite{Schoenberg1935}, Schoenberg also gave theorems for embedding points into spheres (see also Ref.~\cite{BaiQiXiu2015}) as well as embedding points into Minkowski space. In Ref.~\cite{BoutinKemper2003}, Boutin and Kemper showed that there exist configurations in Euclidean space that give rise to the same distance matrix yet cannot be mapped to each other under an isometry of Euclidean space. This latter result emphasizes that MDS just provides one realization as an embedding in Euclidean space, when it exists, and any two embeddings need not be related by a global orthogonal matrix or an affine transformation. 
We also mention that there are some examples of familiar metric spaces that cannot be embedded in Euclidean space for \emph{any} finite dimension.
For example, in Ref.~\cite{AdBlKa2020MDS}, the authors showed that the unit circle equipped with the geodesic distance metric is not embeddable into any finite-dimensional Euclidean space, and that an infinite-dimensional embedding is possible, but only as a fractal. As another example, let $X=\{x_1,x_2,x_3,x_4\}$ be the set of four points on the two-dimensional sphere of radius $\frac{2}{\pi}$, with one point on the north pole and the other three points on the equator positioned at azimuthal angles $0, \frac{\pi}{2}$, and $\pi$. Equipping $X$ with the geodesic distance from the sphere, Linial showed that there does not exist an embedding of this metric space into Euclidean space of \emph{any} dimension~\cite{BrBrKi2008}. We will revisit this example in more detail in Example~\ref{ex:linial}.
\end{remark}

Regarding Proposition~\ref{thm:Schoenberg}, if the matrix $G_{X}$ is not positive semidefinite, then one can nevertheless take its positive part $G_{X}^{(+)}$, which comes from a spectral decomposition $G_{X}=G_{X}^{(+)}-G_{X}^{(-)}$, and apply RSVD to $G_{X}^{(+)}$ as before. Namely, the matrix $G^{(+)}$ is the closest positive semidefinite matrix to $G_{X}$ in the Frobenius norm~\cite{Mardia1978,Strang2019Data,Dax2013}, i.e., 
\begin{equation}
G_{X}^{(+)}=\argmin_{\substack{M\in\matr_{n}\\M\ge0}}\Tr\big[(M-G_{X})^2\big].
\end{equation}
One would then apply RSVD to $G^{(+)}_{X}$ to obtain $G_{X}^{(+)}=Q\Lambda Q^{T}$ and set the elements of $Y=\{\mathbf{y}_1,\dots,\mathbf{y}_n\}$ to be the column vectors of the $k\times n$ matrix $\sqrt{\Lambda}Q^{T}$. However, the mapping $\Phi:X\to Y$ sending $x_{i}\in X$ to $\mathbf{y}_{i}\in\R^{k}$ for all $i$ in this case would only be an approximate embedding of $(X,d_{X})$ into $(Y,d_{Y})$. We make this statement precise by introducing the definitions of stress and strain, as well as formalizing the preceding construction, which is due to Refs.~\cite{Torgerson1952MDS,Gower1966,Mardia1978}. Textbook references include Refs.~\cite{BoGr2005MDS,MaKeBi1979Multivariate}. 

\begin{definition}
Given two finite metric spaces $(X,d_{X})$ and $(Y,d_{Y})$ and an injective function $f:X\to Y$, write $X=\{x_1,\dots,x_n\}$ and let $(X,d_{f})$ be the pullback metric space. 
The \define{stress} of $(X,d_{X})$ and $f$ is the nonnegative real number given by 
\begin{equation}
\label{eqn:stressf}
\mathrm{stress}(f):=\lVert D_{X}-D_{f}\rVert_{2},
\end{equation}
where $D_{X}$ is the $n\times n$ distance matrix for $(X,d_{X})$ and $D_{f}$ is the $n\times n$ distance matrix for $(X,d_{f})$. 
Meanwhile, the \define{strain} of $(X,d_{X})$ and $f$ is 
\begin{equation}
\label{eqn:strainf}
\mathrm{strain}(f):=\lVert G_{X}-G_{f}\rVert_{2}
\end{equation}
where 
\begin{equation}
G_{f}:=-\frac{1}{2}C D_{f}^{(2)} C, 
\end{equation}
$D_{f}^{(2)}$ is the distance-squared matrix associated with $(X,d_{f})$, and $C$ is the centering matrix.
\end{definition}

\begin{remark}
Just like the distortion $\mathrm{dis}(f)$ can be expressed as the $\ell^{\infty}$ distance between the distance matrices $D_{X}$ and $D_{f}$ as in Remark~\ref{rmk:distortionasdistancedifference}, the stress $\mathrm{stress}(f)$ can be expressed as the $\ell^{2}=S_{2}$ distance between $D_{X}$ and $D_{f}$ (see Remark~\ref{rmk:ellpSchattenp} to see why the $\ell^{2}$ and $S_{2}$ metrics coincide). As such, we sometimes refer to the stress as the \define{$\ell^{2}$/Hilbert--Schmidt/Frobenius distortion of $f$}, and we might denote it by $\mathrm{dis}_{2}(f)$ when emphasizing this connection between stress and distortion (Ref.~\cite{BrBrKi2008} calls it the \emph{$L_{2}$-stress} and denotes it by $\sigma_{2}$, while others call it the \emph{Kruskal stress}). 
The distortion $\mathrm{dis}(f)=\lVert D_{X}-D_{f}\rVert_{\infty}$ has also been used in Ref.~\cite{BrBrKi2006GMDS} to define a non-symmetric semi-metric on the space of compact metric spaces, called the \emph{partial embedding distance}, and it has been used for generalized multidimensional scaling. 
Note that some authors divide the right-hand-side of~\eqref{eqn:stressf} by $\lVert D_{X}\rVert_{2}$ and call that the \emph{normalized stress}~\cite{BoGr2005MDS}, while others divide by the number of samples~\cite{BrBrKi2006GMDS}. Additionally, some authors square the right-hand-sides of~\eqref{eqn:stressf} and~\eqref{eqn:strainf} when defining these quantities.
\end{remark}

\begin{definition}
\label{defn:EuclideanMDS}
Let $(X,d_{X})$ be a finite metric space, write $X=\{x_1,\dots,x_n\}$ so that the elements of $X$ are ordered, and set $G_{X}=-\frac{1}{2}C D_{X}^{(2)} C$, where $D_{X}^{(2)}$ is the distance-squared matrix for $(X,d_{X})$ and $C$ is the $n\times n$ centering matrix. Write $G_{X}=G_{X}^{(+)}-G_{X}^{(-)}$ thereby splitting $G_{X}$ into the difference of two positive semi-definite matrices. Let $r$ be the rank of $G_{X}^{(+)}$, so that $r\le n$, and let $G_{X}^{(+)}=Q\Lambda Q^{T}$ be an RSVD of $G_{X}^{(+)}$. Set $Y=\{\mathbf{y}_1,\dots,\mathbf{y}_n\}$ to be the column vectors of the $r\times n$ matrix $\sqrt{\Lambda}Q^{T}$. Set $\Phi_{\mathrm{MDS}}:X\to \R^{r}$ to be the function sending $x_{i}\in X$ to $\mathbf{y}_{i}$ for all $i$. The function $\Phi_{\mathrm{MDS}}$ is called a \define{Euclidean MDS mapping} of maximal rank. 
\end{definition}

\begin{remark}
Since the RSVD of a positive semidefinite matrix need not be unique, one should note that a Euclidean MDS mapping is not canonically defined for an arbitrary finite metric space $(X,d_{X})$. As such, it would be more accurate to say that the Euclidean MDS mapping is associated with $(X,d_{X})$ together with an RSVD of $G_{X}^{(+)}$. 
We will ignore this subtle point. 
\end{remark}

A Euclidean MDS mapping is sometimes called a \emph{representation function}~\cite{BoGr2005MDS}. The procedure for starting with a distance matrix, which is equivalent to a finite metric space $(X,d_{X})$ and the choice of an ordering of the elements of $X$, and constructing the Euclidean MDS mapping has many names, including \emph{classical MDS} (cMDS), \emph{principal coordinate analysis} (PCoA), \emph{Torgerson scaling}, or \emph{Torgerson--Gower scaling}~\cite{AdBlKa2020MDS}. Because the Euclidean MDS mapping involves obtaining eigenvectors of some matrix (in this case, $G_{X}^{(+)}$), it is an example of a spectral embedding or eigenmap~\cite{BrBrKi2008}. 

\begin{proposition}
\label{prop:cMDSexistenceanduniqueness}
Let $(X,d_{X})$ be a finite metric space and $\Phi_{\mathrm{MDS}}:X\to\R^{r}$ its Euclidean MDS mapping as in Definition~\ref{defn:EuclideanMDS}. Then, 
\begin{equation}
\label{eqn:strainMDS}
\mathrm{strain}(\Phi)\ge\mathrm{strain}(\Phi_{\mathrm{MDS}})
\end{equation}
for all functions $\Phi:X\to\R^{k}$. In other words, $\Phi_{\mathrm{MDS}}$ is a global minimum of strain for $(X,d_{X})$. 
\end{proposition}

\begin{proof}
See Theorem 14.4.2 in Ref.~\cite{MaKeBi1979Multivariate} (see also Ref.~\cite{Mardia1978}). 
\end{proof}

Euclidean MDS utilizes the structure of Euclidean space to compare a certain matrix to a Gram matrix, which itself is defined in terms of inner products. Minimizing strain led to an explicit and analytic algebraic expression for the optimal approximate embedding. For more general metric spaces, one often uses some form of stress or distortion as the function to minimize~\cite{BrBrKi2008}. Minima of these functions tend to provide different approximate embedding schemes. 

\begin{example}
\label{ex:linial}
As an example, let us revisit Linial's example of a 4 point metric space obtained from the geodesic distances on a sphere of radius $\frac{2}{\pi}$ as in Remark~\ref{rmk:Schoenberg} (see Example 7.1 in Ref.~\cite{BrBrKi2008}). The four points $x_1,x_2,x_3,x_4$ can be visualized in $\R^{3}$, though the distances between points are \emph{geodesic} distances (and not the Euclidean distances). 
\begin{center}
\includegraphics{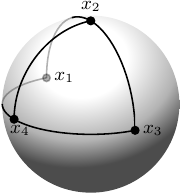}
\end{center}
In this case, the distance matrix is 
\begin{equation}
\label{eqn:LinialDistance}
D_{X}=\begin{bmatrix}0&1&2&1\\1&0&1&1\\2&1&0&1\\1&1&1&0\end{bmatrix}.
\end{equation} 
Hence,  
\begin{equation}
G_{X}=-\frac{1}{2}C D_{X}^{(2)} C=\frac{1}{16}\begin{bmatrix}15&1&-17&1\\1&3&1&-5\\-17&1&15&1\\1&-5&1&3\end{bmatrix}.
\end{equation}
Diagonalizing $G_{X}$ gives $G_{X}=Q_{X}\Lambda_{X}Q_{X}^{T}$, where 
\begin{equation}
\Lambda_{X}=\begin{bmatrix}2&0&0&0\\0&\frac{1}{2}&0&0\\0&0&-\frac{1}{4}&0\\0&0&0&0\end{bmatrix}
\quad\text{and}\quad
Q_{X}=\begin{bmatrix}
\frac{1}{\sqrt{2}}&0&\frac{-1}{2}&\frac{1}{2}\\
0&\frac{1}{\sqrt{2}}&\frac{1}{2}&\frac{1}{2}\\
\frac{-1}{\sqrt{2}}&0&\frac{-1}{2}&\frac{1}{2}\\
0&\frac{-1}{\sqrt{2}}&\frac{1}{2}&\frac{1}{2}
\end{bmatrix}
\end{equation}
One then finds that $G_{X}^{(+)}=Q\Lambda Q^{T}$, where 
\begin{equation}
\Lambda=\begin{bmatrix}2&0\\0&\frac{1}{2}\end{bmatrix} \quad\text{ and }\quad Q^{T}=\frac{1}{\sqrt{2}}\begin{bmatrix}1&0&-1&0\\0&1&0&-1\end{bmatrix}.
\end{equation}
Hence, 
\begin{equation}
\label{eqn:MDSgeodesicsphere4points}
Y_{\mathrm{strain}}:=
\sqrt{\Lambda}Q^{T}=\begin{bmatrix}1&0&-1&0\\0&\frac{1}{2}&0&-\frac{1}{2}\end{bmatrix}.
\end{equation}
As a form of comparison, the new distance matrix associated with the approximately embedded data is 
\begin{equation}
D_{Y_{\mathrm{strain}}}
=\begin{bmatrix}
0&\frac{\sqrt{5}}{{2}}&2&\frac{\sqrt{5}}{{2}}\\ 
\frac{\sqrt{5}}{{2}}&0&\frac{\sqrt{5}}{{2}}&1\\
2&\frac{\sqrt{5}}{{2}}&0&\frac{\sqrt{5}}{{2}}\\
\frac{\sqrt{5}}{{2}}&1&\frac{\sqrt{5}}{{2}}&0
\end{bmatrix}
\approxeq
\begin{bmatrix}0&1.12&2&1.12\\1.12&0&1.12&1\\2&1.12&0&1.12\\1.12&1&1.12&0\end{bmatrix}
,
\end{equation}
whic looks quite close to~\eqref{eqn:LinialDistance}. 
Meanwhile, finding a numerical local minimum of stress near the points represented by the columns of~\eqref{eqn:MDSgeodesicsphere4points} while fixing the first vector to be the first column leads to 
\begin{equation}
Y_{\mathrm{stress}}:=
\begin{bmatrix}1&0.0528&-0.8944&0.0528\\0&0.4737&0.0001&-0.4735\end{bmatrix}.
\end{equation}
As a form of comparison, the new distance matrix associated with the approximately embedded data is 
\begin{equation}
D_{Y_{\mathrm{stress}}}
\approxeq
\begin{bmatrix}0&1.06&1.89&1.06\\1.06&0&1.06&0.95\\1.89&1.06&0&1.06\\1.06&0.95&1.06&0\end{bmatrix}
.
\end{equation}
Finally, finding a numerical local minimum of distortion near the points represented by the columns of~\eqref{eqn:MDSgeodesicsphere4points} while fixing the first vector to be the first column leads to 
\begin{equation}
Y_{\mathrm{dis}}:=\begin{bmatrix}1&0&-1&0\\0&\frac{1}{2}&0&-\frac{1}{2}\end{bmatrix},
\end{equation}
which happens to coincide with $Y_{\mathrm{strain}}$. We visualize these different embeddings in $\R^2$ as 
\begin{center}
\begin{tikzpicture}[scale=2]
\draw[thin,opacity=0.5,step=0.25] (-1.375,-0.875) grid (1.375, 0.875);
\draw[->] (-1.375,0) -- (1.375,0);
\draw[->] (0,-0.875) -- (0,0.875);
\node at (1,0) {$\boldsymbol{\times}$};
\node at (0,0.5) {$\boldsymbol{\times}$};
\node at (-1,0) {$\boldsymbol{\times}$};
\node at (0,-0.5) {$\boldsymbol{\times}$};
\node at (1,0) {\LARGE${\circ}$};
\node at (0.0528,0.4737) {\LARGE${\circ}$};
\node at (-0.8944,0.0001) {\LARGE${\circ}$};
\node at (0.0528,-0.4735) {\LARGE${\circ}$};
\end{tikzpicture}
\end{center}
where $\boldsymbol{\times}$ represents a point in the Euclidean MDS embedding (minimizing strain) and {\LARGE$\circ$} represents a point in the metric MDS embedding (minimizing stress). It turns out that in this example, the Euclidean MDS embedding coincides with the minimum distortion embedding. One can see from this example that the three approximate embeddings are all close to each other, though they need not all be the same embedding, even up to isometry.  
\end{example}

Example~\ref{ex:linial} illustrates that optimizing over stress, strain, or distortion may lead to different solutions for approximate embeddings of finite metric spaces (datasets equipped with a notion of distance). Among all of these, minimizing strain has a simple algebraic solution given by finding a closest positive matrix of low rank and then applying Singular Value Decomposition to that closest positive matrix in order to extract out the new MDS coordinates. On the one hand, minimizing stress or distortion does not seem to admit as simple of an algebraic solution as it does for strain, suggesting that minimizing strain is desirable. On the other hand, the notion of strain has not, to the best of our knowledge, been extended to the general context of Riemannian manifolds beyond Euclidean space due to its reliance on centering data and the Gram matrix. Some few exceptions include spherical and hyperbolic geometries of constant curvature~\cite{Schoenberg1935,KeNa2020hydra}, where the double centering operation on a distance matrix is replaced by the trigonometric cosine function in the case of spheres and the hyperbolic cosine function in the case of hyperbolic space. Moreover, if the ultimate goal is to minimize the error in distance, then stress and distortion seem more appropriate as they involve the distance matrix directly, rather than a centered version of its squared entries. As a result, when extending Multidimensional Scaling to quantum systems, we will primarily use stress and distortion.  

We therefore finally discuss optimal quantum encodings. As just mentioned in the previous paragraph, although it remains an open question whether there exists a strain-minimizing projective space embedding for data, we can nevertheless explore a variety of optimal quantum encodings based on minimizing stress or distortion. Moreoever, such optimality also depends on which distance measure is used on the space of quantum states. For concreteness and simplicity, we stick to one definition. 
\begin{definition}
\label{defn:quantumMDS}
Let $(X,d_{X})$ be a finite metric space such that 
$\mathrm{diam}(X,d_{X})=1$ (if $\mathrm{diam}(X,d_{X})\ne1$, let $\lambda=\frac{1}{\mathrm{diam}(X,d_{X})}$ and use the rescaled metric space $(X,\lambda\, d_{X})$ in what follows). Equip the family of quantum state spaces $\states(\C^{m})$ for varying $m\in\N$ with the Bures fidelity metric $d_{F}$ from~\eqref{eq:bures-distance}. 
If $r\in\N$, an $r$-dimensional \define{quantum MDS encoding} is a quantum state encoding of the form $\rho_{\mathrm{QMDS}}:X\to\states(\C^{r})$ such that 
\begin{equation}
\mathrm{stress}(\rho)\ge\mathrm{stress}(\rho_{\mathrm{QMDS}})
\end{equation}
for all other quantum state encodings $\rho:X\to\states(\C^{r})$. 
\end{definition}

The reason for specifying the diameter of $(X,d_{X})$ to be $1$ is simply because the diameter of $\big(\states(\C^{m}),d_{F}\big)$ is $1$ and we would like to approximately embed our data set $X$ into a quantum state space in a way that is spread out over the quantum state space as much as possible. However, in practice we will often not rescale by the diameter and instead will rescale by some parameter in order to minimize stress. The reason for doing this is because there may be some other scalar factor besides the diameter that will lead to a lower scale-free distortion (or stress). Indeed, an outlier data point may cause the diameter to be unnecessarily large. 

Definition~\ref{defn:quantumMDS} can be viewed as a limitation to how well a quantum encoding can preserve distances in a classical dataset. In other words, it provides a theoretical lower bound for the distortion of any quantum encoding. What this implies, practically, is that if we want to create a quantum encoding for a data set using actual circuits, then we can test how good our quantum encoding is at preserving the structure for persistent homology by calculating its stress and then comparing it to the lower bound as obtained through QMDS. Conversely, if we find a numerically optimal quantum encoding satisfying this definition, there is no reason that it will be practically implementable. Therefore, we may restrict attention to a subclass of quantum encodings when searching for one that minimizes stress. Although we will not solve this problem for any specific such subclass of quantum encodings, we formulate a precise version of this problem in the case that the quantum encodings can be viewed as Hamiltonian time evolution, or equivalently as one-parameter subgroups. We do so in the next section. 

%%%%%%%%%%%%%%%%%%%%%%%%%%%%%%%%%%%%%%%%%%%%%%%
\subsection{One-parameter subgroup Quantum MDS}
\label{sec:QMDSOPG}
%%%%%%%%%%%%%%%%%%%%%%%%%%%%%%%%%%%%%%%%%%%%%%%

Some quantum machine learning algorithms assume a quantum state encoding of the form 
\begin{equation}
\R^{m}\ni x \mapsto\rho(x)=U(x)\ket{\psi_0}\bra{\psi_0}U^{\dag}(x),
\end{equation}
where $\ket{\psi_{0}}\in\C^{s}$ is some fiducial vector and $U(x)$ is a special unitary matrix of the form 
\begin{equation}
U(x)=e^{i\mathcal{L}(x)},
\end{equation}
where $\mathcal{L}:\R^{m}\to\mathcal{B}(\C^{s})$ is a linear transformation and $\mathcal{L}(x)$ is self-adjoint and traceless. This guarantees that $U$ defines what is called a \emph{one parameter subgroup}, occasionally called a \emph{Hamiltonian evolution}, since each direction $\hat{\mathbf{x}}$ in $\R^{m}$ specifies a line $\R\ni t\mapsto t\hat{\mathbf{x}}$ and the family of unitary matrices $\R\ni t\mapsto U(t\hat{\mathbf{x}})$ is a one-parameter group of unitary operators~\cite{PBVP24,SchuldPetruccione21}. Therefore, the 
Bures fidelity distance
between quantum states $\rho(x)$ and $\rho(y)$  
under such a transformation becomes 
\begin{equation}
d_{F}\big(\rho(x),\rho(y)\big)=\sqrt{1-\big|\langle\psi_{0}|e^{-i\mathcal{L}(x)}e^{i\mathcal{L}(y)}|\psi_{0}\rangle\big|}.
\end{equation}
Thus, given a centered finite subset $X=\{x_1,\dots,x_n\}$ of Euclidean space (a point cloud) satisfying $\mathrm{diam}(X)=1$, a quantum encoding $\rho$ as above minimizing the stress between the original Euclidean distance and the associated 
Bures fidelity distance is equivalent to solving the minimization problem 
\begin{equation}
\argmin_{\substack{\mathcal{L}:\R^m\to\mathcal{B}(\C^s)\\\mathcal{L}(x)^{\dag}=\mathcal{L}(x)\\\Tr[\mathcal{L}(x)]=0}}
\sum_{1\le j<k\le n}\Big(\lVert x_j-x_k\rVert-
d_{F}\big(\rho(x_j),\rho(x_k)\big)
\Big)^2.
\end{equation}
By setting $L_{\alpha}:=\mathcal{L}(e_{\alpha})$ to be image of the unit vector $e_{\alpha}$ under the linear transformation $\mathcal{L}$, we can reexpress this minimization problem as
\begin{equation}
\argmin_{\substack{L_{1},\dots,L_{m}\\L_{\alpha}^{\dag}=L_{\alpha}\\\Tr[L_{\alpha}]=0}}
\sum_{1\le j<k\le n}\Big(\lVert x_j-x_k\rVert-
d_{F}\big(\rho(x_j),\rho(x_k)\big)
\Big)^2.
\end{equation}
Therefore, the problem of finding a quantum encoding that minimizes the stress becomes a minimization problem over a finite-dimensional space of self-adjoint traceless matrices of a fixed dimension. We do not find the solution of this optimization problem in this work, but we leave this as an open problem for the future. Similarly, one can formulate analogous versions of such a minimization problem by restricting to an appropriate finite-dimensional class of quantum encodings.

%%%%%%%%%%%%%%%%%%%%%%%%%%%%%%%%%%%%%%%%%%%%%%%
\section{Experiments}
\label{sec:experiments}
%%%%%%%%%%%%%%%%%%%%%%%%%%%%%%%%%%%%%%%%%%%%%%%

In this section, we will look at simple datasets in $\R^2$ and apply several different quantum encoding techniques in order to test how much these techniques distort the metric structure of the dataset. The datasets will first consist of an ideal circle of the $n^{\text{th}}$ roots of unity at some radius and a noisy version of this dataset. These datasets therefore should exhibit the topology of a circle, and the goal is to see which quantum encodings best preserve this. By analyzing the distance matrices of the encodings, we will visually see how the encoding distorts the original geometry by comparing the distance matrices.

%%%%%%%%%%%%%%%%%%%%%%%%%%%%%%%%%%%%%%%%%%%%%%%%
\subsection{Ideal circle}
\label{subsec:idealcircle}
%%%%%%%%%%%%%%%%%%%%%%%%%%%%%%%%%%%%%%%%%%%%%%%%

Before going to experiments with noisy data, we begin by testing the different encoding techniques for an ideal circle of varying radius, described by the $n^{\text{th}}$ roots of unity viewed as a dataset in $\R^{2}$. For concreteness, we will take $n=200$. This will allow us to precisely identify how to rescale the domain of the encoding before applying some of the commonly used quantum encoding techniques such as angle, dense angle, and IQP encodings. By focusing on such an ideal dataset first, we can avoid several issues of topology change merely due to a dataset being outside (or close to the boundary of) the fundamental domain of the quantum encoding, rather than issues arising due to noise from sampling. 

To set the notation, we begin with a point cloud in $\R^{2}$ of the form 
\begin{equation}
X=\left\{x_{k}:=\Big(r c_{k},rs_{k}\right)\;:\;k\in\{0,1,\dots,n-1\}\Big\},
\end{equation}
where 
\begin{equation}
c_{k}:=\cos\left(\frac{2\pi k}{n}\right) \, ,
\quad
s_{k}:=\sin\left(\frac{2\pi k}{n}\right)\,,
\end{equation}
and $r$ is some positive number, as in Figure~\ref{fig:200rootsofunity}.

\begin{figure}
\begin{tikzpicture}[scale=1.25]
\draw[->] (-1.25,0) -- (1.25,0) node[right]{};
\draw[->] (0,-1.25) -- (0,1.25) node[above]{};
\draw[gray,thin,opacity=0.5,step=0.25cm] (-1.25,-1.25) grid (1.25,1.25); 
\foreach \k in {0,1,2,...,199} {
\node at ({cos(deg(2*pi*\k/200))},{sin(deg(2*pi*\k/200))}) {\scalebox{0.3}{$\bullet$}};
}
\end{tikzpicture}
\qquad
\raisebox{-6pt}{
\includegraphics[width=4.25cm]{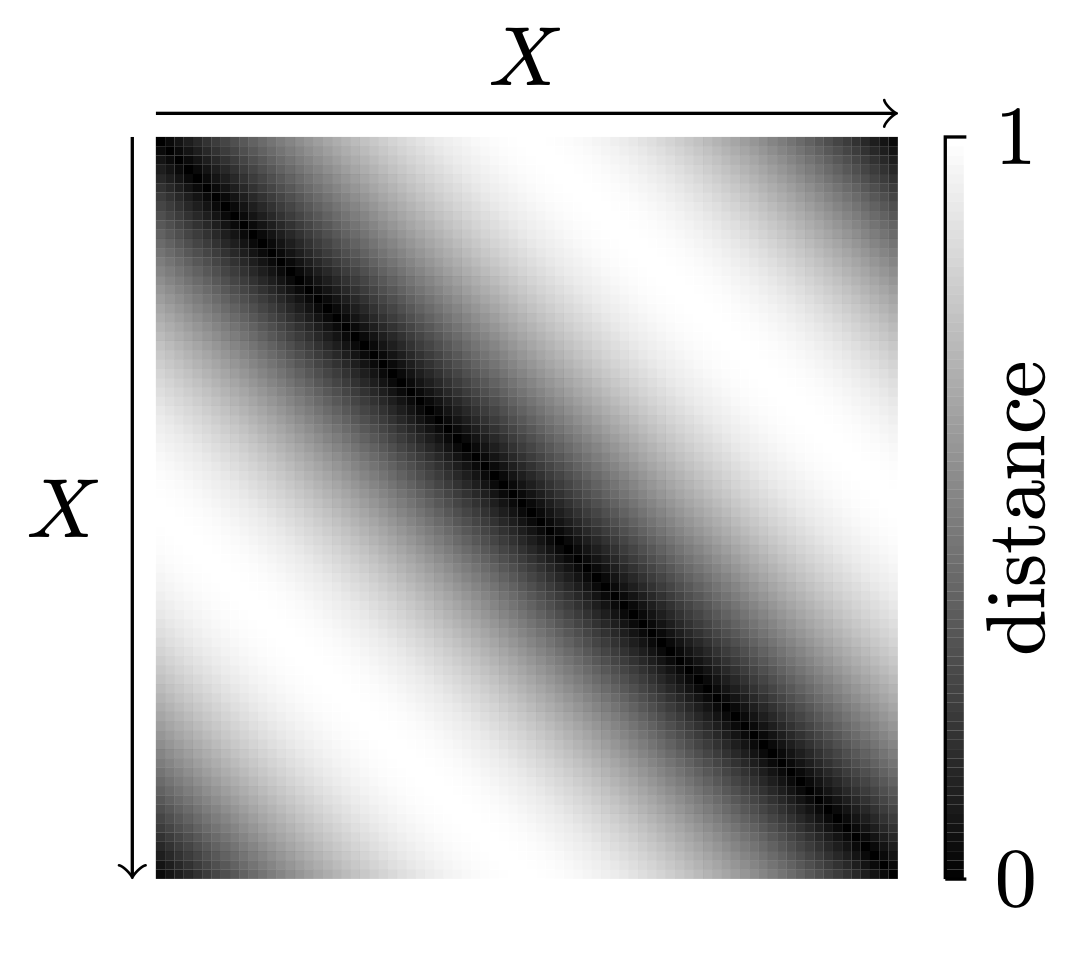}}
\caption{The set $X$ consisting of the $n^{\text{th}}$ roots of unity, where $n=200$. The heat map associated with the (Euclidean) distance matrix is depicted on the right. The distance has been rescaled so that the diameter of $X$ is $1$.}
\label{fig:200rootsofunity}
\end{figure}
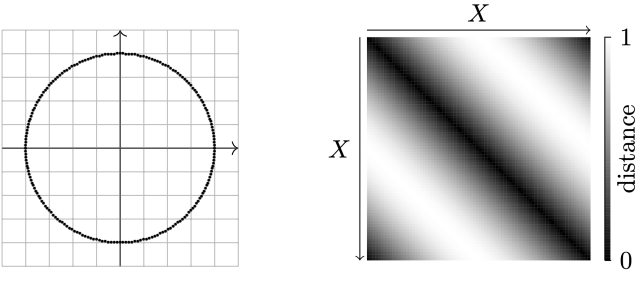

For simplicity, we take $n$ to be an even positive integer (this assumption avoids the small technical point that the diameter of $X$ is not a constant when $n$ is odd). We first rescale the Euclidean distance so that 
\begin{equation}
d_{X}(x_j,x_k)=\frac{1}{2r}d_{\R^2}(x_j,x_k)
=\sqrt{\frac{1-c_{k-j}}{2}},
\end{equation}
so that the diameter of $X$ with respect to $d_{X}$ is $1$. This normalizes the point cloud distance, which will be useful when we compare with the distance in quantum state space, which has a maximum value of $1$ as well. 

In the subsequent examples, we will implement the following general procedure for each quantum encoding. 
\begin{enumerate}
\item Identify the fundamental domain of the quantum encoding centered at $0$.  
\item For several values of $r$, calculate the associated Bures fidelity distance matrix $d_{F}$ and visualize it as a heat map while simultaneously comparing it to the original distance matrix associated with $d_{X}$. Then calculate the strain, distortion, and stress associated with the quantum encoding. 
\item If the codomain of the quantum encoding is the space of qubits, visualize the quantum encoded data on the Bloch sphere. If the dimension of the quantum state space exceeds three, apply cMDS to the distance matrix of the quantum encoded data to visualize the relative positions of the data points back in Euclidean space $\R^2$. This provides an alternative visualization (besides the heat map) for us to identify how much the distances have been distorted in the quantum encoding. Such a distortion may significantly alter the inferred topology of the data manifold. 
\end{enumerate}

\subsubsection{Square-root encoding for ideal circle}

For square-root encoding as defined in Example~\ref{ex:squarerootencoding}, we first must apply the uniform transformation~\eqref{eqn:TXmapping} to map the $n^{\text{th}}$ roots of unity onto the standard 2-simplex $\Delta^2$ in $\R^{3}$. After applying the uniform transformation, each such root of unity has a vector representation in $\Delta^2\subset\R^3$ as 
\begin{equation}
\frac{1}{3}\begin{bmatrix}1+\frac{1+\sqrt{3}}{2\sqrt{2}}c_k+\frac{1-\sqrt{3}}{2\sqrt{2}}s_k\\
1+\frac{1-\sqrt{3}}{2\sqrt{2}}c_k+\frac{1+\sqrt{3}}{2\sqrt{2}}s_k\\
1-\frac{1}{\sqrt{2}}c_k-\frac{1}{\sqrt{2}}s_k\\
\end{bmatrix}
\end{equation}
for $k\in\{0,1,2,\dots,n-1\}$ and $n\in\N$ with $n\ge2$. The diagonal encoding therefore takes the point $x_k$ to the qutrit state 
\begin{equation}
\label{eqn:sqrtencodingcircle}
\ket{x_k}=\frac{1}{\sqrt{3}}
\begin{bmatrix}\sqrt{1+\frac{1+\sqrt{3}}{2\sqrt{2}}c_k+\frac{1-\sqrt{3}}{2\sqrt{2}}s_k}\\
\sqrt{1+\frac{1-\sqrt{3}}{2\sqrt{2}}c_k+\frac{1+\sqrt{3}}{2\sqrt{2}}s_k}\\
\sqrt{1-\frac{1}{\sqrt{2}}c_k-\frac{1}{\sqrt{2}}s_k}\\
\end{bmatrix}
\end{equation}
(alternatively, one could send it to a 2-qubit state by setting the $\ket{11}$ component to be zero). The Bures fidelity distance between $\ket{x_j}$ and $\ket{x_k}$ is 
$d_{F}\big(\ket{x_j},\ket{x_k}\big)=\sqrt{1-|\langle x_{j}|x_{k}\rangle|}$
and is calculated 
using~\eqref{eqn:sqrtencodingcircle}. A heatmap of the associated distance matrix is shown in 
Figure~\ref{fig:ideal-circle-sqrt-heatmapsandCMDS}. 
Note that because of the rescaling in the definition of the uniform transformation, this distance matrix does not depend on $r$. 

\begin{figure}
\begin{tabular}{cc}
\begin{tikzpicture}
\node at (1.3375,-1.3375) {\includegraphics[width=2.675cm]{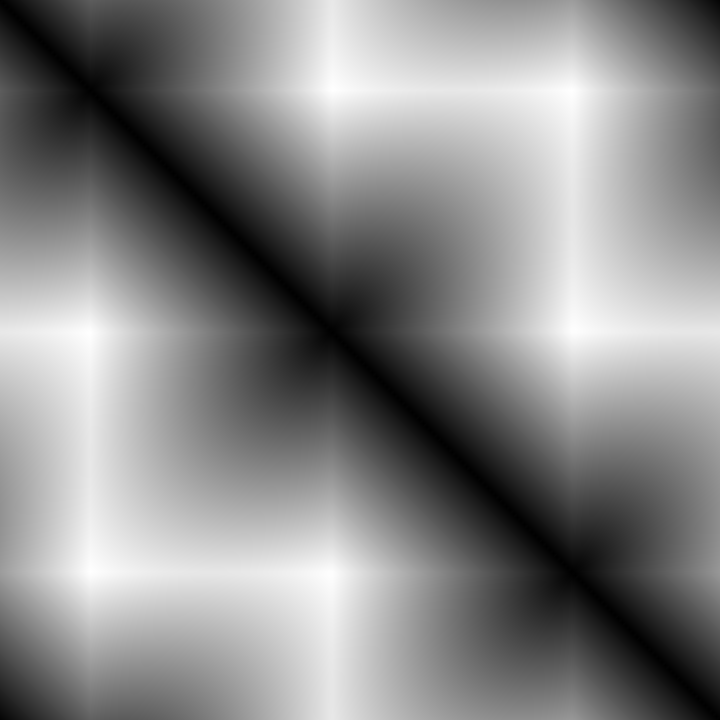}};
\draw[->] (-0.1,0) -- node[left]{$Y$} (-0.1,-2.675);
\draw[->] (0,0.1) -- node[above]{$Y$} (2.675,0.1);
  \foreach \x in {0,0.025,0.05,...,1.00} {
      \fill[black,opacity={\x}] ({2.675+0.2},{2.675*(-\x)}) rectangle ({2.675*(1+0.025)+0.2},{2.675*(-\x-0.025)});
  }
  \draw ({2.675+0.2},{-2.675}) -- ({2.675+0.2},{0}) -- ({2.675*(1+0.025)+0.2+0.1},{0}) node[right]{$1$};
\draw ({2.675+0.2},{-2.675}) -- ({2.675*(1+0.025)+0.2+0.1},{-2.675}) node[right]{$0$};
\node at ({2.675+0.5},{-2.675/2}) {\rotatebox{90}{distance}};
\end{tikzpicture}
\;\;&\;\;
\includegraphics[width=3cm]{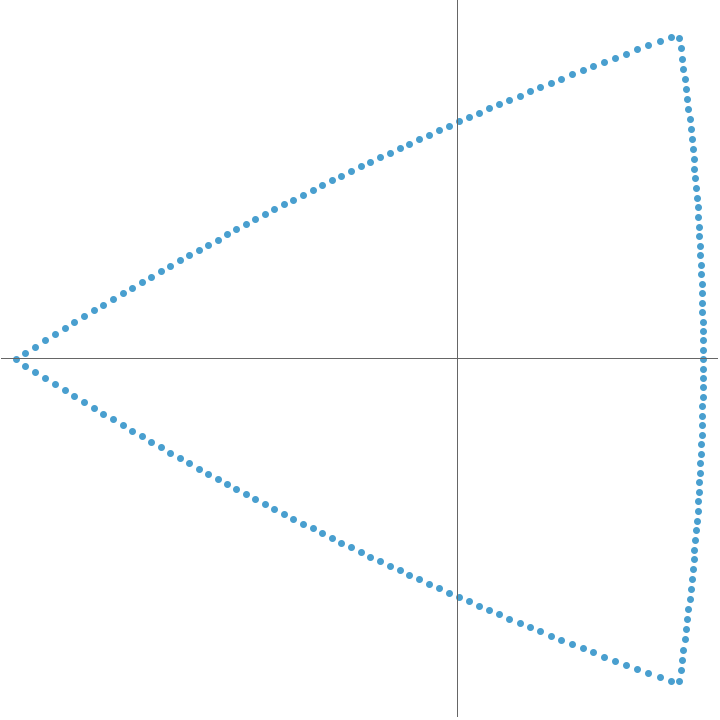}\\
\includegraphics[width=3.25cm]{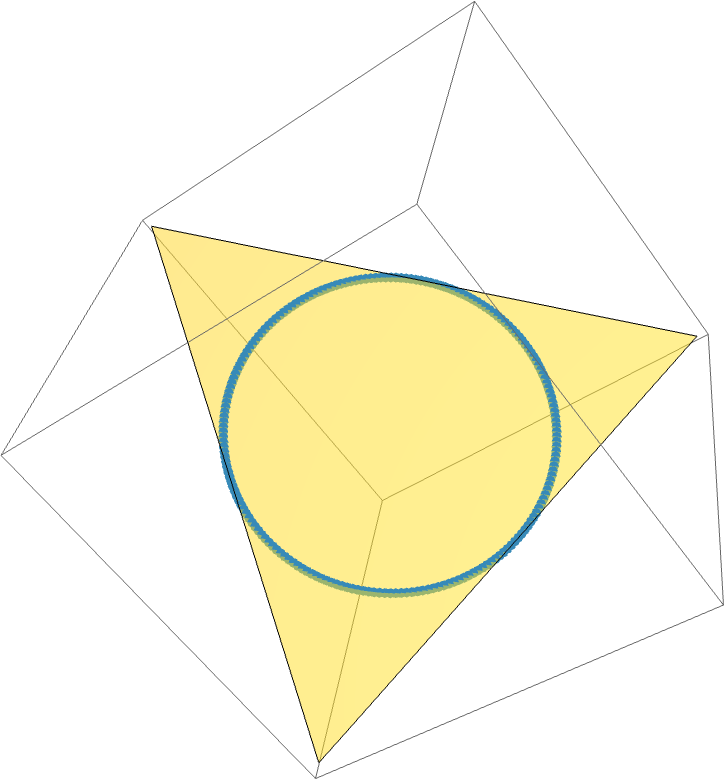}&
\includegraphics[width=3.25cm]{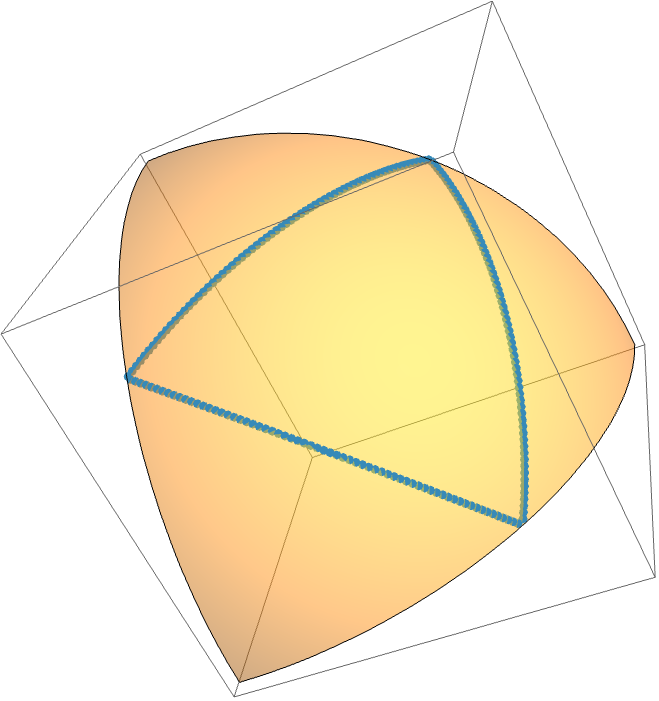}
\end{tabular}
\caption{A heatmap of the Bures fidelity distance matrix associated with the quantum encoding given by first applying the uniform transformation to the set of $200^{\text{th}}$ roots of unity and then applying the square-root encoding (top left). A visualization of the uniformly transformed data is shown on the bottom left, while taking the square-root of these points is shown on the unit sphere on the bottom right. The top right picture applies cMDS to the distance matrix to the quantum encoded data. The sharp corners are obtained due to the proximity of the data near the faces of the simplex.}
\label{fig:ideal-circle-sqrt-heatmapsandCMDS}
\end{figure}
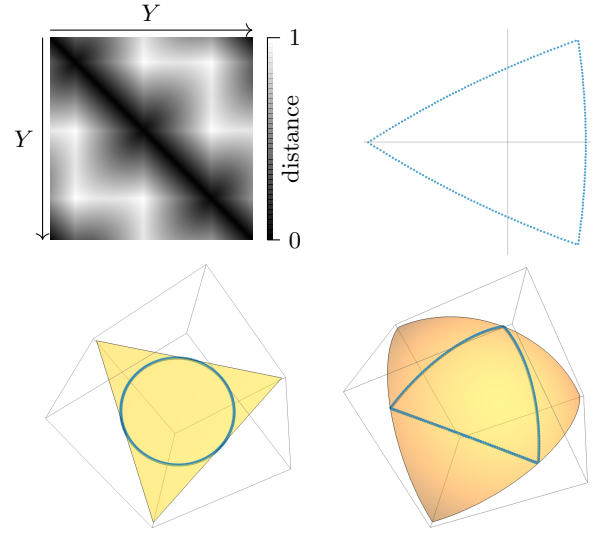

As can be seen from Figure~\ref{fig:ideal-circle-sqrt-heatmapsandCMDS}, this transformation approximately preserves the topology, but heavily distorts the geometry. The reason for this is because some data points end up on or near the faces of the simplex after applying the uniform transformation~\eqref{eqn:TXmapping}. When the square-root encoding is applied, points near these faces can concentrate. As a result, in order to better preserve the dataset, one could rescale the dataset by an additional factor to guarantee that the data points are farther away from the faces of the simplex. Figure~\ref{fig:ideal-circle-smaller-sqrt-heatmapsandCMDS} illustrates how this works for the $n^{\text{th}}$ roots of unity when rescaled by an additional factor of $\frac{3}{4}$. 

\begin{figure}
\begin{tabular}{cc}
\begin{tikzpicture}
\node at (1.3375,-1.3375) {\includegraphics[width=2.675cm]{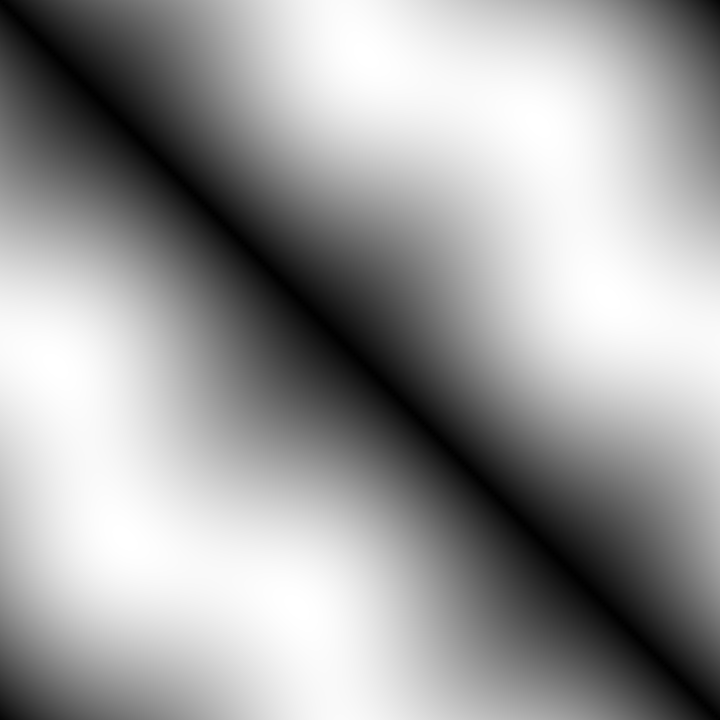}};
\draw[->] (-0.1,0) -- node[left]{$Y$} (-0.1,-2.675);
\draw[->] (0,0.1) -- node[above]{$Y$} (2.675,0.1);
  \foreach \x in {0,0.025,0.05,...,1.00} {
      \fill[black,opacity={\x}] ({2.675+0.2},{2.675*(-\x)}) rectangle ({2.675*(1+0.025)+0.2},{2.675*(-\x-0.025)});
  }
  \draw ({2.675+0.2},{-2.675}) -- ({2.675+0.2},{0}) -- ({2.675*(1+0.025)+0.2+0.1},{0}) node[right]{$1$};
\draw ({2.675+0.2},{-2.675}) -- ({2.675*(1+0.025)+0.2+0.1},{-2.675}) node[right]{$0$};
\node at ({2.675+0.5},{-2.675/2}) {\rotatebox{90}{distance}};
\end{tikzpicture}
\;\;&\;\;
\includegraphics[width=3cm]{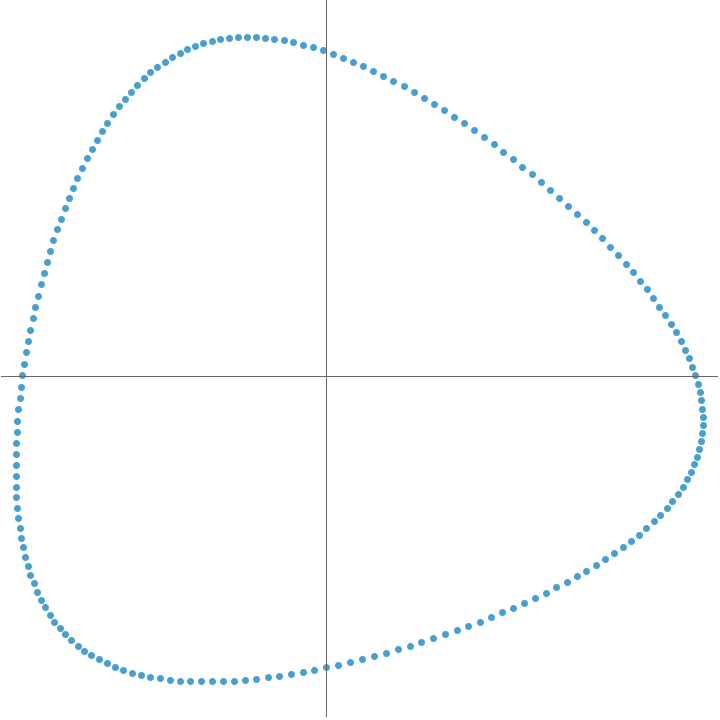}\\
\includegraphics[width=3.25cm]{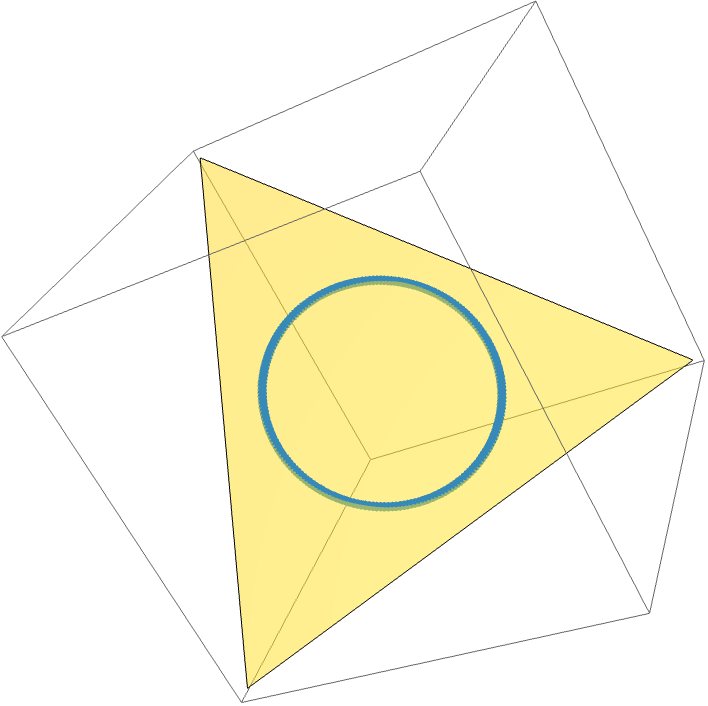}&
\includegraphics[width=3.25cm]{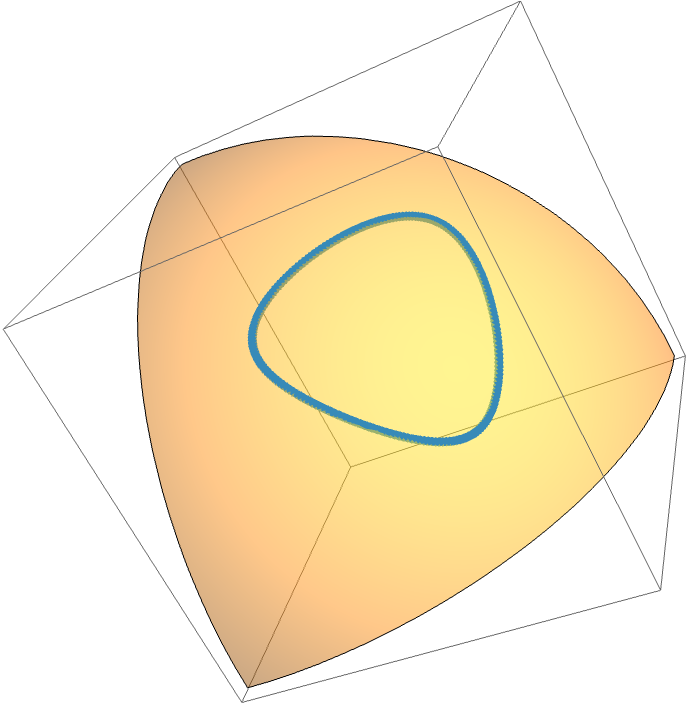}
\end{tabular}
\caption{This is the same as Figure~\ref{fig:ideal-circle-sqrt-heatmapsandCMDS} with the main difference being that the original dataset is rescaled by an additional factor of $\frac{3}{4}$ in the definition of the uniform transformation in order to guarantee that the data are sufficiently far away from the faces of the simplex after square-root encoding is applied.}
\label{fig:ideal-circle-smaller-sqrt-heatmapsandCMDS}
\end{figure}
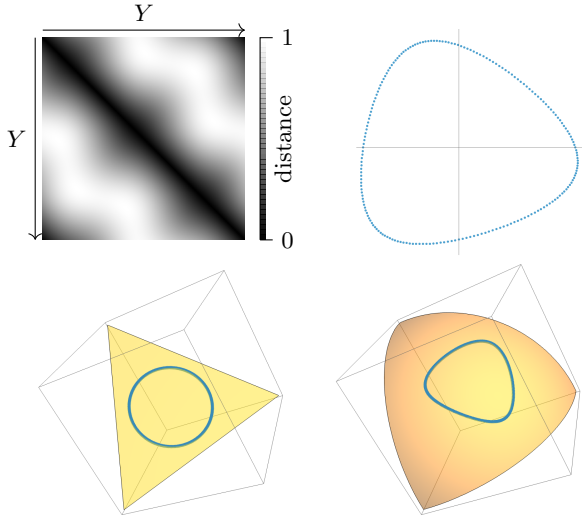

\subsubsection{Angle encoding for ideal circle}
\label{sec:angleencodingidealcircle}

For angle encoding, we refer back to Example~\ref{ex:angleencoding}. In this example, we specifically use the pure state from Eqn.~\ref{eqn:angleencodingstate} or equivalently the density matrix from Eqn.~\ref{eqn:angleencodingdensitymatrix}. The element $x_{k}$ of the set $X$ gets sent to the 2-qubit pure quantum state 
\begin{align}
\ket{x_{k}}&=\left(\cos\left(\frac{r c_{k}}{2}\right)\ket{0}+\sin\left(\frac{rc_{k}}{2}\right)\ket{1}\right) \nonumber\\
&\,\otimes \left(\cos\left(\frac{r s_{k}}{2}\right)\ket{0}+\sin\left(\frac{rs_{k}}{2}\right)\ket{1}\right).
\end{align}
The inner product between $\ket{x_{j}}$ and $\ket{x_{k}}$ is therefore 
\begin{equation}
\label{eqn:angleencodinginnerproduct}
\langle x_{j}|x_{k}\rangle=\cos\left(\frac{r(c_{k}-c_{j})}{2}\right)\cos\left(\frac{r(s_{k}-s_{j})}{2}\right).
\end{equation}
Thus, the Bures fidelity distance between $\ket{x_j}$ and $\ket{x_k}$ is 
$d_{F}\big(\ket{x_j},\ket{x_k}\big)=\sqrt{1-|\langle x_{j}|x_{k}\rangle|}$
and is calculated using~\eqref{eqn:angleencodinginnerproduct}.
A heatmap of this distance is shown in 
Figure~\ref{fig:ideal-circle-angle-heatmapsandCMDS}
for various values of $r$ (when $r<\frac{\pi}{2}$, the maximum distance between the quantum encoded data is less than $1$, so we rescale the distance to make the maximum distance $1$). 

\begin{figure*}[!t]
\begin{tabular}{cccc}
\includegraphics[width=3.75cm]{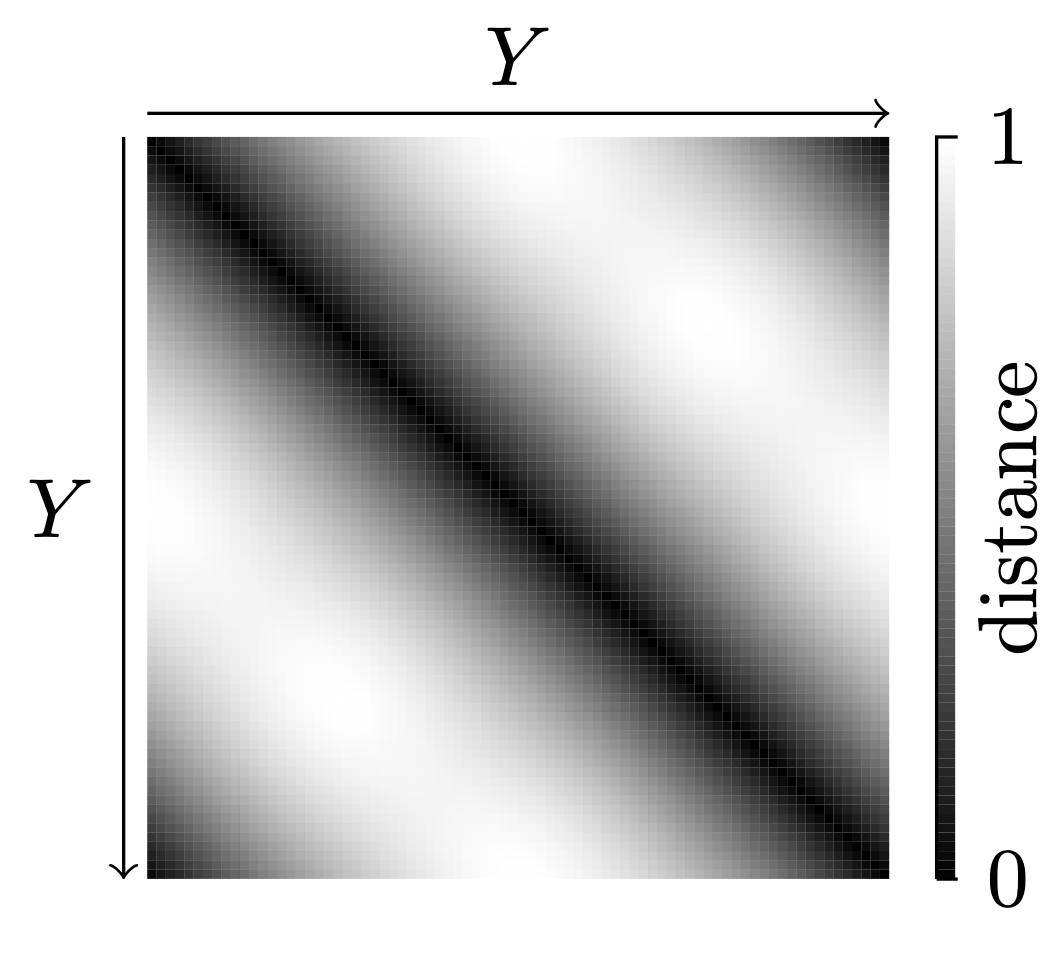}
& 
\includegraphics[width=3.75cm]{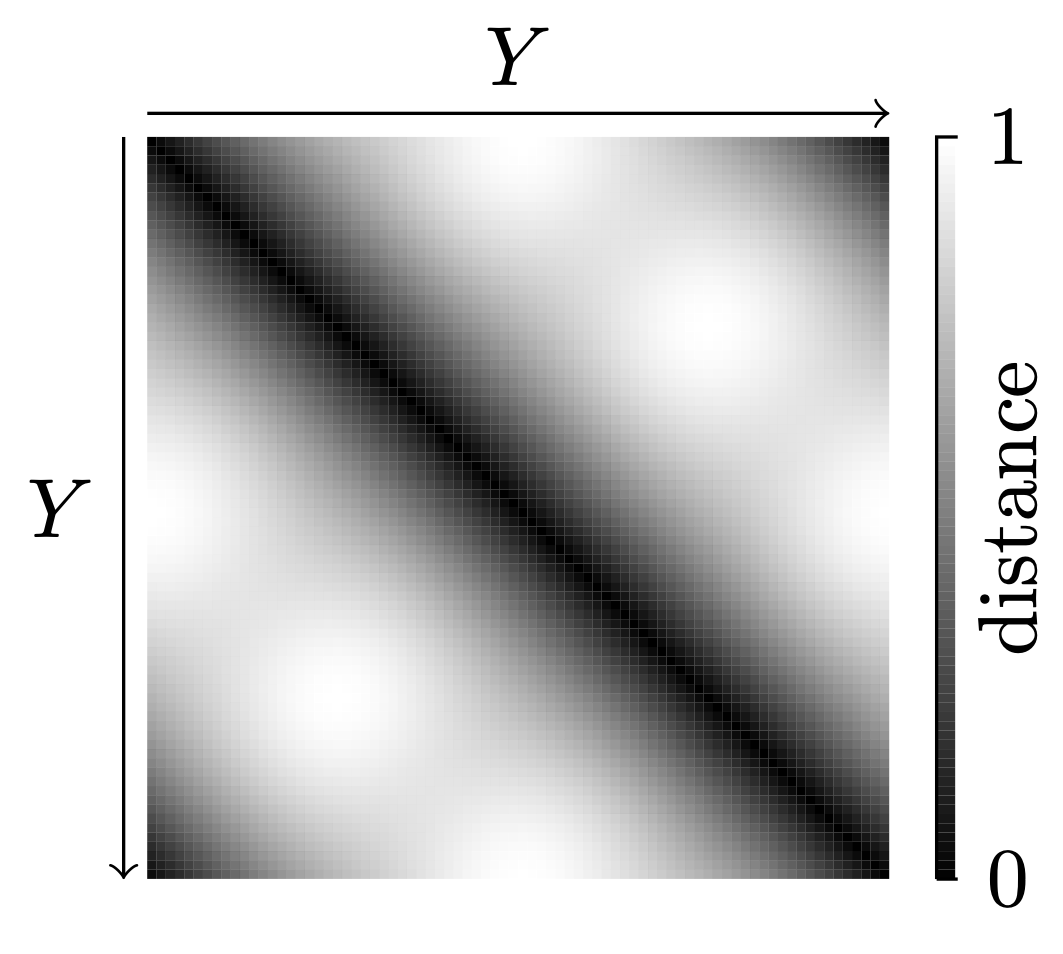}
&
\includegraphics[width=3.75cm]{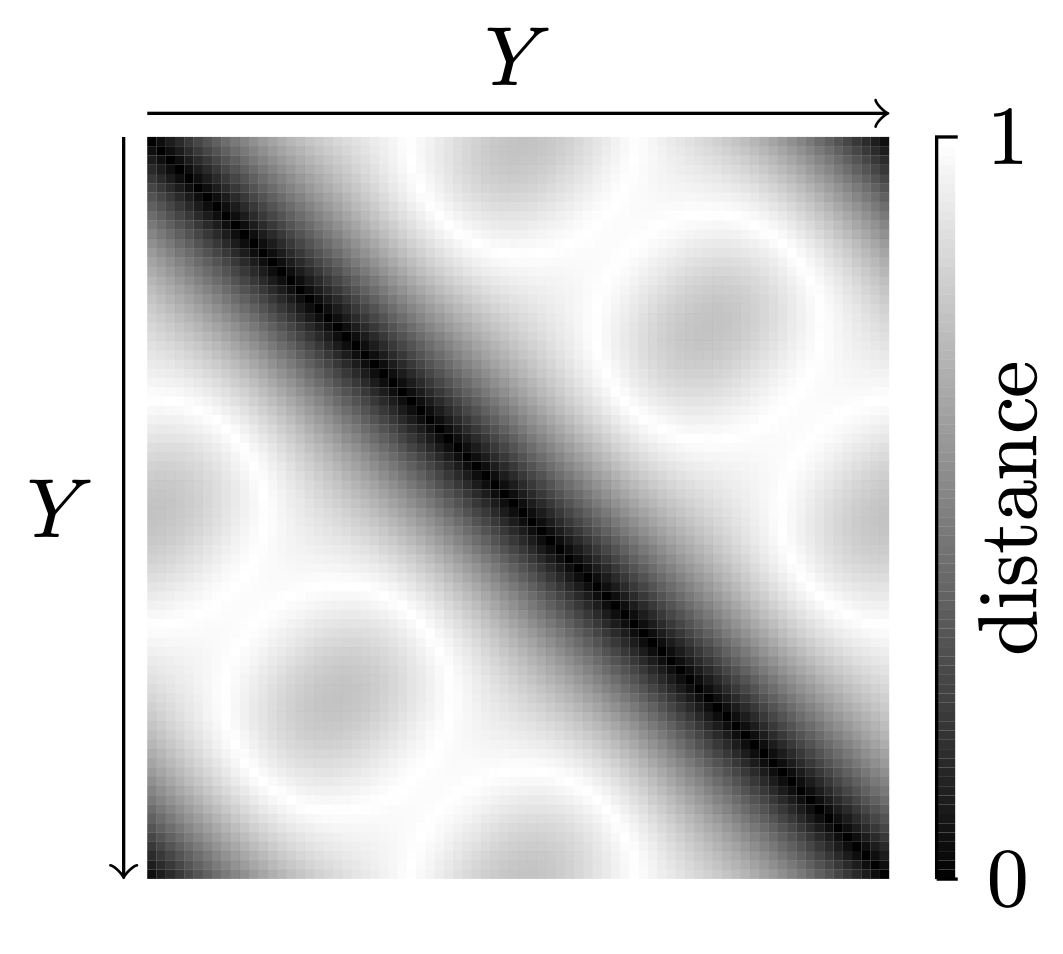}
&
\includegraphics[width=3.75cm]{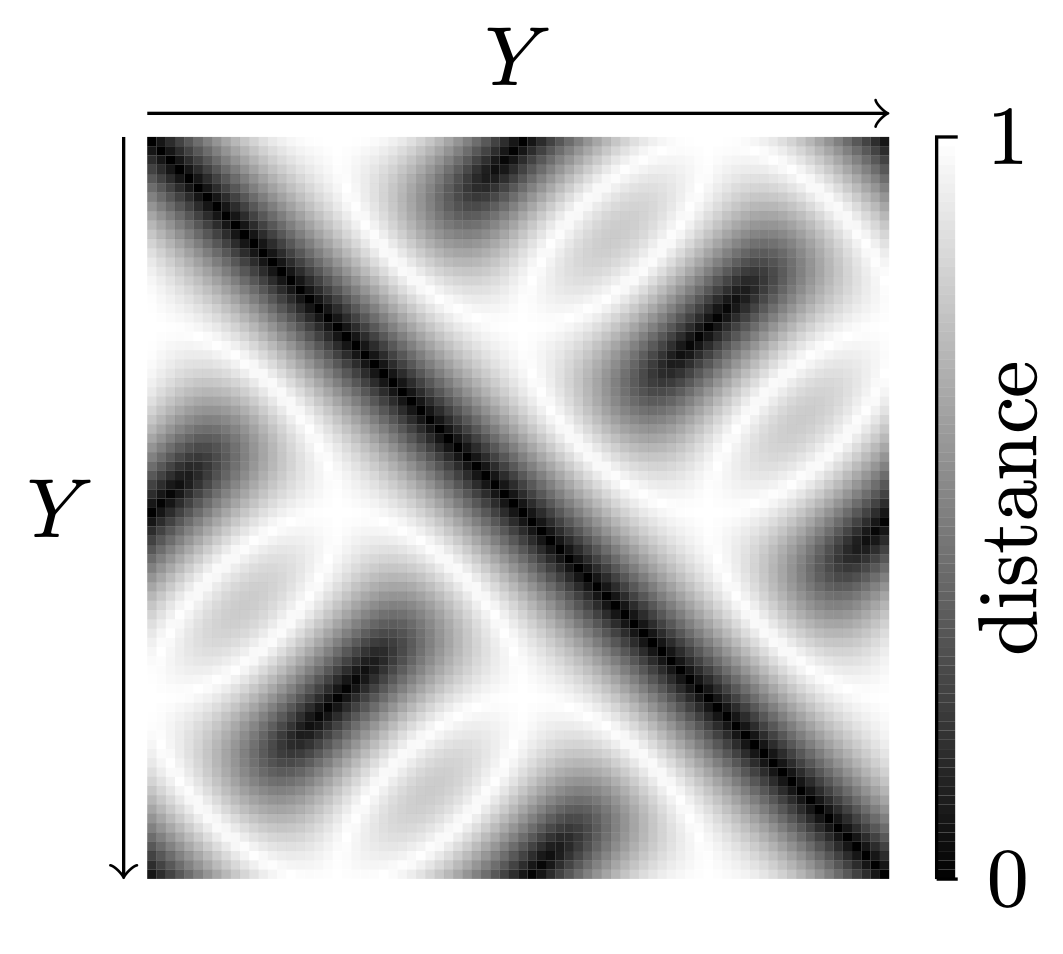} 
\\
\includegraphics[width=3.75cm]{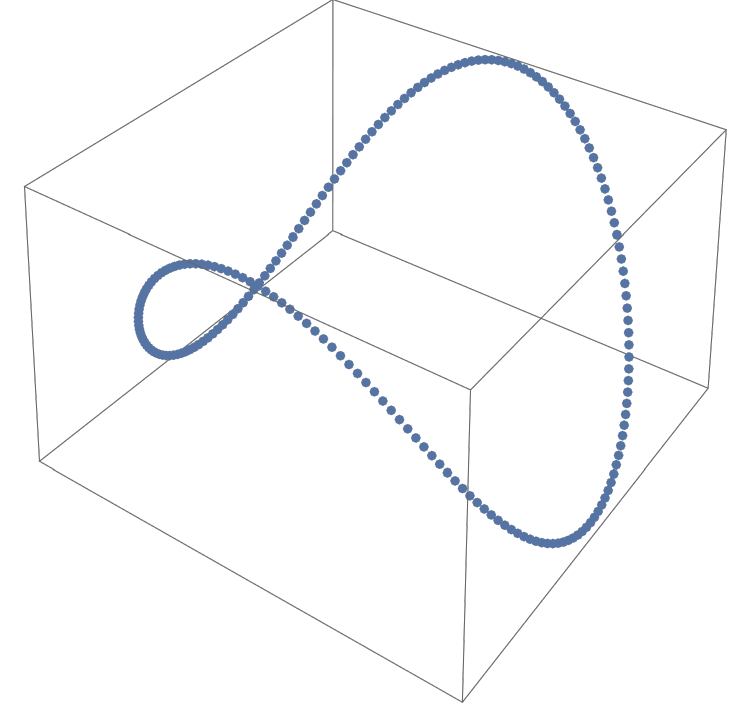}
&
\includegraphics[width=3.75cm]{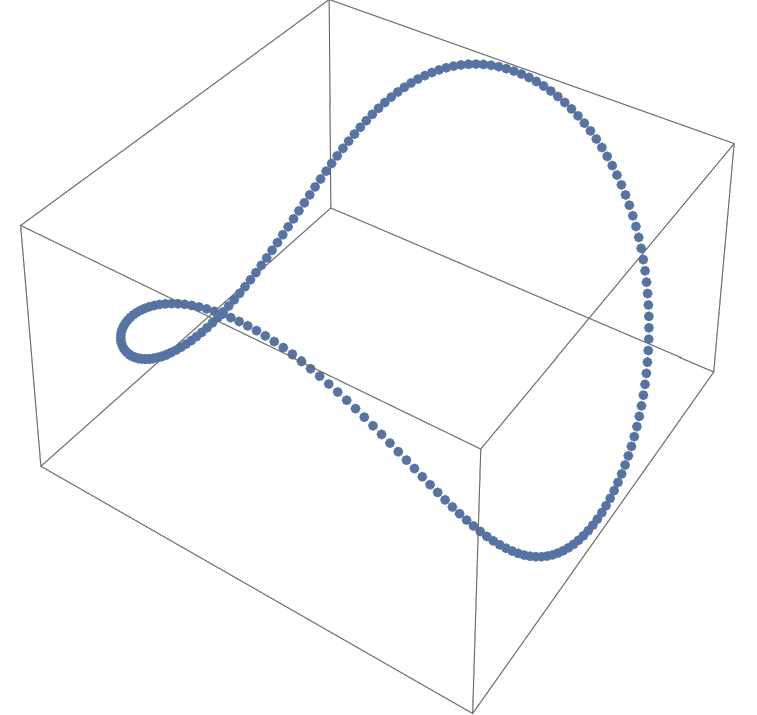}
&
\includegraphics[width=3.75cm]{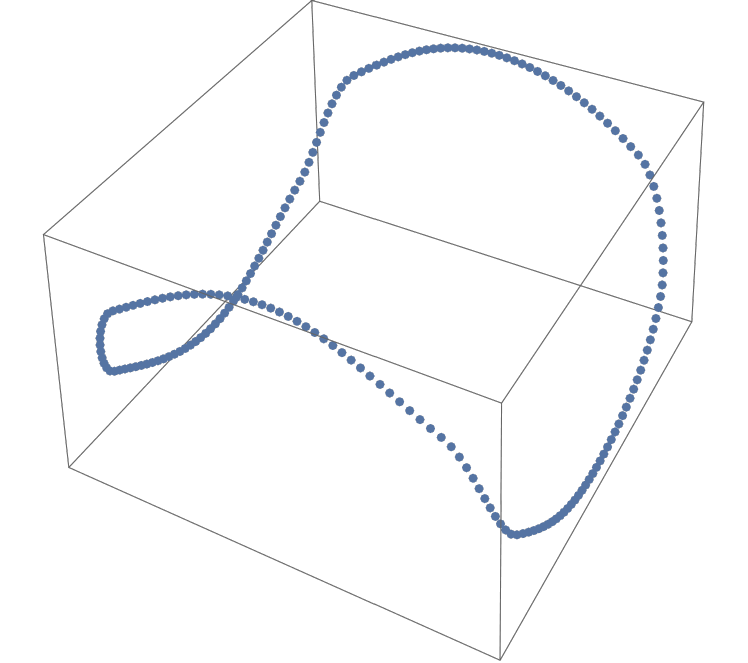}
&
\includegraphics[width=3.75cm]{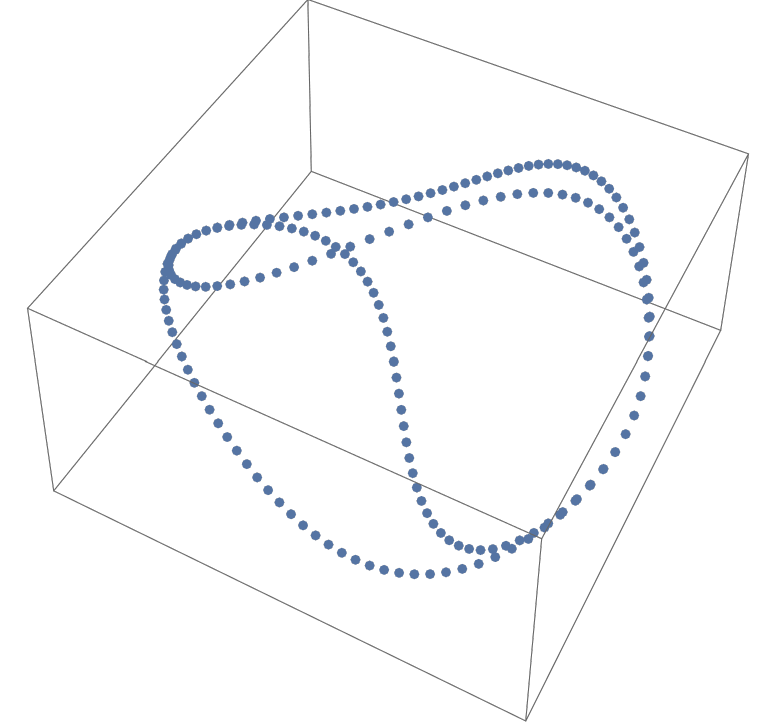}
\end{tabular}
\caption{The top row depicts the heatmaps of the Bures fidelity distance matrix after applying angle encoding to the dataset $X$ from Figure~\ref{fig:200rootsofunity} for values of $r$ given by (from left to right) $r=1$, $r=\frac{\pi}{2}$, $r=2$, and $r=\pi$. When $r>\frac{\pi}{2}$, points that should be far away begin to come closer together. When $r=\pi$, antipodal points along the horizontal and vertical axes get mapped to the same quantum state. The bottom row applies classical Multidimensional Scaling (cMDS) to the four distance matrices. Rather than using the first two principal coordinates associated with cMDS, we use the first \emph{three} principal coordinates because the circle is wrapping around the quantum state space in a highly nontrivial way that cannot be captured adequately using cMDS with the first two principal components. Indeed, even when $r=1$, we can already see that cMDS depicts a circle where one set of antipodal points are raised up along a third axis so that they are closer to each other, while the opposite antipodal points are lowered down along that third axis so that they are also closer to each other. These visualizations also show how the topology of the original point cloud is distorted as the radius increases. When $r=\pi$, the topological shape is that of a graph with two vertices and four edges connecting those two vertices.}
\label{fig:ideal-circle-angle-heatmapsandCMDS}
\end{figure*}

We can apply cMDS to these Bures fidelity distance matrices in order to visualize how much angle encoding distorts distances for these various values of $r$. This is depicted in %Figure~\ref{fig:idea-circle-angle-cMDS} 
Figure~\ref{fig:ideal-circle-angle-heatmapsandCMDS} 
underneath the heatmaps. 
A visualization of the original dataset $X$ for the various radii is shown in Figure~\ref{fig:circle-angle-encoding-torus} as if the square $[-\pi,\pi]\times[-\pi,\pi]$ had its boundary identified as if it was a torus. This further supports how the inferred topology of the quantum encoded data could change if the dataset becomes too large and approaches the boundary of the fundamental domain of the encoding. 

\begin{figure}
\begin{tabular}{ccccccc}
    \begin{tikzpicture}[scale=0.30] 
    \draw[gray,thin,opacity=0.75,step=2.0cm,xshift=-1cm,yshift=-1cm] (-2.01,-2.01) grid (4.01,4.01); 
    \foreach \cx in {-2,0,2} {
    \foreach \cy in {-2,0,2} {
    \draw[samples=100,smooth,domain=0:360,very thick] plot ({\cx+(1/pi)*cos(\x)},{\cy+(1/pi)*sin(\x)});
    }
    }
    \end{tikzpicture}
    & & 
    \begin{tikzpicture}[scale=0.30] 
    \draw[gray,thin,opacity=0.75,step=2.0cm,xshift=-1cm,yshift=-1cm] (-2.01,-2.01) grid (4.01,4.01); 
    \foreach \cx in {-2,0,2} {
    \foreach \cy in {-2,0,2} {
    \draw[samples=100,smooth,domain=0:360,very thick] plot ({\cx+(1/2)*cos(\x)},{\cy+(1/2)*sin(\x)});
    }
    }
    \end{tikzpicture}
    & \phantom{x} & 
    \begin{tikzpicture}[scale=0.30] 
    \draw[gray,thin,opacity=0.75,step=2.0cm,xshift=-1cm,yshift=-1cm] (-2.01,-2.01) grid (4.01,4.01); 
    \foreach \cx in {-2,0,2} {
    \foreach \cy in {-2,0,2} {
    \draw[samples=100,smooth,domain=0:360,very thick] plot ({\cx+(2/pi)*cos(\x)},{\cy+(2/pi)*sin(\x)});
    }
    }
    \end{tikzpicture}
    & \phantom{x} &
    \begin{tikzpicture}[scale=0.30]
    \draw[gray,thin,opacity=0.75,step=2.0cm,xshift=-1cm,yshift=-1cm] (-2.01,-2.01) grid (4.01,4.01); 
    \foreach \cx in {-2,0,2} {
    \foreach \cy in {-2,0,2} {
    \draw[samples=100,smooth,domain=0:360,very thick] plot ({\cx+cos(\x)},{\cy+sin(\x)});
    }
    }
    \end{tikzpicture}
\end{tabular}
\caption{The dataset $X$ for values of $r$ given by $r=1$, $r=\frac{\pi}{2}$, $r=2$, and $r=\pi$ drawn on the domain $[-3\pi,3\pi)\times[-3\pi,3\pi)$, with the fundamental domain being $[-\pi,\pi)\times[-\pi,\pi)$ and all other domains obtained from the fundamental one by imposing periodic boundary conditions. In the case that $r=\pi$, one sees that the point $(r,0)$ becomes identified with the point $(-r,0)$ under angle encoding, while the point $(0,r)$ becomes identified with the point $(0,-r)$. This gives a shape (topologically a graph with 2 vertices and 4 edges connecting those two vertices) matching the one shown in Figure~\ref{fig:ideal-circle-angle-heatmapsandCMDS}.}
\label{fig:circle-angle-encoding-torus}
\end{figure}
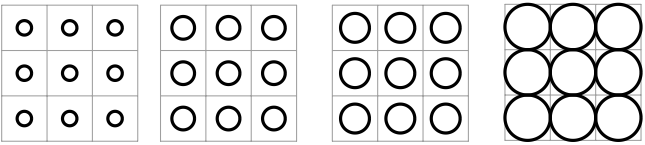

What this tells us is that in order to use angle encoding for mapping classical data from $\R^{d}$ to quantum states in $(\C^{2})^{\otimes d}$ in a way that somewhat appropriately preserves the topology of the original dataset, it is crucial to center the data and to rescale the data so that they fit not only inside of $[-\pi,\pi)^{d}$, but a suitable distance away from the boundary of this set. Namely, the data should be uniformly rescaled so that they fit inside of $\left[-\frac{\pi}{2},\frac{\pi}{2}\right)^{d}$. Although this point may have already been clear to some, we found that emphasizing this point with simple examples illustrates its importance. Therefore, when we compare different quantum encodings, we will center and rescale all the data so that they lie within a proper subset of the fundamental domain that would not distort the topology purely due to boundary and periodicity effects. 

\subsubsection{Dense angle encoding for ideal circle}

The benefit of dense angle encoding for the $n^{\text{th}}$ roots of unity data set $X$ is that we can actually visualize the quantum encoded data on the Bloch sphere. From that Bloch sphere, we can see the topology change directly from the data without the need for cMDS as in the case for ordinary angle encoding. Moreover, the fundamental domain for dense angle encoding is $[-\pi,\pi)\times[-\pi,\pi)$ so we will be able to make a direct comparison with the ordinary angle encoding case. 

\begin{figure*}[t!]
\begin{tabular}{cccc}
\includegraphics[width=3.75cm]{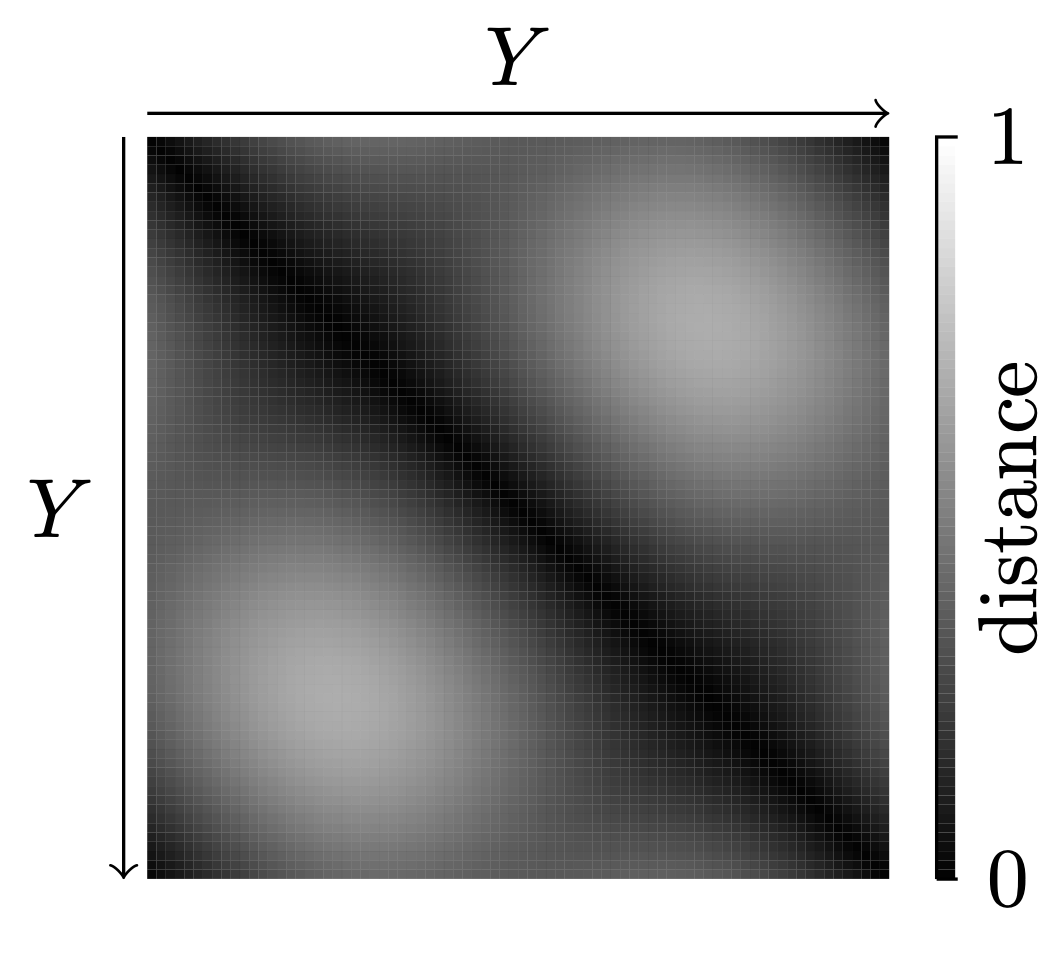}
&
\includegraphics[width=3.75cm]{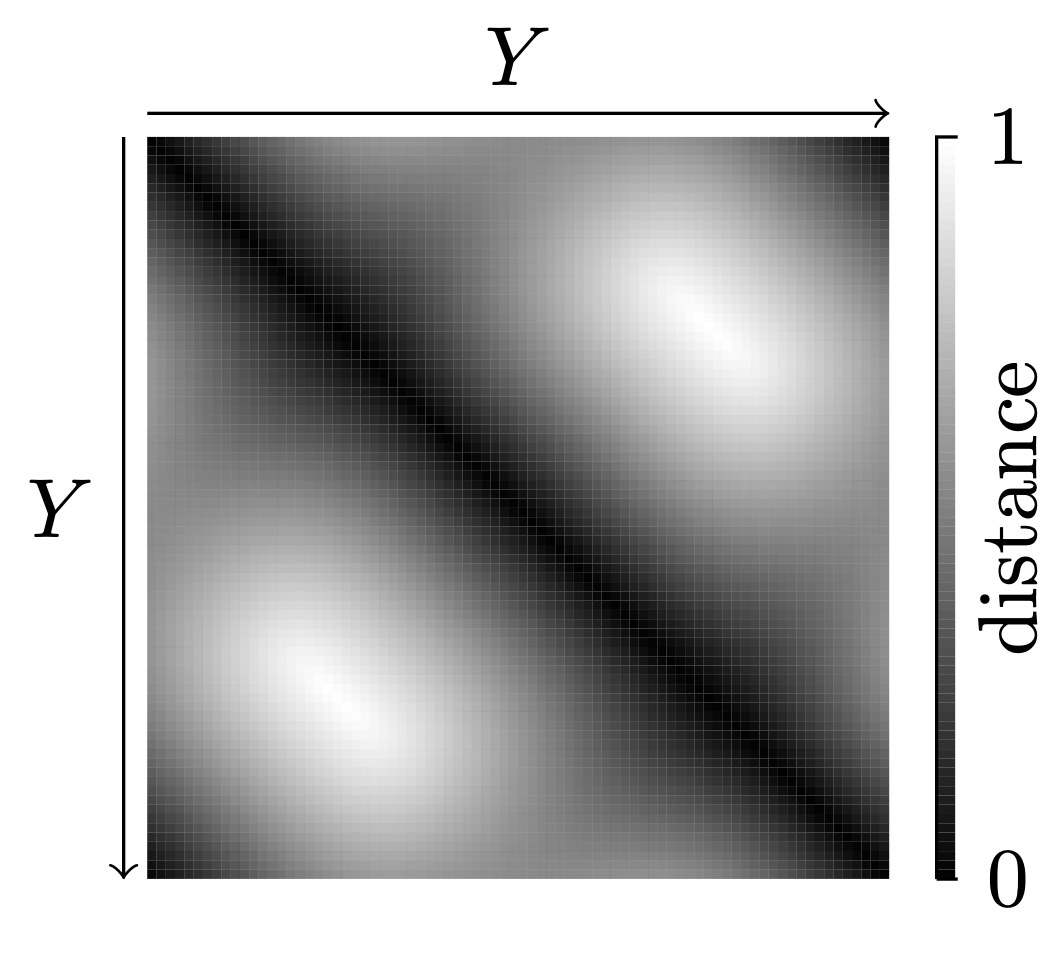}
&
\includegraphics[width=3.75cm]{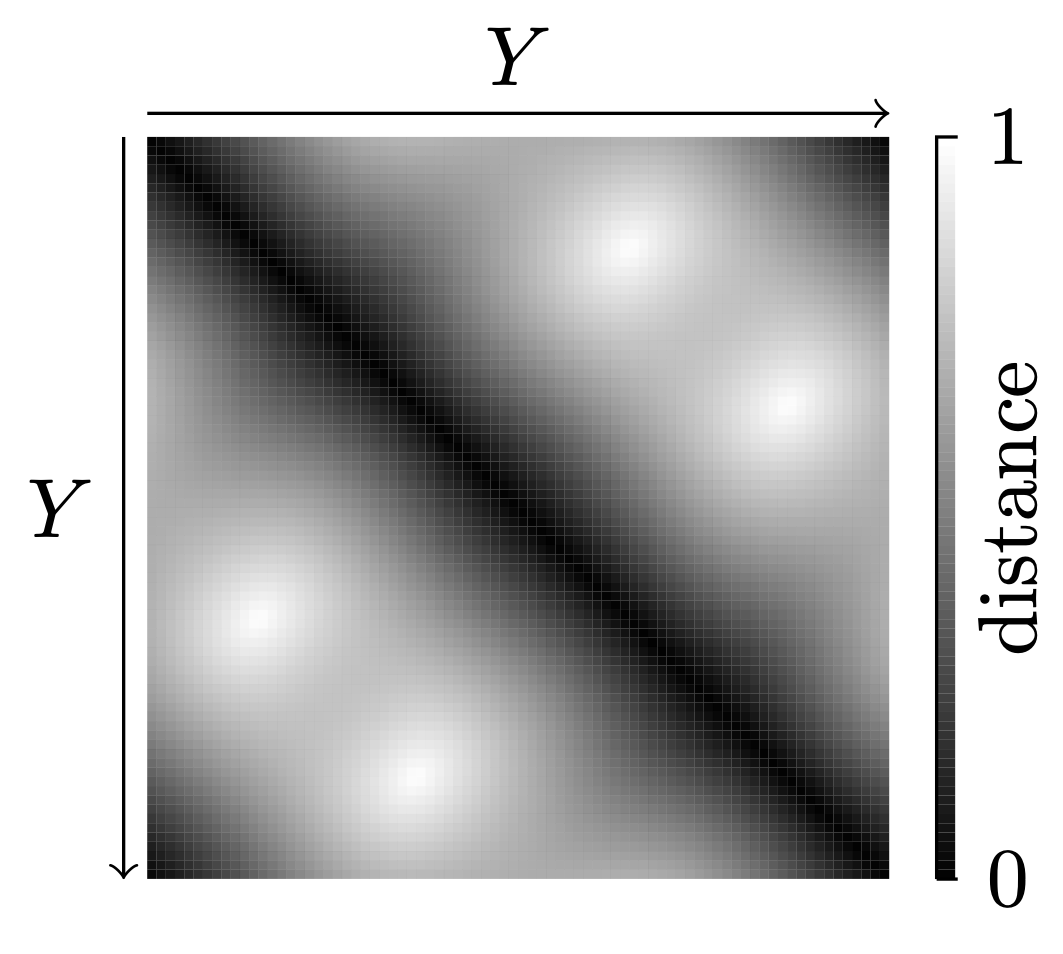}
&
\includegraphics[width=3.75cm]{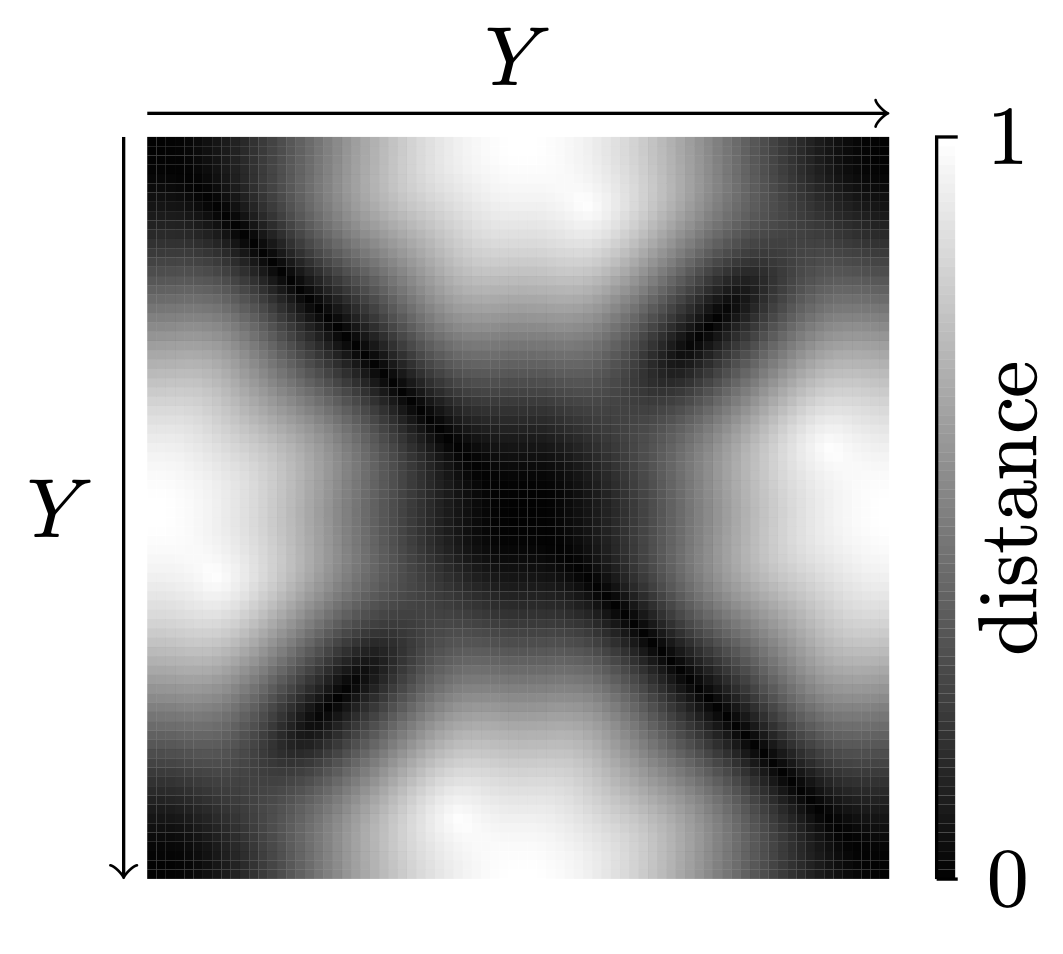}
\\
\includegraphics[width=3.25cm]{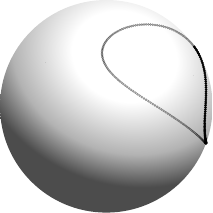}
&
\includegraphics[width=3.25cm]{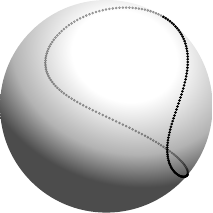}
&
\includegraphics[width=3.25cm]{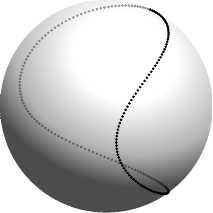}
&
\includegraphics[width=3.25cm]{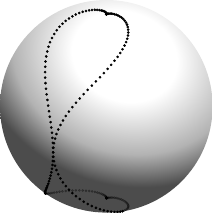}
\end{tabular}
\caption{The top row depicts heatmaps of the Bures fidelity distance matrix after applying dense angle encoding to the dataset $X$ for values of $r$ given by (from left to right) $r=1$, $r=\frac{\pi}{2}$, $r=2$, and $r=\pi$. When $r>\frac{\pi}{2}$, points that should be far away begin to come closer together. When $r=\pi$, antipodal points along the horizontal and vertical axes get mapped to the same quantum state. The bottom row depicts the images of the ideal circle data after applying dense angle encoding for the values of $r$ corresponding to the distance matrices above. When $r=\pi$, the encoding has the (persistent) topology of a figure $8$ (a graph with a single vertex and two edges from the vertex to itself), which has persistent Betti dimension $2$ for the degree $1$ homology.}
\label{fig:ideal-circle-denseangle-heatmapsBlochsphere}
\end{figure*}

Using similar notation as in the case of angle encoding, dense angle encoding sends a data point $x_k=(rc_k,rs_k)$ to 
\begin{equation}
\ket{x_k}=\cos\left(\frac{rc_k+\pi}{4}\right)\ket{0}+e^{irs_k}\sin\left(\frac{rc_k+\pi}{4}\right)\ket{1},
\end{equation}
which leads to the inner product
\begin{align}
\langle x_{j}&|x_{k}\rangle=\cos\left(\frac{rc_{j}+\pi}{4}\right)\cos\left(\frac{rc_{k}+\pi}{4}\right) \nonumber \\
&\quad+e^{ir(s_{k}-s_{j})}\sin\left(\frac{rc_j+\pi}{4}\right)\sin\left(\frac{rc_k+\pi}{4}\right).
\label{eqn:denseangleinnerproduct}
\end{align}
Hence, the Bures fidelity distance between $\ket{x_j}$ and $\ket{x_k}$ is 
$d_{F}\big(\ket{x_j},\ket{x_k}\big)=\sqrt{1-|\langle x_{j}|x_{k}\rangle|}$
and is calculated using~\eqref{eqn:denseangleinnerproduct}. Heatmaps of this distance matrix for various values of $r$ are shown in Figure~\ref{fig:ideal-circle-denseangle-heatmapsBlochsphere} together with the corresponding images of the datasets on the Bloch sphere underneath.

The images of the data set on the Bloch sphere are generated by calculating the $(x,y,z)$ coordinates associated with the states $\ket{x_k}$, which one can show are 
\begingroup
\allowdisplaybreaks
\begin{align}
x&=\langle x_k|\sigma_1|x_k\rangle=\sin\left(\frac{rc_{k}+\pi}{2}\right)\cos(rs_{k}) \nonumber \\
y&=\langle x_k|\sigma_2|x_k\rangle=\sin\left(\frac{rc_{k}+\pi}{2}\right)\sin(rs_{k}) \nonumber \\
z&=\langle x_k|\sigma_3|x_k\rangle=\cos\left(\frac{rc_{k}+\pi}{2}\right).
\end{align}
\endgroup

\subsubsection{IQP encoding for ideal circle}

For IQP encoding, we refer back to Example~\ref{ex:IQPencoding}. Under this encoding, the element $x_{k}=(rc_k,rs_k)$ of the set $X$ gets sent to the 2-qubit pure quantum state 
\begin{equation}
\ket{x_k}=e^{i \big( rc_{k} \sigma_{1}\otimes\mathds{1}_{2}+rs_{k}\mathds{1}_{2}\otimes\sigma_{1}+(\pi-rc_{k})(\pi-rs_{k})\sigma_{1}\otimes\sigma_{1} \big)} \ket{00}. 
\end{equation}
In this case, the explicit calculation of the inner product $\langle x_j| x_k\rangle$ is not particularly illuminating and is quite cumbersome to write. As a result, we omit the formula but mention that it is significantly more sensitive to perturbations than the previous encodings. The associated heatmaps associated with the pointwise distances between the IQP quantum encoded data are shown in Figure~\ref{fig:heatmapIQPencoding} for the same values of $r$ chosen before. 
We have skipped applying cMDS to the IQP encoded distance matrix due to the result being heavily distorted. 

\begin{figure}
\begin{tabular}{ccc}
\begin{tikzpicture}
\node at (1.3375,-1.3375) {\includegraphics[width=2.675cm]{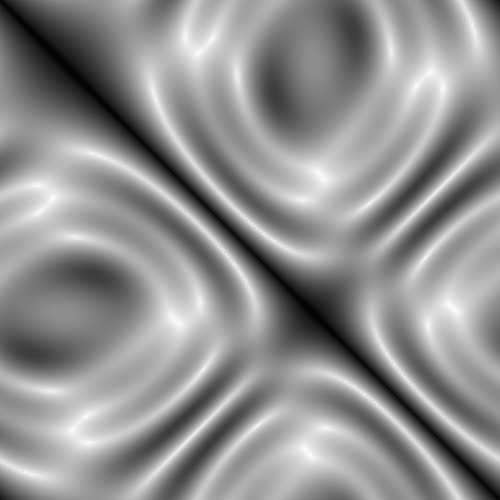}};
\draw[->] (-0.1,0) -- node[left]{$Y$} (-0.1,-2.675);
\draw[->] (0,0.1) -- node[above]{$Y$} (2.675,0.1);
  \foreach \x in {0,0.025,0.05,...,1.00} {
      \fill[black,opacity={\x}] ({2.675+0.2},{2.675*(-\x)}) rectangle ({2.675*(1+0.025)+0.2},{2.675*(-\x-0.025)});
  }
  \draw ({2.675+0.2},{-2.675}) -- ({2.675+0.2},{0}) -- ({2.675*(1+0.025)+0.2+0.1},{0}) node[right]{$1$};
\draw ({2.675+0.2},{-2.675}) -- ({2.675*(1+0.025)+0.2+0.1},{-2.675}) node[right]{$0$};
\node at ({2.675+0.5},{-2.675/2}) {\rotatebox{90}{distance}};
\end{tikzpicture}
&\phantom{x}&
\begin{tikzpicture}
\node at (1.3375,-1.3375) {\includegraphics[width=2.675cm]{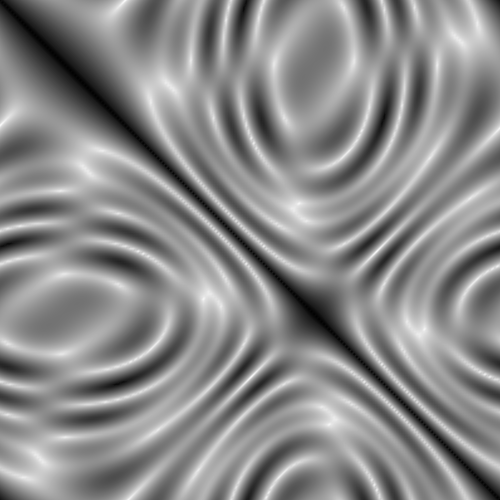}};
\draw[->] (-0.1,0) -- node[left]{$Y$} (-0.1,-2.675);
\draw[->] (0,0.1) -- node[above]{$Y$} (2.675,0.1);
  \foreach \x in {0,0.025,0.05,...,1.00} {
      \fill[black,opacity={\x}] ({2.675+0.2},{2.675*(-\x)}) rectangle ({2.675*(1+0.025)+0.2},{2.675*(-\x-0.025)});
  }
  \draw ({2.675+0.2},{-2.675}) -- ({2.675+0.2},{0}) -- ({2.675*(1+0.025)+0.2+0.1},{0}) node[right]{$1$};
\draw ({2.675+0.2},{-2.675}) -- ({2.675*(1+0.025)+0.2+0.1},{-2.675}) node[right]{$0$};
\node at ({2.675+0.5},{-2.675/2}) {\rotatebox{90}{distance}};
\end{tikzpicture}
\\
\begin{tikzpicture}
\node at (1.3375,-1.3375) {\includegraphics[width=2.675cm]{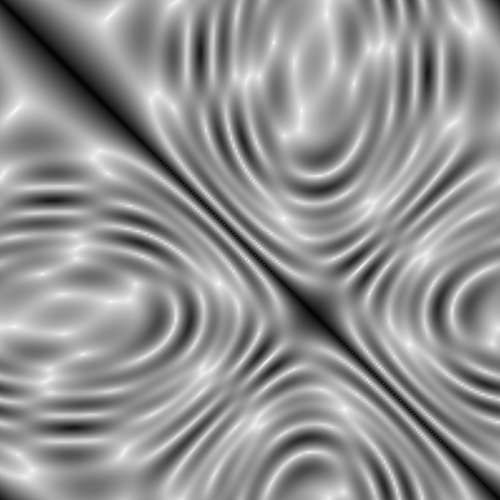}};
\draw[->] (-0.1,0) -- node[left]{$Y$} (-0.1,-2.675);
\draw[->] (0,0.1) -- node[above]{$Y$} (2.675,0.1);
  \foreach \x in {0,0.025,0.05,...,1.00} {
      \fill[black,opacity={\x}] ({2.675+0.2},{2.675*(-\x)}) rectangle ({2.675*(1+0.025)+0.2},{2.675*(-\x-0.025)});
  }
  \draw ({2.675+0.2},{-2.675}) -- ({2.675+0.2},{0}) -- ({2.675*(1+0.025)+0.2+0.1},{0}) node[right]{$1$};
\draw ({2.675+0.2},{-2.675}) -- ({2.675*(1+0.025)+0.2+0.1},{-2.675}) node[right]{$0$};
\node at ({2.675+0.5},{-2.675/2}) {\rotatebox{90}{distance}};
\end{tikzpicture}
&\phantom{x}&
\begin{tikzpicture}
\node at (1.3375,-1.3375) {\includegraphics[width=2.675cm]{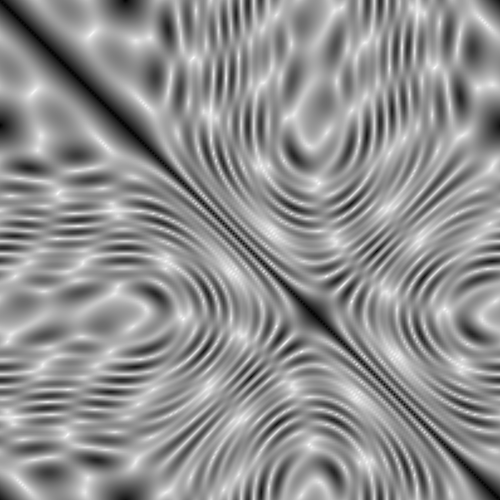}};
\draw[->] (-0.1,0) -- node[left]{$Y$} (-0.1,-2.675);
\draw[->] (0,0.1) -- node[above]{$Y$} (2.675,0.1);
  \foreach \x in {0,0.025,0.05,...,1.00} {
      \fill[black,opacity={\x}] ({2.675+0.2},{2.675*(-\x)}) rectangle ({2.675*(1+0.025)+0.2},{2.675*(-\x-0.025)});
  }
  \draw ({2.675+0.2},{-2.675}) -- ({2.675+0.2},{0}) -- ({2.675*(1+0.025)+0.2+0.1},{0}) node[right]{$1$};
\draw ({2.675+0.2},{-2.675}) -- ({2.675*(1+0.025)+0.2+0.1},{-2.675}) node[right]{$0$};
\node at ({2.675+0.5},{-2.675/2}) {\rotatebox{90}{distance}};
\end{tikzpicture}
\end{tabular}
\caption{Heatmaps of the Bures fidelity distance matrix after applying IQP encoding to the dataset $X$ for values of $r$ given by $r=1$ (top left), $r=\frac{\pi}{2}$ (top right), $r=2$ (bottom left), and $r=\pi$ (bottom right). These images illustrate that IQP encoding highly distorts the distances from the original dataset, rendering it essentially unsuitable for preserving the topology of the original dataset. Small perturbations of the dataset could result in drastically different distances. This is further supported later in Figure~\ref{fig:barcodes} in the case of noisy data.}
\label{fig:heatmapIQPencoding}
\end{figure}
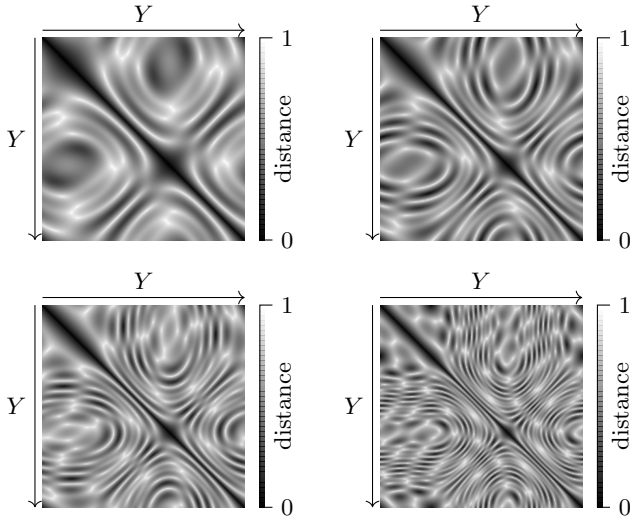

%%%%%%%%%%%%%%%%%%%%%%%%%%%%%%%%%%%%%%%%%%%%%%%%
\subsubsection{MDS encoding}
\label{subsubsec:MDScircle}
%%%%%%%%%%%%%%%%%%%%%%%%%%%%%%%%%%%%%%%%%%%%%%%%

We next briefly analyze MDS encoding by finding a function $f:X\to \CP^{1}$ that minimizes stress~\eqref{eqn:stressf}. Here, $\CP^{1}$ is complex projective space of complex dimension $1$ (i.e., the Bloch sphere). Since the classical dataset already lies on a circle, we will try to find an MDS encoding whose stress minimizer lies near the points on the equator of the Bloch sphere given by the points 
\begin{equation}
\rho(x_k):=\frac{1}{2}(\mathds{1}_{2}+c_k \sigma_1 + s_k \sigma_2), 
\end{equation}
which has corresponding pure state given by 
\begin{equation}
\label{eqn:MDSstartingpoint}
\ket{x_k}=\frac{1}{\sqrt{2}}\Big(\ket{0}+(c_k+i s_k)\ket{1}\Big).
\end{equation}
The corresponding inner product is 
\begin{equation}
\langle x_j|x_k\rangle
=\frac{1}{2}\big(1+c_{k-j} + i s_{k-j}\big),
\end{equation}
whose magnitude is 
\begin{equation}
\big|\langle x_j|x_k\rangle\big|=\sqrt{\frac{1+c_{k-j}}{2}}.
\end{equation}
Therefore, 
\begin{align}
\frac{1}{2}d_{\mathrm{Tr}}\big(\ket{x_j},\ket{x_k}\big)&=\sqrt{1-\big|\langle x_j|x_k\rangle\big|^2} \nonumber \\
&=\sqrt{\frac{1-c_{k-j}}{2}} \nonumber \\
&=d_{X}(x_j,x_k),
\end{align}
which shows that this encoding would be an isometry (and hence have zero distortion) if the quantum state space was equipped with (half) the trace distance~\eqref{eqn:Trpure}. This justifies that the encoding~\eqref{eqn:MDSstartingpoint} is a reasonable starting point to minimize the stress assuming that $\CP^{1}$ is equipped with the Bures fidelity distance. More generally, if 
\begin{equation}
\mathbf{y}=\big(\sin(\phi)\cos(\theta),\sin(\phi)\sin(\theta),\cos(\phi)\big)
\end{equation}
is a unit vector in Euclidean space in terms of spherical angular coordinates $\theta\in[0,2\pi)$ and $\phi\in[0,\pi]$, then the density matrix $\rho=\frac{1}{2}(\mathds{1}+\mathbf{y}\cdot\boldsymbol{\sigma})$ has corresponding vector representation as 
\begin{equation}
\label{eqn:genericBlochstate}
\ket{\mathbf{y}}=\cos\left(\frac{\phi}{2}\right)\ket{0}+e^{i\theta}\sin\left(\frac{\phi}{2}\right)\ket{1}.
\end{equation}
Therefore, the inner product between $\ket{\mathbf{y}}$ and $\ket{\mathbf{y}'}$ is 
\begin{equation}
\langle\mathbf{y}'|\mathbf{y}\rangle=\cos\left(\frac{\phi}{2}\right)\cos\left(\frac{\phi'}{2}\right)+e^{i(\theta'-\theta)}\sin\left(\frac{\phi}{2}\right)\sin\left(\frac{\phi'}{2}\right),
\end{equation}
and the magnitude-squared can be expressed as
\begin{equation}
\big|\langle\mathbf{y}'|\mathbf{y}\rangle\big|^2=\cos^2\left(\frac{\phi+\phi'}{2}\right)+\frac{1+\cos(\theta-\theta')}{2}\sin(\phi)\sin(\phi').
\end{equation}
After minimizing stress using Mathematica's function FindMinimum~\cite{Mathematica2024}, the local minimum of stress with respect to Bures fidelity distance was achieved at exactly the pure states given by~\eqref{eqn:MDSstartingpoint}. In other words, the $n^{\text{th}}$ roots of unity get mapped under the MDS encoding to the equator of the Bloch sphere at the locations given by~\eqref{eqn:MDSstartingpoint}. 

%%%%%%%%%%%%%%%%%%%%%%%%%%%%%%%%%%%%%%%%%%%%%%%%
\subsection{Noisy circle}
\label{subsec:noisycircle}
%%%%%%%%%%%%%%%%%%%%%%%%%%%%%%%%%%%%%%%%%%%%%%%%

We next look at a point cloud $X\subset\R^2$ consisting of $200$ points sampled from a probability distribution concentrated \emph{near} the unit circle centered at the origin, as shown in Figure~\ref{fig:circle}. In other words, we analyze how noise in the ideal circle data explored in Section~\ref{subsec:idealcircle} affects the quantum encodings and their associated persistent homologies. There are some differences with how the data are handled in this section as compared to the section on the ideal $n^{\text{th}}$ roots of unity. We outline these differences in the following first few paragraphs.

\begin{figure}%[ht!]
     \centering
    \includegraphics[width=.4\textwidth]{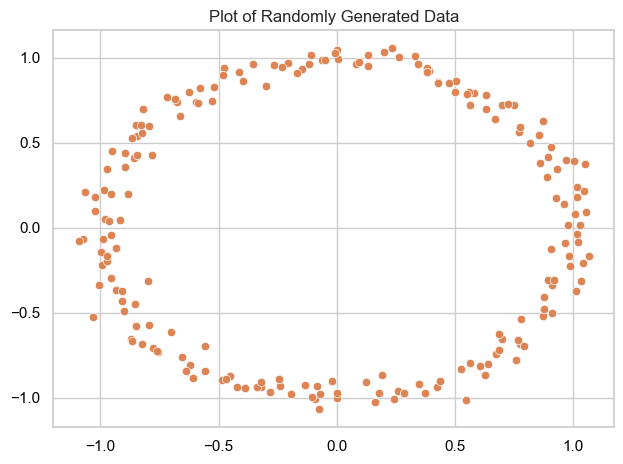}
\caption{Randomly perturbed circle with 200 synthetic data points, generated for the toy example used in Section~\ref{subsec:noisycircle}.}
\label{fig:circle}
\end{figure}

For the numerical experiments, we analyze how the distances and associated persistent homologies are affect by the following four transformations (1) uniform transformation from~\eqref{eqn:TXmapping} (which we know should leave the persistent homology invariant under a rescaling), (2) uniform transformation followed by square-root encoding (Example~\ref{ex:squarerootencoding}), (3) uniform transformation followed by a rescaling by $\pi$ and then followed by angle encoding~\eqref{eq:angle} (with $\vec{n}_{j}=(0,1,0)$ and $d=2$), and (4) uniform transformation followed by a rescaling by $\frac{\pi}{2}$ and then followed by IQP encoding~\eqref{eq:iqp-second} (with $d=2$). We denote the uniform transformed data as \emph{normed data} for brevity. The reason we precede each of the encodings by uniform transformation is to compare all encodings on as equal of a footing as possible (as if they were defined on the same space). The additional rescaling for angle encoding by $\pi$ after the uniform transformation guarantees that the input data into angle encoding all take values between $0$ and $\pi$, which is based on working with as large of a domain as possible without encountering the boundary effects discussed in Section~\ref{sec:angleencodingidealcircle}. 
For square-root, angle, and IQP encodings, pairwise distances between the quantum states are calculated using Qiskit's AerSimulator~\cite{aleksandrowicz2019qiskit} by calculating the fidelity between pairwise states using the SWAP test~\cite{BuhrmanQuantumFingerprinting2001}. 
The Bures angle distance (as opposed to the Bures fidelity distance used in Section~\ref{subsec:idealcircle}) is then calculated using these values for fidelity. 

\begin{figure*}[!t]
    \centering
    \renewcommand{\arraystretch}{1.3}%
    \begin{tabular}{cccc}
    (a) Normed Data
    & (b) Square Root
    & (c) Angle
    & (d) IQP
    \\
    \includegraphics[width=4.3 cm]{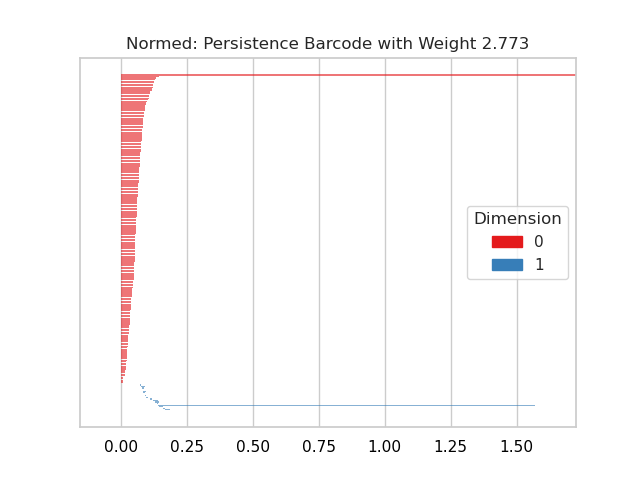} 
    & \includegraphics[width=4.3 cm]{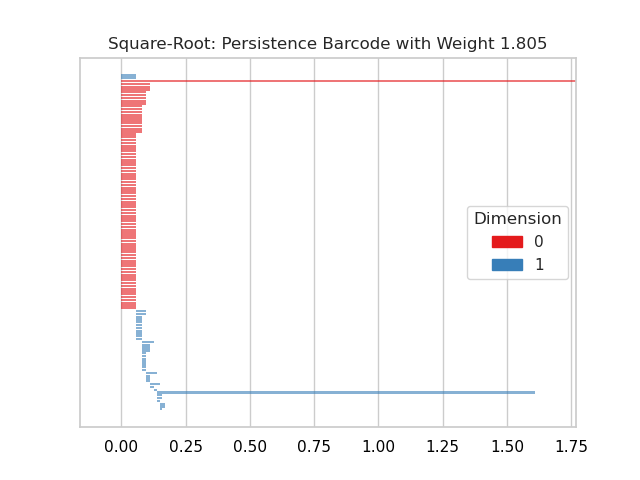} 
    & \includegraphics[width=4.3 cm]{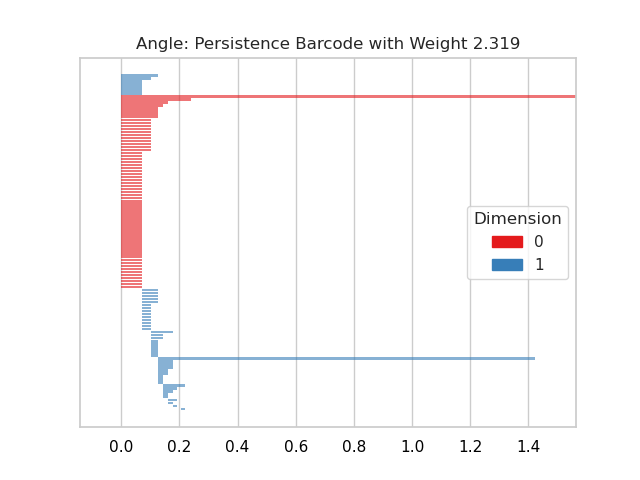} 
    & \includegraphics[width=4.3 cm]{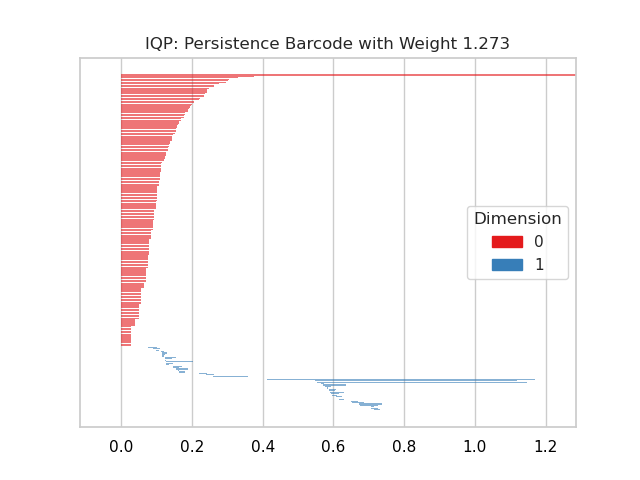} 
    \\  
    \includegraphics[width=4.3 cm]{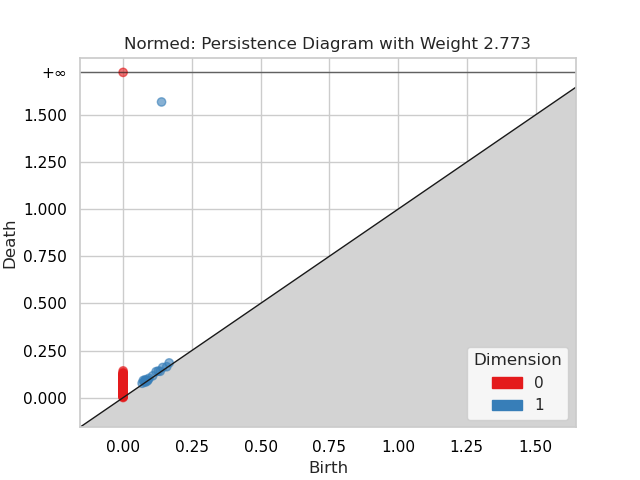}
    &  \includegraphics[width=4.3 cm]{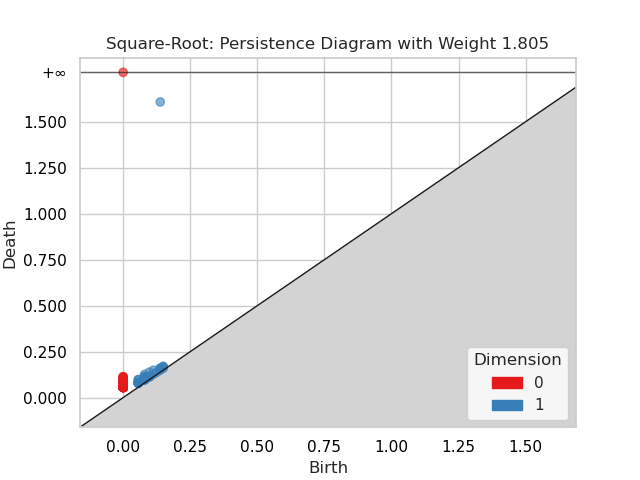}
    &  \includegraphics[width=4.3 cm]{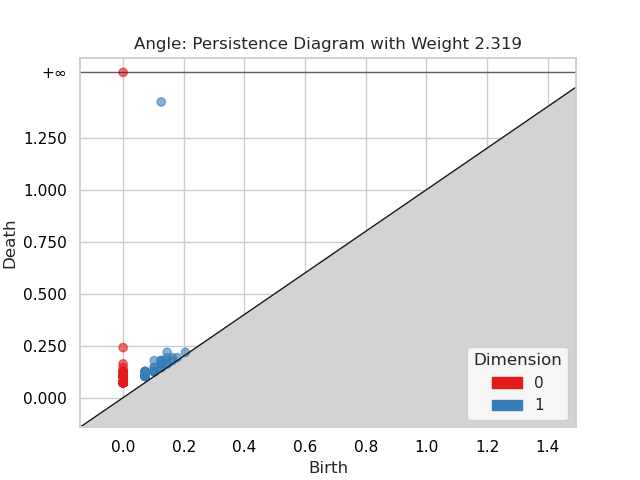} 
    & \includegraphics[width=4.3 cm]{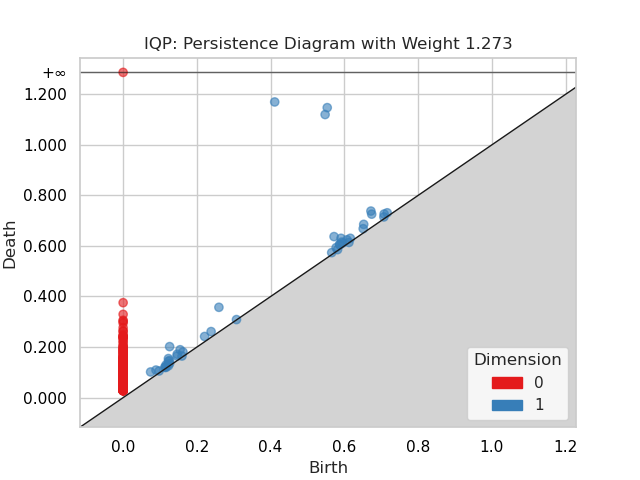} 
    \\ 
    \includegraphics[width=4.3 cm]{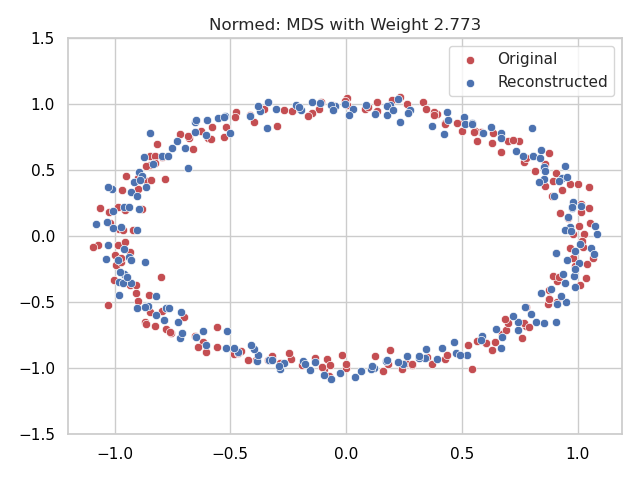}
    & \includegraphics[width=4.3 cm]{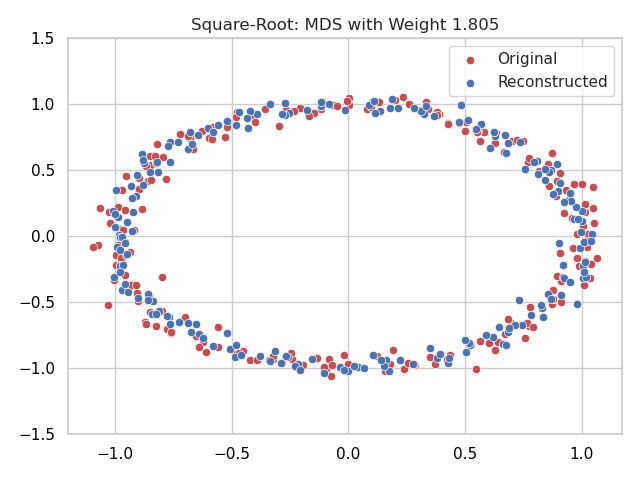}
    & \includegraphics[width=4.3 cm]{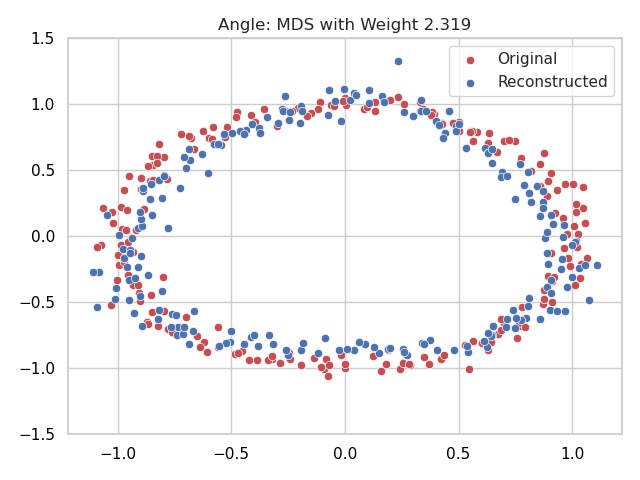}
    & \includegraphics[width=4.3 cm]{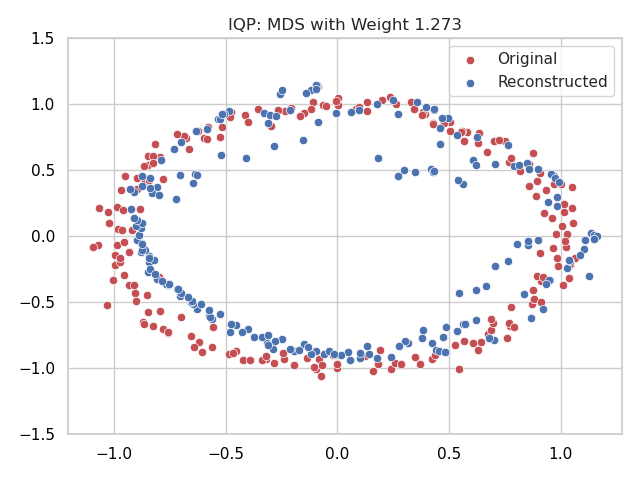}
    \end{tabular}
    \caption{Barcodes, persistence diagrams, and MDS reconstructions for the four maps applied to the circle data: uniform transformation (normed data), square-root encoding, angle encoding, and IQP encoding, respectively. In the latter three cases, distances for the encoded data points were calculated with the Fubini--Study metric, where the fidelity was obtained by the swap test. The barcodes, persistence diagrams, and MDS reconstructions were obtained directly from the distances between the encoded data points.}  
    \label{fig:barcodes}
\end{figure*}

The resulting distance matrices associated with the quantum encoded data are then used in Figure~\ref{fig:barcodes} to obtain the barcodes and persistence diagrams using the Gudhi package~\cite{gudhirips}, closely following what was done in Ref.~\cite{vlasic2023qtda}. The distance matrices are also used to obtain Euclidean MDS reconstructions in an attempt to visually compare the new distance matrix with the original dataset in $\R^{2}$. We first note that MDS here is obtained by minimizing (the square of) \emph{stress} as defined in~\eqref{eqn:stressf}, as opposed to \emph{strain} as defined in~\eqref{eqn:strainf}.  Secondly, the specific distance matrix used for computing both the barcodes and MDS was rescaled by an additional factor in order to properly compare the initial dataset with the quantum encoded dataset. Namely, to appropriately compare the distances between the encoded states with the distances between the original data, because the original data were rescaled in the process of the encoding, and because the encoding itself might change the diameter of the dataset, we additionally rescale the distance matrix obtained from the quantum encoded data. More precisely, for each quantum encoded dataset $Y$, let $D_{Y}$ denote the associated distance matrix from the Fubini--Study distance from~\eqref{eqn:FubiniStudy}, except that the $\frac{2}{\pi}$ factor has been removed, i.e., $D_{Y}\big(\ket{x_j},\ket{x_k}\big)=\arccos\big(|\langle x_{j}|x_{k}\rangle\big)$. Multiply $D_{Y}$ by a parameter $w$, called the \define{weight}. Let $D_{X}$ denote the distance matrix of the data set $X$. To find the optimal weight, we would like to minimize the distortion $\lVert D_{X}-wD_{Y}\rVert_{\infty}$ from Definition~\ref{defn:distortionf}, which we view as the error associated with the encoding (we could also minimize the stress $\lVert D_{X}-wD_{Y}\rVert_{2}$ or another suitable cost function). Note that the optimization problem
\begin{equation}
\argmin_{w\ge0}\left(\sup_{j,k}\Big|D_{X}(x_j,x_k)-w\arccos\big(\langle x_j|x_k\rangle\big)\Big|\right).
\end{equation}
is convex. To find the minimum, we utilized the minimize method optimization package within SciPy~\cite{2020SciPy-NMeth} and the Nelder--Mead algorithm~\cite{NelderMead1965}, which is rigorously known to converge in 1d~\cite{Lagarias1998NelderMead}.

For example, although the distortion is large between the original dataset and the uniformly-transformed dataset, multiplying the latter distance matrix by $w=r(X)\sqrt{m(m+1)}\approx2.7726$ yields a zero distortion. In other words, the \emph{scale-free distortion} of the uniformly transformed dataset as compared to the original is zero, as expected from its construction in Section~\ref{sec:MAE}. The same value of this weight was also obtained using the Nelder--Mead algorithm proposed in the previous paragraph, and therefore in some sense provides a control for this numerical experiment. Meanwhile, square-root encoding had a weight of $1.805$, angle encoding had a weight of $2.319$, and IQP encoding had a weight of $1.273$ (cf. Figure~\ref{fig:barcodes}).

For IQP encoding, two other preprocessing techniques were tested, though their results are not shown here. In one case, the original data vectors were multiplied by $\pi$ and values beyond $(-\pi,\pi)$ were clipped. In another case, the normed data vectors were multiplied by $2\pi$ so that their values lied between $0$ and $2 \pi$. Afterwards, the IQP encoding map was applied to the resulting preprocessed data. 
However, both approaches resulted in a higher distortion than the one used in Figure~\ref{fig:barcodes}. The amount of distortion can be visualized by comparing the original dataset and MDS applied to the IQP encoded data, which is also shown in Figure~\ref{fig:barcodes}. The reason we visualize things by applying MDS back after the encoding is because MDS attempts to preserve distance geometry (and hence some topology as described by persistent homology), so the two sets should appear close.

Table~\ref{tab:metric-distances} provides some numerical differences between the metric spaces of the rescaled data versus what happens after the data are quantum encoded by calculating the distances between the corresponding metric spaces. 
The distances between metric spaces that we calculated were the Gromov--Hausdorff and bottleneck distances. The Gudhi package~\cite{gudhirips} was used to compute the bottleneck distances. 
The Gromov--Hausdorff distance~\eqref{eq:gromov-hausdorff-final}, on the other hand, is not simple to compute and is computationally intractable~\cite{oles2023computing}, in the sense that it is NP-hard to compute or approximate~\cite{Adams2025Hausdorff,Schmiedl2017GHdistance,Agarwal2018ComputingGHdistance}. 
This motivated Ref.~\cite{oles2023computing} to derive a computable upper bound for the Gromov--Hausdorff distance along with an implementation in Python (some lower bounds were obtained in Refs~\cite{memoli2012some,Adams2025Hausdorff}). 
As such, the subsequent calculations involving the Gromov--Hausdorff distance between metric spaces will actually use the upper bound obtained in Ref.~\cite{oles2023computing}. In particular, when the Gromov--Hausdorff distance is referenced in Table \ref{tab:metric-distances}, what is actually being calculated is the upper bound for it derived in Ref.~\cite{oles2023computing}.

\begin{table} %[ht]
\centering
\resizebox{\columnwidth}{!}{
\begin{tabular}{|c|c|c|c|c|c|c|}
\hline
\multicolumn{7}{|c|}{Distances Between Spaces} \\ \hlineB{2.7}
\multicolumn{2}{|c|}{Basis}  & \multicolumn{5}{|c|}{Feature Maps}  \\ \hlineB{2.7}
Metric & Feature Maps        & Original & Normed   & Square-Root  & Angle  & IQP  \\ \hlineB{2.7}
\multirow{ 5 }{*}{ \stackanchor{Gromov--}{Hausdorff} }
& Original & 0.0 & 0.1201 & 0.116 & 0.2063 & 0.6549 \\ \cline{2-7}
& Normed & 0.1201 & 0.0 & 0.116 & 0.2063 & 0.5093 \\ \cline{2-7}
& Square-Root & 0.116 & 0.116 & 0.0 & 0.2129 & 0.6861 \\ \cline{2-7}
& Angle & 0.2063 & 0.2063 & 0.2129 & 0.0 & 0.5406 \\ \cline{2-7}
& IQP & 0.6549 & 0.5093 & 0.6861 & 0.5406 & 0.0 \\ \hlineB{2.7}
\multirow{ 5 }{*}{Bottleneck} 
& Original & 0.0 & 0.0 & 0.0387 & 0.1468 & 0.3994 \\ \cline{2-7} 
& Normed & 0.0 & 0.0 & 0.0387 & 0.1468 & 0.3994 \\ \cline{2-7} 
& Square-Root & 0.0387 & 0.0387 & 0.0 & 0.1855 & 0.438 \\ \cline{2-7} 
& Angle & 0.1468 & 0.1468 & 0.1855 & 0.0 & 0.2968 \\ \cline{2-7} 
& IQP & 0.3994 & 0.3994 & 0.438 & 0.2968 & 0.0 \\ \hlineB{2.7} 
\end{tabular}
}
\caption{This tables gives upper bounds for the Gromov--Hausdorff distance and the exact bottleneck distances of the persistence diagrams when comparing the datasets using the respective metrics shown. The distance matrices have been multiplied by their respective weights based on the discussion associated with Figure~\ref{fig:barcodes}. Note that the Gromov--Hausdorff distances between any of the encodings with the original and any encodings with the normed cells should be exactly the same. However, they have different numerical values because only an upper bound for the Gromov--Hausdorff distance was obtained using Ref.~\cite{oles2023computing}. We also note that these are the absolute distances, and they are not relativized. Analogous to the relative error in numerical methods~\cite{TrefethenBau2022}, the Gromov--Hausdorff distance is considered to be small when it is small compared to the diameters of the two metric spaces.}
\label{tab:metric-distances}
\end{table}

All the quantum encodings in Figure~\ref{fig:barcodes} were based on standard quantum encodings used in the literature with some slight modifications. Alternatively, we may attempt to find an MDS encoding based on minimizing the stress between the original point cloud and its image in the space of pure quantum states, which is a complex projective space that can be equipped with the Fubini--Study metric. 
Even for our small dataset of 200 data points, with a desired encoding into the standard qubit Bloch sphere, where each data point is specified by two real parameters as in~\eqref{eqn:genericBlochstate}, this minimization procedure involves minimizing a function consisting of 400 variables. 

To simplify the problem, we will change two things in the procedure to find an MDS encoding that minimizes the stress. First, rather than using the angular coordinates of~\eqref{eqn:genericBlochstate}, we will instead write a qubit state as an arbitrary linear combination $\alpha\ket{0}+\beta\ket{1}$, with $\alpha,\beta\in\C$, which therefore amounts to four real, but arbitrary, parameters, and then we will normalize to guarantee that the vector describes a quantum state. Although this amounts to minimizing a function consisting of 800 variables, the benefit is that this allows us to use ordinary gradient descent in Euclidean space, where the coefficients appear linearly in the description of the vectors. Second, motivated by the projective space extension of principal component analysis (PCA) of Ref.~\cite{Perea2018Multiscale}, we approximate the Fubini--Study metric~\eqref{eqn:FubiniStudy} 
by 
\begin{equation}
\label{eqn:FSquadapprox}
d_{\mathrm{FS}}(\ket{\psi},\ket{\phi}) = \frac{2}{\pi}\mathrm{arccos}\Big( \big| \langle \psi| \phi \rangle \big|\Big) \approx 1 - \big| \langle \psi| \phi \rangle \big|^2. 
\end{equation} 
This is an approximation in the sense that the unique quadratic polynomial $f(x)=a_0+a_1x+a_2x^2$ on the unit interval $[0,1]$ that satisfies the boundary conditions $f(0)=1$ and $f(1)=0$ and which also minimizes the $L_{\infty}$-distance $\lVert \frac{2}{\pi}\mathrm{arccos}-f\rVert_{\infty}$ for functions on the interval $[0,1]$ is precisely $f(x)=1-x^2$. These two functions are shown in Figure~\ref{fig:approxarccos} for comparison. 
\begin{figure}
\includegraphics[width=6cm]{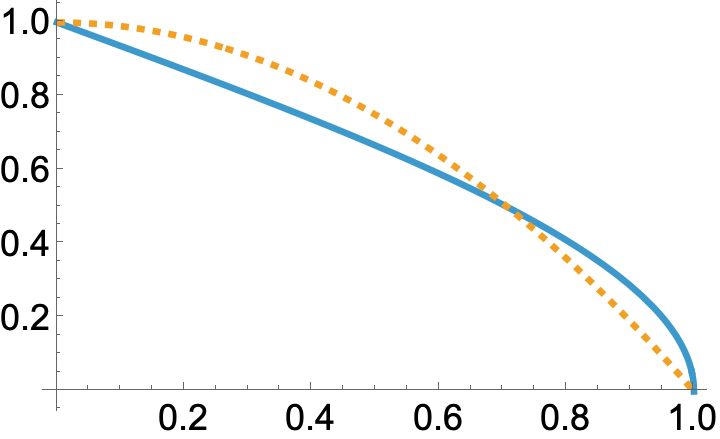}
\caption{The function $\frac{2}{\pi}\mathrm{arccos}(x)$ (solid blue) is plotted together with the function $1-x^2$ (dashed orange) on the interval $[0,1]$. Among all quadratic polynomials, $1-x^2$ is the unique one that has the same boundary values and also minimizes the $L_{\infty}$-distance to $\frac{2}{\pi}\mathrm{arccos}(x)$. It is therefore used as an approximation to illustrate MDS encoding, which is discussed more in the text.}
\label{fig:approxarccos}
\end{figure}
Note that this approximation allows for the easier computation of the MDS encoding, as it avoids the divergence of the derivative of $\frac{2}{\pi}\mathrm{arccos}(x)$ as $x$ approaches $1$ (which happens when the quantum states are close to each other) and because the derivatives are simple.

As such, we instead minimize the approximated stress function 
\begin{equation}
\label{eqn:quadapproxFS}
\sum_{j,k}\Big(D_{X}(x_j,x_k)-w \big(1-|\langle x_j|x_k\rangle|^2\big)\Big)^2
\end{equation}
using the approximation~\eqref{eqn:FSquadapprox} of the Fubini--Study distance function in order to obtain an approximate quantum MDS encoding. For our example, the weight $w$ is taken to be $w=\mathrm{diam}(X)$, i.e., the largest entry of the distance matrix $D_{X}$; equivalently,  the distance matrix $D_{X}$ is divided by the largest entry $\mathrm{diam}(X)$. This guarantees that the two distance functions have the same maximum value and so that there is a fair comparison between the two metric spaces. 
As for the algorithm used for minimizing stress, we use the SymPy package \cite{meurer2017sympy} combined with ordinary gradient descent to find a minimum~\cite{Cauchy1847GD,Strang2022} (it would be interesting to see how this changes if one utilizes natural/geometric gradient descent as in Refs.~\cite{stokes2020quantum,Amari1998NaturalGradient}). More specifically, we start with Haar random initial unit states, and we use learning rates of $.001$ and $.0001$, where each learning rate is applied after the stress function has stopped decreasing. After the parameters have been updated, they are normalized. The result is displayed in Figure~\ref{fig:quanutm_mds}. We note one caveat, which is that convergence of the algorithm is not guaranteed, especially when there are many entries with small fidelity (approximately orthogonal states), i.e., large distance. Consequently, it is important to identify a transformation acting on the distance matrices that would allow for a cleaner separation of the values enough for the algorithm to converge while approximately preserving the associated topologies.

With the specific noisy circle dataset of Figure~\ref{fig:circle}, this MDS procedure of mapping the data into the Bloch sphere yielded an average error of $.0049$ for each summand in the stress function and it resulted in a distortion value of $.2486$.
The resulting quantum states are depicted on the Bloch sphere in Figure \ref{fig:quanutm_mds}. The quantum states almost lie on a great circle. 
As with ordinary/classical MDS, encoding the data into a projective space of larger dimension might provide an encoding with smaller distortion and stress. Also, as in the case of cMDS, the initialization may affect the evolution and final result of the algorithm.

\begin{figure}%[ht!]
     \centering
    \includegraphics[width=.45\textwidth]{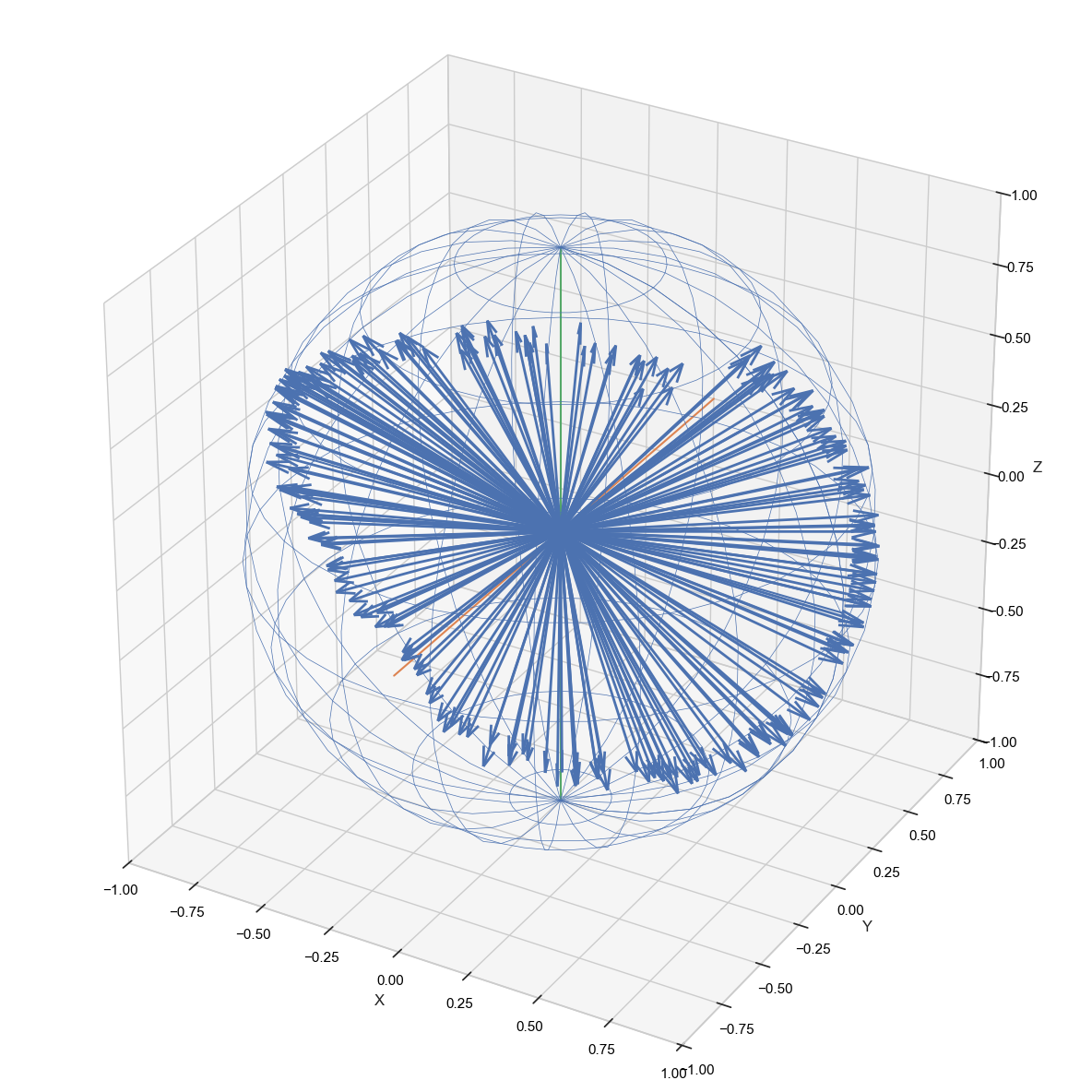}
    \caption{The Bloch sphere together with an MDS encoding of the noisy circle dataset from Figure~\ref{fig:circle}. The stress~\eqref{eqn:quadapproxFS} was minimized between the original Euclidean distance and a rescaled quadratic approximation to the Fubini--Study distance. Since the Fubini--Study metric is invariant under unitary transformations, the MDS encoding is not in general unique. Moreover, what is shown in this figure is only one realization for a particular random initialization.}
\label{fig:quanutm_mds}
\end{figure}

One major open question regarding quantum MDS that is important to address is to find an appropriate analogue of strain that could be used to obtain an analytic solution to strain-minimized MDS for complex projective space equipped with the Fubini--Study metric. This would extend the results achieved for real Euclidean space~\cite{Schoenberg1935,Torgerson1952MDS}, spheres~\cite{Schoenberg1935}, and hyperbolic space~\cite{KeNa2020hydra}.

%%%%%%%%%%%%%%%%%%%%%%%%%%%%%%%%%%%%%%%%%%%%%%%
\section{Discussion}
\label{sec:discussion}
%%%%%%%%%%%%%%%%%%%%%%%%%%%%%%%%%%%%%%%%%%%%%%%

In this paper, our main achievements can be summarized as follows. 
\begin{enumerate}
\item Using ideas from Ref.~\cite{PBVP24}, we isolated the relevant structure that a quantum encoding should preserve in order for the overall topology of a point cloud to be unchanged. These are isometries, or more generally homotheties (isometries up to an overall scalar), of metric spaces. Since not all commonly used quantum encodings preserve distances exactly or up to a scalar, in order to formulate a way in which distances are minimally altered, we reviewed the notions of distortion, codistorion, and Gromov--Hausdorff distances. 

\item We reviewed Topological Data Analysis (TDA), focusing on persistent homology and the stability theorem, the latter of which provides bounds relating the bottleneck distance to the Gromov--Hausdorff distance. We included several simple examples to illustrate the main ideas as well as how the associated quantities are computed, such as birth-death diagrams and bottleneck distances. 

\item Our first main result is Corollary~\ref{cor:UTDpreservespersistentH} (a consequence of Theorem~\ref{thm:UTDpreservesmetric}), which provides a quantum encoding that \emph{exactly} preserves (up to an overall constant) all of the distances from a point cloud in Euclidean space, and therefore preserves all of the topological invariants that could be computed from TDA including persistent Betti numbers. This explicitly shows that there is indeed a quantum encoding that can (in principle) perfectly preserve the topology of a dataset. 

\item
Since the quantum encoding that perfectly preserves the topology involves the preparation of mixed quantum states, we supplied an alternative quantum encoding that only requires the preparation of pure quantum states. We then used metric Multidimensional Scaling to formalize the idea that \emph{every} quantum encoding \emph{must} distort the classical data by at least a certain amount. In other words, we defined the theoretically optimal way to encode classical data preserving distances into quantum state space. One can then use the associated lowest distortion bounds to gauge how much of the topological structure of a dataset might be preserved by any quantum encoding. 

\item
We posed the next major problem for this line of research, which is to find an \emph{efficient} quantum encoding that minimizes distortion/Gromov--Hausdorff distance. The theoretically optimal solution is not necessarily an implementable and efficient quantum encoding. Ideally, one would hope that there exists an efficient implementation that is at least close to this optimal encoding. We sketched out a clear formulation of this problem in the case that the class of quantum encodings is given by Hamiltonian evolution and one-parameter subgroups. 
\end{enumerate}

As we mentioned in the introduction, one could ask why we would want to encode classical data directly onto the quantum computer if the originally proposed quantum TDA algorithm in Ref.~\cite{LloydGarneroneZanardi16} instead first constructed the associated Vietoris--Rips complex and encoded the latter onto a quantum computer via bit encoding in order to achieve a quantum advantage over classical algorithms. We have several reasons for this, which lead to some interesting open problems. Our main purpose in doing so was to analyze how one should encode classical data if one were to work directly with the data on a quantum computer. 

The important message of our manuscript is that the distances between the data points must be approximately preserved (at least up to a scalar) in order for the topological inference to be similar to what one would get from the classical data. Furthermore, the simplicial complex associated with the classical data grows exponentially with the number of data points (in the worst case), which makes it somewhat difficult to encode onto a quantum computer. Moreover, it has recently been shown in Refs.~\cite{KingKohler2024,CrichignoKohler2024,SchmidhuberLloyd2023} that determining homology from the graph (the $0$- and $1$-simplices from the Vietoris--Rips complex) is hard, even for a quantum computer (Refs.~\cite{AdamaszekStacho2016,SchmidhuberLloyd2023} show that the homology problem is NP-hard, while Refs.~\cite{KingKohler2024,CrichignoKohler2024} show that the clique homology problem is $\mathrm{QMA}_{1}$-hard or $\mathrm{BQP}$-hard depending on the formulation of the problem). This suggests that if one were to utilize the power of quantum computers for analyzing the topological features of data and performing topological inference, perhaps there might be a better and more efficient route at doing so by directly encoding the data onto a quantum computer rather than first constructing the associated simplicial complex. Although this seems to require Quantum Random Access Memory (QRAM) to load all the data onto a quantum computer~\cite{GLM2008QRAM}, exploring whether this is indeed possible is an important question to address, as it might provide an alternative route towards quantum advantage in quantum topological data analysis that bypasses some of its current bottlenecks.

\bigskip
\noindent
\textbf{Acknowledgements}

\noindent
The authors thank Jonathan Bloom, Chian Yeong Chuah, Cheyne Glass, Jean-Claude Hausmann, Seth Lloyd, Thomas Needham, Anh Pham, Alexander Schmidhuber, and Rapha{\"e}l Tinarrage for discussions. AJP thanks Kosha Upadhyay for discussions on multidimensional scaling. Some parts of Figures~\ref{fig:ideal-circle-sqrt-heatmapsandCMDS}, \ref{fig:ideal-circle-smaller-sqrt-heatmapsandCMDS}, \ref{fig:ideal-circle-angle-heatmapsandCMDS}, \ref{fig:heatmapIQPencoding}, and~\ref{fig:approxarccos} were created using Mathemetica~\cite{Mathematica2024}. 

\bigskip
\noindent
\textbf{Conflict of Interest Statement}

\noindent
AJP has received financial support from Deloitte in his involvement with this project.
He carried out this project as a consultant to Deloitte and not as part of his MIT responsibilities.

\bibliography{refs}
\end{document}